\documentclass[a4paper,11pt]{article}

\usepackage{jheppub} 

\usepackage{graphbox}

\usepackage[numbers,sort&compress]{natbib} 
\usepackage{url}
\usepackage{amsmath}
\usepackage{rotating}

\usepackage{multirow}

\usepackage[T1]{fontenc} 

\usepackage{float, array, xspace, amscd, amsthm}
\usepackage{fancyhdr}
\usepackage{longtable}

\newtheorem{theorem}{Theorem}

\newtheorem{proposition}{Proposition}
\newtheorem{definition}{Definition}

\allowdisplaybreaks

\usepackage[usenames,dvipsnames,svgnames,table,x11names]{xcolor}
\definecolor{MagentaXD}{RGB}{204, 48, 152}
\definecolor{MagentaXDdetail}{RGB}{150, 79, 126}
\definecolor{GreenMAF}{RGB}{28, 112, 46}
\definecolor{GreenMAFdetail}{RGB}{80, 117, 88}
\definecolor{detail}{RGB}{110,110,110}
\definecolor{quantumviolet}{HTML}{53257F} 
\definecolor{quantumgray}{HTML}{555555} 
\definecolor{quantumgreen}{HTML}{007474} 
\definecolor{quantumblue}{HTML}{002366} 
\definecolor{quantumpurple}{HTML}{66023C} 
\definecolor{quantumdarkviolet}{HTML}{5D3954} 

\newcommand{\FA}{\mathfrak{A}}
\newcommand{\FB}{\mathfrak{B}}
\newcommand{\FF}{\mathfrak{F}}
\newcommand{\FM}{\mathfrak{M}}
\newcommand{\FN}{\mathfrak{N}}

\newcommand\vPsi {|\Psi\rangle}

\newcommand{\be}{\begin{equation}}
	\newcommand{\ee}{\end{equation}}
\def\bea#1\eea{\begin{align}#1\end{align}}

\newcommand{\sslash}{\mathrm{/\mkern-6mu/}}    

\usepackage{tikz}
\usetikzlibrary{positioning,intersections}
\usetikzlibrary{calc}
\usetikzlibrary{shapes.geometric}
\usepackage{braids}
\usepackage{tikz-cd}

\newif\ifcomments
\commentstrue

\newif\ifdetails
\detailstrue

\usepackage{upgreek}
\usepackage[normalem]{ulem}
\usepackage{mathtools}
\usepackage{datetime}
\usepackage{cancel}

\usepackage[all,cmtip,knot]{xy}
\usepackage{xypic}

\newcommand{\orcid}[1]{\href{https://orcid.org/#1}{\includegraphics[width=8pt]{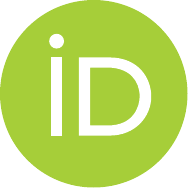}}}

\usepackage{amsmath,amsfonts,amsthm}

\theoremstyle{definition}
\newtheorem{example}[definition]{Example}

\theoremstyle{remark}
\newtheorem{remark}{Remark}[section]

\usepackage{bm,latexsym,galois,euscript,dsfont}

\newcommand\EB{\EuScript{B}}
\newcommand\EC{\EuScript{C}}
\newcommand\ED{\EuScript{D}}

\newcommand\EM{\EuScript{M}}

\newcommand\EW{\EuScript{W}}

\newcommand\Fun{\textsf{Fun}}
\newcommand\Vect{\textsf{Vect}}

\newcommand\Obj {\mathrm{Obj}}
\newcommand\id {\mathrm{id}}
\newcommand\Hom {\mathrm{Hom}}
\newcommand{\one}{\mathds{1}}

\newcommand\Mod{\textsf{Mod}}

\newcommand{\FPdim}{\operatorname{FPdim}}

\newcommand\Hhat {\hat{H}}

\usepackage{mathtools,amssymb,scalerel}

\newcommand\biopencrossl{%
	\mathrel{\scalerel*{>\kern-.4\LMpt\joinrel\blacktriangleleft}{x}}}
\newcommand\biopencrossr{%
	\mathrel{\scalerel*{\blacktriangleright\joinrel\kern-.4\LMpt<}{x}}}

\settimeformat{ampmtime}

\title{\color{black} Boundary and domain wall theories of 2d generalized quantum double model
}

\author[a,b]{Zhian Jia\orcid{0000-0001-8588-173X},}
\author[a,b]{Dagomir Kaszlikowski,}
\author[c]{and Sheng Tan}
\affiliation[a]{Centre for Quantum Technologies, National University of Singapore, Singapore 117543, Singapore}
\affiliation[b]{Department of Physics, National University of Singapore, Singapore 117543, Singapore}
\affiliation[c]{Department of Mathematics, Purdue University, West Lafayette, IN 47907, USA}

\emailAdd{giannjia@foxmail.com}
\emailAdd{phykd@nus.edu.sg}
\emailAdd{tan296@purdue.edu}

\abstract{
	
	The generalized quantum double lattice realization of $2d$ topological orders based on Hopf algebras is discussed in this work. Both left-module and right-module constructions are investigated. The ribbon operators and the classification of topological excitations based on the representations of the quantum double of Hopf algebras are discussed.
	To generalize the model to a $2d$ surface with boundaries and surface defects, we present a systematic
	construction of the boundary Hamiltonian and domain wall Hamiltonian.
	The algebraic data behind the gapped boundary and domain wall are comodule algebras and bicomodule algebras. The topological excitations in the boundary and domain wall are classified by bimodules over these algebras.
	The ribbon operator realization of boundary-bulk duality is also discussed.
	Finally, via the Hopf tensor network representation of the quantum many-body states, we solve the ground state of the model in the presence of the boundary and domain wall.
	
}

\begin{document}
	
	\maketitle
	\flushbottom

	\section{Introduction}
	\label{sec:intro}

Topologically ordered phases extended our understanding of the notion of phases of matter beyond the Landau-Ginzburg symmetry-breaking picture \cite{Wen1990,Wen2004}. Besides their foundational importance, these exotic quantum phases of matter also found their applications in quantum information processing, such as robust topological quantum error-correction code (QECC) \cite{Dennis2002topological,Terhal2015quantum} and topological quantum computation (TQC) \cite{Kitaev2003,freedman2002modular,Nayak2008}.
	The mathematical structures behind the usual symmetry-breaking phases are symmetry groups $(G_{H}, G_{\Psi})$ with $G_{H}$ the symmetry group of Hamiltonian and $G_{\Psi}$ the symmetry group of the wavefunction, while the mathematical frameworks behind topological orders are tensor categories. 
	More precisely, a gapped $2d$ topological order is characterized by a unitary modular tensor category (UMTC) $\ED$; for the gapless edges, another data, central charge $c$, is needed, namely, $(\ED,c)$ fully characterizes the topologically ordered phase.
	
	A gapped topological phase is an equivalence class of gapped Hamiltonians together with their corresponding ground state spaces $\{(H,\mathcal{H})\}$, which realizes some topological quantum field theory (TQFT) at low energy. 
	The excitations are characterized by a quantum group constructed from the gauge group of the theory, and the quasi-particle types are given by irreducible representations of the quantum group.
	A crucial class of such kind of  $(2+1)D$ model is the twisted quantum double model based on a finite group algebra $\mathbb{C}[G]$ \cite{Mesaros2013classification,Bullivant2017twisted, Hu2012twisted}, which is a Hamiltonian realization of  the Dijkgraaf-Witten TQFT  \cite{dijkgraaf1990topological}. Their topological excitations are characterized by the representation category of the twisted quantum double $D^{\alpha}(G)$ with $\alpha$ a 3-cocycle over $G$.
	When the 3-cocycle is trivial, the model becomes Kitaev's quantum double model \cite{Kitaev2003}, which corresponds to BF theory and to a special case of the Kuperberg invariant \cite{kup1996noninvolutary} as well.
	The Levin-Wen's string-net model \cite{Levin2005} has a more general setting, and realizes the Turaev-Viro-Barrett-Westbury TQFT \cite{turaev1992state,barrett1996invariants}. For arbitrary unitary fusion category (UFC) $\EC$, the topological excitations of the Levin-Wen model based on $\EC$ are given by the Drinfeld center $\mathcal{Z} (\EC)$.

	Various generalizations of the quantum double model have been studied (e.g., \cite{Hu2012twisted,Buerschaper2013a,chang2014kitaev}), among which, Hopf algebraic quantum double model turns out to be a crucial class and attracts extensive studies from both lattice model perspective \cite{Buerschaper2013a, buerschaper2013electric,chen2021ribbon} and Hopf algebraic gauge theory perspective \cite{bais2003hopf,meusburger2017kitaev,meusburger2021hopf}.
	A quantum double model can be mapped into a string-net model \cite{Buerschaper2009mapping, Hu2018full,wang2020electric}; the reverse direction has also been studied \cite{Kitaev2012a,chang2014kitaev}.
	The physical properties (including topological excitation, ribbon operator, electric-magnetic duality, etc.) of the quantum double model for the finite group algebra case has been extensively studied 
	\cite{Kitaev2003,bravyi1998quantum,Bombin2008family,freedman2001projective,Beigi2011the, Kitaev2012a,Levin2013protected,Cong2017,wang2020electric,etingof2003finite,ostrik2003module,andruskiewitsch2007module,natale2017equivalence}.
	However, for the general Hopf algebra case, only some special examples are discussed \cite{Buerschaper2013a}. In this work, we will present a systematic investigation of this generalized quantum double model based on $C^*$ Hopf algebras.
	
    On the other hand, although the topological phases on closed surfaces seem to be natural from a mathematical perspective, real samples of topologically ordered material usually have boundaries and the boundary modes are easier to measure experimentally. Thus the topological phases on surfaces with boundaries are of more practical and theoretical importance.
    Another crucial reason to investigate the boundary theory and domain wall theory is that there is a boundary-bulk duality, with which the boundary phase is obtained from the bulk phase using anyon condensation, and the bulk phase is recovered from the boundary phase by taking the Drinfeld center \cite{kong2014braided,kong2017boundary}.
    The study of the gapped boundary theory of topological phases has garnered significant attention over the past few decades \cite{bravyi1998quantum,Bombin2008family,freedman2001projective,Beigi2011the,Kitaev2012a,Levin2013protected,Cong2017,wang2020electric,Haldane1995stability,Kane1997quantized,kapustin2011topological,Barkeshli2013theory,kong2014braided,Kong2014,Wang2015boundary,Lan2015gapped,seiberg2016gapped,kong2017boundary,hu2018boundary,Lan2020Gapped,freed2021gapped}.
	
	Not all topological phases allow the gapped boundaries (with a lattice realization). One of the crucial observations of the existence of a gapped boundary is that the chiral central charge $c_-=c_L-c_R$ must vanish \cite{Haldane1995stability,Kane1997quantized, Levin2013protected,Wang2015boundary,ganeshan2021ungappable}.  Even for the $c_-=0$ case, there exist some ungappable boundaries \cite{Levin2013protected,Wang2015boundary,ganeshan2021ungappable}. Therefore, a deep and comprehensive understanding of the boundary theory is of great importance.
	For the quantum double phase, the boundary is gappable, and the gapped boundary theories for the finite group algebra case have been extensively explored \cite{bravyi1998quantum,Bombin2008family,freedman2001projective,Beigi2011the,Kitaev2012a,Levin2013protected,Cong2017,wang2020electric,etingof2003finite,ostrik2003module,andruskiewitsch2007module,natale2017equivalence}.
	However, for the general Hopf algebra case, the gapped boundary theory has not been systematically investigated yet.
	This is due, to some extent, to the mathematical difficulties when dealing with general Hopf algebras.
	
    In this work, we will investigate the generalized quantum double model in detail and present the Hamiltonian construction and algebraic theory of gapped boundaries and domain walls.

    In Sec.~\ref{sec:Kitaev}, we first review the generalized quantum double model on a closed surface and stress the problem of constructing the ribbon operators for this model and classifying the topological excitations. For the generalized quantum double model, there also exist electric charges, magnetic charges, and dyons, with all these charges can be created with proper ribbon operators.

    Sec.~\ref{sec:bd} establishes the boundary theory of the Hopf quantum double model.
    We show that gapped boundaries of the Hopf quantum double phase can be equivalently classified by an $H$-comodule algebra $\FA$ or an $H$-module algebra $\FM$.
    The boundary excitations are classified by bimodules over these algebras.
    The lattice model for the gapped boundary is constructed using the symmetric separability idempotent of the $H$-comodule algebra. We show that the boundary stabilizers generate a boundary local algebra whose  representation category is equivalent to the category of boundary excitations.
    The connection between the quantum double boundary theory and Levin-Wen string-net boundary theory is also discussed.
    After these preparations, we present the construction of boundary ribbon operators and discuss how to realize the anyon condensation via these ribbon operators.
    In Appendix~\ref{sec:bdII}, we also give another construction of a gapped boundary lattice model which is parameterized by a triple of Hopf algebras $(K,J,W)$ with some pairings among them. Each boundary site supports a representation of a generalized quantum double induced by the pairing.
    This indirect lattice construction is also of its own interest.

    Sec.~\ref{sec:domainwall} establishes the domain wall theory of the generalized quantum double model.
    Using the folding trick, a domain wall can be transformed into a gapped boundary; and a gapped boundary can be regarded as a special case of the domain wall which separates the quantum double phase with vacuum.
    We build the theory of the gapped domain wall based on the $H_1|H_2$-bicomodule algebra, both the algebraic theory and Hamiltonian theory are discussed.
    The domain wall is characterized by an $H_1|H_2$-bicomodule algebra $\FB$ and the wall excitations are classified by $\FB|\FB$-bimodules.
    In Appendix~\ref{sec:app-domain-wall}, an indirect construction based on the generalized quantum double is also discussed.    The lattice model of the domain wall between two quantum double phases is given by a quadruple of Hopf algebras $(K, J_1,J_2, W)$ with some pairings among them.
    The left and right boundary sites support different representations of quantum doubles induced by different pairings.

    In Sec.~\ref{sec:bdgs}, by utilizing the Hopf tensor network representation of quantum many-body states, we solve the ground state of the model with boundaries and domain walls and obtain the explicit ground states.
    Using this explicit exact ground state of the model, we can investigate various properties of the phase in the presence of boundaries and domain walls. Especially, the entanglement entropy can be calculated directly. This also paves the way for applications of generalized quantum double phase in QECC and TQC.

    The appendices collect some detailed discussions and calculations.

	\section{Generalized quantum double model}\label{sec:Kitaev}
	
	Let us start with a brief review of the Hopf algebraic quantum double model \cite{Buerschaper2013a,cowtan2021quantum,chen2021ribbon} on a closed $2d$ surface.
	Kitaev's original construction of the quantum double model is based on finite group algebra $\mathbb{C}[G]$, and he also pointed out that the model can be generalized to the finite-dimensional Hopf algebra equipped with a Hermitian inner product with certain properties \cite{Kitaev2003}. 
	The first explicit construction is given in \cite{Buerschaper2013a}, and the corresponding ribbon operators are discussed in detail in \cite{cowtan2021quantum,chen2021ribbon}. 
	In this section, we will discuss the construction and stress the chirality of the construction. From the model, we can obtain Turaev-Viro type topological invariant \cite{balsam2012kitaevs,meusburger2017kitaev}.

	\subsection{Generalized quantum double model}
	\label{sec:bulkQD}

	For a given 2$d$ closed manifold $\Sigma$, consider a lattice\,\footnote{Also known as a cellulation of $\Sigma$ or a ribbon graph on $\Sigma$. We assume that the graph corresponding to the lattice is a simple graph, namely, a graph that no edge starts and ends at the same vertex.} $C(\Sigma)=V(\Sigma)\cup E(\Sigma)\cup F(\Sigma)$ on $\Sigma$, where $V(\Sigma)$, $ E(\Sigma)$ and $F(\Sigma)$ are sets of vertices, edges and faces respectively.
	The dual lattice of $C(\Sigma)$ is a lattice $\tilde{C}(\Sigma)$ for which
	the vertices and faces of the original lattice are switched while the edge set remains unchanged.
	We assign a direction for each edge $e \in E(\Sigma)$, and the direction of the corresponding dual edge $\tilde{e}\in \tilde{E}(\Sigma)$ is obtained by rotating the direction of $e$ counterclockwise by $\pi/2$. As shown in Fig.~\ref{fig:sphere}, the original lattice is drawn with black solid lines, and the dual lattice is drawn with red dashed lines.
	A site $s=(v,f)$ is a pair of a vertex $v$ and an adjacent face $f$ (which is a dual vertex of the dual lattice). The site is drawn as a solid line connecting $v$ and $f$.
	Two sites $s,s'$ are called adjacent if they share a common vertex or face.

	\begin{figure}[t]
		\centering
		\includegraphics[width=12cm]{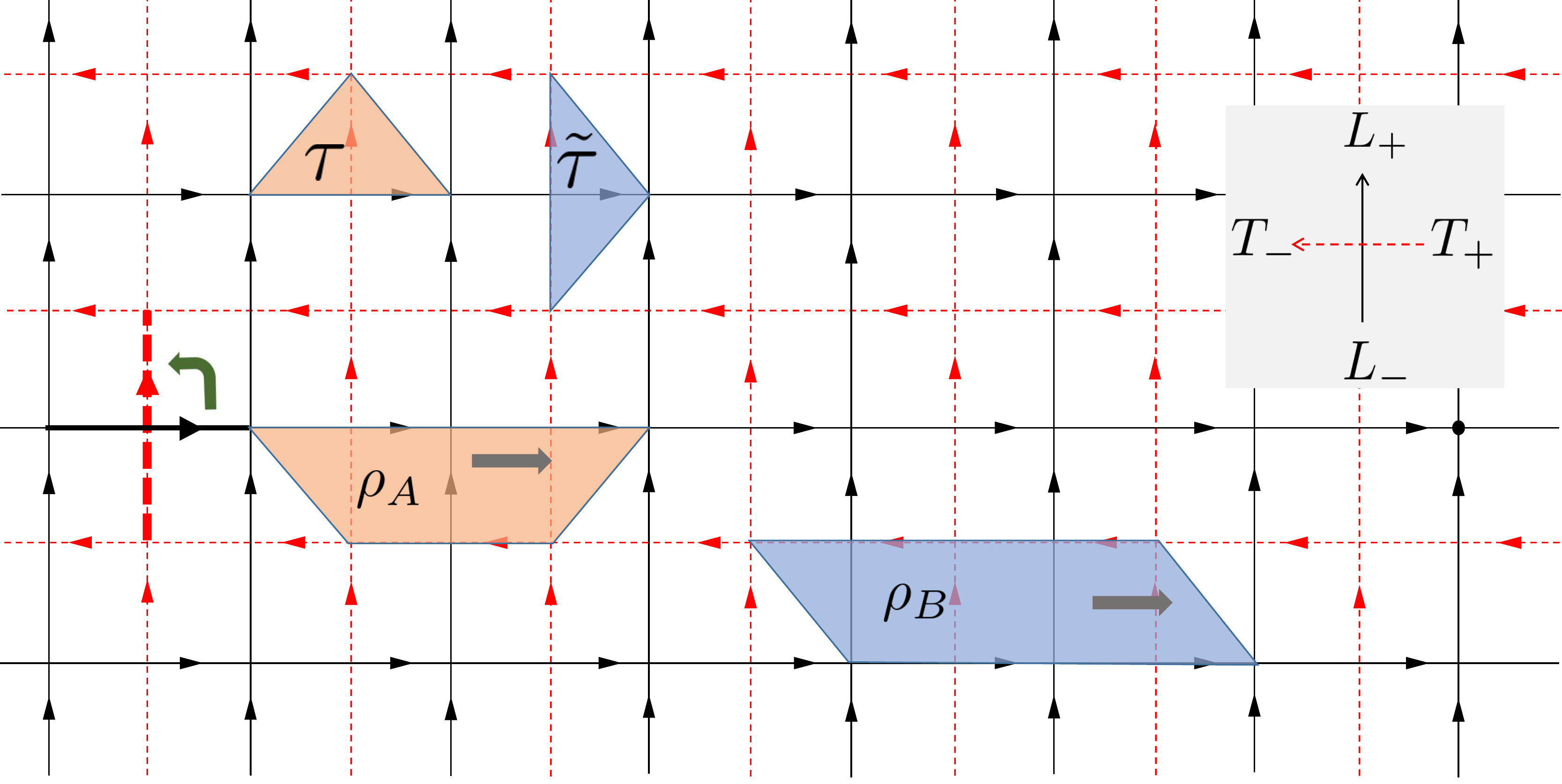}
		\caption{The depiction of geometric objects appeared in the definition of the generalized Kitaev model.  The solid black lines represent direct lattice, and the dashed red lines represent dual lattice. \label{fig:sphere}}
	\end{figure}

	Hereinafter, we assume that $H$ is a finite-dimensional $C^*$ Hopf algebra ($H$ has a Hilbert space structure given by Eq.~(\ref{eq:inner})); its dual Hopf algebra will be denoted as $\hat{H}$ or $H^{\vee}$. Note that such an $H$ is semisimple. For general facts of Hopf algebras, see Appendix \ref{sec:app_hopf}. To each edge of the lattice, we attach a copy of $H$, i.e., $\mathcal{H}_e=H$. The total Hilbert space is  $\mathcal{H}=\mathcal{H}_{tot}=\otimes_{e\in E(\Sigma)} \mathcal{H}_e$.
	To define the model, we need to introduce four types of edge operators $L_+^h,L_{-}^h$, $T_+^{\varphi},T_{-}^{\varphi}$ with $h\in H$ and $\varphi \in \hat{H}$, as follows: 
	\begin{align}
		&L_+^h |x\rangle =|h \triangleright  x\rangle=|hx\rangle,\\
		&L_-^h|x\rangle=|x \triangleleft S(h) \rangle=|xS(h)\rangle,\\
		&T_+^{\varphi}|x\rangle=|\varphi \rightharpoonup x\rangle=| \sum_{(x)}\langle \varphi, x^{(2)}\rangle x^{(1)}\rangle,\\
		&T_-^{\varphi}|x\rangle=|x\leftharpoonup \hat{S}(\varphi) \rangle
		=| \sum_{(x)}\langle \hat{S}(\varphi), x^{(1)}\rangle x^{(2)}\rangle
		=| \sum_{(x)}\langle \varphi, S(x^{(1)})\rangle x^{(2)}\rangle,
	\end{align}
	where we have adopted the Sweedler arrow notations.
	Notice that $H$ can be regarded as left $H$-modules with actions $h \triangleright  x$ and $ x\triangleleft S(h)$ (recall that $S(hg)=S(g)S(h)$ and $S(1_H)=1_H$), so that $L_+^h$ and $L_-^h$ are the corresponding operator representations.
	$H$ can also be regarded as left $\hat{H}$-modules with the actions $\varphi \rightharpoonup x$
	and $x\leftharpoonup \hat{S}(\varphi)$  (recall that $\hat{S}(\varphi \psi)=\hat{S}(\psi)\hat{S}(\varphi)$ and $\hat{S}(\hat{1})=\hat{1}$), so that $T_+^{\varphi}$ and $T_-^{\varphi}$ are the corresponding operator representations. 
	
	Since the antipode is involutive $S^2=\id$, the reverse of the edge direction is given by $x_e\mapsto \bar{x}_e=S(x_e)$. This is compatible with four edge actions, e.g.,  
	$S(L_-^h |x\rangle) = | S(xS(h)) \rangle = |hS(x)\rangle = L_{+}^h| S(x)\rangle $ and  $S(T_+^{\varphi}  |x\rangle) =\sum_{(x)}  \langle \varphi, x^{(2)} \rangle |S(x^{(1)}) \rangle=T_-^{\varphi}  |S(x)\rangle$. This means that all patterns of the edge directions are equivalent, hence the quantum double model can be constructed from arbitrary given pattern.
	
	Let $j$ be a directed edge and $v$ one of its endpoints. We define $L^h(j,v)$ as follows: if $v$ is the origin of $j$, set  $L^h(j,v)=L_-^h(j)$, otherwise, set $L^h(j,v)=L_+^h(j)$. Similarly, let $j$ be a directed dual edge and $f$ one of its endpoints; if $f$ is the origin of $j$, set $T^{\varphi}(j,f)=T_{+}^{\varphi}(j)$, otherwise, set $T^{\varphi}(j,f)=T_{-}^{\varphi}(j)$. See Fig.~\ref{fig:sphere} for an illustration of these choices. 
	For a given site $s=(v,f)$, we order the edges around the vertex  $v$ and around the face $f$ counterclockwise with the origin $s$.
	Using these conventions, the vertex operators and face operators on a site $s=(v,p)$ are defined as
	\begin{align}
		A^{h}(s)= \sum_{(h)} L^{h^{(1)}}(j_1,v)\otimes \cdots \otimes L^{h^{(n)}}(j_n,v),\\
		B^{\varphi}(s)= \sum_{(\varphi)} T^{\varphi^{(1)}}(j_1,f)\otimes \cdots \otimes T^{\varphi^{(n)}}(j_n,f).
	\end{align}
    Notice that hereinafter we assume that comultiplication $\Delta (\varphi)$ is taken in $\hat{H}$, thus the order around the face is counterclockwise; if we use comultiplication $\Delta$ of $\hat{H}^{\rm cop}$, the orientation around the face must be clockwise.

	Since $A^{h}(s)$ and $A^{h'}(s')$ can only share at most one edge (with opposite directions for $s\neq s'$), from the fact that $[L^{x}_+,L^{y}_{-}]=0$ for all $x,y \in H$,  we see $[A^{h}(s), A^{h'}(s')]=0$ for all sites $s\neq s'$.  Similarly, consider the dual lattice,  from the fact that $[T^{\psi}_+,T^{\zeta}_{-}]=0$ for all $\psi,\zeta \in \hat{H}$, we obtain $[	B^{\varphi}(s),B^{\varphi'}(s')]=0$ for all sites $s\neq s'$.
	For the non-adjacent sites $s,s'$, it is clear that $[A^h(s),B^{\varphi}(s')]=0$, but for the adjacent sites, $A^h(s)$ and $B^{\varphi}(s')$  are in general not commutative.
	It can be checked that, when $\varphi$ and $h$ are cocommutative elements, we have $[A^h(s),   B^{\varphi}(s')]=0$ for $s\neq s'$.
	
	Now consider a fixed site $s=(v,f)$, the corresponding vertex operator and face operator can only share two common edges. For example
	\begin{equation}
		\begin{aligned}
			\begin{tikzpicture}
				\draw[-latex,black] (-1,0) -- (0,0); 
				\draw[-latex,black] (0,0) -- (1,0); 
				\draw[-latex,black] (0,0) -- (0,1); 
				\draw[-latex,black] (0,-1) -- (0,0); 
				\draw[-latex,black] (0,1) -- (1,1);
				\draw[-latex,black] (1,0) -- (1,1);  
				\draw[line width=0.5pt, red] (0,0) -- (0.5,0.5);
				\draw [fill = black] (0,0) circle (1.2pt);
				\draw [fill = black] (0.5,0.5) circle (1.2pt);
				\node[ line width=0.2pt, dashed, draw opacity=0.5] (a) at (0.2,0.6){$s$};
				\node[ line width=0.2pt, dashed, draw opacity=0.5] (a) at (-0.2,-0.2){$v$};
				\node[ line width=0.2pt, dashed, draw opacity=0.5] (a) at (0.7,0.7){$f$};
				\node[ line width=0.2pt, dashed, draw opacity=0.5] (a) at (-1.2,0){$x_5$};
				\node[ line width=0.2pt, dashed, draw opacity=0.5] (a) at (0,-1.2){$x_6$};
				\node[ line width=0.2pt, dashed, draw opacity=0.5] (a) at (0.8,-0.3){$x_1$};
				\node[ line width=0.2pt, dashed, draw opacity=0.5] (a) at (1.3,0.5){$x_2$};
				\node[ line width=0.2pt, dashed, draw opacity=0.5] (a) at (0.5,1.2){$x_3$};
				\node[ line width=0.2pt, dashed, draw opacity=0.5] (a) at (-0.2,0.5){$x_4$};
			\end{tikzpicture}
		\end{aligned}
		\quad 
		\begin{aligned}
			&A^h(s)=\sum_{(h)} L_-^{h^{(1)}}(j_4)\otimes L_+^{h^{(2)}}(j_5)\otimes L_+^{h^{(3)}}(j_6) \otimes L_-^{h^{(4)}}(j_1) ,\\
			& B^{\varphi}(s)=\sum_{(\varphi)}T_{-}^{\varphi^{(1)}} (j_1) \otimes T_{-}^{\varphi^{(2)}} (j_2) \otimes T_{+}^{\varphi^{(3)}} (j_3) \otimes T_{+}^{\varphi^{(4)}} (j_4),\\
			&  A^h(s)B^{\varphi}(s)= \sum_{(h)} B^{\varphi (S^{-1} (h^{(3)} )    \bullet h^{(1)}  )}   (s)     A^{h^{(2)}}(s),
		\end{aligned}
	\end{equation}
	where ``$\bullet$'' means the argument of the function. In fact, this establishes an algebra homomorphism for arbitrary given site $s$:
	\begin{equation}
		\Phi : D(H) \to \operatorname{End} (\mathcal{H}_{tot} (s)), \quad \varphi \otimes h \mapsto D^{\varphi \otimes h}(s):=B^{\varphi}(s) A^h(s),
	\end{equation}
	where $D(H)=\hat{H}^{\rm cop}\Join H$ is the quantum double of $H$ (see Appendix \ref{sec:app_hopf}), and $\mathcal{H}_{tot} (s) =\otimes_{j\in \partial v, j\in \partial f} \mathcal{ H}_j$. When $h$ and $\varphi$ are cocommutative elements, we have $[A^h(s),B^{\varphi} (s)] =0$.
	We denote $\mathcal{D}(s) =\Phi (D(H))$. Then the mapping $h\mapsto \hat{1} \otimes h \mapsto B^{\hat{1} }  (s) A^{h}(s)$ provides a representation of $H$, and the mapping $\varphi \mapsto \varphi \otimes 1 \mapsto B^{\varphi} (s)A^1(s)$ provides a representation of $\hat{H}$. 
	Therefore $[A^{h}(s),A^{g}(s)]=A^{[h,g]} (s)$, and similarly, $[B^{\varphi}(s),B^{\psi}(s)]=B^{[\varphi,\psi]} (s)$. We see that if $h\in Z(H)$ and $\varphi \in Z(\hat{H})$ (by $Z(H)$ we mean the center of the Hopf algebra $H$), the commutators vanish.
	
	Using the $C^*$ structure, for involution invariant elements $h^*=h$ and $\varphi^*=\varphi$, the corresponding operators are Hermitian
	\begin{equation}
		(A^h(s))^{\dagger} = A^{h^*}(s) =A^{h}(s), \quad (B^{\varphi}(s'))^{\dagger} =B^{\varphi^*}(s')=B^{\varphi}(s').
	\end{equation}
	If we further require that $h, \varphi$ are idempotent $h^2=h$ and $\varphi^2=\varphi$, then $A^h(s)$ and $B^{\varphi}(s')$ become projectors.

	Now we are in a position to give the Hamiltonian construction for the Hopf algebraic quantum double model.
	The input data will be a finite-dimensional $C^*$ Hopf algebra $H$, which introduces the total space $\mathcal{H}_{tot}=\otimes_{e\in E(\Sigma)} \mathcal{H}_e$ with $\mathcal{H}_e = H$. The Haar integrals $h_H\in H$ and $\varphi_{\hat{H}}\in \hat{H}$ exist and are unique, involutive, idempotent, cocommutative, and in the center of $H$ and $\hat{H}$, respectively.
	The corresponding operator $A^{h_H}(s)$ only depends on vertex $v$ and is thus denoted as $A^H_v$, and similarly  $B^{\varphi_{ \hat{H}  }} ( s)$ only depends on face $f$ and is thus denoted as $B^{H}_f$. The local operators are projectors and commutative with each other. The  frustration-free Hamiltonian is  
	\begin{equation}\label{eq:kitaev}
		H_H(\Sigma)=\sum_{v\in V(\Sigma)} (I-A_v^{H})+\sum_{f\in F(\Sigma)} (I-B^{H}_f).
	\end{equation}
	This model will be called a generalized quantum double model.
	
	\begin{remark}
		Notice that in the above construction, both $A_v$ and $B_f$ are defined from the left-module structures of $H$. We can also introduce a right-module construction.
		To define the model, we need to introduce four types of right-module edge operators $\tilde{L}_{-}^h,\tilde{L}_{+}^h$, $\tilde{T}_{-}^{\varphi},\tilde{T}_{+}^{\varphi}$, as follows: 
		\begin{align}
			&\tilde{L}_-^h |x\rangle =|x \triangleleft  h\rangle=|xh\rangle,\\
			&\tilde{L}_+^h|x\rangle=|S(h) \triangleright x\rangle=|S(h) x\rangle,\\
			&\tilde{T}_-^{\varphi}|x\rangle = |x \leftharpoonup \varphi \rangle  
			=|\sum_{(x)}\langle \varphi, x^{(1)}\rangle x^{(2)}\rangle,\\
			&\tilde{T}_+^{\varphi}|x \rangle=| \hat{S}(\varphi) \rightharpoonup x\rangle
			=| \sum_{(x)}\langle \hat{S}(\varphi), x^{(2)}\rangle x^{(1)}\rangle
			=| \sum_{(x)}\langle \varphi, S(x^{(2)})\rangle x^{(1)}\rangle,
		\end{align}
		where $h\in H$ and $\varphi \in \hat{H}$.
		Here $H$ can be regarded as right $H$-modules via actions $x \triangleleft  h$ and $ S(h)\triangleright x$, with $\tilde{L}_-^h$ and $\tilde{L}_+^h$ being the corresponding operator representations.
		$H$ can also be regarded as right $\hat{H}$-modules via the actions $x \leftharpoonup \varphi$
		and $\hat{S}(\varphi) \rightharpoonup  x$,  with $\tilde{T}_-^{\varphi}$ and $\tilde{T}_+^{\varphi}$ being the corresponding operator representations.
		For site $s=(v,f)$, we can order the edges around vertex $v$ and around face $f$ clockwise starting from $s$. The convention for choosing ``$+$'' or ``$-$'' remains unchanged. In this way, for Haar integrals $h_H\in H$ and $\varphi_{\hat{H}} \in \hat{H}$, the vertex operator $\tilde{A}^H_v$ and the face operator $\tilde{B}^H_f$ can be constructed. A right-module Hamiltonian $\tilde{H}_H(\Sigma)=-\sum_{v\in V(\Sigma)} \tilde{A}^H_v -\sum_{f\in F(\Sigma)} \tilde{B}^H_f$ is thus obtained. 
	\end{remark}
	
	The ground state space of the model  (\ref{eq:kitaev}) is given by
	\begin{equation}
		\mathcal{H}_{GS}=	(\prod_{v\in V(\Sigma) } A^H_v \prod_{f\in F(\Sigma)} B_f^{H} )\mathcal{H}_{tot}.
	\end{equation}
	The ground state degeneracy depends on the topology of the surface $\Sigma$, and is independent of the choice of the cellulation $C(\Sigma)$ (since different cellulations can be related with Pachner moves),
	\begin{equation}
		\operatorname{GSD}=\operatorname{Tr} (\prod_{v\in V(\Sigma) } A^H_v \prod_{f\in F(\Sigma)} B_f^{H} ).
	\end{equation}
	On a sphere (or infinite plane), $\operatorname{GSD}=1$, i.e., there is a unique ground state $|\Psi_{GS}\rangle$.

\subsection{Ribbon operators}
	\label{sec:ribbon}
	
The ribbon operators are crucial for us to study the topological excitations.
In this subsection, following the work of \cite{chen2021ribbon}, we will construct ribbon operators and determine the corresponding ribbon operator algebra over a given ribbon.
There are two kinds of ribbons, called type-A and type-B here, classified by the chirality of the triangles composing them.
Their corresponding ribbon operator algebras are slightly different.
A detailed presentation of ribbon operator algebras and the properties of ribbon operators will be given in Appendices \ref{sec:AA} and \ref{sec:AB}.
	
To begin with, we introduce several geometric objects \cite{Kitaev2003,Bombin2008family,Cong2017,chen2021ribbon} (see Fig.~\ref{fig:sphere} for illustrations and see Appendix \ref{sec:AA} for a more comprehensive discussion):
	\begin{itemize}
		\item A direct triangle $\tau=(s_0,s_1,e)$ consists of two adjacent sites $s_0$ and $s_1$ connected via a directed edge $e$. 
		Similarly, the dual triangle $\tilde{\tau}=(s_0,s_1,\tilde{e})$ consists of two adjacent sites $s_0$ and $s_1$ connected via a dual edge $\tilde{e}$.
		The direction of the triangle is defined as the direction from $s_0$ to $s_1$. Notice that the direction of the triangle may or may not match the direction of the edge.  
		
		\item For a given (direct or dual) triangle, the chirality  (which is called local orientation in \cite{chen2021ribbon}) of the triangle is defined as follows: it is called a left-handed (right-handed) triangle if the edge of the triangle is on the left-hand side (right-hand side) when we pass through the triangle along its positive direction. Notice that the chirality is fixed when the direction of the triangle is fixed. We will denote left-handed (right-handed) direct triangle  and dual triangle as $\tau_L$ and $\tilde{\tau}_L$ ($\tau_R$ and $\tilde{\tau}_R$) respectively.
		
		\item A ribbon $\rho$ is a collection of triangles $\tau_1,\cdots,\tau_n$ with  a given direction such that $\partial_1\tau_j=\partial_0 \tau_{j+1}$ and there is no self-intersection. A closed ribbon is defined as a ribbon that does not have any open ends, \emph{viz}., $\partial_1 \tau_n=\partial_0 \tau_1$. For a given directed ribbon, the direct triangle and dual triangle in it must have different chirality.
		A ribbon consisting of left-handed direct triangles and right-handed dual triangles is called a type-A ribbon and will be denoted as $\rho_{A}$; similarly, a type-B ribbon consists of right-handed direct triangles and left-handed dual triangles and will be denoted as $\rho_{B}$. 
	\end{itemize}

	The ribbon operator can be defined recursively. First, we define the triangle operator, then by introducing recursive relation, the ribbon operator is determined.
	To define triangle operators, we need to consider different cases separately (the convention we use here is following Ref.~\cite{Kitaev2003}, which is slightly different from the one in Ref.~\cite{chen2021ribbon}).

	For right-handed direct triangles, we have
	\begin{align}
		&	\begin{aligned}
			\begin{tikzpicture}
				\draw[-latex,black] (1,0) -- (-1,0); 
				\node[ line width=0.2pt, dashed, draw opacity=0.5] (a) at (0,0.2){$x$};
				\draw[line width=0.5pt, red] (-1,0) -- (0,-1);
				\node[ line width=0.2pt, dashed, draw opacity=0.5] (a) at (-0.6,-0.7){$s_1$};
				\draw[line width=0.5pt, red] (1,0) -- (0,-1);
				\node[ line width=0.2pt, dashed, draw opacity=0.5] (a) at (0.6,-0.7){$s_0$};
				\draw[-stealth,gray, line width=3pt] (0.5,-0.4) -- (-0.5,-0.4); 
			\end{tikzpicture}
		\end{aligned}
		\quad
		\begin{aligned}
			F^{h,\varphi}( \tau_R)|x\rangle =  
			\varepsilon(h)  T_{-}^{\varphi} |x\rangle
			=  \varepsilon(h) | x \leftharpoonup \hat{S}( \varphi) \rangle,
		\end{aligned} \label{eq:tri1} \\
		&\begin{aligned}
			\begin{tikzpicture}
				\draw[-latex,black] (-1,0) -- (1,0); 
				\node[ line width=0.2pt, dashed, draw opacity=0.5] (a) at (0,0.2){$x$};
				\draw[line width=0.5pt, red] (-1,0) -- (0,-1);
				\node[ line width=0.2pt, dashed, draw opacity=0.5] (a) at (-0.6,-0.7){$s_1$};
				\draw[line width=0.5pt, red] (1,0) -- (0,-1);
				\node[ line width=0.2pt, dashed, draw opacity=0.5] (a) at (0.6,-0.7){$s_0$};
				\draw[-stealth,gray, line width=3pt] (0.5,-0.4) -- (-0.5,-0.4); 
			\end{tikzpicture}
		\end{aligned}
		\quad
		\begin{aligned}
			F^{h,\varphi}( \tau_R)|x\rangle = 	\varepsilon(h)  T_{+}^{\varphi} |x\rangle
			= \varepsilon(h) | \varphi \rightharpoonup x \rangle.
		\end{aligned}
	\end{align}
		For right-handed dual triangles, we have
	\begin{align}
		&	\begin{aligned}
			\begin{tikzpicture}
				\draw[-latex,dashed,black] (1,0) -- (-1,0); 
				\node[ line width=0.2pt, dashed, draw opacity=0.5] (a) at (0,0.2){$x$};
				\draw[line width=0.5pt, red] (-1,0) -- (0,-1);
				\node[ line width=0.2pt, dashed, draw opacity=0.5] (a) at (-0.6,-0.7){$s_1$};
				\draw[line width=0.5pt, red] (1,0) -- (0,-1);
				\node[ line width=0.2pt, dashed, draw opacity=0.5] (a) at (0.6,-0.7){$s_0$};
				\draw[-stealth,gray, line width=3pt] (0.5,-0.4) -- (-0.5,-0.4); 
			\end{tikzpicture}
		\end{aligned}
		\quad
		\begin{aligned}
			F^{h,\varphi}( \tilde{\tau}_R)|x\rangle =   \hat{\varepsilon}(\varphi) L^{h}_{-}|x\rangle
			= \hat{\varepsilon}(\varphi) |  x \triangleleft  S(h) \rangle  ,
		\end{aligned}\\
		&\begin{aligned}
			\begin{tikzpicture}
				\draw[-latex,dashed,black] (-1,0) -- (1,0); 
				\node[ line width=0.2pt, dashed, draw opacity=0.5] (a) at (0,0.2){$x$};
				\draw[line width=0.5pt, red] (-1,0) -- (0,-1);
				\node[ line width=0.2pt, dashed, draw opacity=0.5] (a) at (-0.6,-0.7){$s_1$};
				\draw[line width=0.5pt, red] (1,0) -- (0,-1);
				\node[ line width=0.2pt, dashed, draw opacity=0.5] (a) at (0.6,-0.7){$s_0$};
				\draw[-stealth,gray, line width=3pt] (0.5,-0.4) -- (-0.5,-0.4); 
			\end{tikzpicture}
		\end{aligned}
		\quad
		\begin{aligned}
			F^{h,\varphi}( \tilde{\tau}_R)|x\rangle =  
			\hat{\varepsilon}(\varphi) L^{h}_{+}|x\rangle= \hat{\varepsilon}(\varphi) | h\triangleright x \rangle.
		\end{aligned} \label{eq:tri8}
	\end{align}
	For left-handed dual triangles, we have
	\begin{align}
		&	\begin{aligned}
			\begin{tikzpicture}
				\draw[-latex,dashed,black] (1,0) -- (-1,0); 
				\node[ line width=0.2pt, dashed, draw opacity=0.5] (a) at (0,0.2){$x$};
				\draw[line width=0.5pt, red] (-1,0) -- (0,-1);
				\node[ line width=0.2pt, dashed, draw opacity=0.5] (a) at (-0.6,-0.7){$s_0$};
				\draw[line width=0.5pt, red] (1,0) -- (0,-1);
				\node[ line width=0.2pt, dashed, draw opacity=0.5] (a) at (0.6,-0.7){$s_1$};
				\draw[-stealth,gray, line width=3pt] (-0.5,-0.4) -- (0.5,-0.4); 
			\end{tikzpicture}
		\end{aligned}
		\quad
		\begin{aligned}
			F^{h,\varphi}( \tilde{\tau}_L)|x\rangle =  \hat{\varepsilon}(\varphi)   \tilde{L}_-^{h} |x\rangle = \hat{\varepsilon}(\varphi)  
			| x\triangleleft  h \rangle.
		\end{aligned}\\
		&\begin{aligned}
			\begin{tikzpicture}
				\draw[-latex,dashed,black] (-1,0) -- (1,0); 
				\node[ line width=0.2pt, dashed, draw opacity=0.5] (a) at (0,0.2){$x$};
				\draw[line width=0.5pt, red] (-1,0) -- (0,-1);
				\node[ line width=0.2pt, dashed, draw opacity=0.5] (a) at (-0.6,-0.7){$s_0$};
				\draw[line width=0.5pt, red] (1,0) -- (0,-1);
				\node[ line width=0.2pt, dashed, draw opacity=0.5] (a) at (0.6,-0.7){$s_1$};
				\draw[-stealth,gray, line width=3pt] (-0.5,-0.4) -- (0.5,-0.4); 
			\end{tikzpicture}
		\end{aligned}
		\quad
		\begin{aligned}
			F^{h,\varphi}( \tilde{\tau}_L)|x\rangle =   \hat{\varepsilon}(\varphi) \tilde{L}_{+}|x\rangle
			= \hat{\varepsilon}(\varphi)    	
			|S(h)\triangleright  x\rangle.
		\end{aligned}
	\end{align}
	For left-handed direct triangles, we have 
	\begin{align}
		&	\begin{aligned}
			\begin{tikzpicture}
				\draw[-latex,black] (1,0) -- (-1,0); 
				\node[ line width=0.2pt, dashed, draw opacity=0.5] (a) at (0,0.2){$x$};
				\draw[line width=0.5pt, red] (-1,0) -- (0,-1);
				\node[ line width=0.2pt, dashed, draw opacity=0.5] (a) at (-0.6,-0.7){$s_0$};
				\draw[line width=0.5pt, red] (1,0) -- (0,-1);
				\node[ line width=0.2pt, dashed, draw opacity=0.5] (a) at (0.6,-0.7){$s_1$};
				\draw[-stealth,gray, line width=3pt] (-0.5,-0.4) -- (0.5,-0.4); 
			\end{tikzpicture}
		\end{aligned}
		\quad
		\begin{aligned}
			F^{h,\varphi}( \tau_L)|x\rangle =  
			\varepsilon(h)  \tilde{T}_{-}^{\varphi} |x\rangle
			=  \varepsilon(h) 
			| x \leftharpoonup \varphi \rangle ,
		\end{aligned}\\
		&\begin{aligned}
			\begin{tikzpicture}
				\draw[-latex,black] (-1,0) -- (1,0); 
				\node[ line width=0.2pt, dashed, draw opacity=0.5] (a) at (0,0.2){$x$};
				\draw[line width=0.5pt, red] (-1,0) -- (0,-1);
				\node[ line width=0.2pt, dashed, draw opacity=0.5] (a) at (-0.6,-0.7){$s_0$};
				\draw[line width=0.5pt, red] (1,0) -- (0,-1);
				\node[ line width=0.2pt, dashed, draw opacity=0.5] (a) at (0.6,-0.7){$s_1$};
				\draw[-stealth,gray, line width=3pt] (-0.5,-0.4) -- (0.5,-0.4); 
			\end{tikzpicture}
		\end{aligned}
		\quad
		\begin{aligned}
			F^{h,\varphi}( \tau_L)|x\rangle = 
				\varepsilon(h)  \tilde{T}_{+}^{\varphi} |x\rangle
			= \varepsilon(h) 	|  \hat{S}( \varphi) \rightharpoonup  x \rangle.
		\end{aligned}  \label{eq:tri16}
	\end{align}
	The reason for the above choice of convention we have made is to make the vertex and face operators be special cases of ribbon operators.

For general ribbon $\rho$, the ribbon operator on it can be defined recursively as follows. We begin with the definition of type-B ribbon operators.
They are built from the dual Hopf algebra $D(H)^{\vee}=D_B(H)^{\vee}=H^{\rm op} \otimes \hat{H}$, with 
    \begin{align}
    	&	(h \otimes \varphi)(g \otimes \psi) =g h \otimes \varphi \psi, \\
    	&	\Delta_{D(H)^{\vee}}(h \otimes \alpha) =\sum_{k,(k),(h)} (h^{(1)}\otimes \hat{k})\otimes (S(k^{(3)} )  h^{(2)} k^{(1)} \otimes \alpha (k^{(2)}  \bullet)), \label{eq:coprod-dual}\\
    	&	1_{D(H)^{\vee}}=1_{H} \otimes \varepsilon_{H}, \quad \varepsilon_{D(H)^{\vee}}(h \otimes \alpha)=\varepsilon(h) \alpha(1), 
    \end{align}
    where $\{k\}$ is an orthogonal basis of $H$ with $\{\hat{k}\}$ its dual basis.
    For a ribbon $\rho=\rho_B$ and an element $h\otimes \varphi \in  D(H)^{\vee}$, we can define the ribbon operator $F^{h,\varphi}(\rho)=F^{h\otimes \varphi} (\rho_B)$.
    The operator acts non-trivially only on edges contained in $\rho_B$, and it commutes with all vertex and face operators except the ones that act on the ending sites of the ribbon.
    These operators form an algebra $\mathcal{A}_{\rho_B}=\{F^{h,\varphi}(\rho_B)~|~ h\in H^{\rm op},\varphi \in \hat{H} \}$ called the ribbon operator algebra. Consider the decomposition $\rho=\rho_B =\tau_1\cup  \tau_2$ where both $\tau_1$ and $\tau_2$ have the same direction with $\rho$ and $\partial_1 \tau_1 =\partial_0 \tau_2$.
	For $h\otimes \varphi \in  H^{\rm op}\otimes \hat{H}$, the ribbon operator on this composite ribbon can be defined as
	\begin{equation}\label{eq:ribb_op}
		\begin{aligned}
			F^{h,\varphi}(\rho) =&\sum_{(h\otimes \varphi)} F^{(h\otimes \varphi)^{(1)}} (\tau_1) F^{(h\otimes \varphi)^{(2)}} (\tau_2)\\
			=& \sum_{k} \sum_{(k), (h)} F^{h^{(1)}, \hat{k}}(\tau_1) F^{S(k^{(3)} )  h^{(2)} k^{(1)} , \varphi (k^{(2)}  \bullet) }(\tau_2). 
		\end{aligned}
	\end{equation}
From the co-associativity of the Hopf algebra $D(H)^{\vee}$, we see that this definition is independent of the decomposition $\rho=\tau_1\cup \tau_2$.
The construction for type-A ribbon is similar, but the ribbon operator is now built from $D(H)^{\vee, \rm op}=D_A(H)^{\vee}=H\otimes\hat{H}^{\rm op}$.

For a closed ribbon, there is only one end $\partial\rho=\partial_0 \rho =\partial_1 \rho$.
The vertex  and face operators can be regarded as special cases of closed ribbon operators. Specifically, $A_v=F^{h_H,\hat{1}}(\sigma_A)$ is a type-A dual closed ribbon operator, and $B_f=F^{1,\varphi_{\hat{H}}}(\sigma_B)$ is a type-B direct closed ribbon operator. Likewise,
$\tilde{A}_v=F^{h_H,\hat{1}}(\sigma_B)$ is a type-B dual closed  ribbon operator, and $\tilde{B}_f=F^{1,\varphi_{\hat{H}}}(\sigma_A)$ is a type-A direct closed ribbon operator.
A more comprehensive discussion of ribbon operators is given in Appendices \ref{sec:AA} and \ref{sec:AB}.

	\subsection{Topological excitations}

	The topological excitations of the model are point-like quasi-particles, and the corresponding state $|\Psi\rangle$ violates some of the stabilizer conditions $A_v^{h_H}|\Psi\rangle = |\Psi\rangle$ and  $B_f^{\varphi_{\hat{H}}}|\Psi\rangle = |\Psi\rangle$ for some local vertex and face operators.
	For a ribbon $\rho$ connecting sites $s_0$ and $s_1$, ribbon operator $F^{g, \psi}(\rho)$ commutes with all the vertex and face operators in the Hamiltonian of Eq.~(\ref{eq:kitaev}) except the ones at sites $s_0=\partial_0 \rho$ and $\partial_1 \rho=s_1$.
	Thus the ribbon operator creates excitations only at the ends of the ribbon.

	Before we discuss the topological excitations of the Hopf algebraic quantum double model, let us first recall the case that $H=\mathbb{C}[G]$ with $G$ a finite group \cite{Kitaev2003}.
	In this case, topological excitations are classified by $([g], \pi)$ where $[g]$ is a conjugacy class of the group $G$, and $\pi$ is an irreducible representation of the centralizer $C_{G}(g)$. Notice that there is a $\mathbb{C}C_{G}(g)$-module $M_{\pi}$ corresponding to $\pi$, thus a topological charge can also be expressed as 
	\begin{equation}\label{eq:anyon}
		a_{[g],\pi}= \mathbb{C}[G] \otimes _{\mathbb{C} C_G(g)} M_{\pi}.
	\end{equation} 
	The vacuum charge corresponds to $g=e_G$ and $\pi=\mathds{1}$ (the trivial representation).
	The antiparticle of the one in Eq.~(\ref{eq:anyon})  is given by (note that $C_G(g)=C_G(g^{-1})$)
	\begin{equation}
		a_{[g^{-1}],\pi^{\dagger}}= \mathbb{C}[G] \otimes _{\mathbb{C} C_G(g^{-1})} M_{\pi^{\dagger}}.
	\end{equation} 
	The conjugacy class $[g]$ of a topological charge is called magnetic charge and the irrep $\pi$ is called electric charge. When $g=e_G$,  $	a_{[e],\pi}$ is characterized by a representation of $G$ and is called a chargeon; when  $\pi=\mathds{1}$, $a_{[g],\mathds{1}}$ is called a fluxion; and  when both $g\neq e_G$ and $\pi\neq \mathds{1}$, $	a_{[g],\pi}$ is called a dyon.
	The quantum dimension of the topological excitation is given by
	\begin{equation}
		\FPdim a_{[g],\pi} =|[g]| \dim \pi.
	\end{equation}
	 These topological excitations form a UMTC $\mathsf{Rep} (D(G))$, the representation category of the quantum double of the finite group $G$.

	For a general Hopf algebra $H$, a similar picture for classifying topological excitations exists, but it is much more complicated (to our knowledge, this has not been discussed in physical literature so far).
	To introduce such a classification, let us first present several crucial mathematical notions.
	A fusion category $\EC$ is called $G$-graded if there exists a decomposition
	\begin{equation}
		\EC=\oplus_{g\in G} \EC_g,
	\end{equation}
	where $\EC_g$'s are some full Abelian subcategories, and the tensor product of $\EC$ maps $\EC_g\times \EC_h$ to $\EC_{gh}$ for all $g,h$ in the finite group $G$. $G$ is called a grading group of $\EC$; when $G$ is maximal in the sense that any other grading is obtained by a quotient group of $G$, it is called a universal grading group and we denote it as $G=U(\EC)$.
	It follows from \cite{gelaki2008nilpotent} that there is a universal grading group $G=U(\mathsf{Rep}(H))$ for any semisimple Hopf algebra $H$.
	
	Consider the largest central Hopf subalgebra $K(\hat{H}^{\rm cop})$ of $\hat{H}^{\rm cop}$, we have $K(\hat{H}^{\rm cop})=\mathbb{C}G^
	{\vee}$, which is commutative, and
	$\mathsf{Rep} (\hat{H}^{\rm cop}) =\oplus_{g\in G}	\mathsf{Rep} (\hat{H}^{\rm cop}) _g$.
	Suppose that $H_g$ is a Hopf subalgebra of $H$ such that 
	$\mathsf{Rep}(\hat{H}_g) =\oplus_{x\in C_G(g)}	\mathsf{Rep} (\hat{H}^{\rm cop}) _x$.
	It is proved that $K(\hat{H}^{\rm cop})=\mathbb{C}G^{\vee}$  is a normal Hopf subalgebra of $D(H)$ \cite{burciu2012irreducible,burciu2017grothendieck}, and $g\in G$ becomes an irreducible character of  $K(\hat{H}^{\rm cop})$.
	We denote $\mathcal{I}_g=\{M_g\}$ the set of all irreducible representations  of $\hat{H}^{\rm cop} \Join H_g$ (here ``$ \Join$'' denotes bicrossed product) such that the character $\chi_{M_g}$, when restricted on $K(\hat{H}^{\rm cop})$,  satisfies ${\chi_{M_g}}_{ |_{K(\hat{H}^{\rm cop})} }=g \dim M_g$.
	With the above preparation, we are now at a position to present the classification of topological excitations, namely 
	\begin{equation}
		a_{g, M_g}= H \otimes_{H_g} M_g,
	\end{equation}
	where $g\in G$ and $M_g\in \mathcal{I}_g$.
	This completely classifies the irreducible representations of the quantum double $D(H)$, see \cite{burciu2012irreducible,burciu2017grothendieck} for rigorous proof.
	The $g\in G$ can be regarded as a magnetic charge and $M_g$ an electric charge. 
	When $g=e_G$,  $	a_{e_G,M_{e_G}}$ is called a chargeon; when  $M_g=\mathds{1}$, $a_{g,\mathds{1}}$ is called a fluxion; and  when both $g\neq e$ and $M_g \neq \mathds{1}$, $	a_{g, M_g}$ is called a dyon.
	The quantum dimension of the topological excitation is
	\begin{equation}
		\FPdim 	a_{g, M_g} =\frac{|G|}{\dim H_g} \dim M_g.
	\end{equation}
	These topological excitations form a UMTC $\mathsf{Rep} (D(H))$, the representation category of the quantum double of the Hopf algebra $H$.

\section{Gapped boundary theory}
\label{sec:bd}
	
In this section, we will establish the theory of gapped boundaries for the Hopf quantum double model. While the boundary theories for the special case where $H=\mathbb{C}[G]$ have been investigated from different aspects in previous works \cite{bravyi1998quantum,Bombin2008family,freedman2001projective,Beigi2011the, Kitaev2012a,Levin2013protected,Cong2017,wang2020electric,etingof2003finite,ostrik2003module,andruskiewitsch2007module,natale2017equivalence}, the problem has not yet been systematically studied for the general Hopf algebra case. Therefore, we will present the lattice construction and algebraic theory of the gapped boundary for a general bulk Hopf algebra $H$, based on an $H$-comodule algebra $\mathfrak{A}$ (or equivalently, based on an $H$-module algebra $\mathfrak{M}$).
Furthermore, we will demonstrate that there exists a one-to-one correspondence between all the gapped boundaries of quantum double phase and $H$-comodule algebras. Specifically, this means that any gapped boundary of the quantum double phase with input bulk Hopf algebra $H$ can be realized by an $H$-comodule algebra.
	
Before we start, let us first recall the following fundamental definitions (see \cite{montgomery1993hopf,majid2000foundations} and Appendix \ref{sec:app_hopf} for more details).
	
  \begin{definition}
	Let $H$ be a Hopf algebra and $(\mathfrak{A},\mu_{\mathfrak{A}},\eta_{\mathfrak{A}})$ an algebra. If $\mathfrak{A}$ is a left $H$-comodule with left coaction $\beta_{\mathfrak{A}}:\mathfrak{A}\to H\otimes \mathfrak{A}$ such that $\beta_{\mathfrak{A}}(xy)=\beta_{\mathfrak{A}}(x)\beta_{\mathfrak{A}}(y)$ and $\beta_{\mathfrak{A}}(1_{\mathfrak{A}})=1_H\otimes 1_{\mathfrak{A}}$, then $\mathfrak{A}$ is called a left $H$-comodule algebra. A right $H$-comodule algebra can be defined similarly.
	\end{definition}

	\begin{definition}
	Let $H$ be a Hopf algebra and $\FM$ an algebra. If $\FM$ is a left $H$-module such that $h\triangleright (xy)=\sum_{(h)}(h^{(1)}\triangleright x)(h^{(2)}\triangleright y)$ and $h\triangleright 1_\mathfrak{M}=\varepsilon(h)1_\mathfrak{M}$, then $\FM$ is called a left $H$-module algebra. A right $H$-module algebra can be defined similarly.
	\end{definition}

\noindent
{\it Notation.} For left and right comodules, we shall adopt the Sweedler's notation for left and right coactions $\beta(x)=\sum_{(x)}x^{[-1]}\otimes x^{[0]}$ and $\beta(x)=\sum_{(x)}x^{[0]}\otimes x^{[1]}$ respectively.

By invoking the construction of Frobenius algebra in a rigid category in Ref.~\cite{fuchs2009frobenius}, we can equip an $H$-comodule algebra $\FA$ with a Frobenius algebra structure.
Consider the rigid category $\Vect_{\mathbb{C}}$, which contains $\FA$ as an object. We denote the dual object corresponding to $\FA$ in $\Vect_{\mathbb{C}}$ as $\FA^{\vee}$. Let $\{k,\hat{k}\}_k$ be dual basis
for $\FA$ and $\FA^{\vee}$.
The coevaluation map is a linear map $\operatorname{coev}_{\FA}:\mathbb{C}\to \FA\otimes \FA^{\vee}$ defined as $\operatorname{coev}_{\FA}(1)=\sum_k k\otimes \hat{k}$. The evaluation map is a linear map $\operatorname{ev}_{\FA}: \FA \otimes \FA^{\vee}\to \mathbb{C}$ defined as $\operatorname{ev}_{\FA}(a\otimes\psi)=\psi(a)$.
We introduce the counit map as $\varepsilon_{\FA}=\operatorname{ev}_{\FA}\comp (\mu_{\FA}\otimes \id_{\FA^{\vee}}) \comp (\id_{\FA} \otimes \operatorname{coev}_{\FA})$, then set $\kappa: \FA\otimes \FA \to \mathbb{C}$ as $\kappa=\varepsilon_{\FA}\comp \mu_{\FA}$, and $\Phi_{\kappa}:\mathfrak{A}\to \FA^{\vee}$ as $\Phi_{\kappa}=(\id_{\FA^{\vee}}\otimes \kappa)\comp (\operatorname{coev}_{\FA} \otimes \id_{\FA})$.
The comultiplication is defined as $\Delta_{\FA}=(\id_{\FA}\otimes \mu_{\FA})\comp (\id_{\FA}\otimes \Phi_{\kappa}^{-1}\otimes \id_{\FA})\comp (\operatorname{coev}_{\mathfrak{A}} \otimes \id_{\FA})$.
Based on Ref.~\cite[Proposition 8]{fuchs2009frobenius}, it is easy to check that $(\FA,\mu_{\FA},\eta_{\FA},\Delta_{\FA},\varepsilon_{\FA})$ is a normalized-special and symmetric Frobenius algebra (see Appendix \ref{sec:app_hopf} for definition).

\begin{remark}
It is worth mentioning that a left (right) $H$-action on $\mathfrak{A}$ canonically corresponds to a right (left) $\hat{H}$-coaction on $\mathfrak{A}$.
If $\mathfrak{A}$ is a left (right) $H$-comodule algebra, then it is a right (left) $\hat{H}$-module algebra.
Thus, the boundary can be equivalently described by $\hat{H}$-module algebras.
Since the quantum double model is self-dual when exchanging $H$ and $\hat{H}$ (known as an electric-magnetic duality), the boundary data can be chosen as $H$-module algebras or $H$-comodule algebras freely.
On the other hand, it is easy to verify that if $\FA$ is a module algebra then $\hat{\FA}$ is a module coalgebra and vice versa; if $\FA$ is a comodule algebra then $\hat{\FA}$ is a comodule coalgebra and vice versa. This means that we can freely choose one
of four kinds of structures as our input data for boundary model: module algebra, comodule algebra, module coalgebra, and comodule coalgebra. In this work, we will choose to interchangeably use the comodule algebra and module algebra structures.
\end{remark}	

In order to ensure the stability of our gapped boundary lattice model, we will require that the $H$-comodule algebra $\FA$ is $H$-indecomposable. This means that there are no non-trivial two-sided $H$-ideals $I$ and $J$ such that $\FA = I \oplus J$. By $H$-ideal we mean an ideal that is also an $H$-comodule. This stability condition is also presented in the Kitaev-Kong construction of the string-net gapped boundary model, and it has been shown \cite{andruskiewitsch2007module} that $\FA$ is $H$-indecomposable if and only if the module category ${_{\FA}}\Mod$ of left $\mathfrak{A}$-modules is indecomposable. We will provide further clarification on this topic in our subsequent discussion.
	
\subsection{Lattice construction} \label{sec:lattice-const}

We will start with a simple construction based on Hopf subalgebras, which are automatically $H$-comodule algebras. Later, we will explore a more general construction based on $H$-comodule algebras.

\subsubsection{Gapped boundary theory I}
\label{sec:bdI}
	
	\begin{figure}[t]
		\centering
		\includegraphics[width=13cm]{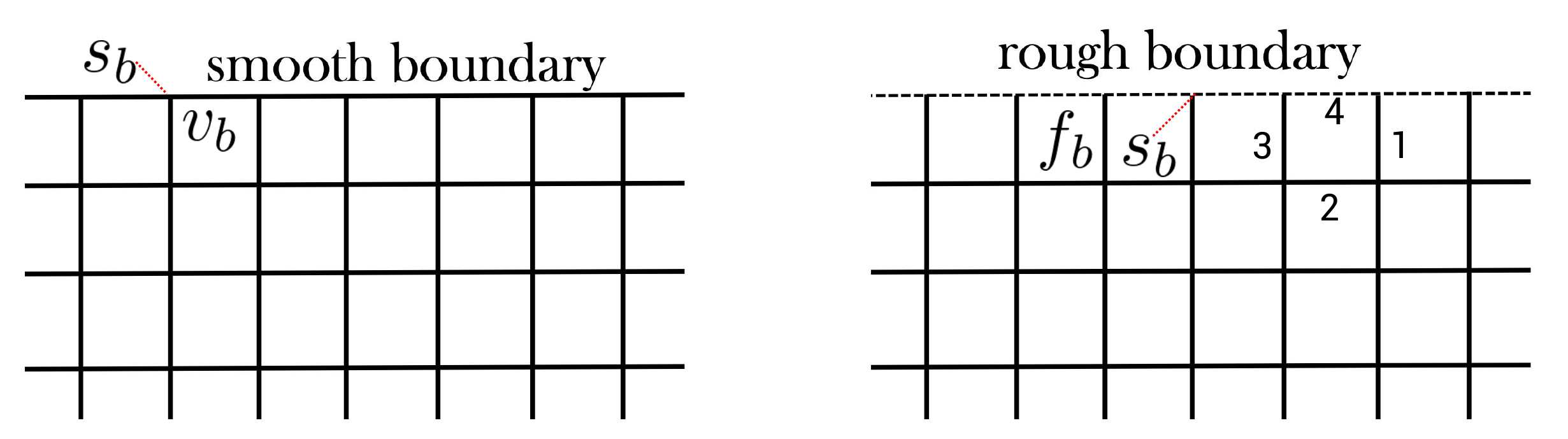}
		\caption{The depiction of smooth boundary and rough boundary. \label{fig:smoothbd}}
	\end{figure}

Consider a surface $\Sigma$ with a single boundary $\partial \Sigma$. For a given cellulation $C(\Sigma)=C(\Sigma\setminus \partial \Sigma) \cup C( \partial \Sigma) $, we need to investigate the vertices, edges, and faces in the vicinity of the boundary.	
Our first construction of the gapped boundary is based on a Hopf subalgebra $K\leq H$.
In this case, $\mathfrak{A}=K$ is naturally an $H$-comodule algebra.
The coaction is just the coproduct map $\beta_K=\Delta_K: K \to K\otimes K$, since $K\otimes K \subset H\otimes K$ and $K\otimes K \subset K\otimes H$. Thus $K$ can be regarded as either a left or a right $H$-comodule algebra. By definition, $\beta_K(ab)=\beta_K(a)\beta_K(b)$ and $\beta_K(1)=1\otimes 1$.

The boundary edges are projected to the subspace 
	\begin{equation} \label{eq:proj-K}
	    \mathcal{H}_{e_b}=\Pi_K H=K.
	\end{equation}
The boundary vertex operator is chosen as 
	\begin{equation}
	    A_{v_b}^K=A_{v_b}^{h_K},
	\end{equation}
where $h_K$ is the Haar integral of $K$.
It is obvious that $[A_{v_b}^K,A_{v'_b}^K]=0=[A_{v_b}^K,A_{v}^H]$ for all $v_b,v'_b,v$. Since $h_K$ and $\varphi_{\hat{H}}$ are Haar integrals, they are cocommutative, which implies that $[A^K_{v_b},B_f^H]=0$ for all $v_b$ and $f$. 
Notice that the face operator near the boundary has the same form as the bulk face operator, but it has a different meaning, since there is one component $\varphi^{(i)}_{\Hhat}$ acting on $K$. This point of view will become clearer in the next subsection for the lattice construction of general $H$-comodule algebra.

This construction is a natural generalization of the constructions for the group algebra case \cite{Bombin2008family, Beigi2011the}. When we choose $H=\mathbb{C}[G]$ and $K=\mathbb{C}[N]$ with $N$ a subgroup of $G$, our model reduces to the group-algebra boundary model. However, the viewpoint for constructing the boundary local operator algebra is different. In our construction, we do not treat the local edge projector $\Pi_K:H\to K$ as a stabilizer, and instead, we set the boundary edge space to be $K$. For the Hopf quantum double model, this approach is more natural, as the projector is not easily realized as a face operator with just one edge.
For the group algebra case, both constructions work, and their local operator algebra is related by a fusion-categorical Morita equivalence. This will be clarified in Proposition~\ref{thm:local-ribbon-alg}.

Let us take a closer look at two typical examples of this kind of construction. 
	
\begin{example}[Smooth boundary]\label{exp:smoothBd}
    For the smooth boundary, the corresponding Hopf subalgebra is $K=H$. 
    On a boundary site $s_b$, we have the local operator algebra $H$ generated by
\begin{equation}
    A^{k}(s_b)A^{l}(s_b)=A^{kl}(s_b).
\end{equation}
The boundary excitations are characterized by the UFC $\mathsf{Rep}(H)$.
\end{example}

\begin{example}[Rough boundary]\label{exp:roughBd}
For the rough boundary, the corresponding Hopf subalgebra is the trivial one $K=\mathbb{C}1_H$ and the Haar integral is $h_{K}=1_H$. In this case, boundary edges are fixed with label $1_H$ and $A^K_{v_b}=\id$.
Consider the boundary face operator as shown in Fig.~\ref{fig:smoothbd} acting on $x_1,x_2,x_3,x_4=1_H$ as 
\begin{equation}
    B^{\varphi}(s_b) |x_1x_2x_3x_4\rangle=  (B^{\varphi}(s_b)|x_1x_2x_3\rangle)\otimes |1_H\rangle.
\end{equation}
Thus the rough boundary can be obtained equivalently by removing all boundary edges and boundary vertices, and the boundary local operator algebra is generated by
\begin{equation}
    B^{\varphi}(s_b) B^{\psi}(s_b)=B^{\varphi \psi}(s_b).
\end{equation}
In this way, the boundary excitations are characterized by the UFC $\mathsf{Rep}(\hat{H})$.
\end{example}

The boundary topological phases for smooth and rough boundaries are Morita equivalent (their Drinfeld double is equivalent as UMTC).
To see this, let us recall the Kitaev-Kong construction of boundary theory of the Levin-Wen model for $2d$ topological phase \cite{Kitaev2012a}, for which the bulk phase is determined by a UFC $\EC$, and the boundary is characterized by an indecomposable $\EC$-module category $\EM$. The boundary excitation is given by the UFC $\mathsf{Fun}_{\EC}(\EM,\EM)$ of all $\EC$-module functors.
By transforming the basis of the Hopf algebraic quantum double model to the fusion basis, the quantum double model can be mapped to a Levin-Wen model \cite{Buerschaper2009mapping,buerschaper2013electric,balsam2012kitaevs}, see also next subsection. In this case, the input UFC for the bulk is $\EC=\mathsf{Rep}(H)$. 
For the smooth boundary, the module category is $\EM_s=\mathsf{Rep}(H)$, and for the rough boundary, the module category is $\EM_r=\mathsf{Vect}$.
It can be proved that there is a monoidal equivalence
\begin{equation}
    \mathsf{Fun}_{\mathsf{Rep}(H)}(\mathsf{Vect},\mathsf{Vect})\simeq \mathsf{Rep}(\hat{H}).
\end{equation}
This implies that the topological excitations of smooth and rough boundaries are Morita equivalent.

\subsubsection{Gapped boundary theory II} \label{sec:bdd-II}

We now consider how to construct the gapped boundary model for general $H$-comodule algebras.
An indirect construction based on the generalized quantum double is given in Appendix \ref{sec:bdII}, where we use a pair of Hopf algebras $J,K$ to realize the gapped boundary.
In this part, we discuss a direct construction proposed recently for a more general weak Hopf quantum double model in our work \cite{jia2023weak}.

From our previous construction for a Hopf subalgebra $K$, it was apparent that the crucial step was to use the Haar integral $h_K$ for the construction of local vertex stabilizers. However, for a general $H$-comodule algebra $\FA$, this approach is not feasible as it is not a Hopf algebra and hence lacks the concept of Haar integral. Therefore, it is necessary to generalize the notion of Haar integral to the general $H$-comodule algebra $\FA$.
As we have demonstrated in Ref.~\cite{jia2023weak}, the key concept for constructing local stabilizers for a general $H$-comodule algebra $\FA$ is the notion of a \emph{symmetric separability idempotent}, as defined in Refs.~\cite{aguiar2000note,koppen2020defects}.

\begin{definition}
A symmetric separability idempotent of an algebra $\mathfrak{A}$ is an element $\lambda=\sum_{\langle\lambda\rangle}\lambda^{\langle 1\rangle}\otimes \lambda^{\langle 2\rangle}\in \FA\otimes \FA$ that satisfies the following conditions:
\begin{itemize}
\item[(1)] $\sum_{\langle\lambda\rangle}x\lambda^{\langle 1\rangle}\otimes \lambda^{\langle 2\rangle} =\sum_{\langle\lambda\rangle} \lambda^{\langle 1\rangle}\otimes \lambda^{\langle 2\rangle} x$, for all $x\in \FA$.
\item[(2)] $\sum_{\langle\lambda\rangle}\lambda^{\langle 1\rangle} \lambda^{\langle 2\rangle} =1$.
\item[(3)] $\sum_{\langle\lambda\rangle}\lambda^{\langle 1\rangle}\otimes \lambda^{\langle 2\rangle}=\sum_{\langle\lambda\rangle}\lambda^{\langle 2\rangle}\otimes \lambda^{\langle 1\rangle}$.
\end{itemize}
\end{definition}

From the above definition we can derive that $\lambda$ is an idempotent of the enveloping algebra $\FA\otimes \FA^{\rm op}$. If $K$ is a Hopf algebra with Haar integral $h_K$, then it is straightforward to verify that $\lambda=\sum_{(h_K)} h_K^{(1)}\otimes S(h_K^{(2)})$ is a symmetric separability idempotent. 
In Ref.~\cite[Corollary 3.1]{aguiar2000note}, it is proved that the symmetric separability idempotent exists and is unique for a finite-dimensional semisimple algebra over an algebraically closed field of characteristic zero. 
Additionally, Ref.~\cite[Proposition 19]{koppen2020defects} shows that the following equation holds:
\begin{equation}\label{eq:lambda}
\sum \lambda^{\langle 1\rangle[0]}\otimes \lambda^{\langle 2\rangle[0]} \otimes \lambda^{\langle 1\rangle[1]}\lambda^{\langle 2\rangle[1]}=\sum \lambda^{\langle 1\rangle}\otimes \lambda^{\langle 2\rangle}\otimes 1,
\end{equation}
i.e., $\sum \beta(\lambda^{\langle 1\rangle})_{13}\beta(\lambda^{\langle 2\rangle})_{23} = \lambda_{12}$ in $\mathfrak{A}\otimes \mathfrak{A}\otimes H$ (the notation we adopt here is standard in Hopf algebra; the subscript indicates the tensor space in which the vectors lie).

If we designate the direction of all boundary edges such that the bulk is located to the right of the boundary edge\,\footnote{We choose this convention to simplify the notation of the coaction index as ``[1]'' instead of ``[-1]'' and to avoid cluttering of equations when doing calculations. However, we could have equivalently chosen the reverse direction. Henceforth, we will interchangeably use both choices whenever necessary.}, there exist only two configurations, depending on how the bulk edges are connected to the boundary vertex.
The choice of boundary edge directions depends on the choice of left or right $H$-comodule algebras for the input data of the gapped boundary.
For these two configurations, there are four corresponding sites.
We define the following operator $A^{a\otimes b}(s_b)$ for $a\otimes b\in \FA \otimes \FA^{\rm op}$ for four possible cases associated to a boundary vertex $v_b$:
\begin{align}
     \begin{aligned}
    \begin{tikzpicture}
				\draw[-latex,black,line width = 1.6pt] (0,0) -- (0,1);
    	        \draw[red,line width = 1pt] (0,1) -- (0.5,0.5); 
    			\draw[-latex,black,line width = 1.6pt] (0,1) -- (0,2); 
				\draw[-latex,black] (1,1) -- (0,1); 
				\node[ line width=0.2pt, dashed, draw opacity=0.5] (a) at (-0.4,0.5){$x$};
    		\node[ line width=0.2pt, dashed, draw opacity=0.5] (a) at (-0.4,1.5){$y$};
            \node[ line width=0.2pt, dashed, draw opacity=0.5] (a) at (0.5,1.3){$h$};
            \node[ line width=0.2pt, dashed, draw opacity=0.5] (a) at (0.7,0.5){$s_b$}; 
	\end{tikzpicture}  \end{aligned} & \quad \quad  A^{a\otimes b}(s_b)|x,y,h\rangle=\sum_{[b]} |ax,yb^{[0]},S(b^{[1]})h\rangle, \label{eq:bdd-sta-A1}\\
       \begin{aligned}
    \begin{tikzpicture}
				\draw[-latex,black,line width = 1.6pt] (0,0) -- (0,1);
    	        \draw[red,line width = 1pt] (0,1) -- (0.5,1.5); 
    			\draw[-latex,black,line width = 1.6pt] (0,1) -- (0,2); 
				\draw[-latex,black] (1,1) -- (0,1); 
				\node[ line width=0.2pt, dashed, draw opacity=0.5] (a) at (-0.4,0.5){$x$};
    		\node[ line width=0.2pt, dashed, draw opacity=0.5] (a) at (-0.4,1.5){$y$};
            \node[ line width=0.2pt, dashed, draw opacity=0.5] (a) at (0.7,1.4){$s_b$};
            \node[ line width=0.2pt, dashed, draw opacity=0.5] (a) at (0.5,0.75){$h$};
	\end{tikzpicture} \end{aligned} & \quad \quad A^{a\otimes b}(s_b)|x,y,h\rangle=\sum_{[a]} | a^{[0]}x, yb, a^{[1]}h\rangle,  \label{eq:bdd-sta-A2}\\
       \begin{aligned}
    \begin{tikzpicture}
				\draw[-latex,black,line width = 1.6pt] (0,0) -- (0,1); 
        	    \draw[red,line width = 1pt] (0,1) -- (0.5,0.5); 
    			\draw[-latex,black,line width = 1.6pt] (0,1) -- (0,2); 
				\draw[-latex,black] (0,1) -- (1,1); 
				\node[ line width=0.2pt, dashed, draw opacity=0.5] (a) at (-0.4,0.5){$x$};
    		\node[ line width=0.2pt, dashed, draw opacity=0.5] (a) at (-0.4,1.5){$y$};
            \node[ line width=0.2pt, dashed, draw opacity=0.5] (a) at (0.5,1.3){$h$};
            \node[ line width=0.2pt, dashed, draw opacity=0.5] (a) at (0.7,0.5){$s_b$};
	\end{tikzpicture}\end{aligned} & \quad\quad A^{a\otimes b} (s_b)|x,y,h\rangle=\sum_{[b]}|ax, yb^{[0]},hb^{[1]}\rangle, \label{eq:bdd-sta-A3} \\
      \begin{aligned} \begin{tikzpicture}
				\draw[-latex,black,line width = 1.6pt] (0,0) -- (0,1); 
        	    \draw[red,line width = 1pt] (0,1) -- (0.5,1.5); 
    			\draw[-latex,black,line width = 1.6pt] (0,1) -- (0,2); 
				\draw[-latex,black] (0,1) -- (1,1); 
				\node[ line width=0.2pt, dashed, draw opacity=0.5] (a) at (-0.4,0.5){$x$};
    		\node[ line width=0.2pt, dashed, draw opacity=0.5] (a) at (-0.4,1.5){$y$};
            \node[ line width=0.2pt, dashed, draw opacity=0.5] (a) at (0.5,0.75){$h$};
            \node[ line width=0.2pt, dashed, draw opacity=0.5] (a) at (0.7,1.4){$s_b$};
	\end{tikzpicture}  \end{aligned}  &\quad \quad A^{a\otimes b} (s_b)|x,y,h\rangle=\sum_{[a]}|a^{[0]}x, yb,hS(a^{[1]})\rangle. \label{eq:bdd-sta-A4} 
\end{align}
The convention here follows from that in Sec.~\ref{sec:bulkQD}.
Note that $x,y\in \FA$ and $a,b\in \FA$ are elements of the $H$-comodule algebra $\FA$. Since $\FA$ is an $H$-comodule algebra, $a$ and $b$ cannot directly act on $h\in H$ as there is no action from $\FA$ to $H$. However, by using the coaction map $\beta:\FA \to \FA \otimes H$, both $a^{[1]}$ and $b^{[1]}$ are elements of $H$ and hence can act on $h\in H$.
From the construction, it is clear that if the bulk is on the right-hand side of the boundary when moving along the positive direction of the boundary, the input $H$-comodule algebra must be chosen as a right $H$-comodule algebra.
If the input data is chosen as a left $H$-comodule algebra, we must change the direction of the boundary edges. This matches well with the construction of the Levin-Wen string-net boundary \cite{Kitaev2012a}, where the boundary direction determines the choice of left or right $\EC$-module structure $\EM$. Since $\EM$ is not the fusion category, there is no duality operation in $\EM$, and we cannot change the boundary direction arbitrarily. We will discuss this in detail later.

To introduce the boundary face operator, we need to introduce the action of $\hat{H}$ on the right $H$-comodule algebra $\FA$. This is given by ${\bar{T}}^{\varphi}_+x=\varphi\rightharpoonup x=\sum_{[x]}\varphi(x^{[1]}) x^{[0]}$, which is well-defined because $x^{[1]}\in H$ for $x\in \FA$. Then the boundary face operator $B^\varphi(s_b)$ for $\varphi\in\hat{H}$ is given as follows: 
\begin{equation}\label{eq:Bf}
	B^{\varphi}(s_b)
	\big{|}	\begin{aligned}
		\begin{tikzpicture}
			\draw[-latex,black] (-0.5,0.5) -- (0.5,0.5);
			\draw[-latex,line width = 1.6pt,black] (-0.5,-0.5) -- (-0.5,0.5); 
			\draw[-latex,black] (0.5,-0.5) -- (0.5,0.5); 
			\draw[-latex,black] (-0.5,-0.5) -- (0.5,-0.5); 
			\draw [fill = black] (0,0) circle (1.2pt);
			\draw[red,line width=1pt] (0,0) -- (-0.5,0.5);
			\node[ line width=0.2pt, dashed, draw opacity=0.5] (a) at (0.75,0){$k$};
			\node[ line width=0.2pt, dashed, draw opacity=0.5] (a) at (-0.75,0){$x$};
			\node[ line width=0.2pt, dashed, draw opacity=0.5] (a) at (0,-0.7){$h$};
			\node[ line width=0.2pt, dashed, draw opacity=0.5] (a) at (0,0.7){$l$};
		\end{tikzpicture}
	\end{aligned}   \big{ \rangle}     
	= 
	\sum_{(\varphi)}
	\big{|}	\begin{aligned}
	\begin{tikzpicture}
		\draw[-latex,black] (-0.5,0.5) -- (0.5,0.5);
		\draw[-latex,line width = 1.6pt,black] (-0.5,-0.5) -- (-0.5,0.5); 
		\draw[-latex,black] (0.5,-0.5) -- (0.5,0.5); 
		\draw[-latex,black] (-0.5,-0.5) -- (0.5,-0.5); 
		\draw [fill = black] (0,0) circle (1.2pt);
		\draw[red,line width=1pt] (0,0) -- (-0.5,0.5);
		\node[ line width=0.2pt, dashed, draw opacity=0.5] (a) at (1.3,0){$T^{\varphi^{(3)}}_{-}k$};
		\node[ line width=0.2pt, dashed, draw opacity=0.5] (a) at (-1.3,0){$\bar{T}^{\varphi^{(1)}}_{+}x$};
		\node[ line width=0.2pt, dashed, draw opacity=0.5] (a) at (0,-0.9){$T^{\varphi^{(2)}}_{-}h$};
		\node[ line width=0.2pt, dashed, draw opacity=0.5] (a) at (0,0.9){$T^{\varphi^{(4)}}_+l$};
	\end{tikzpicture}
\end{aligned}   \big{ \rangle},
\end{equation}
 according to the convention in Sec.~\ref{sec:bulkQD}.

We introduce the crossed product $(\FA\otimes \FA^{\rm op})\star	\hat{H}$ between $\FA\otimes \FA^{\rm op}$ and $\Hhat$ as follows: the underlying vector space is $(\mathfrak{A}\otimes \mathfrak{A})\otimes \hat{H}$, the multiplication is given by
\begin{equation}\label{eq:AHstar}
   ((a\otimes b)\star	\varphi )\cdot ((c\otimes d)\star	\psi)=  \sum (ac^{[0]}\otimes d^{[0]}b) \star	  \varphi(c^{[1]}\bullet d^{[1]})\psi,
\end{equation}
and the unit is $(1_\FA\otimes 1_\FA)\star\varepsilon$. Clearly, $((1_\FA\otimes 1_\FA)\star\varepsilon)\cdot ((a\otimes b)\star	\varphi) = ((a\otimes b)\star	\varphi)\cdot((1_\FA\otimes 1_\FA)\star	\varepsilon) = (a\otimes b)\star	\varphi$. The multiplication is associative: 
\begin{align}
    & \quad [((a\otimes b)\star	\varphi )\cdot ((c\otimes d)\star	\psi)] \cdot((e\otimes f)\star	\theta)  \nonumber \\
    & =  \sum ((ac^{[0]}\otimes d^{[0]}b) \star	\varphi(c^{[1]}\bullet d^{[1]})\psi)  \cdot((e\otimes f)\star	\theta) \nonumber  \\
    & = \sum  (ac^{[0]}e^{[0]}\otimes f^{[0]}d^{[0]}b) \star	  \varphi(c^{[1]} e^{[1]}\bullet f^{[1]}d^{[1]})\psi(e^{[2]}\bullet f^{[2]})\theta \\
    & = \sum ((a\otimes b)\star	\varphi)\cdot ((ce^{[0]}\otimes f^{[0]}d)\star	\psi(e^{[1]}\bullet f^{[1]})\theta)  \nonumber \\
    & = ((a\otimes b)\star	\varphi )\cdot [((c\otimes d)\star	\psi) \cdot((e\otimes f)\star	\theta)]. \nonumber 
\end{align}
Thus, $(\FA\otimes \FA^{\rm op})\star	\hat{H}$ is an algebra. 
By the identifications $\FA\otimes \FA^{\rm op} \cong (\FA\otimes \FA^{\rm op})\star	\varepsilon$ and $\hat{H}\cong (1_\FA\otimes 1_\FA)\star	\hat{H}$, the multiplication of $(\FA\otimes \FA^{\rm op})\star	\hat{H}$ is determined by the straightening formula 
\begin{equation}\label{eq:stra-a-1}
    \varphi(a\otimes b) = \sum (a^{[0]}\otimes b^{[0]})\varphi(a^{[1]}\bullet b^{[1]}). 
\end{equation}

\begin{proposition}\label{prop:bdLocalAlg}
    At a boundary site $s_b$, the boundary face and vertex operators generate the algebra $(\FA\otimes \FA^{\rm op})\star	 \Hhat$, with the straightening relation
    \begin{equation}\label{eq:stra-a-2}
        B^{\varphi}(s_b) A^{a\otimes b}(s_b) =\sum_{[a],[b]}  A^{a^{[0]}\otimes b^{[0]}}(s_b) B^{\varphi(a^{[1]}\bullet b^{[1]})}(s_b),
    \end{equation}
    where $a\otimes b \in \FA \otimes \FA^{\rm op}$ and $\varphi \in \hat{H}$. Therefore, the map 
    \begin{equation}
        \Psi_{s_b}:(\FA\otimes \FA^{\rm op})\star	\hat{H} \to \operatorname{End}(\mathcal{H}(s_b)),\quad (a\otimes b)\star	\varphi \mapsto A^{a\otimes b}(s_b)B^{\varphi}(s_b)
    \end{equation}
    is an algebra homomorphism. That is, every boundary site supports a representation of $(\FA\otimes\FA^{\rm op})\star	\Hhat$. 
\end{proposition}
\begin{proof}
Let us consider the following boundary configuration (the others are similar)
\begin{equation*}
    \begin{aligned}
		\begin{tikzpicture}
			\draw[-latex,black] (0.5,0.5) -- (-0.5,0.5);
			\draw[-latex,line width = 1.6pt,black] (-0.5,-0.5) -- (-0.5,0.5); 
                \draw[-latex,line width = 1.6pt,black] (-0.5,0.5) -- (-0.5,1.5); 
			\draw[-latex,black] (0.5,-0.5) -- (0.5,0.5); 
			\draw[-latex,black] (-0.5,-0.5) -- (0.5,-0.5); 
			\draw [fill = black] (0,0) circle (1.2pt);
			\draw[red,line width=1pt] (0,0) -- (-0.5,0.5);
			\node[ line width=0.2pt, dashed, draw opacity=0.5] (a) at (0.75,0){$k$};
			\node[ line width=0.2pt, dashed, draw opacity=0.5] (a) at (-0.75,0){$x$};
      	  \node[ line width=0.2pt, dashed, draw opacity=0.5] (a) at (-0.75,1){$y$};       
			\node[ line width=0.2pt, dashed, draw opacity=0.5] (a) at (0,-0.7){$h$};
			\node[ line width=0.2pt, dashed, draw opacity=0.5] (a) at (0,0.7){$l$};
                \node[ line width=0.2pt, dashed, draw opacity=0.5] (a) at (0.2,-0.2){$s_b$};
		\end{tikzpicture}
	\end{aligned} 
\end{equation*}
It is straightforward to compute that  
\begin{align*}
    &\quad B^\varphi(s_b)A^{a\otimes b}(s_b)|x,y,h,k,l\rangle \\
    & = \sum B^\varphi(s_b)|ax,yb^{[0]},h,k,S(b^{[1]})l\rangle \\
    & = \sum \varphi(a^{[1]}x^{[1]}S(h^{(1)})S(k^{(1)})S(S(b^{[2]})l^{(1)})) \\
    &\quad \quad |a^{[0]}x^{[0]},yb^{[0]},h^{(2)},k^{(2)},S(b^{[1]})l^{(2)}\rangle \\
    & = \sum \varphi(a^{[1]}x^{[1]}S(h^{(1)})S(k^{(1)})S(l^{(1)})b^{[2]}) \\
    &\quad \quad |a^{[0]}x^{[0]},yb^{[0]},h^{(2)},k^{(2)},S(b^{[1]})l^{(2)}\rangle \\
    & = \sum A^{a^{[0]}\otimes b^{[0]}}(s_b)\varphi(a^{[1]}x^{[1]}S(h^{(1)})S(k^{(1)})S(l^{(1)})b^{[1]}) \\
    &\quad \quad |x^{[0]},y,h^{(2)},k^{(2)},l^{(2)}\rangle \\
    & = \sum A^{a^{[0]}\otimes b^{[0]}}(s_b)B^{\varphi(a^{[1]}\bullet b^{[1]})}(s_b)|x,y,h,k,l\rangle. 
\end{align*}
The second assertion follows immediately from Eqs.~\eqref{eq:stra-a-1} and \eqref{eq:stra-a-2}. 
\end{proof}

To construct the lattice model for a gapped boundary, we introduce the boundary vertex operator $A^{\lambda}(s_b)$, where $\lambda$ is the symmetric separability idempotent of the input $H$-comodule algebra $\FA$. It can be verified straightforwardly from Eqs.~\eqref{eq:lambda} and \eqref{eq:stra-a-2} that $[A^{\lambda}(s_b),B^{\varphi_{\Hhat}}(s_b)]=0$ for all boundary sites $s_b$, and that all boundary vertex operators commute with each other.
The boundary Hamiltonian is of the form
\begin{equation}
  H[C(\partial \Sigma)]=\sum_{s_b}   (I-A^{\lambda}(s_b))+\sum_{s_b}(I-B^{\varphi_{\Hhat}}(s_b)).
\end{equation}
Notice that here we have used the boundary site $s_b$ to stress that the local stabilizers are constructed on boundary sites.
A direct calculation shows that when the input algebra $\FA=K$ for some Hopf subalgebra $K\subset H$,  $A^{\lambda}(s_b)$ is the same as the vertex operators $A^{h_K}_{v_b}$ that we constructed in the last subsection.
This is because in this case,  the symmetric separability idempotent $\lambda$ and the Haar integral $h$ are related by $\lambda=\sum_{(h)}h^{(1)}\otimes S(h^{(2)})$. 
Thus the boundary model in the last subsection is a special case of the model presented here.

\subsection{Boundary topological excitation} \label{sec:bdd-top-exc}

We have presented a Hamiltonian theory for a gapped boundary in the previous subsection. Let us now provide a detailed characterization of the topological excitations for these boundary models.
We first give an algebraic theory of the topological excitations based on an $H$-comodule algebra $\FA$ and modules over $\FA$ (see Appendix \ref{sec:app_hopf} for some necessary details of mathematical concepts used here),
and then we show the topological nature of the boundary Hamiltonian by mapping the Hopf quantum double boundary lattice model to a Levin-Wen string-net boundary whose topological nature has been shown previously in Ref.~\cite{Kitaev2012a}.
The main results are summarized in Tables \ref{tab:bdTop1} and \ref{tab:bdTop2}.

\subsubsection{Algebraic theory}

There are several different but equivalent ways to understand the topological excitations of the gapped boundary for $2d$ non-chiral topological order (for which quantum double phase is a special example).
Roughly they can be divided into two types (see Refs.~\cite{Kitaev2012a,Kong2014,Cong2017,Jia2022electric}):
\begin{itemize}
    \item From anyon condensation point of view, the bulk phase is characterized by a UMTC $\ED$. The gapped boundary phase is determined by a Lagrangian algebra $L\in \ED$, which plays the role of boundary vacuum charge. When a bulk anyon moves to the boundary, it is surrounded by the boundary vacuum, which mathematically results in a right $L$-module structure of the condensed anyon. The boundary phase is thus given by the category of all right $L$-modules in $\ED$, denoted as $\EB_L=\mathsf{Mod}_L(\ED)$, which is a UFC (the boundary is a $1d$ phase, and there is no braiding structure).
    The condensation process is given by a monoidal functor $\mathbf{Cond}_{L}:\ED \to \EB_L, X\mapsto X\otimes L$. Another different but equivalent construction of anyon condensation from bulk to the boundary is based on a Frobenius algebra $\mathfrak{F}\in \ED$. 
    In this case, the condensation becomes a two-step procedure: first, take the quotient category $\ED/\mathfrak{F}$; then take idempotent completion (Karoubi envelope) of $\ED/\mathfrak{F}$ to obtain the boundary phase $\EB_{\mathfrak{F}}=\EuScript{Q}(\ED,\mathfrak{F})$. See our previous work \cite{Jia2022electric} for a detailed discussion of the connection and the difference for both of the approaches. Notice that the Lagrangian algebra $L$ naturally has a Frobenius algebra structure  \cite{Kong2014}.\footnote{For definitions and basic properties of monoidal category and monoidal functor, see, e.g., \cite{etingof2016tensor}. }
    
    \item Since the bulk phase $\ED$ is non-chiral, this means that there is a UFC $\EC$ such that $\ED=\mathcal{Z}(\EC)$. Thus the boundary theory can also be understood at the $\EC$ level.
    Kitaev and Kong show that the boundary is characterized by an indecomposable $\EC$-module category $\EM$ \cite{Kitaev2012a} (notice that $\EM$ is in general not a monoidal category, it is only required to be a finite semisimple Abelian category).
    The boundary excitations are given by the category of $\EC$-module functors $\EC_{\EM}^{\vee}=\Fun_{\EC}(\EM,\EM)$, which is a UFC\,\footnote{When dealing with the boundary defects, it is more natural to set the boundary phase as $\Fun_{\EC}(\EM,\EM)^{\otimes\rm op}$. In this work, to avoid the cluttering of equations, we will omit this opposite tensor.}.
    The boundary vacuum charge is just the identity functor $\id:\EM\to\EM$, which plays the same role as Lagrangian algebra in $\ED$.
    Since the category $\ED=\mathcal{Z}(\EC)$ is equivalent to the category $\Fun_{\EC|\EC}(\EC,\EC)$ (the category of all $\EC|\EC$-bimodule category functors from $\EC$ to $\EC$), the anyon condensation from bulk to the boundary is thus a monoidal functor $\mathbf{Cond}_{\EM}:\Fun_{\EC|\EC}(\EC,\EC)\to \Fun_{\EC}(\EM,\EM)$, whose explicit form is given by \cite{Kong2014} as follows: 
    \begin{equation}
        (\EC\xrightarrow{F} \EC)\mapsto (\EM \simeq \EC\boxtimes \EM \xrightarrow{F\boxtimes \id_{\EM}} \EC\boxtimes \EM\simeq \EM).
    \end{equation}
    The Lagrangian algebra determined by $\mathbf{Cond}_{\EM} $ can be obtained by $L_{\EM}=\mathbf{Cond}_{\EM}^{\vee} (\one_{\EC})$, \emph{viz.}, acting the right adjoint functor of $\mathbf{Cond}_{\EM} $ on the tensor unit of $\EC$.
\end{itemize}

For our model, the bulk is determined by the UFC $\EC=\mathsf{Rep}(H)$ and the bulk topological phase is given by $\ED=\mathcal{Z}(\EC)\simeq \mathsf{Rep}(D(H))$.
As we have presented above, our boundary model is determined by an $H$-comodule algebra $\mathfrak{A}$.
It is natural to ask what the $\EC$-module category corresponding to this boundary is; the answer is 
\begin{equation}
   \text{boundary module category:}\quad {_{\mathfrak{A}}}\EM={_\mathfrak{A}}\mathsf{Mod},
\end{equation}
where ${_\mathfrak{A}}\mathsf{Mod}$ is the category of all left $\FA$-modules (this is mathematically guaranteed by the result in \cite{andruskiewitsch2007module}).
To show that this is indeed a $\mathsf{Rep}(H)$-module, first notice that ${_\mathfrak{A}}\mathsf{Mod}$ is finite semisimple, where the simple objects are just simple $\mathfrak{A}$-modules. Since $\dim \mathfrak{A} < \infty$, there are finite simple objects up to equivalence.
The action of $\mathsf{Rep}(H)$ over ${_\mathfrak{A}}\mathsf{Mod}$ is given by the usual tensor product $X\otimes M$ for $X\in \mathsf{Rep}(H)$ and $M\in {_\mathfrak{A}}\mathsf{Mod}$.
We only need to check that there is a left $\mathfrak{A}$-module structure over $X\otimes M$, more explicitly, a map $\mu_{X\otimes M}:\mathfrak{A}\otimes (X\otimes M)\to X\otimes M$.
In fact, it is easy to verify that this structure is given by (in diagrammatic representation)
\begin{equation}
    \mu_{X\otimes M}= \begin{aligned}
        \begin{tikzpicture}
			\draw[black, line width=1.0pt] (-0.1,-0.8) arc (180:360:0.3);
			\draw[black, line width=1.0pt] (1,0) arc (180:90:0.5);
			\draw[black, line width=1.0pt] (-0.1,0) arc (180:90:0.6);
			\draw[black, line width=1.0pt] (-0.1,0)--(-0.1,-0.8);
			\braid[
				width=0.5cm,
				height=0.3cm,
				line width=1.0pt,
				style strands={1}{black},
				style strands={2}{black}] (Kevin)
				s_1^{-1} ;
	    	\draw[black, line width=1.0pt] (1.0,-0.8)--(1.0,-1.35);
	    	\draw[black, line width=1.0pt] (0.2,-1.1)--(0.2,-1.35);
	    	\draw[black, line width=1.0pt] (1.5,0.9)--(1.5,-1.35);
	    	 \draw[black, line width=1.0pt] (0.5,0)--(0.5,0.9);
	    	\node[ line width=0.2pt, dashed, draw opacity=0.5] (a) at (0.2,-1.6){$\mathfrak{A}$};
	    	\node[ line width=0.2pt, dashed, draw opacity=0.5] (a) at (1,-1.6){$X$};
	    	\node[ line width=0.2pt, dashed, draw opacity=0.5] (a) at (1.5,-1.6){$M$};
	    	\node[ line width=0.2pt, dashed, draw opacity=0.5] (a) at (0.55,1.1){$X$};
	    	\node[ line width=0.2pt, dashed, draw opacity=0.5] (a) at (1.5,1.1){$M$};
	     	\node[ line width=0.2pt, dashed, draw opacity=0.5] (a) at (0.7,-1){$\mathfrak{A}$};
			\node[ line width=0.2pt, dashed, draw opacity=0.5] (a) at (-0.4,-1){$H$};
			\node[ line width=0.2pt, dashed, draw opacity=0.5] (a) at (0.2,-1.1){$\bullet$};
			\node[ line width=0.2pt, dashed, draw opacity=0.5] (a) at (0.5,0.6){$\bullet$};
			\node[ line width=0.2pt, dashed, draw opacity=0.5] (a) at (1.5,0.5){$\bullet$};
			\end{tikzpicture}
    \end{aligned},
\end{equation}
\emph{viz.}, $\mu_{X\otimes M}=(\mu_X\otimes \mu_M)\comp(\id_H\otimes \tau_{\mathfrak{A},X} \otimes \id_{M}) \comp (\beta_{\mathfrak{A}} \otimes \id_X\otimes \id_M)$, where $\mu_X$ and $\mu_M$ are the $H$-module structure and $\mathfrak{A}$-module structure maps of $X$ and $M$ respectively, $\beta_{\mathfrak{A}}$ is the $H$-comodule structure map of $\mathfrak{A}$, and $\tau_{\mathfrak{A},X}$ is the swap map.
We say that $\mathfrak{A}$ is exact if ${_{\mathfrak{A}}}\mathsf{Mod}_{\mathfrak{A}}$ is exact; when $\mathfrak{A}$ is exact, the category ${_{\mathfrak{A}}}\mathsf{Mod}_{\mathfrak{A}}$ is semisimple.
When $\mathfrak{A}$ is $H$-indecomposable,  the category ${_{\mathfrak{A}}}\mathsf{Mod}_{\mathfrak{A}}$ is indecomposable.

Since the $\mathfrak{A}$-boundary excitation can be regarded as a point defect between two $\mathfrak{A}$-boundaries, the boundary anyon is thus described by a functor from the $\mathsf{Rep}(H)$-module category ${_\mathfrak{A}}\mathsf{Mod}$ to itself.
It can be proved \cite[Theorem 1.21]{andruskiewitsch2007module} that for each such functor $F$ there is an $\mathfrak{A}|\mathfrak{A}$-bimodule $Y$ such that $F$ is naturally isomorphic to $Y\otimes_{\mathfrak{A}} \bullet$.
Thus the boundary anyons are classified by $\mathfrak{A}|\mathfrak{A}$-bimodules $Y$.
A left $H$-covariant $\FA|\FA$-bimodule is an $\FA|\FA$-bimodule equipped with a left $H$-coaction such that the coaction is a morphism of $\FA|\FA$-bimodules.
To be more clear, for two $H$-comodule algebras $\FA$ and $\FB$, the $\FA|\FB$ biaction on $H\otimes M$ (with $M$ an $\FA|\FB$-bimodule) is defined as follows
\begin{equation}
    a\cdot(h\otimes m)\cdot b= \sum_{[a],[b]} a^{[-1]}hb^{[-1]} \otimes a^{[0]}mb^{[0]}.
\end{equation}
An $H$-covariant $\FA|\FB$-module $M$ is left $H$-comodule with coaction $\beta: M\to H\otimes M$ such that $\beta$ is an $\FA|\FB$-bimodule map.
This means that 
\begin{equation}
   \sum_{[a\cdot m \cdot b]} (a\cdot m \cdot b)^{[-1]} \otimes  (a\cdot m \cdot b)^{[0]}= \sum_{[a],[m],[b]} a^{[-1]}\cdot m^{[-1]} \cdot b^{[-1]} \otimes a^{[0]}\cdot m^{[0]} \cdot b^{[0]}.
\end{equation}
We denote the category of all $H$-covariant $\FA|\FB$-bimodules as ${^H_{\mathfrak{A}}}\mathsf{Mod}_{\mathfrak{B}}$\,\footnote{For right $H$-comodule algebras, we can similarly define ${_{\mathfrak{A}}}\mathsf{Mod}_{\mathfrak{B}}^H$.}.
It can be proved that the module functor category from ${_\mathfrak{A}}\mathsf{Mod}$ to itself is equivalent to ${^H_{\mathfrak{A}}}\mathsf{Mod}_{\mathfrak{A}}$ \cite{andruskiewitsch2007module}.
This implies that the $\mathfrak{A}$-boundary topological phase is given by 
\begin{equation}
 \text{boundary excitation:}\quad   \EB_{\mathfrak{A}}\simeq \mathsf{Fun}_{\mathsf{Rep}(H)}({_\mathfrak{A}}\mathsf{Mod},{_\mathfrak{A}}\mathsf{Mod})\simeq {^H_{\mathfrak{A}}}\mathsf{Mod}_{\mathfrak{A}}.
\end{equation}
This point of view can be naturally generalized to the boundary defect between two boundaries determined by two $H$-comodule algebras $\mathfrak{A}$ and $\mathfrak{B}$. The direction of the boundary now matters, an $\mathfrak{A}|\mathfrak{B}$-defect is different from a $\mathfrak{B}|\mathfrak{A}$-defect. 
For boundary defect from $\mathfrak{A}$-boundary to $\mathfrak{B}$-boundary, the boundary defects are classified by functors from the $\mathsf{Rep}(H)$-module category ${_{\mathfrak{A}}}\mathsf{Mod}$ to the $\mathsf{Rep}(H)$-module category ${_{\mathfrak{B}}}\mathsf{Mod}$. Equivalently, they are classified by $\mathfrak{B}|\mathfrak{A}$-bimodules $Y\in {_{\mathfrak{B}}}\mathsf{Mod}_{\mathfrak{A}}$ with the associated functor given by $Y\otimes_{\mathfrak{A}}\bullet$. Similar to the boundary excitations,  $Y$ needs to be $H$-covariant.
This result of a classification of boundary defects is a generalization of the Eilenberg-Watts theorem \cite{eilenberg1960abstract,watts1960intrinsic}.
To summarize, we have the following result:

\begin{proposition} \label{prop:bdd-top-exc1}
For a quantum double phase with bulk Hopf algebra $H$:
\begin{enumerate}
    \item For a gapped boundary determined by a left $H$-comodule algebra $\mathfrak{A}$, the boundary excitations are classified by left $H$-covariant $\FA|\FA$-bimodules.
\item For two boundaries determined by $H$-comodule algebras $\mathfrak{A}$ and $\FB$, the boundary defects from $\FA$-boundary to $\FB$-boundary are classified by $H$-covariant $\FB|\FA$-bimodules.
\end{enumerate}
\end{proposition}

\begin{table}[t]
\centering \small 
\begin{tabular} {|l|c|c|} 
\hline
   &Hopf quantum double model & String-net model   \\ \hline
 Bulk  & Hopf algebra $H$ & UFC $\EC=\mathsf{Rep}(H)$  \\ \hline
 Bulk phase  & 
 $\ED=\mathsf{Rep}(D(H))$ & $\mathsf{Fun}_{\mathsf{Rep}(H)|\mathsf{Rep}(H)}(\mathsf{Rep}(H),\mathsf{Rep}(H)) $  \\ \hline
 Boundary & $H$-comodule algebra $\mathfrak{A}$ & $\mathsf{Rep}(H)$-module category ${_{\mathfrak{A}}}\EM={_{\mathfrak{A}}}\mathsf{Mod}$ \\\hline
 Boundary phase & $\EB\simeq {^H_{\mathfrak{A}}}\mathsf{Mod}_{\mathfrak{A}}$ & $\mathsf{Fun}_{\mathsf{Rep}(H)}({_{\mathfrak{A}}}\EM,{_{\mathfrak{A}}}\EM)$ \\\hline
Boundary defect & ${^H_{\mathfrak{B}}}\mathsf{Mod}_{\mathfrak{A}}$ & $\mathsf{Fun}_{\mathsf{Rep}(H)}({_{\mathfrak{A}}}\EM,{_{\mathfrak{B}}}\EM)$ \\\hline
\end{tabular}
\caption{The dictionary between left $H$-comodule algebra description of Hopf quantum double boundary and string-net boundary.\label{tab:bdTop1}}
\end{table}

We have shown that from our boundary model, there is a corresponding module category.
Conversely, for a given $\mathsf{Rep}(H)$-module $\EM$, we can also find a corresponding $H$-comodule algebra $\mathfrak{A}$ such that $\EM\simeq {_{\mathfrak{A}}}\mathsf{Mod}$.
The main mathematical tool for doing this is the internal Hom proposed in Refs.~\cite{etingof2003finite,ostrik2003module}.
We say that $M\in \EM$ generates $\EM$ if for any $N\in \EM$, there exists $X\in \EC$ such that $\Hom(X\otimes M,N)\neq 0$.
For indecomposable $\EM$, all non-zero simple objects are generators.
For finite $\EC$-module category $\EM$, the functor $\Hom(\bullet\otimes M_1,M_2)$ from $\EC$ to $\mathsf{Vect}$ is a representable functor, and 
the object in $\EC$ that represents this functor is called the internal Hom and denoted as $\underline{\Hom}(M_1,M_2)$; more explicitly,
\begin{equation}\label{eq:HomInt1}
    \Hom(X\otimes M_1,M_2)\cong \Hom(X,\underline{\Hom}(M_1,M_2))
\end{equation}
for all $X\in \EC$ and $M_1,M_2\in \EM$.
In diagrammatic representation, this reads
\begin{equation} \label{eq:IntHom}
   \begin{aligned}
        \begin{tikzpicture}
        	\draw (-0.1,0) rectangle (0.4,0.4);
			\draw[black, line width=1.0pt] (0.3,-0.5)--(0.3,0);
			\draw[black, line width=1.0pt] (0,-0.5)--(0,0);
			\draw[black, line width=1.0pt] (0.15,0.4)--(.15,0.9);
   	   	    \node[ line width=0.2pt, dashed, draw opacity=0.5] (a) at (-0.3,-0.3){$X$};
   	        \node[ line width=0.2pt, dashed, draw opacity=0.5] (a) at (0.7,-0.3){$M_1$};
   	        \node[ line width=0.2pt, dashed, draw opacity=0.5] (a) at (0.7,0.7){$M_2$};
   	        \node[ line width=0.2pt, dashed, draw opacity=0.5] (a) at (0.15,0.2){$u$};
			\end{tikzpicture}
    \end{aligned}  
    \xleftrightarrow[]{~\simeq~~}\ \ 
       \begin{aligned}
        \begin{tikzpicture}
        	\draw (-0.1,0) rectangle (0.4,0.4);
			\draw[black, line width=1.0pt] (0.15,0.4)--(.15,0.9);
			\draw[black, line width=1.0pt] (0.15,-0.5)--(.15,0);
   	        \node[ line width=0.2pt, dashed, draw opacity=0.5] (a) at (0.7,-0.3){$X$};
   	        \node[ line width=0.2pt, dashed, draw opacity=0.5] (a) at (1.5,0.7){$\underline{\Hom}(M_1,M_2)$};
   	        \node[ line width=0.2pt, dashed, draw opacity=0.5] (a) at (0.15,0.2){$\tilde{u}$};
			\end{tikzpicture}
    \end{aligned}. 
\end{equation}
Choosing $M\in \EM$ which generates $\EM$, then 
the internal Hom $\mathfrak{M}:=\underline{\Hom}(M,M)$ is an algebra in $\EC$.
To see this, set $X=\mathfrak{M}\otimes \mathfrak{M}$ and $M_1=M_2=M$ in Eq.~\eqref{eq:HomInt1}, we obtain 
\begin{equation}\label{eq:IntHom2}
    \Hom(\mathfrak{M}\otimes \mathfrak{M},\mathfrak{M})\simeq \Hom ((\mathfrak{M}\otimes \mathfrak{M})\otimes M,M).
\end{equation}
And set $X=\mathfrak{M}$ in Eq.~\eqref{eq:IntHom} and $M_1=M_2=M$, we obtain a map $\mu_M:\mathfrak{M}\otimes M\to M$ which corresponds to $\id_{\mathfrak{M}}:\mathfrak{M}\to \mathfrak{M}$.
The multiplication map of the algebra $\mathfrak{M}$ is, under the isomorphism \eqref{eq:IntHom2}, given by 
\begin{equation}
   \mu_{\mathfrak{M}}= \mu_M\comp (\id_{\mathfrak{M}} \otimes \mu_M)\comp \alpha_{\mathfrak{M},\mathfrak{M},M},
\end{equation}
where $\alpha_{\mathfrak{M},\mathfrak{M},M}$ is the associator mapping $(\mathfrak{M}\otimes \mathfrak{M})\otimes M$ to $\mathfrak{M}\otimes( \mathfrak{M}\otimes M)$.
It can be proved that the $\EC$-module category $\EM$ is equivalent to the category of all right $\mathfrak{M}$-modules in $\EC$, $\EM\simeq \mathsf{Mod}_{\mathfrak{M}}(\EC)=:\EC_{\mathfrak{M}}$ (see \cite{ostrik2003module}). Here an $\FM$-module $M\in \EC_{\FM}$ is also a left $H$-module, and the $\mathfrak{M}$-action and $H$-action are compatible, thus we can also denote this category as ${_H}\Mod_{\mathfrak{M}}$.

The opposite algebra $\mathfrak{M}^{\rm op}$ is an $H^{\rm cop}$-module algebra. Taking the smash product of $\mathfrak{M}^{\rm op}$ and $H^{\rm cop}$ we obtain an $H$-comodule algebra (see Appendix \ref{sec:app_hopf} for the definition of smash product)
\begin{equation} \label{eq:mod-comod}
    \mathfrak{A}:=\mathfrak{M}^{\rm op} \# H^{\rm cop}.
\end{equation}
It is showed in \cite{andruskiewitsch2007module} that
${_{\mathfrak{A}}}\mathsf{Mod}\simeq \EC_{\FM} \simeq \EM$.
To summarize, for any $\EC$-module category $\EM$, we can find a corresponding $H$-comodule algebra $\mathfrak{A}$ such that this $\EM$-module boundary can be regarded as an $\mathfrak{A}$-boundary of the Hopf quantum double phase.
Thus the boundary can be equivalently characterized by an $H$-module algebra $\mathfrak{M}$ or an $H$-comodule algebra $\mathfrak{A}$. 
Note that for a right $\FM$-module $M$, the left $\FA$-action on $M$ is given by
\begin{equation}
    (x\# h)\triangleright m =(h\triangleright m) \triangleleft x
\end{equation}
for all $x\in \FM$, $h\in H$ and $m\in M$.

\begin{proposition} \label{prop:bdd-top-exc2}
For a quantum double phase with bulk Hopf algebra $H$: 
\begin{enumerate}
    \item For a gapped boundary determined by $\FM$, the boundary excitations are classified by $\FM|\FM$-bimodules in $\mathsf{Rep}(H)$.
    \item For two gapped boundaries determined by  $\FM$ and $\FN$, the boundary defects from $\FM$-boundary to $\FN$-boundary are classified by $\FN|\FM$-bimodules in $\mathsf{Rep}(H)$.
\end{enumerate}
\end{proposition}

For the bulk Hopf algebra $H$, the boundary $H$-comodule algebra $\mathfrak{A}$ and the corresponding boundary $H$-module algebra $\mathfrak{M}$ are also correlated by the Yan-Zhu stabilizer \cite{yan1998stabilizer}. It is proved \cite{andruskiewitsch2007module} that,
as $\hat{H}$-comodule algebra (equivalently, as $H$-module algebra)
\begin{equation}
    \mathfrak{M}=\underline{\Hom}(M,M) \cong \operatorname{Stab}_{\mathfrak{A}}(M),
\end{equation}
where $M$ is a nonzero $\mathfrak{A}$-module and $\operatorname{Stab}_{\mathfrak{A}}(M)$ is the Yan-Zhu stabilizer.

\begin{table}[t]
\centering \small 
\begin{tabular} {|l|c|c|} 
\hline
   &Hopf quantum double model & String-net model   \\ \hline
 Bulk  & Hopf algebra $H$ & UFC $\EC=\mathsf{Rep}(H)$  \\ \hline
 Bulk phase  & 
 $\ED=\mathsf{Rep}(D(H))$ & $\mathsf{Fun}_{\mathsf{Rep}(H)|\mathsf{Rep}(H)}(\mathsf{Rep}(H),\mathsf{Rep}(H)) $  \\ \hline
 Boundary & $H$-module algebra $\mathfrak{M}$ & $\EC_{\FM}$ \\\hline
 Boundary phase & $\EB\simeq  {_{\FM}}\EC_{\FM}$ & $\mathsf{Fun}_{\EC}(\EC_{\FM},\EC_{\FM})$ \\\hline
Boundary defect & ${_{\FN}}\EC_{\FM}$ & $\mathsf{Fun}_{\EC}(\EC
_{\FM},\EC_{\FN})$ \\\hline
\end{tabular}
\caption{The dictionary between $H$-module algebra description of Hopf quantum double boundary and string-net boundary.\label{tab:bdTop2}}
\end{table}

We have established a complete theory about the gapped boundary of the Hopf quantum double phase.
However, it is possible that two different $H$-comodule algebras give the same topological boundary $\EM$.
This has been reflected in the above discussion of constructing the $H$-comodule algebra $\mathfrak{A}$ from a given $\EC$-module $\EM$: by choosing different non-zero simple objects $M\in\EM$, we obtain different $H$-module algebras $\mathfrak{M}=\underline{\Hom}(M,M)$, and they all give the same $\EC$-module category $\EM$.
In Hopf algebraic language, this is captured by the equivariant Morita equivalence of $H$-comodule algebras (see Appendix \ref{sec:app_hopf} for a detailed discussion).
If two $H$-comodule algebras $\mathfrak{A}$ and $\mathfrak{B}$ are equivariantly Morita equivalent, then ${_{\mathfrak{A}}}\mathsf{Mod} \simeq {_{\mathfrak{B}}}\mathsf{Mod} $, \emph{viz.}, they give the same topological boundary.
For $H$-module algebra description, the same result holds.

\begin{proposition}
   Two boundary models determined by $H$-comodule algebras  $\mathfrak{A}$ and $\mathfrak{B}$ (resp. $H$-module algebras $\mathfrak{M}$ and $\mathfrak{N}$) are topologically equivalent if and only if $\mathfrak{A}$ and $\mathfrak{B}$ (resp. $\mathfrak{M}$ and $\mathfrak{N}$) are equivariantly Morita equivalent.
\end{proposition}

Notice that from Yan-Zhu stabilizer and Eq.~\eqref{eq:mod-comod}, we will obtain an $H$-comodule algebra $\FA'$. This $\FA'$ may be different from $\FA$, but they are Morita equivalent, thus they give the same gapped boundary.

To summarize, a gapped boundary for Hopf quantum double phase is determined by Morita equivalent class of $H$-comodule algebras or $H$-module algebras.   
The module category of the corresponding boundary is given by the module category over these algebras. The topological excitations are classified by bimodules over these algebras.

\begin{example}
Consider the smooth and rough boundaries in Examples \ref{exp:smoothBd} and \ref{exp:roughBd}.
For smooth boundary, $\mathfrak{M}_s=\mathbb{C}$, and the boundary module category is $\EM_{s}\simeq \mathsf{Mod}_{\mathbb{C}}(\mathsf{Rep}(H))\simeq \mathsf{Rep}(H)$; for rough boundary, $\mathfrak{M}=\hat{H}$, and the boundary module is $\EM_r\simeq \mathsf{Mod}_{\hat{H}}(\mathsf{Rep}(H))\simeq \mathsf{Vect}$. This matches well with the result of lattice construction.  
\end{example}

\begin{example}
    It is worth discussing how to recover the result of boundary theory for finite group quantum double $H=\mathbb{C}[G]$ in our framework.
    From Refs.~\cite{Bombin2008family,Beigi2011the, Kitaev2012a,Cong2017,etingof2003finite,ostrik2003module,andruskiewitsch2007module,natale2017equivalence},  the boundary is described by a pair $(N,\sigma)$ where $N\leq G$ is a subgroup up to conjugacy and $\sigma \in H^2(N,\mathbb{C}^{\times})$ is a $2$-cocycle.
    The twisted group algebra $\mathbb{C}_{\sigma}[N]$ is a $\mathbb{C}[G]$-comodule algebra via $\beta_{\mathbb{C}_{\sigma}[N]}(x)=x\otimes x$ for all $x\in N$.
    The bulk UFC is $\EC=\mathsf{Rep}(G)$ and the boundary is the $\mathsf{Rep}(G)$-module category $\EM(N,\sigma)=\mathsf{Rep}(\tilde{N})$, where $\tilde{N}$ is the extension of $N$ under $\sigma$.
    The $H$-comodule algebra for this boundary is $\FA=\mathbb{C}_{\sigma}[N]$.
    The $H$-module algebra for this boundary is $\FM=\operatorname{Stab}_{\mathbb{C}_{\sigma}[N]}(M)=\mathbb{C}[G]\otimes_{\mathbb{C}[N]} \operatorname{End}(M)$ for some $\mathbb{C}_{\sigma}[N]$-module $M$.
\end{example}

\subsubsection{Mapping quantum double boundary to string-net boundary}

We have proved the topological nature of our Hopf quantum double boundary from an algebraic theory perspective. Let us now prove it from the perspective of lattice Hamiltonian.
We will show how to map a quantum double boundary to a string-net boundary (which realizes an extended Turaev-Viro TQFT, thus is topological), and also discuss, how a string-net boundary can be mapped to a quantum double boundary.

We first introduce the Fourier transformation for a general Hopf algebra $H$ with Haar integral $h_H$:
\begin{equation}
    |\nu; a,b\rangle=|[\nu]_{ab}\rangle = \sqrt{\frac{\dim \nu}{\dim H}} \sum_{(h_H)}D^{\nu}(h_H^{(1)})_{ab}h_H^{(2)},
\end{equation}
where $D^{\nu}$ is a fixed matrix representation of $\nu \in \operatorname{Irr}(H)$, and $a,b=1,\cdots, \dim \nu$. Here $\operatorname{Irr}(H)$ is the set of all equivalence classes of irreducible representations of $H$. 
This orthonormal basis is usually called the fusion basis. We have (Peter–Weyl theorem, a.k.a., Artin-Wedderburn theorem)
\begin{equation}
    H\cong \sum_{\nu \in \operatorname{Irr}(H)} V_{\nu}\otimes V_{\nu^{\vee}},
\end{equation}
and thus $\dim H=\sum_{\nu\in \operatorname{Irr}(H)}(\dim \nu)^2$.  So they form a complete basis.
For the convenience of later discussion, we assume that the Hopf quantum double model is built on a trivalent lattice with boundary\,\footnote{This is just for convenience for discussion; for any other lattice, the generalization is straightforward.}.
The general approach to mapping the quantum double model on a closed surface to an extended string-net model is presented in \cite{Buerschaper2009mapping,buerschaper2013electric,Hu2018full}. A mathematical proof that shows Hopf quantum double model gives the same topological invariant as Turaev-Viro TQFT is presented in \cite{balsam2012kitaevs}.
The input data of an extended string-net model is $(\EC,\omega_{\EC})_{\rm SN}$ with $\EC$ a UFC and $\omega_{\EC}$ a fiber functor, \emph{viz.}, a monoidal functor $\omega_{\EC}:\EC\to \mathsf{Vect}$.
The input data of a Hopf quantum double model is a Hopf algebra $(H)_{\rm QD}$. By setting $\EC=\mathsf{Rep}(H)$, the Hopf quantum double model can be mapped into a string-net model.  
Using Tannaka–Krein reconstruction \cite{etingof2016tensor}, and by setting $H=\operatorname{End}(\omega_{\EC})$ (the natural endomorphism of $\omega_{\EC}$ which has a canonical Hopf algebra structure), the extended string-net model can be mapped into a Hopf quantum double model.

For our case, we only need to consider how to map the boundary of the Hopf quantum double model to the extended string-net boundary.
Recall that extended string-net boundary is determined by a $\EC$-module category $\EM$ together with an additive functor $\omega_{\EM}: \EM \to \mathsf{Vect}$.
They satisfy\,\footnote{Notice that two tensor products are different. The upper one is for $\EC$-module structure over $\EM$, while the lower one is for $\mathsf{Vect}$-module structure over $\mathsf{Vect}$.} 
\begin{equation}\label{eq:CM}
    \begin{tikzcd}
\EC \times \EM \arrow[r, "\otimes"] \arrow[d,"\omega_{\EC}\times \omega_{\EM}"]
& \EM \arrow[d, "\omega_{\EM}"] \\
\mathsf{Vect}\times \mathsf{Vect} \arrow[r, "\otimes" ]
& \mathsf{Vect}
\end{tikzcd}.
\end{equation}

From the previous discussion of algebraic theory we know that $\EM={_{\FA}}\mathsf{Mod}$. We can choose $\omega_{\EM}$ as a forgetful functor. We will establish the following correspondence:

\begin{proposition}
   The Hopf quantum double model with boundary $(H,\mathfrak{A})_{\rm QD}$ is equivalent to the extended string-net model with boundary $(\EC,\omega_{\EC},\EM,\omega_{\EM})_{\rm SN}$:
   \begin{enumerate}
       \item A Hopf quantum double model can be mapped to an extended string-net model by setting: bulk $\EC=\mathsf{Rep}(H)$, and boundary $\EM={_{\mathfrak{A}}}\mathsf{Mod}$.
       \item An extended string-net model with boundary $(\EC,\omega_{\EC},\EM,\omega_{\EM})$ can be mapped to a Hopf quantum double model by setting: $H=\operatorname{End}(\omega_{\EC})$ and $\FA=\FM^{\rm op}\# H^{\rm cop}$ with $\FM=\underline{\rm Hom}(M,M)$ for some nonzero $M\in \EM$.
   \end{enumerate}
\end{proposition}

\begin{proof}

A rigorous proof for the general $H$-comodule algebra is complicated and lengthy, it will be given in our forthcoming work \cite{jia2023extended}. Here we only prove that the commutative diagram in Eq.~\eqref{eq:CM}
holds when we set $\EC=\mathsf{Rep}(H)$ and $\EM={_{\FA}}\mathsf{Mod}$.
When acting on objects $X \in \mathsf{Rep}(H)$ and $M\in {_{\FA}}\mathsf{Mod}$, for the upper right path, the tensor product is given by $X\otimes_{\mathbb{C}} M$, then acting $\omega_{{_{\FA}}\mathsf{Mod}}$ on it we obtain $ \omega_{{_{\FA}}\mathsf{Mod}}\comp \otimes (X,M)= X\otimes_{\mathbb{C}} M$.
For the lower left path, it is clear that we still obtain $\otimes \comp (\omega_{\EC}\times \omega_{{_{\FA}}\mathsf{Mod}}) (X,M)=X\otimes_{\mathbb{C}} M$.
Now consider the morphisms $(f,g): X\times M \to Y\times N$, it is clear that  $\otimes \comp (\omega_{\EC}\times \omega_{{_{\FA}}\mathsf{Mod}})(f,g)= \omega_{{_{\FA}}\mathsf{Mod}}\comp \otimes(f,g)$. 

The following proof is for the special case that $\FA=K\leq H$, where all data are clear for both the extended string-net model and for the Hopf quantum double model.

1.
Suppose that the boundary is characterized by an $H$-comodule algebra $\mathfrak{A}=K \leq H$. For a boundary vertex $v_b$ connecting two boundary edges $e_1,e_2$ and a bulk edge $e_3$, the boundary edge space is $\mathcal{H}_{e_1}=\mathcal{H}_{e_2}=\mathfrak{A}$ and the bulk edge space is $\mathcal{H}_{e_3}=H$. Thus $\mathcal{H}(v_b)=H\otimes \mathfrak{A}\otimes \mathfrak{A}$.
By taking the Fourier transform for $H$ and $\mathfrak{A}$, we obtain
\begin{equation}
  \Big{ | }   \begin{aligned}
           \begin{tikzpicture}
			\draw[black, line width=1.5pt] (0,-0.7)--(0,0.7);
			\draw[-latex,black, line width=1.5pt] (0,-0.7)--(0,-0.2);
			\draw[-latex,black, line width=1.5pt] (0,0)--(0,0.5);
			\draw[black, line width=1.0pt] (-0.6,-0.6)--(0,0);
			\draw[-latex,black, line width=1.0pt] (-0.6,-0.6)--(-0.2,-0.2);
			\node[ line width=0.2pt, dashed, draw opacity=0.5] (a) at (-0.5,-0.1){$g_3$};
   	        \node[ line width=0.2pt, dashed, draw opacity=0.5] (a) at (0.4,-0.4){$x_1$};
   	        \node[ line width=0.2pt, dashed, draw opacity=0.5] (a) at (0.4,0.4){$y_2$};
   	        \node[ line width=0.2pt, dashed, draw opacity=0.5] (a) at (0.2,0){$v_b$};
			\end{tikzpicture}
    \end{aligned}
\Big{ \rangle} = \sum_{[\nu_1]_{a_1 b_1},[\nu_2]_{a_2b_2},[\mu_3]_{a_3b_3}} \Psi([\nu_1]_{a_1 b_1},[\nu_2]_{a_2b_2},[\mu_3]_{a_3b_3}) \Big{ | }   \begin{aligned}
           \begin{tikzpicture}
			\draw[black, line width=1.5pt] (0,-0.7)--(0,0.7);
			\draw[-latex,black, line width=1.5pt] (0,-0.7)--(0,-0.2);
			\draw[-latex,black, line width=1.5pt] (0,0)--(0,0.5);
			\draw[black, line width=1.0pt] (-0.6,-0.6)--(0,0);
			\draw[-latex,black, line width=1.0pt] (-0.6,-0.6)--(-0.2,-0.2);
			\node[ line width=0.2pt, dashed, draw opacity=0.5] (a) at (-1,-0.2){$[\mu_3]_{a_3b_3}~$};
   	        \node[ line width=0.2pt, dashed, draw opacity=0.5] (a) at (0.9,-0.4){$[\nu_1]_{a_1 b_1}$};
   	        \node[ line width=0.2pt, dashed, draw opacity=0.5] (a) at (0.9,0.4){$[\nu_2]_{a_2b_2}$};
   	        \node[ line width=0.2pt, dashed, draw opacity=0.5] (a) at (0.2,0){$v_b$};
			\end{tikzpicture}
    \end{aligned}
\Big{ \rangle}. 
\end{equation}
The coefficient can be obtained by taking the inner product of each, e.g., $\langle \nu;a,b|g\rangle=\sqrt{\frac{\dim \nu}{\dim H}} \sum_{(h_H)}D^{\nu}(h_H^{(1)})_{ab}^*\langle h_H^{(2)}|g\rangle$\,\footnote{Recall that the inner product of $H$ is given by $\langle h|g\rangle =\varphi_{\hat{H}}(h^*g)$}.

To obtain the Fourier transform of the vertex operator $A_{v_b}^K$,
it is sufficient to consider the $K$ action over the fusion basis of $H$ and $K$. Taking $\mathfrak{A}=K$ fusion basis as an example, it is easy to verify that 
\begin{equation}
  L^{k}_+|[\nu]_{ab}\rangle = \sum_{c}D^{\nu}(S(k))_{ac}|[\nu]_{cb}\rangle, \quad L_-^k|[\nu]_{ab} \rangle =\sum_c D^{\nu}(k)_{cb}|[\nu]_{ac}\rangle,
\end{equation}
where the second one can be derived from the first one by using Eq.~\eqref{eq:LTS}.
Then we see that 
\begin{equation}
    A^K_{v_b}|[\nu_1]_{a_1b_1},[\nu_2]_{a_2b_2},[\mu_3]_{a_3b_3}\rangle=\sum_{c_1,c_2,c_3}[W^{\nu_2}_{\mu_3\nu_1}]_{\{a_1,b_2,a_3\},\{c_1,c_2,c_3\}} |[\nu_1]_{c_1b_1},[\nu_2]_{a_2c_2},[\mu_3]_{c_3b_3}\rangle,
\end{equation}
where the $3j$-symbol
\begin{equation}\label{eq:3jsymbol}
    [W^{\nu_2}_{\mu_3\nu_1}]_{\{a_1,b_2,a_3\},\{c_1,c_2,c_3\}}=\sum_{(h_K)} D^{\nu_1}(S(h_K^{(1)}))_{a_1,c_1}\otimes D^{\nu_2}(h_K^{(2)})_{c_2,b_2}\otimes D^{\mu_2}(S(h_K^{(3)}))_{a_3,c_3}
\end{equation}
is the projector on the space $\Hom (\mu_3\otimes \nu_1,\nu_2)$.
This coincides with the string-net boundary vertex operator.

In the same spirit, for boundary face $f_b$ (with one boundary edge $e_1$ and three bulk edges $e_2,e_3,e_4$), we need to consider the $\hat{H}$ action $T^{\varphi}_{\pm}$ on fusion bases of $H$ and $K$. It is clear that for $H$ (and similar for $K$)
\begin{align}
    T^{\varphi}_+|[\nu]_{ab}&\rangle= \sqrt{\frac{\dim \nu}{\dim H}} \sum_{(h_H)}D^{\nu}(h_H^{(1)})_{ab}   \varphi( h_H^{(3)} ) |h_H^{(2)}\rangle,\label{eq:TBF1}\\
     T^{\varphi}_{-}|[\nu]_{ab}\rangle&= \sqrt{\frac{\dim \nu}{\dim H}} \sum_{(h_H)}D^{\nu}(h_H^{(1)})_{ab}   \varphi( S(h_H^{(2)}) )  |h_H^{(3)}\rangle. \label{eq:TBF2}
\end{align}
Notice that we have 
\begin{equation}
    \varphi_{\Hhat}=\sum_{\nu \in \operatorname{Irr}(H)} \frac{\dim \nu }{\dim H}\chi_{\nu},
\end{equation}
where $\chi_\nu$ is the character of $\nu$. The boundary phase operator can be written as
\begin{equation}
B_{f_b}=\sum_{\nu} \frac{\dim \nu}{\dim H} B_{f_b}^{\nu},   
\end{equation}
where 
\begin{equation}\label{eq:faceQDW}
    B_{f_b}^{\nu}= B_{f_b}^{\chi_{\nu}}=\bigotimes_{i=1}^4 \operatorname{ev}_i \left(\bigoplus_{\mu_i,\zeta_i} \bigotimes_{i=1}^4 \sqrt{\dim \mu_i \dim \zeta_i} W_{\mu_i,\zeta_i,\nu_i} \right).
\end{equation}
Suppose all edges surrounds $f_b$ counterclockwise, then using Eqs.~\eqref{eq:TBF1} and \eqref{eq:TBF2} one has 
\begin{equation}
    \begin{aligned}
        &\langle \{[\nu_i]_{a_ib_i}\}|B_{f_b}^{\chi_{\nu}}|\{[\nu'_j]_{a'_jb'_j}\}\rangle \\
        =&  \sum_{(\chi_{\nu}),(h_{K,1}),(h_{H,2,3,4})}\sqrt{\frac{\dim \nu'_1}{\dim K}}
        D^{\nu'_1}(h_{K,1}^{(1)})_{a'_1b'_1} \chi_{\nu}^{(1)}(h_{K,1}^{(3)}) \langle [\nu_1]_{a_1b_1}|h_{K,1}^{(2)}\rangle\\
        &\times \prod_{j=2}^4 \sqrt{\frac{\dim \nu'_j}{\dim H}}
        D^{\nu'_j}(h_{H,j}^{(1)})_{a'_jb'_j} \chi_{\nu}^{(j)}(h_{H,j}^{(3)}) \langle [\nu_j]_{a_jb_j}|h^{(2)}_{H,j}\rangle.
    \end{aligned}
\end{equation}
It is easy to verify that this matches well with the string-net face operator.

2. Consider an extended string-net model with $\EC=\mathsf{Rep}(H)$ and $\EM=\mathsf{Rep}(K)$ with $K\leq H$.
We attach to each bulk and boundary edge the Hilbert spaces
\begin{equation}
\begin{aligned}
    \begin{tikzpicture}
			\draw[black, line width=1.5pt] (0,-0.7)--(0,0.7);
			\draw[-latex,black, line width=1.5pt] (0,-0.7)--(0,-0.2);
			\draw[-latex,black, line width=1.5pt] (0,0)--(0,0.5);
			\draw[black, line width=1.0pt] (-0.6,-0.6)--(0,0);
			\draw[-latex,black, line width=1.0pt] (-0.6,-0.6)--(-0.2,-0.2);
			\node[ line width=0.2pt, dashed, draw opacity=0.5] (a) at (-0.5,-0.1){$i$};
   	        \node[ line width=0.2pt, dashed, draw opacity=0.5] (a) at (0.4,-0.4){$x$};
   	        \node[ line width=0.2pt, dashed, draw opacity=0.5] (a) at (0.4,0.4){$y$};
   	        \node[ line width=0.2pt, dashed, draw opacity=0.5] (a) at (0.2,0){$v_b$};
			\end{tikzpicture}
    \end{aligned}\quad \quad \quad 
    \begin{aligned}
           \mathcal{H}_e&=\bigoplus_{i\in \operatorname{Irr} (\EC)} \omega_{\EC}(i) \otimes \omega_{\EC} (i)^{\vee} \cong H,\\
           \mathcal{H}_{e_b}&=\bigoplus_{x\in \operatorname{Irr} (\EM)} \omega_{\EM}(x) \otimes \omega_{\EM} (x)^{\vee}\cong K, 
    \end{aligned}
\end{equation}
where $\omega_{\EC}$ and $\omega_{\EM}$ are forgetful functors of the representation.
For a trivalent boundary vertex $v_b$, the corresponding space is $\mathcal{H}_{v_b}=\mathcal{H}_e\otimes \mathcal{H}_{e_b}\otimes \mathcal{H}_{e_b}$.
As we have shown before, since $K$ is a left $H$-comodule algebra, the tensor product $i\otimes x$ for representations $i\in \operatorname{Irr}(H)$ and $x\in  \operatorname{Irr}(K)$ is a representation of $K$, the fusion rule is
\begin{equation}
    i\otimes x =\bigoplus_{y\in \operatorname{Irr}(K)} N_{i,x}^yy.
\end{equation}
Since $\omega_{\EM}$ preserves that tensor product (Eq.~\eqref{eq:CM}) and direct sum (this holds by definition), we have
\begin{equation}
    \omega_{\EC}(i) \otimes \omega_{\EM}(x)\cong \bigoplus_{y\in \operatorname{Irr}(K)} N_{i,x}^y\omega_{\EM}(y).
\end{equation}
The boundary vertex operator is defined as the projection that characterizes the above isomorphism $ Q_{v_b}^{i,x;y}: \omega_{\EC}(i)\otimes \omega_{\EM}(x)\otimes \omega_{\EM}(y) \to \omega_{\EC}(i)\otimes \omega_{\EM}(x)\otimes \omega_{\EM}(y)$. Namely, if $N_{i,x}^y \neq 0$, $Q_{v_b}^{i,x;y}=1$, otherwise, $Q_{v_b}^{i,x;y}=0$. The boundary vertex operator is thus given by 
\begin{equation}
    Q_{v_b}:=\bigoplus_{i\in \operatorname{Irr}(\EC),\,x,y\in\operatorname{Irr}(\EM)}Q_{v_b}^{i,x;y} \otimes \id_{\omega_{\EC}(i)^{\vee}} \otimes \id_{\omega_{\EM}(x)^{\vee}} \otimes \id_{\omega_{\EM}(y)^{\vee}}.
\end{equation}
To show that $Q_{v_b}$ coincides with the quantum double boundary vertex operator $A_{v_b}$, we only need to use the fact that the $3j$-symbol we defined in Eq.~\eqref{eq:3jsymbol} coincides with $Q_{v_b}^{i,x;y}$.

We will employ the definition of the face operator provided in Ref.~\cite{buerschaper2013electric}. Through our analysis, we will demonstrate that, in this particular instance, the face operator aligns with the quantum double face operator as per its definition.
\begin{equation}
    \begin{aligned}
		\begin{tikzpicture}
  			\draw[black,line width =1.6pt] (0.5,0.5) -- (0.8,0.8);
                \draw[black] (-0.5,0.5)--(-0.8,0.8);
                \draw[black] (-0.8,-0.8)--(-0.5,-0.5);
			\draw[-latex,black] (-0.5,0.5) -- (0.5,0.5);
			\draw[-latex,black] (-0.5,-0.5) -- (-0.5,0.5); 
   	      \draw[black,line width =1.6pt] (0.8,-0.8) -- (0.5,-0.5);
			\draw[-latex,black,line width = 1.6pt] (0.5,-0.5) -- (0.5,0.5);
			\draw[-latex,black] (-0.5,-0.5) -- (0.5,-0.5); 
			\draw [fill = black] (0,0) circle (1.2pt);
			\node[ line width=0.2pt, dashed, draw opacity=0.5] (a) at (0.75,0){$x$};
			\node[ line width=0.2pt, dashed, draw opacity=0.5] (a) at (-0.75,0){$j$};
			\node[ line width=0.2pt, dashed, draw opacity=0.5] (a) at (0,-0.7){$i$};
			\node[ line width=0.2pt, dashed, draw opacity=0.5] (a) at (0,0.7){$k$};
		\end{tikzpicture}
	\end{aligned} \quad \quad 
B_f: \mathcal{H}_e\otimes \mathcal{H}_e \otimes \mathcal{H}_e \otimes \mathcal{H}_{e_b} \to \mathcal{H}_e\otimes \mathcal{H}_e \otimes \mathcal{H}_e \otimes \mathcal{H}_{e_b}.
\end{equation}
If we denote the domain and codomain face spaces of $B_f$ as 
\begin{equation} 
\begin{aligned}
    \mathcal{H}_{f_b}= &(\oplus_{i\in \operatorname{Irr}(\EC)} (\omega_{\EC}(i)\otimes \omega_{\EC}(i)^{\vee} )) \otimes (\oplus_{j\in \operatorname{Irr}(\EC)} (\omega_{\EC}(j)\otimes \omega_{\EC}(j)^{\vee})) \\
    &\otimes (\oplus_{k\in \operatorname{Irr}(\EC)} (\omega_{\EC}(k)\otimes \omega_{\EC}(k)^{\vee}))\otimes (\oplus_{x\in \operatorname{Irr}(\EM)} (\omega_{\EC}(x)\otimes \omega_{\EC}(x)^{\vee}))
\end{aligned}
\end{equation}
and 
\begin{equation}
    \begin{aligned}
        \mathcal{H}'_{f_b}= &(\oplus_{i'\in \operatorname{Irr}(\EC)} (\omega_{\EC}(i')\otimes \omega_{\EC}(i')^{\vee})) \otimes (\oplus_{j'\in \operatorname{Irr}(\EC)} (\omega_{\EC}(j')\otimes \omega_{\EC}(j')^{\vee})) \\
        & \otimes (\oplus_{k'\in \operatorname{Irr}(\EC)} (\omega_{\EC}(k')\otimes \omega_{\EC}(k')^{\vee}))\otimes (\oplus_{x'\in \operatorname{Irr}(\EM)} (\omega_{\EC}(x')\otimes \omega_{\EC}(x')^{\vee})),
    \end{aligned}
\end{equation} 
then the face operator can be regarded as an element in 
\begin{equation}
    \begin{aligned}
        \displaystyle \mathcal{H}_{f_b}\otimes (\mathcal{H}_{f_b}')^{\vee}= \oplus_{i,j,k,x}\oplus_{i',j',k',x'} &( (\omega_{\EC}(i)\otimes \omega_{\EC}(i)^{\vee})\otimes (\omega_{\EC}(i')^{\vee}\otimes \omega_{\EC}(i')) ) \\
         \otimes \,&( (\omega_{\EC}(j)\otimes \omega_{\EC}(j)^{\vee})\otimes (\omega_{\EC}(j')^{\vee}\otimes \omega_{\EC}(j')) )  \\
        \otimes  \,& ( (\omega_{\EC}(k)\otimes \omega_{\EC}(k)^{\vee})\otimes (\omega_{\EC}(k')^{\vee}\otimes \omega_{\EC}(k')) )  \\
        \otimes  \,& ( (\omega_{\EC}(x)\otimes \omega_{\EC}(x)^{\vee})\otimes (\omega_{\EC}(x')^{\vee}\otimes \omega_{\EC}(x')) ). 
    \end{aligned}
\end{equation}
Fix the $\nu \in \operatorname{Irr}(\EC)$, we defined 
$P^{i,i'^{\vee},\nu}$ as the projection of the tensor product of three representations into the vacuum, $i\otimes i'^{\vee} \otimes \nu \to \one$. 
Notice that for the boundary edge, the projection to vacuum is to the vacuum of $\one_{K}\in\mathsf{Rep}(K)$, and for the bulk edge, it is the projection to vacuum $\one_H$.
For each edge, we assign the same $\nu$, and label the edges around the faces as $s=1,\cdots,4$, then the $\nu$-component of string-net face operator is defined as
\begin{equation}
    B_{f_b}^{\nu}=\bigotimes_{s=1}^4 \operatorname{ev}_s \left(\bigoplus_{\mu_s,\zeta_s} \bigotimes_{s=1}^4 \sqrt{\dim \mu_s \dim \zeta_s} P^{\mu_s,\zeta^{\vee}_s,\nu} \right).
\end{equation}
We see that it is of the same form as Eq.~\eqref{eq:faceQDW}.
The boundary face operator 
\begin{equation}
B_{f_b}=\sum_{\nu} \frac{\dim \nu}{\dim H} B_{f_b}^{\nu},   
\end{equation}
thus coincides with the face operator of that of the quantum double model, as expected.
\end{proof}

\subsection{Anyon condensation via ribbon operator}

In this section, we will delve into the lattice construction of local operator algebra, as well as ribbon operators and their application in realizing anyon condensation and boundary-bulk duality.

From the perspective of anyon condensation theory, the boundary vacuum anyon can be viewed as a composite of bulk anyons, represented by the equation $\one_{\EB}=\one \oplus X_1\oplus \cdots \oplus X_n$. 
In the lattice realization, the boundary vacuum anyon sector is the space of states that are invariant under all boundary local stabilizers, i.e.,
\begin{equation}
V_{\one_{\EB}}=\{|\psi\rangle:
A_{v_b}|\psi\rangle = B_{f_b}|\psi\rangle=|\psi\rangle,\forall~ v_b,f_b\}.
\end{equation}
The anyon condensation process occurs when a bulk anyon moves near the boundary and becomes a boundary anyon.
Since the bulk phase is described by a UMTC $\ED$ and the boundary phase is described by a UFC $\EB$, the condensation is characterized mathematically by a monoidal functor $\mathbf{Cond}:\ED\to \EB$ (roughly speaking, this is a functor that preserves the fusion structure).
The bulk anyons that become boundary vacuum anyon are called condensed.
Our goal is to construct ribbon operators that can facilitate the anyon condensation process and explore the boundary-bulk duality.

Before we start, let us introduce some notions that we will use:
(i) the bulk-to-boundary ribbon $\rho_{\downarrow}$ is a directed ribbon that starts at some bulk site $s=\partial_0 \rho$ and ends at some boundary site $s_b=\partial_1 \rho$;
(ii) the boundary-to-bulk ribbon $\rho_{\uparrow}$ starts at some boundary site and ends at some bulk site;
(iii) the boundary ribbon $\rho_{b}$ is a ribbon that starts and terminates on boundary sites $s_b^0=\partial_0\rho_{b}$ and $s_b^1=\partial_1\rho_b$. These ribbons can also be called type-A or type-B according to the same convention as we have adopted in the bulk case (see Sec.~\ref{sec:ribbon}).

Topological excitations are characterized as point-like excitations in the lattice model, and are represented by the local operator algebra. Specifically, for a given bulk site $s$, the vertex operator $A^g(s)$ and $B^{\psi}(s)$ generate an algebra $\mathcal{D}(s)$, which is isomorphic to the quantum double of the input Hopf algebra. Topological excitations correspond to irreducible representations of these local operator algebras.
The bulk ribbon operator is constructed to generate an algebra $\mathcal{A}$, which is isomorphic to the dual algebra of $\mathcal{D}(s)$. Local stabilizers at the internal sites of the ribbon commute with the ribbon operator, but at the two ends, they do not commute, and their commutators correspond to topological excitations when the ribbon operators act on the ground state.
For a gapped boundary, a similar picture exists. Boundary excitations correspond to representations of the boundary local operator algebra, and the bulk-to-boundary ribbon operator algebra is dual to the boundary local operator algebra.
In the previous section, we demonstrated that the gapped boundaries of a bulk Hopf algebra $H$ are classified by $H$-indecomposable left (or right) $H$-comodule algebras $\FA$. We will now determine the algebraic structure of the boundary local operator algebra and bulk-to-boundary ribbon operators.

\begin{table}[t]
\centering \small 
\begin{tabular} {|l|c|c|c|} 
\hline
   &$\EC$-module &  Boundary phase & Boundary local algebra \\ \hline
 Left $H$-comodule  & ${_{\FA}}\Mod$  & ${^H_{\FA}}\Mod_{\FA}$ & $(\FA\otimes \FA^{\rm op}) \star	 \Hhat$ \\ \hline
 Right $H$-comodule  & $\Mod_{\FA}$ & ${_{\FA}}\Mod_{\FA}^H$ &  $(\FA\otimes \FA^{\rm op}) \star\Hhat$ \\ \hline
\end{tabular}
\caption{The boundary local operator algebras, see Eq.~\eqref{eq:AHstar} for the definition of the crossed product $(\FA\otimes \FA^{\rm op}) \star\Hhat$ 
 for a right $H$-comodule algebra $\FA$. The left comodule algebra case can be defined in a similar way. Notice that, by employing a slight abuse of notation, both cases are denoted as $(\FA\otimes \FA^{\rm op}) \star \Hhat$.
 \label{tab:bdLocalAlg}}
\end{table}

\begin{proposition}
\label{thm:local-ribbon-alg}
Consider a gapped boundary determined by an $H$-comodule algebra $\FA$: 
\begin{enumerate}
    \item The boundary local operator algebra is fusion-categorical
Morita equivalent to $\mathcal{D}(s_b)\simeq  (\FA\otimes \FA^{\rm op})\star \Hhat$.\,\footnote{In this paper, we use $A\cong B$ to denote the isomorphism between algebraic objects, and use $A\simeq B$ to denote the Morita equivalence between them.} See Eq.~\eqref{eq:AHstar} for the definition of $ (\FA\otimes \FA^{\rm op})\star \Hhat$, also see Table~\ref{tab:bdLocalAlg} for a summary.
    \item The bulk-to-boundary ribbon operator algebra $\mathcal{A}_{\rho_{\downarrow}}$ is  the dual of the boundary local operator algebra $\mathcal{A}_{\rho_{\downarrow}}\simeq \mathcal{D}(s_{b})^{\vee}$. Or equivalently, the dual of 
     the  bulk-to-boundary ribbon operator algebra is fusion-categorical Morita equivalent to the boundary local operator algebra $\mathcal{A}_{\rho_{\downarrow}}^{\vee}\simeq  (\FA\otimes \FA^{\rm op})\star \hat{H}$.
\end{enumerate}
\end{proposition}

\begin{remark}
Before we prove the main assertions, it is important to recall that two algebras are considered fusion-categorical Morita equivalent if and only if their respective representation categories are equivalent as fusion categories.
In the case of a gapped boundary, the local operator algebra is unique up to fusion-categorical Morita equivalence. The only requirement is that the representation category of the local operator algebra is equivalent to the boundary topological phase as UFC. As a result, the category of $\mathcal{D}(s_b)$-modules is equivalent to either $\EB={_{\FA}^H}\Mod{_{\FA}}$ or $\EB={_{\FA}}\Mod{_{\FA}^H}$, depending on the choice of boundary directions (with the bulk either on the left- or right-hand side of the boundary), see Table~\ref{tab:bdTop1} for a summary.
The bulk-to-boundary ribbon operator algebra is dual to the boundary local operator algebra.
At the boundary end, the ribbon operator algebra generates all the boundary excitations.
Because the local operator algebra is not unique, it follows that the ribbon operator algebra is also not unique. This is reflected in the fact that there exist many different lattice constructions that realize the same phase.
\end{remark}

\begin{proof}[Proof of Proposition \ref{thm:local-ribbon-alg}]
We only need to prove that 1 and 2 are statements about the requirement that bulk-to-boundary ribbon operator algebra must satisfy.
  
Without loss of generality, let us consider the right-$H$ comodule algebra $\FA$, and suppose that the boundary phase is given by $\EB={_{\FA}}\Mod{_{\FA}^H}$.
A module $M\in \EB$ is an $\FA|\FA$-bimodule with a right $H$-comodule structure with the coaction 
\begin{equation}
    \beta: M\to M\otimes H,
\end{equation}
such that the coaction is an $\FA|\FA$-bimodule map (the $\FA|\FA$-bimodule structure of  $M \otimes H$ is given by $a\cdot (m\otimes h)\cdot b=\sum(a^{[0]}\cdot m \cdot b^{[0]})\otimes (a^{[1]}\cdot h \cdot b^{[1]})$).
This means that we have 
\begin{equation} \label{eq:RightAA}
   \sum_{[a\cdot m \cdot b]} (a\cdot m \cdot b)^{[0]} \otimes  (a\cdot m \cdot b)^{[1]}= \sum_{[a],[m],[b]} a^{[0]}\cdot m^{[0]} \cdot b^{[0]} \otimes a^{[1]}\cdot m^{[1]} \cdot b^{[1]}.
\end{equation}
A left $\Hhat$-module structure of $M$ can be defined as 
\begin{equation}
    \varphi \rightharpoonup m=\sum_{[m]} m^{[0]}\varphi(m^{[1]}).
\end{equation}
The condition of Eq.~\eqref{eq:RightAA} can be expressed in the left $\Hhat$-action as 
\begin{equation}
    \varphi \rightharpoonup (a\cdot m \cdot b)=\sum_{[a],[m],[b]} a^{[0]}\cdot m^{[0]} \cdot b^{[0]} \varphi( a^{[1]}\cdot m^{[1]} \cdot b^{[1]}).
\end{equation}
Now a direct calculation gives ($\star$ is just tensor product here)
\begin{equation}
    ((c\otimes d)\star \psi)\cdot ((a\otimes b)\star \varphi) \cdot m= \sum_{[a],[m],[b]} (c a^{[0]} m^{[0]} b^{[0]} d) \star \psi (a^{[1]} m^{[1]} b^{[1]}) \varphi(m^{[2]}),
\end{equation}
where $a,b,c,d\in \FA$ and $\psi,\varphi \in \Hhat$.
This implies that 
\begin{equation}
      ((a\otimes b)\star	\varphi )\cdot ((c\otimes d)\star	\psi)=  \sum (ac^{[0]}\otimes d^{[0]}b) \star	  \varphi(c^{[1]}\bullet d^{[1]})\psi,
\end{equation}
which coincides with Eq.~\eqref{eq:AHstar}.
This means that $M\in \EB$ implies that $M \in {_{(\FA\otimes \FA^{\rm op}) \star \Hhat}} \Mod$.
The proof of the converse direction can be showed similarly since the left $\Hhat$-action and the right $H$-coaction correspond to each other in a canonical way.
\end{proof}

In Proposition \ref{prop:bdLocalAlg}, we demonstrated that the boundary local vertex and face operators generate the boundary local algebra $\mathcal{D}(s_b)$. This result supports the claim that the lattice model we have constructed realizes the boundary phase of  $\EB={_{\FA}}\Mod{_{\FA}^H}$.
We would like to note that the above proposition holds not only when the input $\FA$ for the boundary is a right $H$-comodule algebra, but also when $\FA$ is a left $H$-comodule algebra. This is summarized in Table \ref{tab:bdLocalAlg}.
When $\FA$ is a faithfully flat $H$-Galois extension, the local operator algebra can be showed to be fusion-categorically equivalent to either $\FA\Box_H \bar{\FA}$ or $\bar{\FA}\Box_H \FA$ \cite{caenepeel2007morita}. For further details, please refer to Eq.~\eqref{eq:boxtensor} and the accompanying discussions. This equivalence enables the construction of a gapped boundary model via a generalized quantum double construction, as outlined in Appendix~\ref{sec:bdII}. Using this approach can simplify the lattice realization of ribbon operators and other related concepts.

The case where $\FA=K\leq H$ is of interest in its own right, we can provide a more explicit lattice construction of ribbon operators and anyon condensation. 
By discussion in Sec.~\ref{sec:bdd-top-exc}, we have known that for the boundary model determined by an $H$-comodule algebra $\FA=K\leq H$, the boundary topological excitations are classified by $H$-covariant $\mathfrak{A}|\mathfrak{A}$-bimodules. Or equivalently, they are also classified by $\mathfrak{M}|\mathfrak{M}$-bimodules in $\mathsf{Rep}(H)$, where $\mathfrak{M}\cong\mathrm{Stab}_K(M)$ is the Yan-Zhu stabilizer for some nonzero $K$-module $M$ and is an $H$-module algebra. One can take $M=\mathbb{C}$ the trivial $K$-module. In this case, it is showed \cite[Example 2.19]{andruskiewitsch2007module} that $\mathfrak{M}\cong\mathrm{Stab}_K(\mathbb{C})\cong \underline{K}^\vee$ as right $\hat{H}$-comodule algebras with $\underline{K} := H/S^{-1}(K^+)H$ and $K^+:=\ker(\varepsilon)\cap K$ the augmentation ideal of $K$. 
The $\Hhat$-comodule algebra structure will canonically induces an $H$-module algebra structure on $\underline{K}^\vee$.
The bulk-to-boundary ribbon operators are operators supported on bulk-to-boundary ribbons and commute with all stabilizer operators in the Hamiltonian except at the two end sites of the ribbons. These operators form an algebra $\mathcal{A}_{\rho_{\downarrow}}$. Proposition~\ref{thm:local-ribbon-alg} has provided one way to identify this algebra up to a fusion-categorical Morita equivalence.

\subsubsection{Local operator algebra}

In this subsection, we will determine the local operator algebra $\mathcal{D}(s_b)$ over the boundary site $s_b$ in the case that $\mathfrak{A}=K$ is a Hopf subalgebra of $H$.

Since boundary topological excitations are given as the monoidal category ${^H_K}\mathsf{Mod}_K$, identifying the local operator algebra amounts to the reconstruction problem: if there is some kind of algebraic object whose representation (modules or comodules) category is equivalent to ${^H_K}\mathsf{Mod}_K$. A more general case has been studied in \cite{schauenburg02extension,schauenburg02coquasibialg} which applies to our case. Note that $K\leq H$ are finite-dimensional complex Hopf algebras, hence $H$ is cocleft as a left $K$-module coalgebra (see Refs.~\cite{schauenburg02extension,schauenburg02coquasibialg} and references therein for related definition). Denote by $\pi:H\to K$ the convolution invertible cocleaving map which is unital and counital. Denote $\bar{M}:=M/K^+M$ for an $H$-covariant $K|K$-bimodule $M$. It can be showed that $\overline{M\otimes_KN}\cong\bar{M}\otimes\bar{N}$. Let $j:\bar{M}\to M$ be the map $j(\bar{m})=\sum_{(m)}\pi^{-1}(m^{[-1]})m^{[0]}$. One has an isomorphism $M\cong K\otimes \bar{M}, m\mapsto \sum_{(m)}\pi(m^{[-1]})\otimes \overline{{m^{[0]}}}$ with inverse $x\otimes \bar{m}\mapsto xj(\bar{m})$. In particular, $H\cong K\otimes (H/K^+H)$ as left $K$-modules and right $H/K^+H$-comodules, thus $\pi:H\to K$ and $j:H/K^+H\to H$ can be identified with $\pi(x\otimes\bar{h})=\varepsilon(h)x$ and $j(\bar{h})=1_K\otimes \bar{h}$ respectively for $x\in K, \bar{h}\in H/K^+H$.

Define the functor $\omega:{^H_K}\mathsf{Mod}_K\to\mathsf{Vect}$ by $\omega(M)=\bar{M}=M/K^+M$. It is a tensor functor with the factorial isomorphisms $\xi_{M,N}:\bar{M}\otimes \bar{N} \to \overline{M\otimes_KN}$ given by $\xi(\bar{m}\otimes\bar{n})=\overline{m\otimes j(\bar{n})}$, with $\xi_{M,N}^{-1}(\overline{m\otimes n})=\sum_{(n)}\overline{m\pi(n^{[-1]})}\otimes \overline{n^{[0]}}$. 
Consider the category ${^{H/K^+H}}\mathsf{Mod}_K$ of $H/K^+H$-covariant right $K$-modules, that is, right $K$-modules such that the $H/K^+H$-coactions are morphisms of right $K$-modules. It is showed \cite{schauenburg02coquasibialg} by Schneider's structure theorem that there is a category equivalence $\hat{\omega}:{^H_K}\mathsf{Mod}_K\simeq {^{H/K^+H}}\mathsf{Mod}_K, M\mapsto\bar{M}$, whose inverse is $N\mapsto H\Box_{H/K^+H}N$. In fact, the $H/K^+H$-comodule structure over $\bar{M}$ is induced by the $H$-comodule structure over $M$, and $\bar{m}x=\overline{mx}$ for $m\in M, x\in K$ gives the right $K$-action; the left $H$-comodule and $K$-module structures over $H\Box_{H/K^+H}N$ are induced by the left tensor factor, and $(\sum h_i\otimes n_i)x = \sum h_ix^{(1)}\otimes n_ix^{(2)}$ for $x\in K, \sum h_i\otimes n_i\in H\Box_{H/K^+H}N$. There is a unique monoidal category structure over ${^{H/K^+H}}\mathsf{Mod}_K$ making $\hat{\omega}$ a monoidal equivalence such that $\omega$ factors over $\hat{\omega}$ (see \cite[Theorem 3.3.5]{schauenburg02coquasibialg}):
\begin{equation*}
    \begin{tikzcd}
	{{^H_K}\mathsf{Mod}_K} && {{^{H/K^+H}}\mathsf{Mod}_K} \\
	& {\mathsf{Vect}}
	\arrow["{\hat{\omega}}", from=1-1, to=1-3]
	\arrow["\omega"', from=1-1, to=2-2]
	\arrow["", from=1-3, to=2-2]
\end{tikzcd}.
\end{equation*}

The next step is to identify an algebraic object $\tilde{H}$ such that ${^{H/K^+H}}\mathsf{Mod}_K\simeq {^{\tilde{H}}}\mathsf{Mod}$. It turns out that $\tilde{H} = (H/K^+H)\rtimes\hat{K}$ is a coquasi-Hopf algebra which is a crossed coproduct built on $(H/K^+H)\otimes \hat{K}$ with structures given below (cf.~\cite[Sec.~4.2]{schauenburg02extension}). First define $\rightharpoonup\,:\hat{K}\otimes (H/K^+H)\to H/K^+H$, $\leftharpoonup\,:\hat{K}\otimes (H/K^+H)\to \hat{K}$ by 
\begin{equation*} 
    \varphi\rightharpoonup \bar{h} = \hat{\varepsilon}(\varphi)\bar{h},\quad  (\varphi \leftharpoonup \bar{h})(x)=\varepsilon(h)\varphi(x),
\end{equation*}
and $\rho:H/K^+H\to \hat{K}\otimes (H/K^+H),~\rho(\bar{h}) := \sum\rho(\bar{h})^{[-1]}\otimes \rho(\bar{h})^{[0]}$ by
\begin{equation*} 
     \sum\rho(\bar{h})^{[-1]}(x) \rho(\bar{h})^{[0]} = \overline{j(\bar{h})x},\quad\text{for~all}~ x\in K. 
\end{equation*}
Using this definition, one computes that
\begin{equation*}
    \rho(\bar{h}) =\sum_k\sum_{(k)}\hat{k}(S(k^{(3)})\bullet k^{(1)})\otimes \overline{hk^{(2)}},
\end{equation*}
where $\{k\}$ is an orthogonal basis of $K$ with $\{\hat{k}\}$ its dual basis. By this preparation, $(H/K^+H)\rtimes\hat{K}$ becomes a coquasi-bialgebra with $\varepsilon(\bar{h}\otimes \varphi)=\varepsilon_H(h)\varphi(1_K)$, $1 = \overline{1_H}\otimes\varepsilon_K$,
\begin{align}
    (\bar{h}\otimes\varphi)(\bar{g}\otimes & \psi) = \sum_{(g),(\varphi)}\bar{h}(\varphi^{(1)}\rightharpoonup \overline{g^{(1)}})\otimes (\varphi^{(2)}\leftharpoonup \overline{g^{(2)}})\psi=\bar{h}\bar{g}\otimes \varphi\psi, \nonumber \\
    \Delta(\bar{h}\otimes \varphi) &= \sum(\overline{h^{(1)}}\otimes \rho(\overline{h^{(2)}})^{[-1]}\varphi^{(1)})\otimes(\rho(\overline{h^{(2)}})^{[0]}\otimes \varphi^{(2)}) \nonumber \\
     &= \sum_k\sum_{(k),(h)}(\overline{h^{(1)}}\otimes \hat{k})\otimes (\overline{S(k^{(3)})h^{(2)}k^{(1)}}\otimes \varphi(k^{(2)}\bullet)),  \label{eq:comul-quasi-hopf}
\end{align}
and associator 
\begin{equation*}
\Phi((\bar{h}\otimes\varphi)\otimes(\bar{g}\otimes\psi)\otimes(\bar{k}\otimes\theta)) =\varepsilon_H(h)\varepsilon_H(g)\varepsilon_H(k)\varphi(1_K)\psi(1_K)\theta(1_K).
\end{equation*}
Moreover, $(H/K^+H)\rtimes \hat{K}$ can be equipped with a coquasi-antipode $(S,\alpha,\gamma)$ whose existence is guaranteed by the rigidity of the monoidal category ${^H_K}\mathsf{Mod}_K$. Thus $(H/K^+H)\rtimes \hat{K}$ has a structure of coquasi-Hopf algebra. Details can be found in \cite[Sec.~4.4]{schauenburg02coquasibialg} and \cite[Sec.~9.4]{majid2000foundations}.

The equivalence $\tilde{\omega}:{^{H/K^+H}}\mathsf{Mod}_K\simeq{^{(H/K^+H)\rtimes\hat{K}}\mathsf{Mod}}$ is described as follows. Given an $M\in {^{H/K^+H}}\mathsf{Mod}_K$, it can be endowed with a left $(H/K^+H)\rtimes \hat{K}$-comodule structure by
\begin{equation*}
    N\to ((H/K^+H)\rtimes\hat{K})\otimes N,\quad n\mapsto \sum_{k,(n)} (n^{[-1]}\otimes \hat{k})\otimes n^{[0]}k,
\end{equation*}
where $\{k\}$ is an orthogonal basis of $K$ with $\{\hat{k}\}$ its dual basis; conversely, given an $N\in{^{{(H/K^+H)}\rtimes\hat{K}}}\mathsf{Mod}$ with coaction $\rho_N(n) = \sum (n^{[-2]}\otimes n^{[-1]})\otimes n^{[0]}\in ((H/K^+H)\rtimes\hat{K})\otimes N$, it can be endowed with a left $H/K^+H$-comodule structure by 
\begin{equation*}
    N\to (H/K^+H)\otimes N,\quad n\mapsto \sum n^{[-2]}\varepsilon(n^{[-1]})\otimes n^{[0]}, 
\end{equation*}
and a right $K$-action by $nx = \sum\varepsilon(n^{[-2]})\langle n^{[-1]},x\rangle n^{[0]}$ for $x\in K, n\in N$. It is showed that $\tilde{\omega}$ is a monoidal equivalence. In summary, there are monoidal equivalences ${^H_K}\mathsf{Mod}_K\simeq {^{H/K^+H}}\mathsf{Mod}_K\simeq{^{(H/K^+H)\rtimes\hat{K}}}\mathsf{Mod}$, see \cite{schauenburg02coquasibialg} for rigorous proofs. Since ${^A}\mathsf{Mod}\simeq\mathsf{Mod}_{\hat{A}}\simeq {_{\hat{A}^{\rm op}}}\mathsf{Mod}$, we conclude that 
\begin{equation}\label{eq:bdFusion}
    {^H_K}\mathsf{Mod}_K\simeq {_{((H/K^+H)\rtimes\hat{K})^{\vee,\rm op}}}\mathsf{Mod}.
\end{equation}
For the convenience of later discussion, consider the identification ${^H_K}\mathsf{Mod}_K\simeq {^{H^{\rm op}}_{K^{\rm op}}}\mathsf{Mod}_{K^{\rm op}}$ and apply the discussion above, we see that ${^H_K}\mathsf{Mod}_K\simeq {_{\mathcal{K}_{K}}}\mathsf{Mod}$ where 
\begin{equation}
\begin{aligned}
    \mathcal{K}_{K}:&={((H^{\rm op}/K^{\rm op,+}H^{\rm op})\rtimes\widehat{K^{\rm op}})^{\vee,\rm op}} \\
    & \cong ((H/HK^+)^{\rm cop}\rtimes \hat{K}^{\rm cop})^{\rm cop,\vee} \\
    &\cong((H/HK^+)\rtimes \hat{K})^\vee. 
\end{aligned}
\end{equation}
Therefore, the local operator algebra $\mathcal{D}(s_b)$ is fusion-categorical Morita equivalent to $\mathcal{K}_{K}\cong((H/HK^+)\rtimes \hat{K})^\vee$, which has a structure of quasi-Hopf algebra. This is compatible with Examples~\ref{exp:smoothBd} and \ref{exp:roughBd}, since $\mathcal{K}_{H}\cong H$ and $\mathcal{K}_{\mathbb{C}1_H}\cong \hat{H}$. 
It is important to note that the algebra $\mathcal{K}_K$ is also fusion-categorical Morita equivalent to the local algebra described in Proposition~\ref{thm:local-ribbon-alg}. This equivalence arises because the category in Eq.~\eqref{eq:bdFusion} is equivalent to the fusion category associated with the boundary phase.

\subsubsection{Bulk-to-boundary ribbon operators}

Let us now consider the construction of the bulk-to-boundary ribbon operator for a gapped boundary determined by a Hopf subalgebra $K\leq H$. The process is similar to that of the bulk case: first, we define the triangle operators on eight types of triangles, then the general ribbon operators are defined recursively. The main tool is the observation that the local operator algebra $\mathcal{D}(s_b)$ is fusion-categorical Morita equivalent to $((H/HK^+)\otimes \hat{K})^\vee\cong(H/HK^+)^\vee\otimes K$ given above. Note that $(H/HK^+)^\vee$ can be identified with the subspace of $\hat{H}$ consisting of those functions vanishing on $HK^+$.

Denote $H\sslash K:=H/HK^+$. 
For $[h]=\bar{h}\in H\sslash K$, define
\begin{equation}
    L^{[h]}_{\pm}|x\rangle := L^{hh_K}_{\pm}|x\rangle,
\end{equation}
where $h_K\in K$ is the Haar integral of $K$. This is well-defined by the definition of $K^+$: if $h = h' + h''k$ with $k\in K^+$, then $hh_K=h'h_K+h''kh_K= h'h_K + h''\varepsilon(k)h_K = h'h_K$. In fact, $H\sslash K \cong Hh_K, [h]\mapsto hh_K$, as left $H$-modules via the restriction of the quotient map $H\to H\sslash K$, whose kernel is $H(1_H-h_K)$. A $\varphi\in\hat{K}$ acts on $H$ in the following way. First choose a basis $\{e_1,\cdots,e_n\}$ of $H$ such that the first $k$ vectors form a basis of $K$. Then each dual vector $e^i:=\hat{e}_i\in\hat{H}$, and so $\varphi = \sum_{i=1}^k\varphi(e_i)e^i$ can act on $H$ by $\varphi\rightharpoonup x$ and $x\leftharpoonup \varphi$. Therefore, the edge operators $T_\pm^\varphi:H\to H$ are well-defined for $\varphi\in\hat{K}$.

Ribbon operators are parameterized by pairs $([h],\varphi)$ with $[h]\in H\sslash K$ and $\varphi\in \hat{K}$. The building blocks are triangle operators. Basically, these operators are similar to Eqs.~(\ref{eq:tri1}--\ref{eq:tri16}) with some minor modification: $L^h_\pm$ are replaced by $L^{[h]}_{\pm}$. For example,
\begin{equation}
    F^{[h],\varphi}( \tau_R)|x\rangle =\varepsilon_{H}(h)T_{\pm}^{{\varphi}} |x\rangle, \quad  F^{[h],\varphi}( \tilde{\tau}_R)|x\rangle =\varepsilon_{\hat{K}}(\varphi) L^{[h]}_{\pm}|x\rangle,
\end{equation}
for direct right triangle $\tau_R$ and dual right triangle $\tilde{\tau}_R$, etc. 

For a general bulk-to-boundary ribbon $\rho$, say of type-B, ribbon operators are determined by the comultiplication of the dual  $\mathcal{D}(s_b)^\vee=(H\sslash K)\otimes \hat{K}$, whose comultiplication is given by \eqref{eq:comul-quasi-hopf}. Since it is a coquasi-Hopf algebra, the comultiplication is co-associative. To define the ribbon operator, take a decomposition $\rho=\rho_1\cup\rho_2$. For $[h]\in H\sslash K$ and $\varphi\in \hat{K}$, the ribbon operator is then given by 
\begin{equation}\label{eq:ribb_op_bdd_I}
\begin{aligned}
    F^{[h],\varphi}(\rho) & = \sum_{([h]\otimes\varphi)}F^{([h]\otimes\varphi)^{(1)}}(\rho_1)F^{([h]\otimes \varphi)^{(2)}}(\rho_2) \\
			&= \sum_{k} \sum_{(k), (h)} F^{[h^{(1)}], \hat{k}}(\rho_1) F^{[S(k^{(3)} )  h^{(2)} k^{(1)}] , \varphi (k^{(2)}  \bullet) }(\rho_2),  
\end{aligned}
\end{equation}
where $\{k\}$ is a basis of $K$. Note that $[S(k^{(3)})h^{(2)}k^{(1)}]$ only depends on $[h]$ by simple computation. As in the bulk case, this definition is independent of the choice of the ribbon decomposition. Ribbon operators for ribbons of type-A are defined by $((H\sslash K) \otimes \hat{K})^{\rm op}$ in a similar formula.

\begin{remark}
If we choose $H=\mathbb{C}[G]$ and $K=\mathbb{C}[N]$ for a finite group $G$ and a subgroup $N\leq G$, then since $\mathbb{C}[G]/ \mathbb{C}[G]\mathbb{C}[N]^+\cong \mathbb{C}[G/N]$, our definition reduces to the bulk-to-boundary ribbon operator of group-algebra boundary model in Ref.~\cite{Cong2017}. Thus, the group case fits well into our construction. 
\end{remark}

In Appendix~\ref{sec:AB}, we have shown that at any boundary site $s_b$ the following commutation relations hold for $k\in K$ and $\psi\in(H\sslash K)^\vee$: 
\begin{align}
           A^k(s_b)F^{[h],\varphi}(\rho_A) & = \sum_{(k)}F^{[h],\varphi(\bullet k^{(2)})}(\rho_A)A^{k^{(1)}}(s_b), \label{eq:bdd-comm3} \\
           A^k(s_b)F^{[h],\varphi}(\rho_B) & = \sum_{(k)}F^{[h],\varphi(\bullet k^{(1)})}(\rho_B)A^{k^{(2)}}(s_b), \label{eq:bdd-comm7}  \\
           B^{\psi}(s_b)F^{[h],\varphi}(\rho_A)&=\sum_{(h)}\sum_{\ell,(\ell)}\varphi(\ell^{(2)})F^{[h^{(1)}],\hat{\ell}}(\rho_A)B^{\psi(S(\ell^{(3)})h^{(2)}\ell^{(1)}\bullet)}(s_b), \label{eq:bdd-comm4} \\ 
            B^{\psi}(s_b)F^{[h],\varphi}(\rho_B)&=\sum_{(h)}\sum_{\ell,(\ell)}\varphi(\ell^{(2)})F^{[h^{(1)}],\hat{\ell}}(\rho_B)B^{\psi(\bullet S(\ell^{(3)})h^{(2)})\ell^{(1)})}(s_b), \label{eq:bdd-comm8} 
\end{align}
where $\{\ell\}$ is a basis of $K$. 

The local operator algebra $\mathcal{D}(s_b)$ is fusion-categorical Morita equivalent to $\mathcal{K}_K$, which is also fusion-categorical Morita equivalent to $(K\otimes K^{\rm op})\star \hat{H}$ by Proposition~\ref{thm:local-ribbon-alg}. The multiplication in the algebra $\mathcal{D}(s_b)$ is determined by the comultiplication of $\mathcal{D}(s_b)^\vee$, hence is given as
\begin{equation}\label{eq:mul_loc}
    (\varphi\otimes x)\cdot(\psi\otimes y) = \sum_{(x)}\varphi\psi(S(x^{(3)})\bullet x^{(1)})\otimes x^{(2)}y, 
\end{equation}
by Eq.~\eqref{eq:comul-quasi-hopf}. 
It is immediate to show that the relation 
\begin{equation}\label{eq:str_bdd_loc}
     A^k(s_b)B^\varphi(s_b) = \sum_{(k)} B^{\varphi(S(k^{(3)})\bullet k^{(1)})}(s_b)A^{k^{(2)}}(s_b)
\end{equation}
holds for $\varphi\in (H/HK^+)^\vee$ and $k\in K$ (the proof has no essential difference with that on the bulk). This straightening relation totally determines the multiplication in Eq.~\eqref{eq:mul_loc}. Hence every boundary site supports a representation of $\mathcal{D}(s_b)$. 

More precisely, consider a bulk-to-boundary ribbon $\rho$ with $s_0=\partial_0\rho$ and $s_1=\partial_1\rho=s_b$. Let $|\Omega\rangle$ be a ground state. Denote $|\Omega^{[h],\varphi}_\rho \rangle = F^{[h],\varphi}(\rho)|\Omega\rangle$ and and let $\mathcal{V}(s_0,s_1)$ denote the space spanned by $\{|\Omega^{[h],\varphi}_\rho\rangle\,|\,[h]\in H\sslash K,\,\varphi\in\hat{K}\}$. Then $B^\psi(s_b)A^k(s_b)$ defines a representation on $\mathcal{V}(s_0,s_1)$. In fact, the commutation relations above imply that
\[
    A^k(s_b)|\Omega_\rho^{[h],\varphi}\rangle = |\Omega_\rho^{[h],\varphi(\bullet k)}\rangle  \in \mathcal{V}(s_0,s_1),\; B^{\psi}(s_b)|\Omega_\rho^{[h],\varphi}\rangle = \sum\psi(h^{(2)})|\Omega_\rho^{[h^{(1)}],\varphi}\rangle  \in \mathcal{V}(s_0,s_1); 
\]
moreover, 
\begin{align*}
    &\quad (B^\psi(s_b)A^k(s_b))((B^\theta(s_b)A^l(s_b)) |\Omega_\rho^{[h],\varphi}\rangle) \\
    & = \sum \psi(h^{(2)})\theta(h^{(3)})|\Omega_\rho^{[h^{(1)}],\varphi(\bullet kl)}\rangle  \\
    & = \sum \psi(h^{(2)})\theta(S(k^{(3)})h^{(3)}k^{(1)})|\Omega_\rho^{[h^{(1)}],\varphi(\bullet k^{(2)}l)}\rangle  \\
    & = \sum B^{\psi\theta(S(k^{(3)})\bullet k^{(1)})}(s_b)A^{k^{(2)}l}(s_b)|\Omega_\rho^{[h],\varphi}\rangle \\
    & = ((B^\psi(s_b)A^k(s_b))(B^\theta(s_b)A^l(s_b))) |\Omega_\rho^{[h],\varphi}\rangle,
\end{align*}
where in the last equality we used Eq.~\eqref{eq:str_bdd_loc}. In other words, the map 
\[
    \mathcal{D}(s_b) \to \operatorname{End}(\mathcal{V}(s_0,s_1)),\quad \varphi\otimes k\mapsto B^{\varphi}(s_b)A^k(s_b)
\]
is an algebra homomorphism and defines the representation.

\begin{example}
  Let us take the Kac-Paljutkin algebra $H_8$ as an example to illustrate our construction. 
 As an algebra, it is generated by $x, y, z$, with the relations 
\begin{gather*}
        x^2 = y^2 = 1, \quad z^2 = \frac{1}{2}(1+x+y-xy), \\
        xy = yx,\quad zx=yz,\quad zy = xz. 
\end{gather*}
The coalgebra structure and the antipode are determined by 
\begin{gather*}
    \Delta(x) = x\otimes x,\quad \Delta(y) = y\otimes y, \\
    \Delta(z) = \frac{1}{2}(1\otimes 1+y\otimes 1 + 1\otimes x - y\otimes x)(z\otimes z), \\
    \varepsilon(x) = \varepsilon(y) = \varepsilon(z) = 1, \\
    S(x) = x,\quad S(y) = y, \quad S(z) = z,
\end{gather*}
with linear extension. Clearly, $H_8$ has dimension 8 with basis $\{1,x,y,xy,z,zx,zy,zxy\}$.

To classify the gapped boundaries of the $D(H_8)$ quantum double phase, it is necessary to classify the $H_8$-comodule algebras (or equivalently $H_8$-module algebras) up to Morita equivalence. This classification has been provided in Ref.~\cite[Sec.~5.3]{etingof2021tensor}.
There are six kinds of gapped boundaries.
Here to illustrate our ribbon operator construction, we only consider one example of $K\leq H_8$. By Lagrange theorem, $\dim K$ must divide $\dim H_8$.
Let us consider the four-dimensional case, which is generated by the group $G=\{1,x,y,xy\}$ of group-like elements of $H_8$. The corresponding Hopf subalgebra is $K=\mathbb{C}[G]\leq H=H_8$. 
It is easy to verify that $K$ is normal in the sense that $HK^+=K^+H$, thus $H\sslash K=H/HK^+$ is a Hopf algebra. Indeed, since $\dim (H/ HK^+) = 2$, it is clear that $H/ HK^+ \cong \mathbb{C}[\mathbb{Z}_2]$, where one can take $[x]$ as a generator. Thus the local operator algebra is fusion-categorical Morita equivalent to $\mathbb{C}^{\mathbb{Z}_2}\otimes \mathbb{C}[G]$.

We can write down the ribbon operators in this case. Note that the Haar integral of $K$ is $h_K = \frac{1}{4}(1+x+y+xy)$. Let $\delta_a$ denote the dual element of $a\in K$. The edge operators are 
\begin{align*}
    L^{[h]}_+|w\rangle & = \frac{1}{4}|hw+hxw+hyw+hxyw\rangle, \\
    L^{[h]}_-|w\rangle & = \frac{1}{4}|wh+whx+why+whxy \rangle, \\
    T^{\delta_a}_+|w\rangle & = \sum \delta_{a,w^{(2)}}|w^{(1)}\rangle, \\
    T^{\delta_a}_-|w\rangle & = \sum \delta_{a,w^{(1)}}|w^{(2)}\rangle,
\end{align*}
where $\delta_{i,j}=1$ if $i=j$ and $0$ otherwise. 
For $[h]\in H/HK^+$ and $\delta_a\in\hat{K}$, the ribbon operator $F^{[h],\delta_a}(\rho)$ is given by  
\begin{align*}
     & 
    \begin{aligned}
        \begin{tikzpicture}
            \draw[-latex, line width = 1.6pt] (0,0) -- (0,1.5); 
            \draw[-latex, black] (1.5,1.5) -- (0,1.5);
            \draw[-latex, black] (1.5,1.5) -- (1.5,0);
            \draw[-latex, black] (3,1.5) -- (1.5,1.5);
            \draw[-latex, black] (3,1.5) -- (3,0);
            \draw[-latex, black] (4.5,1.5) -- (3,1.5);
            \draw[-latex, black] (4.5,1.5) -- (4.5,0);
            \draw[-latex, black] (1.5,0) -- (0,0);
            \draw[-latex, black] (3,0) -- (1.5,0);
            \draw[-latex, black] (4.5,0) -- (3,0); 
            \draw[red, line width = 0.6pt] (0,1.5) -- (0.75,0.75);
            \draw[red, line width = 0.6pt] (1.5,1.5) -- (0.75,0.75);
            \draw[red, line width = 0.6pt] (1.5,1.5) -- (2.25,0.75);
            \draw[red, line width = 0.6pt] (3,1.5) -- (2.25,0.75);
            \draw[red, line width = 0.6pt] (3,1.5) -- (3.75,0.75);
            \draw[red, line width = 0.6pt] (4.5,1.5) -- (3.75,0.75);
            \draw[red, line width = 0.6pt] (4.5,1.5) -- (5.25,0.75); 
            \draw[dashed, black] (5.25,0.75) -- (0.75,0.75); 
            \node[ line width=0.2pt, dashed, draw opacity=0.5] (a) at (0.75,1.8){$b_3$};
            \node[ line width=0.2pt, dashed, draw opacity=0.5] (a) at (2.25,1.8){$b_2$};
            \node[ line width=0.2pt, dashed, draw opacity=0.5] (a) at (3.75,1.8){$b_1$};
            \node[ line width=0.2pt, dashed, draw opacity=0.5] (a) at (1.3,0.5){$a_3$};
            \node[ line width=0.2pt, dashed, draw opacity=0.5] (a) at (2.8,0.5){$a_2$};
            \node[ line width=0.2pt, dashed, draw opacity=0.5] (a) at (4.3,0.5){$a_1$};
            \node[ line width=0.2pt, dashed, draw opacity=0.5] (a) at (0.3,0.9){$s_b$};
            \node[ line width=0.2pt, dashed, draw opacity=0.5] (a) at (2.28,1.125){$\rho$};
            \node[ line width=0.2pt, dashed, draw opacity=0.5] (a) at (-1,0.75){$F^{[h],\delta_{a}}(\rho)$};
        \end{tikzpicture}
    \end{aligned} 
    \\ & 
    \begin{aligned}
        \begin{tikzpicture}
            \draw[-latex, line width = 1.6pt] (0,0) -- (0,1.5); 
            \draw[-latex, black] (1.5,1.5) -- (0,1.5);
            \draw[-latex, black] (1.5,1.5) -- (1.5,0);
            \draw[-latex, black] (3,1.5) -- (1.5,1.5);
            \draw[-latex, black] (3,1.5) -- (3,0);
            \draw[-latex, black] (4.5,1.5) -- (3,1.5);
            \draw[-latex, black] (4.5,1.5) -- (4.5,0);
            \draw[-latex, black] (1.5,0) -- (0,0);
            \draw[-latex, black] (3,0) -- (1.5,0);
            \draw[-latex, black] (4.5,0) -- (3,0); 
            \draw[red, line width = 0.6pt] (0,1.5) -- (0.75,0.75);
            \draw[red, line width = 0.6pt] (1.5,1.5) -- (0.75,0.75);
            \draw[red, line width = 0.6pt] (1.5,1.5) -- (2.25,0.75);
            \draw[red, line width = 0.6pt] (3,1.5) -- (2.25,0.75);
            \draw[red, line width = 0.6pt] (3,1.5) -- (3.75,0.75);
            \draw[red, line width = 0.6pt] (4.5,1.5) -- (3.75,0.75);
            \draw[red, line width = 0.6pt] (4.5,1.5) -- (5.25,0.75); 
            \draw[dashed, black] (5.25,0.75) -- (0.75,0.75); 
            \node[ line width=0.2pt, dashed, draw opacity=0.5] (a) at (0.75,1.9){$b_3^{(1)}$};
            \node[ line width=0.2pt, dashed, draw opacity=0.5] (a) at (2.25,1.9){$b_2^{(1)}$};
            \node[ line width=0.2pt, dashed, draw opacity=0.5] (a) at (3.75,1.9){$b_1^{(1)}$};
            \node[ line width=0.2pt, dashed, draw opacity=0.5] (a) at (1.3,0.5){$a'_3$};
            \node[ line width=0.2pt, dashed, draw opacity=0.5] (a) at (2.8,0.5){$a'_2$};
            \node[ line width=0.2pt, dashed, draw opacity=0.5] (a) at (4.3,0.5){$a'_1$};
            \node[ line width=0.2pt, dashed, draw opacity=0.5] (a) at (0.3,0.9){$s_b$};
            \node[ line width=0.2pt, dashed, draw opacity=0.5] (a) at (2.28,1.125){$\rho$};
            \node[ line width=0.2pt, dashed, draw opacity=0.5] (a) at (-2.5,0.75){$\displaystyle = \sum_{g,k\in G} 
    \delta_{a,kgb_3^{(2)}}\delta_{g,b_2^{(2)}}\delta_{k,b_1^{(2)}}$};
        \end{tikzpicture} 
    \end{aligned} 
\end{align*}
where 
\begin{gather*}
    a'_1=a_1hh_K,\quad  
    a'_2=a_2k^{-1}hkh_K,\quad 
    a'_3=a_3(gk)^{-1}h(gk)h_K. 
\end{gather*}

\end{example}

\subsubsection{Condensations}

Now let us characterize the algebras that give the condensations of the boundary model. One can move bulk excitations to the boundary by applying bulk-to-boundary ribbon operators. Since the stabilizer operators of the Hamiltonian on the boundary are different from those in the bulk, a bulk excitation may be condensed on the boundary sites. Our goal is to determine the subalgebra of the ribbon operator algebra which gives condensations. The group case has been studied in \cite{Beigi2011the}.

Fix a bulk-to-boundary ribbon $\rho=\rho_{\downarrow}$, say of type-B, connecting a bulk site $s$ and a boundary site $s_b=(v_b,f_b)$. Let $\mathcal{A}_\rho$ be the algebra of ribbon operators $F^{[h],\varphi}(\rho)$, and let $\mathcal{C}_\rho\subset \mathcal{A}_\rho$ be the subalgebra of ribbon operators that commute with both $A^K_{v_b}$ and $B^{H}_{f_b}$:   
\begin{equation}
    \mathcal{C}_\rho = \left\{F\in\mathcal{A}_\rho\,|\,[F,A^K_{v_b}]=[F,B^{H}_{f_b}]=0\right\}. 
\end{equation}
Applying elements in $\mathcal{C}_\rho$ to the ground state of the Hamiltonian will create excitations on $s$ but not on $s_b$. 

Let $F=\sum_{h,\varphi}c_{h,\varphi}F^{[h],\varphi}(\rho)$ be in $\mathcal{C}_\rho$ with $c_{h,\varphi}$ complex coefficients. By Eq.~\eqref{eq:bdd-comm7}, $[F,A^K_{v_b}]=0$ is equivalent to
\begin{equation}
\sum_{h,\varphi}c_{h,\varphi}F^{[h],\varphi}(\rho)A^{h_K}(s_b)=\sum_{h,\varphi}\sum_{(h_K)}c_{h,\varphi}F^{[h],\varphi(\bullet h_K^{(1)})}(\rho)A^{h_K^{(2)}}(s_b). 
\end{equation}
Multiplying $A^{h_K}(s_b)$ on both sides from the right, and using $A^k(s_b)A^l(s_b)=A^{kl}(s_b)$, $kh_K = \varepsilon(k)h_K$ and $h_K^2=h_K$, the above identity holds if and only if 
\begin{equation}
c_{h,\varphi} = c_{h,\varphi(\bullet h_K)}. 
\end{equation}
Similarly, $[F,B^H_{f_b}]=0$ if and only if $c_{h,\varphi} = 0$ for $h\notin K$. 
By the above computation, we conclude that $\mathcal{C}_\rho$ is generated by $\{F^{[k],\varphi(\bullet h_K)}(\rho)~|~k\in K, \varphi\in\hat{K}\}.$ This algebra determines all excitations at $s$ condensed on the boundary site $s_b$.

\subsection{Recover the bulk phase from the boundary phase}

From the boundary-bulk duality, the bulk phase $\ED$ can be recovered from boundary phase $\EB$ by taking monoidal center: $\ED \simeq \mathcal{Z}(\EB)$.
It is interesting to consider how to realize this point of view of boundary-bulk duality in the lattice model.
This, to our knowledge, has not been systematically discussed before.

Here we will follow the approach initially proposed by Kong \cite{Kong2019private,Kong2014} to recover the bulk anyon from its boundary condensate in the lattice model.
We know that for bulk anyons $x$, the physical information it carries is: topological charge $x$, fusion rule $[N_x]=N_{x,y}^z$, $F$-symbol, quantum dimension $d_x$, topological spin $\theta_x$ and braiding $c_x=\{c_{x,y}:x\otimes y\to y\otimes x\}_y$.
After anyon condensation, the boundary anyon $X=\mathbf{Cond}(x)=x_1\oplus \cdots \oplus x_n$ is a condensate of several kinds of bulk anyons and they all lose their braiding information.
To recover the bulk anyons from their boundary condensate, we need to add their braiding information back.
This intuition can be formalized as the following rigorous construction \cite{turaev2017monoidal}:

\begin{definition}
For a given monoidal category $(\EC,\otimes, \mathds{1},\alpha, l, r)$, a left half braiding is a pair $(x,\beta_{x,\bullet})$, where $x\in \Obj(\EC)$ and $\beta_{x,\bullet}=\{\beta_{x,y}: x\otimes y \to y \otimes x\}_{y\in\Obj(\EC)}$ is a natural isomorphism between functors $x\otimes \bullet$ and $\bullet\otimes x$, that is, for any $y,z \in \Obj(\EC)$ and $f:y\to z$, we have $(f\otimes \id_x) \comp \beta_{x,y}= \beta_{x,z}\comp (\id_x\otimes f)$, and which is $\otimes$-multiplicative in the sense that for $\forall~y,z \in \Obj(\EC)$, the following diagram commutes:
\begin{equation}\label{}
 \begin{split}
  \xymatrix{
(x\otimes y)\otimes z  \ar[d]^{\alpha_{x,y,z}}\ar[r]^{\beta_{x,y}\otimes \id_z} & (y\otimes x)\otimes z\ar[r]^{\alpha_{y,x,z}} & y\otimes (x\otimes z) \ar[d]^{\id_y\otimes \beta_{x,z}} \\
x\otimes (y\otimes z) \ar[r]_{\beta_{x,y\otimes z}}          &(y\otimes z)\otimes x \ar[r]_{\alpha_{y,z,x}}  &  y\otimes (z\otimes x) }
\end{split}.
\end{equation}
\end{definition}

The definition of left half braiding implies that $\beta_{x,\mathds{1}}=\id_x$; we call $x$ the underlying object of the left half braiding $(x,\beta_{x,\bullet})$.

\begin{definition}
The monoidal center $\mathcal{Z}(\EC)$ of a monoidal category $\EC$ is a category defined as follows.
\begin{enumerate}
  \item The objects of $\mathcal{Z}(\EC)$ are left half braidings of $\EC$, i.e., $\Obj(\mathcal{Z}(\EC))=\{(x,\beta_{x,\bullet})\}$;
  \item A morphism $(x,\beta_{x,\bullet})\to (y,\beta_{y,\bullet})$ is a morphism $f:x\to y$ in $\EC$ such that for all $z\in \Obj(\EC)$, the following diagram commutes
      \begin{equation}\label{}
        \begin{split}
       \xymatrix{
       x\otimes z  \ar[d]^{\beta_{x,z}}\ar[r]^{f\otimes \id_z}  & y\otimes z \ar[d]^{\beta_{y,z}} \\
       z\otimes x \ar[r]_{\id_z\otimes f}           &  z\otimes y }
        \end{split};
      \end{equation}
  \item The unit object of $\mathcal{Z}(\EC)$ is $(\mathds{1},\beta_{\mathds{1},\bullet}=\{\beta_{\mathds{1},x}=\id_x\})$;
  \item $\mathcal{Z}(\EC)$ is a monoidal category: the monoidal structure is given by $(x,\beta_{x,\bullet})\otimes (y,\beta_{y,\bullet})=(x\otimes y, (\beta_{x,\bullet}\otimes \id_y)\comp (\id_x\otimes \beta_{y,\bullet}))$, i.e., for any $z\in \Obj(\EC)$, the following diagram commutes
      \begin{equation}\label{}
 \begin{split}
  \xymatrix{
x\otimes (y\otimes z)  \ar[r]^{\id_x\otimes \beta_{y,z}} & x \otimes (z\otimes y) \ar[r]^{\alpha^{-1}_{x,z,y}} & (x\otimes z)\otimes y\ar[d]^{\beta_{x,z}\otimes \id_y} \\
(x\otimes y)\otimes z \ar[u]^{\alpha_{x,y,z}}\ar[r]_{\beta_{x\otimes y, z}}          &z\otimes(x\otimes y) \ar[r]_{\alpha^{-1}_{z,x,y}}  &  (z\otimes x)\otimes y }
\end{split};
\end{equation}
  \item $\mathcal{Z}(\EC)$ is a braided monoidal category: the braiding is given by $c_{(x,\beta_{x,\bullet}), (y,\beta_{y,\bullet})}=\beta_{x,y}$.
\end{enumerate}
\end{definition}

By definition, an anyon in the bulk phase (seen from the boundary-to-bulk perspective)  $\mathcal{Z}(\EB)$ is a pair $(X,\beta_{X}=\{\beta_{X,Y}:X\otimes Y \to Y\otimes X\}_{Y\in\Obj(\EB)})$, where $\beta_{X}$ is a half braiding. 
As illustrated in Fig.~\ref{fig:bdRib}, we could consider two situations: (i) the direct condensation of two bulk anyons $x,y$ and obtain $X\otimes Y$ in the boundary (this is realized by two bulk-to-boundary ribbon operators $F_X(\rho^1_{\downarrow})$ and $F_Y(\sigma^1_{\downarrow})$ with $\rho_{\downarrow}\cap \sigma_{\downarrow}=\emptyset$); and (ii) first braid $x,y$, then condense them and obtain $Y\otimes X$ in the boundary (in this case, two corresponding bulk-to-boundary ribbons have non-empty overlap $\rho^2_{\downarrow}\cap \sigma^2_{\downarrow}\neq \emptyset$). 
The braiding information is thus encoded in the morphism 
$\beta_{X,Y}:X\otimes Y \to Y\otimes X$. Since the boundary phase is described by a UFC $\EB$, we could use the boundary ribbon operator $F_{Y}(\xi)$ to realize this process.
$F_Y(\xi)$ create $Y,\bar{Y}$ pairs at two ends of  $\xi$, then we drag $\bar{Y}$ to the condensed anyon $Y$ and fuse them to vacuum. 
Suppose that $|\Omega\rangle$ is the ground state, then we have
\begin{equation} \label{eq:recover-braiding}
    F_Y(\xi)(F_{Y}(\sigma^1_{\downarrow}) F_X(\rho^1_{\downarrow})|\Omega\rangle)\cong F_{Y}(\sigma^2_{\downarrow}) F_X(\rho^2_{\downarrow})|\Omega\rangle,
\end{equation}
where `$\cong$' means that they are topologically equivalent.

\begin{figure}
  \centering
  \includegraphics[width=10cm]{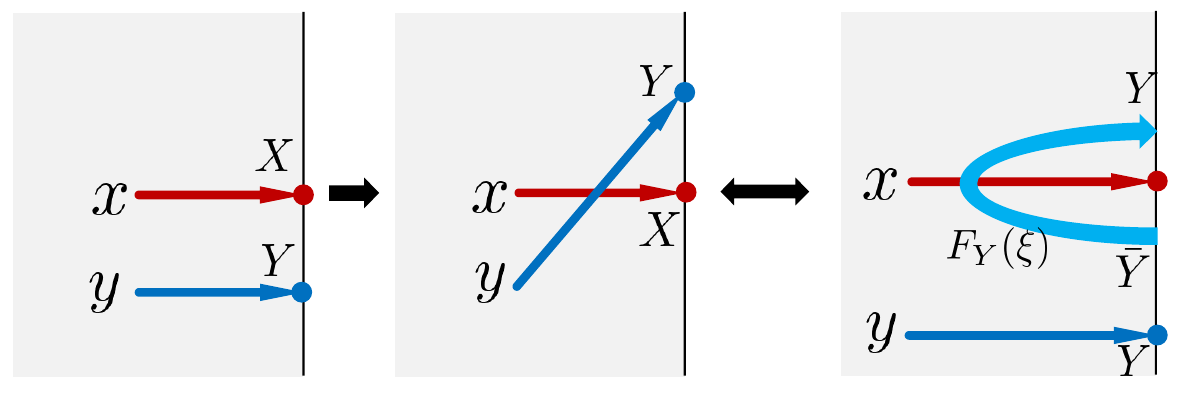}\\
  \caption{Illustration of the boundary ribbon operator. The bulk anyon braiding information is contained in the boundary ribbon.}\label{fig:bdRib}
\end{figure}

A simple example of toric code is given by Kong \cite{Kong2019private}.
For this case $H=\mathbb{C}[\mathbb{Z}_2]$, and the bulk phase is $\ED=\{\one,e,m,\varepsilon\}$.
There are two inequivalent gapped boundaries \cite{bravyi1998quantum}: rough boundary $\EuScript{B}_r=\{I,M\}$, where $I$ is the rough boundary vacuum and $M$ is the boundary magnetic charge (bulk anyons $m,\varepsilon$ condense to $M$, and $\one,e$ condense to $I$),
and smooth boundary $\EuScript{B}_s=\{I,E\}$,  where $I$ represents the boundary vacuum and $E$ represents the boundary electric charge (bulk anyons $e,\varepsilon$ condense to $E$, and $\one,m$ condense to $I$).

To recover the bulk phase from the rough boundary, we first need to calculate the monoidal center $\mathcal{Z}(\EuScript{B}_r)$. 
Each object in $\mathcal{Z}(\EuScript{B}_r)$ is a pair $(x, \beta_{x,\bullet})$. So our goal is to determine each object $x$ and isomorphism $\{\beta_{x,y}:x\otimes y \to y\otimes x\}_{y}$  corresponding to each bulk anyon. 
Notice here $x,y=I, M$; we will give a physical argument based on the lattice model. First, since the bulk particles $\one$ and $e$ condensed as vacuum of $\EuScript{B}_r$, they must be of the forms $(I, \beta_{I,\bullet})$ and $(I, \rho_{I,\bullet})$; similarly, $m$ and $\varepsilon$ are of the forms $(M, \beta_{M,\bullet})$ and $(M, \rho_{M,\bullet})$. Now we need to specify the half-braidings. Firstly, let us take $m$ as an example, the braidings involving $m$ are
\begin{equation}\label{}
\begin{split}
c_{m,\one}&=c_{(M,\beta_{M,\bullet}),(I,\beta_{I,\bullet})}=\beta_{M,I},\\
c_{m,m}&=c_{(M,\beta_{M,\bullet}),(M,\beta_{M,\bullet})}=\beta_{M,M},\\
c_{e,m}&=c_{(I,\rho_{I,\bullet}),(M,\beta_{M,\bullet})}=\rho_{I,M},\\
c_{\varepsilon,m}&=c_{(M,\rho_{M,\bullet}),(M,\beta_{M,\bullet})}=\rho_{M,M}.
\end{split}
\end{equation}
If we write down all $c_{x,y}$ for $\one$, $m$, $e$ and $\varepsilon$, comparing with what we have known for $c_{x,y}$, we can easily get the following result \cite{Kong2019private}:
\begin{equation}\label{}
\begin{split}
\mathds{1} & =(I,I\otimes I=I \xrightarrow{\id_{I}} I=I\otimes I,I\otimes M= M\xrightarrow{\id_{M}}M=I\otimes M), \\
e    &  = (I,I\otimes I=I \xrightarrow{\id_{I}} I=I\otimes I,I\otimes M= M\xrightarrow{-\id_{M}}M=I\otimes M),\\
m   &  = (M,M\otimes I=M\xrightarrow{\id_{M}}M=I\otimes M, M\otimes M=I \xrightarrow{\id_{I}}I=M\otimes M),\\
\varepsilon    &  = (M,M\otimes I=M\xrightarrow{\id_{M}}M=I\otimes M, M\otimes M=I \xrightarrow{-\id_{I}}I=M\otimes M).
\end{split}
\end{equation}

In general, it is believed that braiding has no physical observable effect, we can only observe the mutual statistics that corresponding to double braiding. But we shall argue that for a quantum double with boundary, braiding has its physical reality. This can be seen from toric code boundary model.
We first calculate the statistics of dragging one $e$ particle around a boundary $m$ vortex. Let $|\Psi(m,f_b)\rangle$ be an excited state containing a magnetic vortex at the boundary face $f_b$, that is, $B_{f_b}\vPsi =-\vPsi$. 
The boundary ribbon operator $F_{e}(\xi)$ in this case is just a $Z$-string:
\begin{equation}
    F_{e}(\xi)=\prod_{j\in \xi} \sigma^z_j.
\end{equation}
Notice that $f_b$ must be surrounded by $\xi$.
When implementing this on the boundary excited state $|\Psi(m,f_b)\rangle$, we have
\begin{equation}\label{}
|\Psi(m,f_b)\rangle \mapsto  F_{e}(\xi) |\Psi(m,f_b)\rangle =\prod_{f\, \mathrm{inside}\, \xi}B_f|\Psi(m,f_b)\rangle=-|\Psi(m,f_b)\rangle.
\end{equation}
The boundary ribbon operator $ F_{e}(\xi)$ realizes the half braiding between $m$ and $e$ and it has a physical observable effect.
All other half braidings can be realized in a similar way.

For the general Hopf quantum double boundary, the generalization is straightforward.
To recover the bulk phase from the boundary phase, we implement a boundary-to-bulk ribbon operator $F_X(\rho_{\uparrow})$ which connects boundary site $s_b$
and bulk site $s$. At boundary site $s_b$, there is a boundary excitation $X$ which is a condensate of several bulk anyons $x_i$ (e.g., $M$ is a condensate of $m,\varepsilon$ in toric code).
We have a collection of boundary ribbon operators 
\begin{equation}
    \mathcal{F}_{\EB}(\xi)=\{F_{Y}(\xi),Y\in \EB\}.
\end{equation}
Choosing different values from equivalent classes of $[h],[\varphi]$ for $F_X(\rho_{\uparrow})$ we obtain different bulk anyons.

\section{Gapped domain wall theory}
\label{sec:domainwall}

Let us now consider the gapped domain wall, which is closely related to the gapped boundary: the gapped boundary can be regarded as a domain wall that separates the Hopf quantum double phase from the vacuum phase; conversely, using the folding trick, a gapped domain wall can be transformed into a gapped boundary.
In this section, let us generalize the construction of the lattice model for gapped boundaries to the case of general domain walls and examine the correspondence between boundaries and domain walls.
We will see that the domain wall is characterized by an $H_1|H_2$-bicomodule algebra $\FB$ (or equivalently $H_1|H_2$-bimodule algebra $\FN$):
	
 \begin{definition}
Let $H_1,H_2$ be Hopf algebras, $\mathfrak{B}$ be an algebra, and $\mathfrak{B}$ be an $H_1|H_2$-bicomodule with coaction $\beta_{\mathfrak{B}}:\mathfrak{B}\to H_1\otimes \mathfrak{B}\otimes H_2$. If $\beta_{\mathfrak{B}}$ is a homomorphism of algebras, then $\mathfrak{B}$ is called an $H_1|H_2$-bicomodule algebra.
\end{definition}

\begin{remark}
    For an $H_1|H_2$-bicomodule algebra $\mathfrak{B}$ with coaction $\beta_{\mathfrak{B}}$, the Sweedler's notation reads $\beta_{\mathfrak{B}}(x) = \sum_{[x]}x^{[-1]}\otimes x^{[0]}\otimes x^{[1]} \in H_1\otimes \mathfrak{B}\otimes H_2$. 
\end{remark}

\begin{definition}
Let $H_1,H_2$ be Hopf algebras, $\mathfrak{N}$ be an algebra, and $\FN$ be an $H_1|H_2$-bimodule. $\FN$ is called a bimodule algebra if $h \triangleright (xy) \triangleleft g=\sum_{(h),(g)}(h^{(1)} \triangleright x \triangleleft g^{(2)}) (h^{(2)} \triangleright y \triangleleft g^{(1)})$.
\end{definition}

\subsection{Domain wall lattice}

	\begin{figure}[t]
		\centering
	(a)	\includegraphics[align=c,width=7cm]{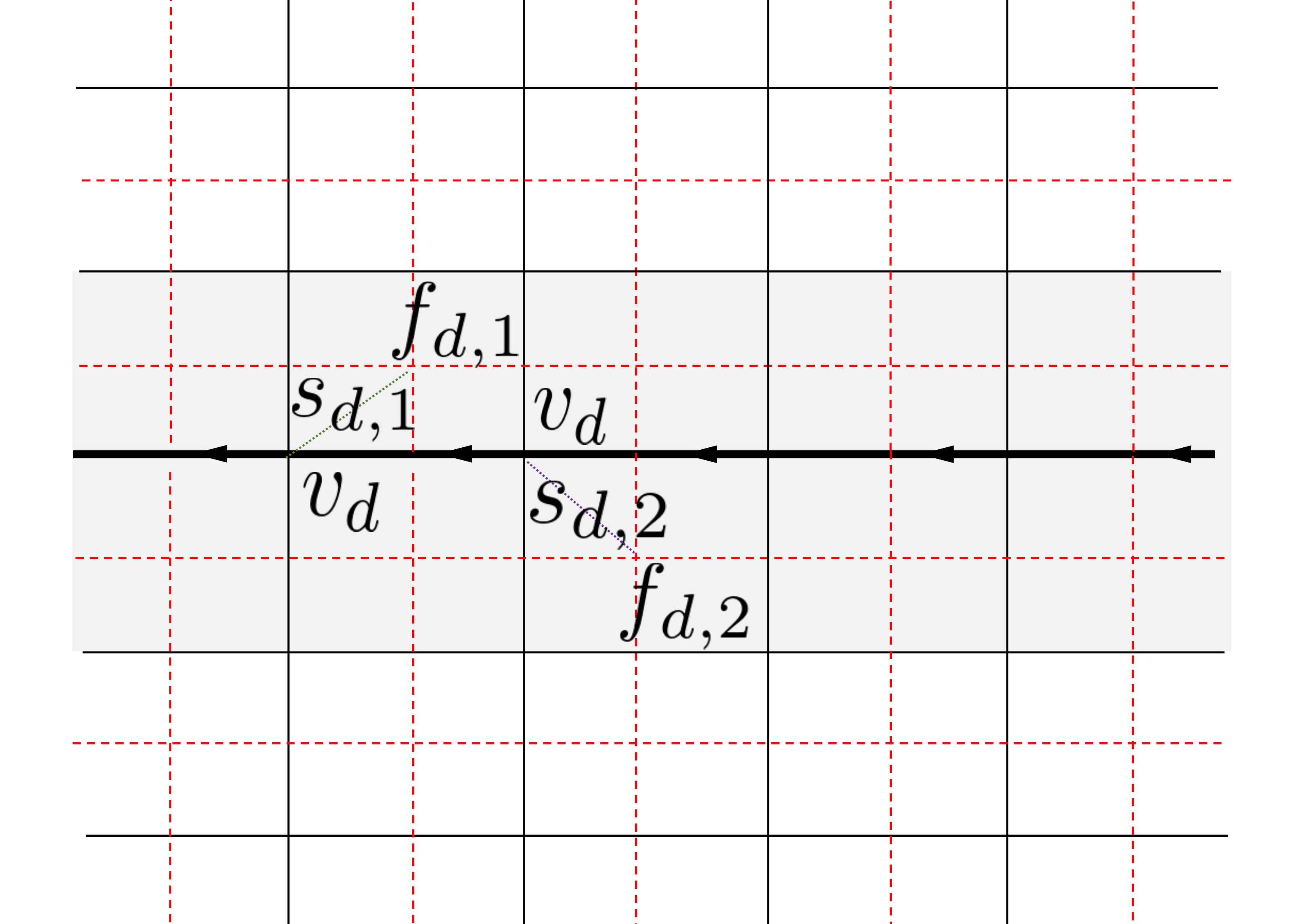}
	(b)	\includegraphics[align=c,width=6.5cm]{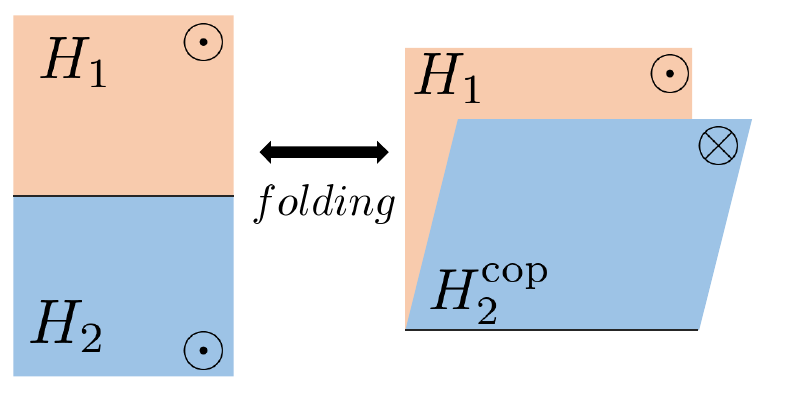}
		\caption{The depictions of (a) the domain wall, domain wall face, domain wall vertex, and domain wall site, and (b) the folding trick. \label{fig:wall}}
	\end{figure} 

Suppose that two $2d$ quantum double phases, which are respectively characterized by Hopf algebras $H_1$ and $H_2$, are separated by a $1d$ domain wall, see Fig.~\ref{fig:wall} (a). Suppose that the domain wall is determined by an $H_1|H_2$-bicomodule algebra $\FB$. The boundary lattice realization in Sec.~\ref{sec:bdd-II} can be generalized to domain wall straightforwardly. By convention, when walking along the direction on the domain wall edges, the $H_1$-bulk is always on the left-hand side while the $H_2$-bulk is always on the right-hand side. According to how the edges on the two bulks are connected to the domain wall vertices, there are four configurations. There are four corresponding sites for each of these configurations. We define the domain wall vertex operator $A^{a\otimes b}(s_{d_i})$ ($i=1,2$) for $a\otimes b\in \mathfrak{B}\otimes\mathfrak{B}^{\rm op}$ as follows: 
\begin{align}
     \begin{aligned}
    \begin{tikzpicture}
				\draw[-latex,black,line width = 1.6pt] (0,0) -- (0,1);
    	        \draw[red,line width = 1pt] (0,1) -- (0.5,0.5); 
             \draw[red,line width = 1pt] (0,1) -- (-0.5,0.5); 
    			\draw[-latex,black,line width = 1.6pt] (0,1) -- (0,2); 
				\draw[-latex,black] (1,1) -- (0,1); 
                    \draw[-latex,black] (0,1) -- (-1,1); 
				\node[ line width=0.2pt, dashed, draw opacity=0.5] (a) at (-0.2,0.5){$x$};
    		\node[ line width=0.2pt, dashed, draw opacity=0.5] (a) at (-0.2,1.5){$y$};
            \node[ line width=0.2pt, dashed, draw opacity=0.5] (a) at (0.5,1.2){$h$};
            \node[ line width=0.2pt, dashed, draw opacity=0.5] (a) at (-0.5,1.2){$k$};
            \node[ line width=0.2pt, dashed, draw opacity=0.5] (a) at (0.85,0.4){$s_{d_2}$}; 
            \node[ line width=0.2pt, dashed, draw opacity=0.5] (a) at (-0.8,0.4){$s_{d_1}$}; 
	\end{tikzpicture}  \end{aligned} & \quad \quad  A^{a\otimes b}(s_{d_i})|x,k,y,h\rangle=\sum_{[b]} |ax,kb^{[-1]},yb^{[0]},S(b^{[1]})h\rangle, \label{eq:bdd-sta-D1}\\
       \begin{aligned}
    \begin{tikzpicture}
				\draw[-latex,black,line width = 1.6pt] (0,0) -- (0,1);
    	        \draw[red,line width = 1pt] (0,1) -- (0.5,1.5); 
             \draw[red,line width = 1pt] (0,1) -- (-0.5,1.5); 
    			\draw[-latex,black,line width = 1.6pt] (0,1) -- (0,2); 
				\draw[-latex,black] (1,1) -- (0,1); 
                    \draw[-latex,black] (0,1) -- (-1,1); 
				\node[ line width=0.2pt, dashed, draw opacity=0.5] (a) at (-0.2,0.5){$x$};
    		\node[ line width=0.2pt, dashed, draw opacity=0.5] (a) at (-0.2,1.5){$y$};
            \node[ line width=0.2pt, dashed, draw opacity=0.5] (a) at (0.85,1.4){$s_{d_2}$};
            \node[ line width=0.2pt, dashed, draw opacity=0.5] (a) at (-0.8,1.4){$s_{d_1}$};
            \node[ line width=0.2pt, dashed, draw opacity=0.5] (a) at (0.5,0.8){$h$};
            \node[ line width=0.2pt, dashed, draw opacity=0.5] (a) at (-0.5,0.8){$k$};
	\end{tikzpicture} \end{aligned} & \quad \quad A^{a\otimes b}(s_{d_i})|x,k,y,h\rangle=\sum_{[a]} | a^{[0]}x, kS(a^{[-1]}), yb,a^{[1]}h\rangle,  \label{eq:bdd-sta-D2}\\
       \begin{aligned}
    \begin{tikzpicture}
				\draw[-latex,black,line width = 1.6pt] (0,0) -- (0,1); 
        	    \draw[red,line width = 1pt] (0,1) -- (0.5,0.5); 
                \draw[red,line width = 1pt] (0,1) -- (-0.5,0.5); 
    			\draw[-latex,black,line width = 1.6pt] (0,1) -- (0,2); 
				\draw[-latex,black] (0,1) -- (1,1); 
                    \draw[-latex,black] (-1,1) -- (0,1); 
				\node[ line width=0.2pt, dashed, draw opacity=0.5] (a) at (-0.2,0.5){$x$};
    		\node[ line width=0.2pt, dashed, draw opacity=0.5] (a) at (-0.2,1.5){$y$};
            \node[ line width=0.2pt, dashed, draw opacity=0.5] (a) at (0.5,1.2){$h$};
            \node[ line width=0.2pt, dashed, draw opacity=0.5] (a) at (-0.5,1.2){$k$};
            \node[ line width=0.2pt, dashed, draw opacity=0.5] (a) at (0.85,0.4){$s_{d_2}$}; 
            \node[ line width=0.2pt, dashed, draw opacity=0.5] (a) at (-0.8,0.4){$s_{d_1}$};
	\end{tikzpicture}\end{aligned} & \quad\quad A^{a\otimes b} (s_{d_i})|x,k,y,h\rangle=\sum_{[b]}|ax, S(b^{[-1]})k, yb^{[0]},hb^{[1]}\rangle, \label{eq:bdd-sta-D3} \\
      \begin{aligned} \begin{tikzpicture}
				\draw[-latex,black,line width = 1.6pt] (0,0) -- (0,1); 
        	    \draw[red,line width = 1pt] (0,1) -- (0.5,1.5); 
                    \draw[red,line width = 1pt] (0,1) -- (-0.5,1.5); 
    			\draw[-latex,black,line width = 1.6pt] (0,1) -- (0,2); 
				\draw[-latex,black] (0,1) -- (1,1); 
                    \draw[-latex,black] (-1,1) -- (0,1);
				\node[ line width=0.2pt, dashed, draw opacity=0.5] (a) at (-0.2,0.5){$x$};
    		\node[ line width=0.2pt, dashed, draw opacity=0.5] (a) at (-0.2,1.5){$y$};
            \node[ line width=0.2pt, dashed, draw opacity=0.5] (a) at (0.5,0.75){$h$};
            \node[ line width=0.2pt, dashed, draw opacity=0.5] (a) at (-0.5,0.8){$k$};
            \node[ line width=0.2pt, dashed, draw opacity=0.5] (a) at (0.85,1.4){$s_{d_2}$};
            \node[ line width=0.2pt, dashed, draw opacity=0.5] (a) at (-0.8,1.4){$s_{d_1}$};
	\end{tikzpicture}  \end{aligned}  &\quad \quad A^{a\otimes b} (s_{d_i})|x,k,y,h\rangle=\sum_{[a]}|a^{[0]}x,a^{[-1]}k, yb,hS(a^{[1]})\rangle, \label{eq:bdd-sta-D4} \\ \begin{aligned}
    \begin{tikzpicture}
				\draw[-latex,black,line width = 1.6pt] (0,0) -- (0,1);
    	        \draw[red,line width = 1pt] (0,1) -- (0.5,0.5); 
             \draw[red,line width = 1pt] (0,1) -- (-0.5,0.5); 
    			\draw[-latex,black,line width = 1.6pt] (0,1) -- (0,2); 
				\draw[-latex,black] (0,1) -- (1,1); 
                    \draw[-latex,black] (0,1) -- (-1,1); 
				\node[ line width=0.2pt, dashed, draw opacity=0.5] (a) at (-0.2,0.5){$x$};
    		\node[ line width=0.2pt, dashed, draw opacity=0.5] (a) at (-0.2,1.5){$y$};
            \node[ line width=0.2pt, dashed, draw opacity=0.5] (a) at (0.5,1.2){$h$};
            \node[ line width=0.2pt, dashed, draw opacity=0.5] (a) at (-0.5,1.2){$k$};
            \node[ line width=0.2pt, dashed, draw opacity=0.5] (a) at (0.85,0.4){$s_{d_2}$}; 
            \node[ line width=0.2pt, dashed, draw opacity=0.5] (a) at (-0.8,0.4){$s_{d_1}$}; 
	\end{tikzpicture}  \end{aligned} & \quad \quad  A^{a\otimes b}(s_{d_i})|x,k,y,h\rangle=\sum_{[b]} |ax,kb^{[-1]},yb^{[0]},hb^{[1]}\rangle, \label{eq:bdd-sta-D5}\\
       \begin{aligned}
    \begin{tikzpicture}
				\draw[-latex,black,line width = 1.6pt] (0,0) -- (0,1);
    	        \draw[red,line width = 1pt] (0,1) -- (0.5,1.5); 
             \draw[red,line width = 1pt] (0,1) -- (-0.5,1.5); 
    			\draw[-latex,black,line width = 1.6pt] (0,1) -- (0,2); 
				\draw[-latex,black] (0,1) -- (1,1); 
                    \draw[-latex,black] (0,1) -- (-1,1); 
				\node[ line width=0.2pt, dashed, draw opacity=0.5] (a) at (-0.2,0.5){$x$};
    		\node[ line width=0.2pt, dashed, draw opacity=0.5] (a) at (-0.2,1.5){$y$};
            \node[ line width=0.2pt, dashed, draw opacity=0.5] (a) at (0.85,1.4){$s_{d_2}$};
            \node[ line width=0.2pt, dashed, draw opacity=0.5] (a) at (-0.8,1.4){$s_{d_1}$};
            \node[ line width=0.2pt, dashed, draw opacity=0.5] (a) at (0.5,0.8){$h$};
            \node[ line width=0.2pt, dashed, draw opacity=0.5] (a) at (-0.5,0.8){$k$};
	\end{tikzpicture} \end{aligned} & \quad \quad A^{a\otimes b}(s_{d_i})|x,k,y,h\rangle=\sum_{[a]} | a^{[0]}x, kS(a^{[-1]}), yb,hS(a^{[1]})\rangle,  \label{eq:bdd-sta-D6}\\
       \begin{aligned}
    \begin{tikzpicture}
				\draw[-latex,black,line width = 1.6pt] (0,0) -- (0,1); 
        	    \draw[red,line width = 1pt] (0,1) -- (0.5,0.5); 
                \draw[red,line width = 1pt] (0,1) -- (-0.5,0.5); 
    			\draw[-latex,black,line width = 1.6pt] (0,1) -- (0,2); 
				\draw[-latex,black] (1,1) -- (0,1); 
                    \draw[-latex,black] (-1,1) -- (0,1); 
				\node[ line width=0.2pt, dashed, draw opacity=0.5] (a) at (-0.2,0.5){$x$};
    		\node[ line width=0.2pt, dashed, draw opacity=0.5] (a) at (-0.2,1.5){$y$};
            \node[ line width=0.2pt, dashed, draw opacity=0.5] (a) at (0.5,1.2){$h$};
            \node[ line width=0.2pt, dashed, draw opacity=0.5] (a) at (-0.5,1.2){$k$};
            \node[ line width=0.2pt, dashed, draw opacity=0.5] (a) at (0.85,0.4){$s_{d_2}$}; 
            \node[ line width=0.2pt, dashed, draw opacity=0.5] (a) at (-0.8,0.4){$s_{d_1}$};
	\end{tikzpicture}\end{aligned} & \quad\quad A^{a\otimes b} (s_{d_i})|x,k,y,h\rangle=\sum_{[b]}|ax, S(b^{[-1]})k, yb^{[0]},S(b^{[1]})h\rangle, \label{eq:bdd-sta-D7} \\
      \begin{aligned} \begin{tikzpicture}
				\draw[-latex,black,line width = 1.6pt] (0,0) -- (0,1); 
        	    \draw[red,line width = 1pt] (0,1) -- (0.5,1.5); 
                    \draw[red,line width = 1pt] (0,1) -- (-0.5,1.5); 
    			\draw[-latex,black,line width = 1.6pt] (0,1) -- (0,2); 
				\draw[-latex,black] (1,1) -- (0,1); 
                    \draw[-latex,black] (-1,1) -- (0,1);
				\node[ line width=0.2pt, dashed, draw opacity=0.5] (a) at (-0.2,0.5){$x$};
    		\node[ line width=0.2pt, dashed, draw opacity=0.5] (a) at (-0.2,1.5){$y$};
            \node[ line width=0.2pt, dashed, draw opacity=0.5] (a) at (0.5,0.75){$h$};
            \node[ line width=0.2pt, dashed, draw opacity=0.5] (a) at (-0.5,0.8){$k$};
            \node[ line width=0.2pt, dashed, draw opacity=0.5] (a) at (0.85,1.4){$s_{d_2}$};
            \node[ line width=0.2pt, dashed, draw opacity=0.5] (a) at (-0.8,1.4){$s_{d_1}$};
	\end{tikzpicture}  \end{aligned}  &\quad \quad A^{a\otimes b} (s_{d_i})|x,k,y,h\rangle=\sum_{[a]}|a^{[0]}x,a^{[-1]}k, yb,a^{[1]}h\rangle. \label{eq:bdd-sta-D8} 
\end{align}
The convention follows that in Sec.~\ref{sec:bulkQD}. 

To construct the lattice model of the domain wall, choose the symmetric separability idempotent $\lambda \in \FB\otimes \FB^{\rm op}$ of the algebra $\FB$ to form the domain wall vertex operator $A^\lambda(s_{d_i})$. 
The domain wall Hamiltonian is thus given as 
\begin{equation}
    H[C(\Sigma_d)] = \sum_{i=1,2}\sum_{d_i}(I-A^\lambda(s_{d_i})) + \sum_{i=1,2}\sum_{d_i}(I-B^{\varphi_{\hat{H}_i}}(s_{d_i})). 
\end{equation}

\subsection{Topological excitation of the domain wall}

In this part, we will establish the algebraic theory of the gapped domain wall and discuss the relation between the gapped domain wall and the gapped boundary in the language of Hopf algebras. The main results are summarized in Table \ref{tab:WallTop}. 
Similar to the boundary, we can analyze the domain wall excitations at two different but equivalent levels:

\begin{itemize}
    \item From the anyon condensation perspective, a domain wall that separates two bulks with topological order given by UMTCs $\ED_i$ ($i=1,2$) is a UFC $\EW$.
    The anyon condensation from $\ED_i$ to domain wall $\EW$ is controlled by a condensable algebra $A_i\in \ED_i$ (i.e., a connected separable commutative algebra) \cite{Kong2014}. The wall excitations are given by $\EW=\mathsf{Mod}_{A_i}(\ED_i)$.
    \item From the bulk UFC $\EC_i$ ($i=1,2$) perspective, two bulk topological order are given by $\ED_i=\mathcal{Z}(\EC_i)$.
    The domain wall is characterized by a $\EC_1|\EC_2$-bimodule category $\EM$.
    The domain wall excitations are given by the category of all module functors $\EW=\mathsf{Fun}_{\EC_1|\EC_2}(\EM,\EM)$.
    Notice that in $\EC_i$ bulk, $\EC_i$ can be regarded as a $\EC_i|\EC_i$-bimodule category; thus the bulk topological excitations can also be described by $\ED_i=\mathcal{Z}(\EC_i)\simeq \mathsf{Fun}_{\EC_i|\EC_i}(\EC_i,\EC_i)$.
\end{itemize}

For our construction, two bulk UFCs are $\EC_i=\mathsf{Rep}(H_i)$, and 
the domain wall is characterized by an $H_1|H_2$-bicomodule algebra $\mathfrak{B}$. 
The category ${_{\mathfrak{B}}}\mathsf{Mod}$ is a $\mathsf{Rep}(H_1)|\mathsf{Rep}(H_2)$-bimodule category. To see this, we only need to check that
for $X_i\in \mathsf{Rep}(H_i)$ and $M\in {_{\mathfrak{B}}}\mathsf{Mod}$,  $X_1\otimes M \otimes X_2$ is also a $\mathfrak{B}$-module.
In fact, it is easy to verify that the module structure map is given by (in diagrammatic representation)
\begin{equation}
    \mu_{X_1\otimes M \otimes X_2}= \begin{aligned}
        \begin{tikzpicture}
		\draw[black, line width=1.0pt] (0.5,0.5) arc (90:180:0.5);
		\draw[black, line width=1.0pt] (1,-1.4) arc (0:-180:0.5);
		\draw[black, line width=1.0pt] (1,0) arc (180:90:0.5);
		\draw[black, line width=1.0pt] (1.5,0)--(1.5,.8);
		\draw[black, line width=1.0pt] (2,0) arc (180:90:0.5);
	    \braid[width=0.5cm,
				height=0.3cm,
				line width=1.0pt,number of strands=4] (braid) a_1^{-1} a_3^{-1} a_2^{-1};
	    	\draw[black, line width=1.0pt] (0.5,0)--(.5,.8);
	    	\draw[black, line width=1.0pt] (0,0)--(0,-1.4);
	    	\draw[black, line width=1.0pt] (2.5,-2.2)--(2.5,0.8);
	    	\draw[black, line width=1.0pt] (0.5,-2.2)--(0.5,-1.3);
	    	\draw[black, line width=1.0pt] (2,-1.4)--(2,-2.2);
	    	\draw[black, line width=1.0pt] (1.5,-1.4)--(1.5,-2.2);
	    	\node[ line width=0.2pt, dashed, draw opacity=0.5] (a) at (0.5,-2.5){$\mathfrak{B}$};
	    	\node[ line width=0.2pt, dashed, draw opacity=0.5] (a) at (0,-2.1){$H_1$};
	    	\node[ line width=0.2pt, dashed, draw opacity=0.5] (a) at (1,-2.1){$H_2$};
	    	\node[ line width=0.2pt, dashed, draw opacity=0.5] (a) at (1.5,-2.5){$X_1$};
	    	\node[ line width=0.2pt, dashed, draw opacity=0.5] (a) at (2,-2.49){$M$};
	    	\node[ line width=0.2pt, dashed, draw opacity=0.5] (a) at (2.5,-2.5){$X_2$};
	    	\node[ line width=0.2pt, dashed, draw opacity=0.5] (a) at (0.55,1.1){$X_1$};
	    	\node[ line width=0.2pt, dashed, draw opacity=0.5] (a) at (1.5,1.12){$M$};
	    	\node[ line width=0.2pt, dashed, draw opacity=0.5] (a) at (2.5,1.1){$X_2$};
	     	\node[ line width=0.2pt, dashed, draw opacity=0.5] (a) at (0.7,-1){$\mathfrak{B}$};
	     	\node[ line width=0.2pt, dashed, draw opacity=0.5] (a) at (0.5,-1.9){$\bullet$};
	    	\node[ line width=0.2pt, dashed, draw opacity=0.5] (a) at (0.5,0.5){$\bullet$};
	    	\node[ line width=0.2pt, dashed, draw opacity=0.5] (a) at (-0.2,0.4){$\mu_{X_1}$};
	    	\node[ line width=0.2pt, dashed, draw opacity=0.5] (a) at (1.5,0.5){$\bullet$};
	    	\node[ line width=0.2pt, dashed, draw opacity=0.5] (a) at (1.85,0.4){$\mu_{M}$};
	    	\node[ line width=0.2pt, dashed, draw opacity=0.5] (a) at (2.5,0.5){$\bullet$};
	    	\node[ line width=0.2pt, dashed, draw opacity=0.5] (a) at (2.9,0.4){$\mu_{X_2}$};
	     	\node[ line width=0.2pt, dashed, draw opacity=0.5] (a) at (0.25,-1.5){$\beta_{\mathfrak{B}}$};
			\end{tikzpicture}
    \end{aligned},
\end{equation}
\emph{viz.}, $\mu_{X_1\otimes M\otimes X_2}=
(\mu_{X_1}\otimes \mu_{M}\otimes \mu_{X_2})
\comp
(\id_{H_1}\otimes \tau_{\mathfrak{B},X_1}\otimes \id_{M\otimes H_2\otimes X_2} )
\comp
(\id_{H_1\otimes \mathfrak{B}\otimes X_1}\otimes \tau_{H_2,M}\otimes \id_{X_2})
\comp
(\id_{H_1\otimes \mathfrak{B}} \otimes \tau_{H_2,X_1}\otimes \id_{M\otimes X_2})
\comp
(\beta_{\mathfrak{B}}\otimes \id_{X_1\otimes M\otimes X_2} )$.
This proves that the boundary bimodule category for our model is $\EM\simeq {_{\mathfrak{B}}}\mathsf{Mod}$.

We can also apply the folding trick to our model. As illustrated in Fig.~\ref{fig:wall} (b), an $H_1|H_2$  gapped domain wall can be transformed into an $H_1\otimes H^{\rm cop}_2$ gapped boundary.
The $H_1|H_2$-bicomodule algebra $\FB$ now becomes an $H_1\otimes H_2^{\rm cop}$-left comodule algebra.
All the results about boundary topological excitations can be applied to the domain wall. Conversely, the gapped boundary can be regarded as a gapped domain wall between Hopf algebras $H_1=H$ and vacuum $H_2=\mathbb{C}$.
From this, we see that the gapped boundary theory and gapped domain wall theory for Hopf quantum double phase are dual to each other.

\begin{proposition}
   The topological excitations of the gapped domain wall characterized by an $H_1|H_2$-bicomodule algebra $\FB$ are equivalently described by the following categories: 
   \begin{itemize}
    \item[(i)] the category of $\mathsf{Rep}(H_1)|\mathsf{Rep}(H_2)$-bimodule functors $\mathsf{Fun}_{\mathsf{Rep}(H_1)|\mathsf{Rep}(H_2)}({_{\FB}}\mathsf{Mod},{_{\FB}}\mathsf{Mod})$; 
    \item[(ii)] the category of $\mathsf{Rep}(H_1\otimes H_2^{\rm cop})$-module functors $\EW\simeq \mathsf{Fun}_{\mathsf{Rep}(H_1\otimes H_2^{\rm cop})} ({_{\FB}}\mathsf{Mod},{_{\FB}}\mathsf{Mod})$; 
    \item[(iii)] the category of all $\FB|\FB$-bimodules ${_{\FB}}\mathsf{Mod}_{\FB}$.
    \end{itemize}
\end{proposition}

\begin{proof}
The proof is straightforward by using the result for boundary topological excitations and the folding trick.
\end{proof}

\begin{table}[t]
\centering \small 
\begin{tabular}{|l|c|c|}
\hline
   &Hopf quantum double model & String-net model   \\ \hline
 Bulks ($i=1,2$) & Hopf algebras $H_i$ & UFC $\EC_i=\mathsf{Rep}(H_i)$   \\ \hline
 Bulk phases  & $\ED_i= \mathsf{Rep}(D(H_i))$ & $\mathsf{Fun}_{\mathsf{Rep}(H_i)|\mathsf{Rep}(H_i)}(\mathsf{Rep}(H_i),\mathsf{Rep}(H_i)) $  \\ \hline
 Domain wall & $H_1|H_2$-bicomodule algebra $\mathfrak{B}$ & $\EC_1|\EC_2$-bimodule category ${_{\mathfrak{B}}}\mathsf{Mod}$ \\\hline
Domain wall phase & ${_{\mathfrak{B}}}\mathsf{Mod}_{\mathfrak{B}}$ & $\mathsf{Fun}_{\EC_1|\EC_2}({_{\mathfrak{B}}}\EM,{_{\mathfrak{B}}}\EM)$ \\\hline
Domain wall defect & ${_{\mathfrak{B}}}\mathsf{Mod}_{\mathfrak{A}}$ & $\mathsf{Fun}_{\mathsf{Rep}(H_1)|\mathsf{Rep}(H_2)}({_{\mathfrak{A}}}\EM,{_{\mathfrak{B}}}\EM)$ \\\hline
\end{tabular}
\caption{The dictionary between Hopf quantum double model and string-net model for gapped domain wall.\label{tab:WallTop}}
\end{table}

Now consider an $H_1|H_2$-bicomodule algebra $\FB$, ${_{\FB}}\mathsf{Mod}$ is a $\mathsf{Rep}(H_1)|\mathsf{Rep}(H_2)$-bimodule category.
If we reverse the direction of the domain wall, we obtain an $H_2|H_1$-domain wall. This domain wall is given by the opposite category ${_{\FB}}\mathsf{Mod}^{\rm op}$, which is a $\mathsf{Rep}(H_2)|\mathsf{Rep}(H_1)$-bimodule category.
If we have
\begin{equation}
    {_{\FB}}\mathsf{Mod}^{\rm op} \boxtimes_{\mathsf{Rep}(H_1)}  {_{\FB}}\mathsf{Mod} \simeq \mathsf{Rep}(H_2), \quad {_{\FB}}\mathsf{Mod} \boxtimes_{\mathsf{Rep}(H_2)}  {_{\FB}}\mathsf{Mod}^{\rm op}\simeq \mathsf{Rep}(H_1),
\end{equation}
this domain wall will be called a transparent domain wall (the corresponding bimodule category is called invertible). 
The box tensor here is the Deligne tensor product, which means that we fuse two domain walls over some given bulk that is in between them.

	\section{Ground states in the presence of boundaries and domain walls}
	\label{sec:bdgs}
	
	In this section, we will solve the extended generalized quantum double model with gapped boundaries and domain walls and obtain the corresponding ground states. 
	To this end, we first review the Hopf tensor network representation of the states proposed in \cite{Buerschaper2013a}, which is convenient for investigating the entanglement properties of the states.
	
\subsection{Hopf tensor network states}\label{sec:tensor-network}

Consider two Hopf algebras $J,W$ with a pairing $\langle \bullet, \bullet \rangle: J\otimes W\to \mathbb{C}$ and $\phi \in J, x\in W$.
The basic ingredient of the Hopf tensor network is the rank-2 Hopf tensor $\Psi(x,\phi)=\langle \phi,x\rangle$,
\begin{equation}
    \langle \phi, x\rangle=	\begin{aligned}
        \begin{tikzpicture}
        \fill[red, opacity=0.8] (-0.4,0) rectangle (0.4,0.4); 
        \draw[-latex,black] (-1,0) -- (1,0);
        \draw[-latex,blue] (-0.9,0.7) arc(180:270:0.5);
        \draw[-latex,blue] (0.4,0.2) arc(-90:0:0.5);
        \node[ line width=0.2pt, dashed, draw opacity=0.5] (a) at (0,0.7){$\phi$};
        \node[ line width=0.2pt, dashed, draw opacity=0.5] (a) at (1.3,0){$x$};
                \end{tikzpicture}
            \end{aligned}
            =
            \begin{aligned}
        \begin{tikzpicture}
        \fill[red, opacity=0.8] (-0.4,0) rectangle (0.4,0.4); 
        \draw[-latex,black] (1,0) -- (-1,0);
        \draw[-latex,blue] (-0.9,0.7) arc(180:270:0.5);
        \draw[-latex,blue] (0.4,0.2) arc(-90:0:0.5);
        \node[ line width=0.2pt, dashed, draw opacity=0.5] (a) at (0,0.7){$\phi$};
        \node[ line width=0.2pt, dashed, draw opacity=0.5] (a) at (1.6,0){$S(x)$};
                \end{tikzpicture}
            \end{aligned},
        \end{equation}
    \begin{equation}
    \langle \phi, x\rangle=		\begin{aligned}
        \begin{tikzpicture}
        \fill[red, opacity=0.8] (-0.4,0) rectangle (0.4,0.4); 
        \draw[-latex,black] (-1,0) -- (1,0);
        \draw[-latex,blue] (-0.9,0.7) arc(180:270:0.5);
        \draw[-latex,blue] (0.4,0.2) arc(-90:0:0.5);
        \node[ line width=0.2pt, dashed, draw opacity=0.5] (a) at (0,0.7){$\phi$};
        \node[ line width=0.2pt, dashed, draw opacity=0.5] (a) at (1.3,0){$x$};
                \end{tikzpicture}
            \end{aligned}
            =
        \begin{aligned}
        \begin{tikzpicture}
        \fill[red, opacity=0.8] (-0.4,0) rectangle (0.4,0.4); 
        \draw[-latex,black] (-1,0) -- (1,0);
        \draw[-latex,blue] (-0.4,0.2) arc(-90:-180:0.5);
        \draw[-latex,blue] (0.9,0.7) arc(0:-90:0.5);
        \node[ line width=0.2pt, dashed, draw opacity=0.5] (a) at (0,0.7){$S(\phi)$};
        \node[ line width=0.2pt, dashed, draw opacity=0.5] (a) at (1.3,0){$x$};
                \end{tikzpicture}
            \end{aligned}.
        \end{equation}
Notice that here the blue edges (corresponding to the faces of the lattice) are labeled with $\phi \in J$ and the black edges (corresponding to the edges of the lattice) are labeled with $x \in W$. For fixed $\phi,x$, the evaluation is determined by the pairing, and the process will be called Hopf trace.
Therefore, the Hopf tensor network representation is a diagrammatic representation of the pairing.		
        
For two Hopf algebras $J$ and $I$	which both have their pairings with Hopf algebra $W$, we can contract the rank-2 tensors to obtain rank-3 tensors. 
There are two types of contractions. The first one is parallel gluing, which is determined by the comultiplication of $W$.
For example, fix $\phi \in J$, $\psi \in I$ and $x \in W$, we define 
        \begin{equation}\label{eq:pglue}
        \begin{aligned}
        \begin{tikzpicture}
        \fill[red, opacity=0.8] (-0.4,-0.4) rectangle (0.4,0.4); 
        \draw[-latex,black] (-1,0) -- (1,0);
        \draw[-latex,blue] (-0.9,0.7) arc(180:270:0.5);
        \draw[-latex,blue] (0.4,0.2) arc(-90:0:0.5);
        \draw[-latex,blue] (-0.4,-0.2) arc(90:180:0.5);
        \draw[-latex,blue] (0.9,-0.7) arc(0:90:0.5);
        \node[ line width=0.2pt, dashed, draw opacity=0.5] (a) at (0,0.7){$\phi$};
        \node[ line width=0.2pt, dashed, draw opacity=0.5] (a) at (0,-0.7){$\psi$};
        \node[ line width=0.2pt, dashed, draw opacity=0.5] (a) at (1.3,0){$x$};
                \end{tikzpicture}
            \end{aligned}
            :=\sum_{(x)}
        \begin{aligned}
        \begin{tikzpicture}
        \fill[red, opacity=0.8] (-0.4,0) rectangle (0.4,0.4); 
        \draw[-latex,black] (-1,0) -- (1,0);
    \draw[-latex,blue] (-0.9,0.7) arc(180:270:0.5);
        \draw[-latex,blue] (0.4,0.2) arc(-90:0:0.5);
        \node[ line width=0.2pt, dashed, draw opacity=0.5] (a) at (0,0.7){$\phi$};
        \node[ line width=0.2pt, dashed, draw opacity=0.5] (a) at (1.3,0.2){$x^{(2)}$};
         \fill[red, opacity=0.8] (-0.4,-0.6) rectangle (0.4,-0.2); 
        \draw[-latex,black] (-1,-0.2) -- (1,-0.2);
        \draw[-latex,blue] (-0.4,-0.4) arc(90:180:0.5);
        \draw[-latex,blue] (0.9,-0.9) arc(0:90:0.5);
        \node[ line width=0.2pt, dashed, draw opacity=0.5] (a) at (0,-0.9){$\psi$};
        \node[ line width=0.2pt, dashed, draw opacity=0.5] (a) at (1.3,-0.4){$x^{(1)}$};
                \end{tikzpicture}
            \end{aligned}
            = \sum_{(x)}\langle \phi, x^{(2)} \rangle \langle \psi, S(x^{(1)})\rangle.
        \end{equation}
The second one is vertical gluing, which is determined by the comultiplication of $J$. For example, fix $\phi \in J$ and $x,y\in W$, we have
    \begin{equation}
    \begin{aligned}
        \begin{tikzpicture}
        \fill[red, opacity=0.8] (-0.4,0) rectangle (0.4,0.4); 
        \draw[-latex,black] (-1,0) -- (1.1,0);
        \draw[-latex,blue] (-0.9,0.7) arc(180:270:0.5);
        \draw[-latex,blue] (0.4,0.2) arc(-90:0:0.5);
        \node[ line width=0.2pt, dashed, draw opacity=0.5] (a) at (0,1.1){$\phi$};
        \node[ line width=0.2pt, dashed, draw opacity=0.5] (a) at (0.8,-0.4){$x$};
        \fill[red, opacity=0.8] (0.7,0.7) rectangle (1.1,1.5); 
        \draw[-latex,black] (1.1,0) -- (1.1,2);
        \draw[-latex,blue] (0.9,1.5) arc(0:90:0.5);			\node[ line width=0.2pt, dashed, draw opacity=0.5] (a) at (1.4,1.8){$y$};
        \end{tikzpicture}
            \end{aligned}
            :=
            \sum_{(\phi)}
        \begin{aligned}
        \begin{tikzpicture}
        \fill[red, opacity=0.8] (-0.4,0) rectangle (0.4,0.4); 
        \fill[red, opacity=0.8] (1.3,1.0) rectangle (1.7,1.8); 
        \draw[-latex,black] (-1,0) -- (1,0);
        \draw[-latex,black] (1.7,0.4) -- (1.7,2.4);
        \draw[-latex,blue] (-0.9,0.7) arc(180:270:0.5);
        \draw[-latex,blue] (1.5,0.4) -- (1.5,1.0);
        \draw[-latex,blue] (0.4,0.2) -- (1,0.2);
        \draw[-latex,blue] (1.5,1.8) arc(0:90:0.5);
        \node[ line width=0.2pt, dashed, draw opacity=0.5] (a) at (0,0.7){$\phi^{(1)}$};
        \node[ line width=0.2pt, dashed, draw opacity=0.5] (a) at (0.9,1.4){$\phi^{(2)}$};
        \node[ line width=0.2pt, dashed, draw opacity=0.5] (a) at (1.3,0){$x$};
        \node[ line width=0.2pt, dashed, draw opacity=0.5] (a) at (2.0,2.0){$y$};
                \end{tikzpicture}
            \end{aligned}=	\sum_{(\phi)}\langle \phi^{(1)},x\rangle \langle \phi^{(2)} ,y\rangle. \label{eq:vglue}
        \end{equation}
Notice that when dealing with a face, the starting site determined the starting point of ordered pairing according to Eq.~\eqref{eq:vglue}.

With the above rules, for arbitrary Hopf tensor network $\Gamma$ with no free indices (which is a lattice for which all edges and faces have been labeled, and each face has its fixed starting site), the above Hopf trace gives us a complex number. We denote the Hopf trace as 
\begin{equation}
  \Psi_{\Gamma}(\{x_e\}_{e\in E_{\Gamma}}, \{ \phi_f\}_{f\in F_{\Gamma}})= \mathrm{ttr}_{\Gamma} (\{x_e\}_{e\in E_{\Gamma}}, \{ \phi_f\}_{f\in F_{\Gamma}}),
\end{equation} 
whose evaluation rules are given by Eqs.~\eqref{eq:pglue} and \eqref{eq:vglue}.

Now consider a lattice $C(\Sigma)$ of a $2d$ surface $\Sigma$ with face set $F$ and edge set $E$. We assign each face $f$ a Hopf algebra $J_f$ and each edge $e$ a Hopf algebra $W_e$. Assume that for each face, $J_f$ has the corresponding pairings with the edge Hopf algebras for edges $e \in \partial f$.
If we set the edge with values $x_e \in W_e$ and face with values $\phi_f\in J_f$, then we obtain a Hopf tensor network 
\begin{equation}
    (\otimes_{e\in E} x_e) \otimes (\otimes_{f\in F} \phi_f) \mapsto    \Psi_{C(\Sigma)}(\{x_e\},\{\phi_f\})=  \mathrm{ttr}_{C(\Sigma)} (\{x_e\},\{\phi_f\}).
\end{equation}
The corresponding Hopf tensor network states are defined as 
\begin{equation}
    |\Psi_{C(\Sigma)} (\{x_e\},\{ \phi_f\})\rangle = \sum_{(x_e)} \mathrm{ttr}_{C(\Sigma)} (\{x_e^{(2)}\},\{\phi_f\}) \bigotimes_{e\in E}|x_e^{(1)}\rangle.
\end{equation}

\subsection{Ground state of the boundary model}

From an anyon-condensation point of view, the bulk phase of the Hopf algebraic quantum double model is characterized by the UMTC $\EB=\mathsf{Rep}(D(H)) \simeq \mathcal{Z} (\mathsf{Rep} (H))$.
A gapped boundary is characterized by a Lagrangian algebra $L$ in $\EB$.
For a sphere $\Sigma$ with $n$ boundaries $\partial^i \Sigma$ and each boundary is characterized by its Lagrangian algebra $L_i$, the ground state degeneracy is given by
\begin{equation}\label{eq:GSDhom}
    \operatorname{GSD}=\dim \Hom(\one, L_1\otimes \cdots \otimes L_n),
\end{equation}
where $\one$ is the bulk vacuum charge.
For the disk, there is only one boundary. From the fact that, for arbitrary  Lagrangian algebra, $\dim \Hom(\one, L)=1$, this implies that the ground state on a disk is unique. Notice that Eq.~\eqref{eq:GSDhom} is universal, namely, it is independent of the lattice realizations.

Let us solve the Hamiltonian of the boundary lattice model constructed in Sec.~\ref{sec:bdd-II} in terms of the Hopf tensor network states. For the boundary model determined by a right $H$-comodule algebra $\FA$, the bulk edge algebra is $H_e=H$, the bulk face algebra is $\hat{H}_f = \hat{H}$ and the boundary edge algebra is $\FA_{e_b}=\FA$. Recall that each $\varphi\in\hat{H}$ can act on $\mathfrak{A}$ by the coaction: $\varphi\rightharpoonup x = \sum x^{[0]}\varphi(x^{[1]})$, since $x^{[1]}\in H$. Hence for $h_e\in H_e$, $\phi_f,\,\varphi_{f_b}\in\hat{H}_f$, and $x_{e_b}\in\FA_{e_b}$, we have the following Hopf tensor network state 
\begin{equation}
\begin{aligned}
    &\quad |\Psi_{C(\Sigma\setminus\partial \Sigma); C(\partial \Sigma)}(\{h_e\},\{\phi_f\};\{x_{e_b}\},\{\varphi_{f_b}\})\rangle \\
    & = \sum_{(h_e)}\sum_{[x_{e_b}]}\operatorname{ttr}_{C(\Sigma\setminus\partial \Sigma); C(\partial \Sigma)}(\{h_e^{(2)}\},\{\phi_f\};\{x_{e_b}^{[1]}\},\{\varphi_{f_b}\}) \bigotimes_{e\in E(\Sigma\setminus\partial \Sigma )} |h_e^{(1)}\rangle  \bigotimes_{e_b\in E(\partial \Sigma)} |x_{e_b}^{[0]}\rangle. 
\end{aligned}
\end{equation}

Suppose that $\FA$ is augmented, i.e., there is an algebra homomorphism $\varepsilon:\FA\to\mathbb{C}$; for example, one can take the counit $\varepsilon$ when $\FA=K$ is a Hopf subalgebra. Let $\lambda = \sum_{\langle \lambda\rangle}\lambda^{\langle 1\rangle}\otimes \lambda^{\langle 2\rangle}\in\FA\otimes\FA^{\rm op}$ be the  symmetric separability idempotent of $\FA$. Denote 
\begin{equation}\label{eq:aaafgh}
    h_\FA := \sum_{\langle \lambda\rangle}\lambda^{\langle 1\rangle}\varepsilon(\lambda^{\langle 2\rangle})\in\FA. 
\end{equation}
When $\FA=K$ and $\varepsilon$ is the counit of $K$, $h_\FA$ is the Haar integral of $K$.

\begin{proposition}[Ground state on a disk] \label{prop:GS-bdd-II}
    Suppose that $\Sigma$ is a disk, then the ground state is unique. Then the ground state of the boundary lattice model determined by an $H$-comodule algebra $\FA$ is the Hopf tensor network state 
    \begin{equation}\label{eq:state-GS}
        |\Psi_{GS}\rangle = |\Psi_{C(\Sigma\setminus\partial \Sigma); C(\partial \Sigma)}(\{h_{H,e}\},\{\phi_{\hat{H},f}\};\{h_{\FA,e_b}\},\{\varphi_{\hat{H},f_b}\})\rangle. 
    \end{equation}
    For each bulk site $s$, we have
    \begin{equation}
	    B^{\varphi}(s) A^{g}(s) |\Psi_{GS}\rangle =\varepsilon(g) \varphi(1_H)  |\Psi_{GS}\rangle.
	\end{equation}
\end{proposition}

\begin{proof}
    The first assertion follows from the general equation \eqref{eq:GSDhom}. Hence by the uniqueness of the ground state on a disk, it suffices to check the state \eqref{eq:state-GS} satisfies the stabilizer conditions. The proof for the bulk part is the same as that in Ref.~\cite{Buerschaper2013a}. Let us prove the boundary part. 

    Consider the following configuration 
    \begin{equation*}
     \begin{aligned}
    \begin{tikzpicture}
                \draw[-latex,blue] (1,1) circle (0.9);
                \draw[-latex,blue] (0.386,0.35) -- (0.35,0.39);
                \draw[-latex,blue] (1,3) circle (0.9);
                \draw[-latex,blue] (0.386,2.35) -- (0.35,2.39);
				\draw[-latex,black,line width = 1.6pt] (0,0) -- (0,2);
    			\draw[-latex,black,line width = 1.6pt] (0,2) -- (0,4); 
				\draw[-latex,black] (2,2) -- (0,2); 
                \fill[red, opacity=1] (-0.35,0.75) rectangle (0.35,1.25);
				\node[ line width=0.2pt, dashed, draw opacity=0.5] (a) at (0,1){$h_{\mathfrak{A}}$};
                \node[ line width=0.2pt, dashed, draw opacity=0.5] (a) at (-0.35,2){$v_b$};
                    \fill[red, opacity=1] (-0.35,2.75) rectangle (0.35,3.25);
				\node[ line width=0.2pt, dashed, draw opacity=0.5] (a) at (0,3){$h_{\mathfrak{A}}$};
                    \fill[red, opacity=1] (0.7,1.65) rectangle (1.3,2.35);
         		\node[ line width=0.2pt, dashed, draw opacity=0.5] (a) at (1,2){$h_H$};
                    \node[ line width=0.2pt, dashed, draw opacity=0.5] (a) at (1,3){$\varphi_{\hat{H}}$};
                    \node[ line width=0.2pt, dashed, draw opacity=0.5] (a) at (1,1){$\varphi_{\hat{H}}$};
	\end{tikzpicture}  \end{aligned}
    \end{equation*}
    In this case, the state \eqref{eq:state-GS} near the vertex reads 
    \begin{equation}
    \begin{aligned}
        |\Psi_{v_b}\rangle & = \sum \varphi_1(S(h_H^{(3)})h_1^{[1]})\varphi_2(h_2^{[1]}h_H^{(2)})|h_1^{[0]},h_H^{(1)},h_2^{[0]}\rangle \\ 
        & = \sum \varepsilon(\lambda_1^{\langle 2 \rangle})\varepsilon(\lambda_2^{\langle 2 \rangle})\varphi_1(S(h_H^{(3)})\lambda_1^{\langle 1 \rangle[1]})\varphi_2(\lambda_2^{\langle 1 \rangle[1]}h_H^{(2)})|\lambda_1^{\langle 1 \rangle[0]}, h_H^{(1)}, \lambda_2^{\langle 1 \rangle[0]} \rangle, 
    \end{aligned}
    \end{equation}
    where $\varphi_1=\varphi_2 = \varphi_{\hat{H}}$,  $h_1=h_2=h_\FA$, and $\lambda_1=\lambda_2=\lambda$. One computes 
    \begin{align*}
        A^\lambda(s_b)|\Psi_{v_b}\rangle & = \sum \varepsilon(\lambda_1^{\langle 2 \rangle})\varepsilon(\lambda_2^{\langle 2 \rangle})\varphi_1(S(h_H^{(3)})\lambda_1^{\langle 1 \rangle[1]})\varphi_2(\lambda_2^{\langle 1 \rangle[1]}h_H^{(2)}) \\ 
        & \quad \quad |\lambda^{\langle 1 \rangle}\lambda_1^{\langle 1 \rangle[0]}, S(\lambda^{\langle 2 \rangle[1]})h_H^{(1)}, \lambda_2^{\langle 1 \rangle[0]}\lambda^{\langle 2 \rangle[0]} \rangle \\
        & = \sum \varepsilon(\lambda_1^{\langle 2 \rangle})\varepsilon(\lambda_2^{\langle 2 \rangle})\varphi_1(S(h_H^{(3)})S(\lambda^{\langle 1 \rangle[2]})\lambda^{\langle 1 \rangle[1]}\lambda_1^{\langle 1 \rangle[1]}) \\ 
        & \quad \quad \varphi_2(\lambda_2^{\langle 1 \rangle[1]}h_H^{(2)}) |\lambda^{\langle 1 \rangle[0]}\lambda_1^{\langle 1 \rangle[0]}, S(\lambda^{\langle 2 \rangle[1]})h_H^{(1)}, \lambda_2^{\langle 1 \rangle[0]}\lambda^{\langle 2 \rangle[0]} \rangle \\
        & = \sum \varepsilon(\lambda_1^{\langle 2 \rangle}\lambda^{\langle 1\rangle[0]})\varepsilon(\lambda_2^{\langle 2 \rangle})\varphi_1(S(h_H^{(3)})S(\lambda^{\langle1\rangle[1]})\lambda_1^{\langle 1 \rangle[1]}) \\ 
        & \quad \quad \varphi_2(\lambda_2^{\langle 1 \rangle[1]}h_H^{(2)})|\lambda_1^{\langle 1 \rangle[0]}, S(\lambda^{\langle 2 \rangle[1]})h_H^{(1)}, \lambda_2^{\langle 1 \rangle[0]}\lambda^{\langle 2 \rangle[0]} \rangle \\
        & = \sum \varepsilon(\lambda_1^{\langle 2 \rangle}\lambda^{\langle 1\rangle[0]})\varepsilon(\lambda_2^{\langle 2 \rangle})\varphi_1(S(\lambda^{\langle2\rangle[2]}h_H^{(3)})S(\lambda^{\langle1\rangle[1]})\lambda_1^{\langle 1 \rangle[1]}) \\ 
        & \quad \quad \varphi_2(\lambda_2^{\langle 1 \rangle[1]}\lambda^{\langle2\rangle[1]}h_H^{(2)})|\lambda_1^{\langle 1 \rangle[0]}, h_H^{(1)}, \lambda_2^{\langle 1 \rangle[0]}\lambda^{\langle 2 \rangle[0]} \rangle \\
        & = \sum \varepsilon(\lambda_1^{\langle 2 \rangle}\lambda^{\langle 1\rangle[0]})\varepsilon(\lambda^{\langle2\rangle[0]}\lambda_2^{\langle 2 \rangle})\varphi_1(S(h_H^{(3)})S(\lambda^{\langle1\rangle[1]}\lambda^{\langle2\rangle[1]})\lambda_1^{\langle 1 \rangle[1]}) \\ 
        & \quad \quad \varphi_2(\lambda_2^{\langle 1 \rangle[1]}h_H^{(2)})|\lambda_1^{\langle 1 \rangle[0]}, h_H^{(1)}, \lambda_2^{\langle 1 \rangle[0]} \rangle \\
        & = \sum \varepsilon(\lambda_1^{\langle 2 \rangle}\lambda^{\langle 1\rangle})\varepsilon(\lambda^{\langle2\rangle}\lambda_2^{\langle 2 \rangle})\varphi_1(S(h_H^{(3)})\lambda_1^{\langle 1 \rangle[1]}) \\ 
        & \quad \quad \varphi_2(\lambda_2^{\langle 1 \rangle[1]}h_H^{(2)})|\lambda_1^{\langle 1 \rangle[0]}, h_H^{(1)}, \lambda_2^{\langle 1 \rangle[0]} \rangle \\
        & = \sum \varepsilon(\lambda_1^{\langle 2 \rangle})\varepsilon(\lambda_2^{\langle2\rangle})\varphi_1(S(h_H^{(3)})\lambda_1^{\langle 1 \rangle[1]}) \varphi_2(\lambda_2^{\langle 1 \rangle[1]}h_H^{(2)})|\lambda_1^{\langle 1 \rangle[0]}, h_H^{(1)}, \lambda_2^{\langle 1 \rangle[0]} \rangle, 
    \end{align*}
    where the third equality follows from $\sum x\lambda_1^{\langle1\rangle}\otimes\lambda_1^{\langle2\rangle} = \sum \lambda_1^{\langle1\rangle}\otimes \lambda^{\langle2\rangle}x$, the fourth from $\sum S(h)h_H^{(1)}\otimes h_H^{(2)} = \sum h_H^{(1)}\otimes hh_H^{(2)}$, the fifth from $\sum \lambda_2^{\langle1\rangle}x\otimes \lambda_1^{\langle2\rangle} = \sum\lambda_2^{\langle1\rangle}\otimes x\lambda_2^{\langle2\rangle}$, the sixth from Eq.~\eqref{eq:lambda}, and the last from $\sum \lambda^{\langle1\rangle}\lambda^{\langle2\rangle} = 1$. 

    Next consider the following configuration 
    \begin{equation*}
     \begin{aligned}
    \begin{tikzpicture}
                \draw[-latex,blue] (1,1) circle (0.9);
                \draw[-latex,blue] (0.386,0.35) -- (0.35,0.39);
				\draw[-latex,black,line width = 1.6pt] (0,0) -- (0,2);
				\draw[-latex,black] (2,2) -- (0,2); 
                \draw[-latex,black] (2,0) -- (0,0);
                \draw[-latex,black] (2,2) -- (2,0);
                \fill[red, opacity=1] (-0.35,0.75) rectangle (0.35,1.25);
                \fill[red, opacity=1] (1.65,0.75) rectangle (2.35,1.25);
                \draw[-latex,blue] (1.85,2.75) .. controls (1.85,2.4) and (1.5,2.25) .. (1.3,2.2);
                \draw[-latex,blue] (2.2,1.25) .. controls (2.25,1.65) and (2.65,1.8) .. (2.75,1.85);
                \draw[-latex,blue] (2.75,0.15) .. controls (2.5,0.2) and (2.25,0.5) .. (2.2,0.75);
                \draw[-latex,blue] (1.3,-0.2) .. controls (1.7,-0.2) and (1.8,-0.5) .. (1.85,-0.75);
                \draw[-latex,blue] (0.7,2.2) .. controls (0.5,2.25) and (0.25,2.35) .. (0.2,2.75);
                \draw[-latex,blue] (0.2,-0.75) .. controls (0.25,-0.35) and (0.5,-0.25) .. (0.7,-0.2);
				\node[ line width=0.2pt, dashed, draw opacity=0.5] (a) at (0,1){$h_{\mathfrak{A}}$};
                \node[ line width=0.2pt, dashed, draw opacity=0.5] (a) at (-0.35,2){$v_b$};
                \fill[red, opacity=1] (0.7,1.65) rectangle (1.3,2.35);
                \fill[red, opacity=1] (0.7,-0.35) rectangle (1.3,0.35);
         		\node[ line width=0.2pt, dashed, draw opacity=0.5] (a) at (1,2){$h_H$};
                \node[ line width=0.2pt, dashed, draw opacity=0.5] (a) at (1,0){$h_H$};
                \node[ line width=0.2pt, dashed, draw opacity=0.5] (a) at (2,1){$h_H$};
                \node[ line width=0.2pt, dashed, draw opacity=0.5] (a) at (1,2.75){$\varphi_{\hat{H}}$};
                \node[ line width=0.2pt, dashed, draw opacity=0.5] (a) at (1,-0.7){$\varphi_{\hat{H}}$};
                \node[ line width=0.2pt, dashed, draw opacity=0.5] (a) at (2.8,1){$\varphi_{\hat{H}}$};
                \node[ line width=0.2pt, dashed, draw opacity=0.5] (a) at (1,1){$\varphi_{\hat{H}}$};
	\end{tikzpicture}  \end{aligned}
    \end{equation*}
    In this case, the state \eqref{eq:state-GS} near the face is 
    \begin{equation}
        \begin{aligned}
            |\Psi_{f_b}\rangle & = \sum \varphi_1(h_1^{(2)})\varphi_2(S(h_2^{(3)}))\varphi_3(S(h_3^{(3)}))\varphi(h^{[1]}S(h_1^{(3)})h_2^{(2)}h_3^{(2)}) \\
            & \quad \quad |h^{[0]},h_1^{(1)},h_2^{(1)},h_3^{(1)} \rangle, 
        \end{aligned}
    \end{equation}
    where $h_i=h_H$, $\varphi_i=\varphi=\varphi_{\hat{H}}$, and $h = h_\FA$. By computation, one has for $\varphi'=\varphi_{\hat{H}}$,  
    \begin{align*}
        B_{f_b}^H|\Psi_{f_b}\rangle & =  \sum \varphi_1(h_1^{(3)})\varphi_2(S(h_2^{(4)}))\varphi_3(S(h_3^{(4)}))\varphi(h^{[2]}S(h_1^{(4)})h_2^{(3)}h_3^{(3)}) \\
            & \quad \quad \varphi'(h^{[1]}S(h_1^{(1)})h_2^{(2)}h_3^{(2)})|h^{[0]},h_1^{(2)},h_2^{(1)},h_3^{(1)} \rangle \\ 
            & =  \sum \varphi_1(h_1^{(1)})\varphi_2(S(h_2^{(4)}))\varphi_3(S(h_3^{(4)}))\varphi(h^{[2]}S(h_1^{(2)})h_2^{(3)}h_3^{(3)}) \\
            & \quad \quad \varphi'(h^{[1]}S(h_1^{(3)})h_2^{(2)}h_3^{(2)})|h^{[0]},h_1^{(4)},h_2^{(1)},h_3^{(1)} \rangle \\ 
            & =  \sum \varphi_1(h_1^{(1)})\varphi_2(S(h_2^{(3)}))\varphi_3(S(h_3^{(3)}))\varphi(h^{[1]}S(h_1^{(2)})h_2^{(2)}h_3^{(2)}) \\
            & \quad \quad |h^{[0]},h_1^{(3)},h_2^{(1)},h_3^{(1)} \rangle, 
    \end{align*}
    where the last equality follows from $\varphi_{\hat{H}}^2 = \varphi_{\hat{H}}$. The last assertion is immediate from the property of Haar integral. 
\end{proof}

In Appendix~\ref{sec:bdII}, we give an alternative construction of the boundary lattice model based on the generalized quantum double. We can also solve the Hamiltonian of this kind of boundary model by the Hopf tensor network states. For this construction, all internal edges are set as the Hopf algebra $H$, and internal faces the Hopf algebra $\hat{H}$. The boundary edge is set as $W$ and the boundary face is set as $J$.
In this way, we obtain the corresponding Hopf tensor network states of the model,
	\begin{equation}
	\begin{aligned}
&	|\Psi_{C(\Sigma\setminus\partial \Sigma ); C(\partial \Sigma) }(\{x_e\},\{\phi_f\}; \{y_{e_b}\},\varphi_{f_b})\rangle\\
	    =&\sum_{(x_e),(y_{e_b}) } \mathrm{ttr}_{C(\Sigma\setminus\partial \Sigma ); C(\partial \Sigma) } (\{x_e^{(2)}\},\{\phi_f\}; \{y_{e_b}^{(2)}\},\varphi_{f_b})    \bigotimes_{e\in E(\Sigma\setminus\partial \Sigma )} |x_e^{(1)}\rangle  \bigotimes _{e_{b} \in E(\partial \Sigma)} |y_{e_b}^{(1)}\rangle.
	\end{aligned}
	\end{equation}
The ground state of the model is the one for which we choose all elements as Haar integrals.
	
\begin{proposition}[Ground state on a disk] Suppose that $\Sigma$ is a disk, then the ground state is unique.
The ground state of the  extended generalized quantum double model is the Hopf tensor network state with all edge and face elements chosen as Haar integrals
    \begin{equation}\label{eq:gstate}
        |\Psi_{GS}\rangle= |\Psi_{C(\Sigma\setminus\partial \Sigma ); C(\partial \Sigma) }(\{h_{H,e}\},\{\varphi_{\hat{H},f}\}; \{h_{K,e_b}\},\{ \varphi_{J,f_b} \})\rangle. 
    \end{equation}
For each bulk site $s$, we have
\begin{equation}\label{BACom1}
    B^{\varphi}(s) A^{g}(s) |\Psi_{GS}\rangle =\varepsilon(g) \varphi(1_H)  |\Psi_{GS}\rangle.
\end{equation}
For the boundary site $s_b$, we have	
\begin{equation}\label{BACom2}
    B^{\varphi}(s_b) A^{g}(s_b) |\Psi_{GS}\rangle =\varepsilon(g) \varphi(1_H)  |\Psi_{GS}\rangle.
\end{equation}
\end{proposition}

\begin{proof}
Since the ground state is unique, we only need to check that the state we construct here satisfies the stabilizer conditions. 
For the bulk part, the proof is the same as given in Ref.~\cite{Buerschaper2013a}. 
We only need to check the boundary stabilizers.
Recall that in our construction, the boundary can be treated as a big face (see Appendix~\ref{sec:bdII}). Using the fact that $\varphi_J \varphi_J=\varphi_J$ and $h_W\in \operatorname{Cocom}(W)$, it is easy to check $B^{\varphi_J} |\Psi_{GS}\rangle =|\Psi_{GS}\rangle$.
For vertex stabilizers, using the fact that $h_K\in \operatorname{Comm}(K)$, $h_W\in \operatorname{Cocom}(W)$ and $k h_W=\varepsilon(k) h_W= h_W k$, we have $A^{h_K} |\Psi_{GS} \rangle =|\Psi_{GS}\rangle$.
Eqs.~\eqref{BACom1} and \eqref{BACom2} are direct consequences of the definition of Haar integral.
\end{proof}

\subsection{Ground state of the domain wall model}
	
Using the Hopf tensor network representation, we can also solve the model with a domain wall $\Sigma_d$ separating two bulks $\Sigma_i$ ($i=1,2$).
The spirit is the same as the construction of the boundary. Consider the domain wall determined by an $H_1|H_2$-bicomodule algebra $\FB$, constructed in Sec.~\ref{sec:domainwall}. For $h_{i,e_i}\in H_{e_i}$, $\phi_{i,f_i},\,\varphi_{i,d_i}\in \hat{H}_{f_i}$ ($i=1,2$), and $x_{e_d}\in\FB$, we have the corresponding Hopf tensor network state 
\begin{equation*}
    \begin{aligned}
        &\quad |\Psi_{C(\Sigma_1),C(\Sigma_d),C(\Sigma_2)}(\{h_{1,e_1}\},\{\phi_{1,f_1}\};\{\varphi_{1,d_1}\},\{x_{e_d}\},\{\varphi_{2,d_2}\};\{h_{2,e_2}\},\{\phi_{2,f_2}\})\rangle  \\
        & = \sum\operatorname{ttr}_{C(\Sigma_1),C(\Sigma_d),C(\Sigma_2)}(\{h_{1,e_1}^{(2)}\},\{\phi_{1,f_1}\};\{\varphi_{1,d_1}\},\{x_{e_d}^{[-1]}\};\{\varphi_{2,d_2}\},\{x_{e_d}^{[1]}\};\{h_{2,e_2}^{(2)}\},\{\phi_{2,f_2}\}) \\
        & \quad \bigotimes_{e_1\in E(\Sigma_1)} |h_{1,e_1}^{(1)}\rangle \bigotimes_{e_d\in E(\Sigma_d)} |x^{[0]}_{e_d}\rangle \bigotimes_{e_2\in E(\Sigma_2)} |h^{(1)}_{2,e_2} \rangle,  
    \end{aligned}
\end{equation*}
where $e_i$ and $f_i$ are the edges and faces respectively in region $i=1,2$, $e_d$ are edges on the wall, and $f_{d_i}$ are wall faces in region $i=1,2$.	
By a similar computation as in Proposition~\ref{prop:GS-bdd-II}, one can show that the ground state $|\Psi_{GS}\rangle$ is given by 
\begin{equation*}
    |\Psi_{C(\Sigma_1),C(\Sigma_d),C(\Sigma_2)}(\{h_{H_1,e_1}\},\{\phi_{\hat{H}_1,f_1}\};\{\varphi_{\hat{H}_1,d_1}\},\{h_{\FB,e_d}\},\{\varphi_{\hat{H}_2,d_2}\};\{h_{H_2,e_2}\},\{\phi_{\hat{H}_2,f_2}\})\rangle. 
\end{equation*}

In Appendix~\ref{sec:app-domain-wall}, we construct a domain wall lattice model based on pairings between Hopf algebras. We can apply the same method to identify the ground state of this type of domain wall model, that is, we put the corresponding Haar integrals over each edge, and all the faces also take the Haar integrals.
The ground state $|\Psi_{GS}\rangle$ is therefore given by
\begin{equation*}
		    |\Psi_{C(\Sigma_1),C(\Sigma_d),C(\Sigma_2) }(\{h_{H_1,e_1}\},\{\varphi_{\hat{H}_1,f_1}\}; \{ \varphi_{J_1,f_{d,1}} \};
		   \{h_{W,e_d}\};
		    \{ \varphi_{J_2,f_{d,2}} \};
		    \{h_{H_2,e_2}\},\{\varphi_{\hat{H}_2,f_2}\} )\rangle. 
\end{equation*}
The generalization to $N$-bulk domain wall is straightforward.	
Notice that these states are local quasi-product states \cite{jia2020entanglement}, thus they satisfy the entanglement area law.
Recently, it is discovered (in group algebra case) that the entanglement entropy of the quantum double phase is sensitive to the existence of the boundary and boundary types \cite{chen2018entanglement,lou2019ishibashiI,shen2019ishibashiII,hu2019entanglement} (and domain walls \cite{Brown2013topological,Shi2021domain}).
Using this explicit ground state, we can also analyze the entanglement feature of the extended generalized quantum double model in more detail; this will be done in our future work.

\section{Conclusion and discussion} 
	
In this work, we establish the theory of gapped boundary and domain wall for the generalized quantum double model.
A gapped boundary is equivalently characterized by an $H$-comodule algebra or an $H$-module algebra.
A domain wall is characterized by an $H_1|H_2$-bicomodule algebra.
We present the lattice realization and the ground states of these extended quantum double models are solved by using the Hopf tensor network representations of the quantum states.
The boundary and wall excitations and ribbon operators are also discussed in detail.

While significant advancements have been achieved in constructing the gapped boundary and domain wall of the generalized quantum double model using Hamiltonian methods, there are still numerous areas that require further exploration to enhance our comprehension of the theory:

(i) We have successfully constructed the gapped boundary model for a general $H$-comodule algebra $\FA$, demonstrating that the boundary local algebra provides the desired boundary phase. However, further investigation is needed to develop a general construction of ribbon operators that realize anyon condensation.

(ii) We have mentioned that the entanglement features of the generalized quantum double model are sensitive to the existence of types of gapped boundaries and domain walls; this part is left for our future study.
   
(iii) Another interesting aspect is the symmetry-enriched case; some simple cases (cyclic group, Abelian group, etc.) are considered in Refs.~\cite{Barkeshli2019symmetry,barkeshli2020relative,barkeshli2020reflection,williamson2017symmetryenriched,wang2021exactly,Jia2022electric,Heinrich2016symmetry}. However, the general Hopf algebra case is still largely left open.

(iv) The study of higher dimensional Hopf quantum double models holds significant importance in the investigation of higher dimensional topological order. While the $\mathbb{Z}_2$ group case has been extensively explored, particularly in relation to topological defects characterized by higher categories \cite{Levin2005,Hamma2005string,kong2020defects,delcamp2021tensornet}, the more general group cases and Hopf algebra cases still largely remain open topics for research. Further investigations are needed to deepen our understanding of these models and their associated higher-dimensional topological phenomena.
   
These topics will be covered in our future studies. 
	
	\subsection*{Acknowledgments}
	Z.J. acknowledges Liang Kong and Zhenghan Wang for many helpful discussions on topological order, boundary-bulk duality, and the mathematical theory of various topological defects. He also acknowledges Yuting Hu for his help.
	S.T. would like to thank Uli Walther for his constant encouragement and stimulating conversations. He also appreciates Shawn X. Cui and Bowen Yan for many helpful discussions. 
    We also acknowledge Juven C. Wang and Eric Samperton for bringing our attention to several related references.
    All authors are grateful for the referee's valuable suggestions. 
	Z.J. and D.K. are supported by the National Research Foundation and the Ministry of Education in Singapore through the Tier 3 MOE2012-T3-1-009 Grant for Random numbers from quantum processes. S.T. was in part supported by NSF grant DMS-2100288 and by Simons Foundation Collaboration Grant for Mathematicians \#580839.

	\appendix

	\section{$C^*$ Hopf algebra} 
	\label{sec:app_hopf} 
	In this appendix, we will introduce the input data of our construction of lattice realization of $2d$ topological ordered phase, that is, a finite-dimensional $C^*$ Hopf algebra and module coalgebra over it.
	And we will collect some properties of them that will be used in this paper. All vector spaces are over the field of complex numbers $\mathbb{C}$.

	\begin{definition}
	A (complex) Hopf algebra is a complex vector space $H$ equipped with several structure morphisms: multiplication $\mu: H\otimes H\to H$, unit $\eta: \mathbb{C}\to H$, comultiplication $\Delta: H\to H\otimes H$, counit $\varepsilon: H\to \mathbb{C}$ and antipode $S:H\to H$, for which some consistency conditions are satisfied:
	\begin{enumerate}
		\item 	$(H,\mu,\eta)$	is an algebra: $\mu \comp (\mu \otimes \id) =\mu \comp (\id \otimes \mu)$,  and $\mu \comp (\eta \otimes \id)=\id =\mu \comp (\id \otimes \eta)$.
		
		\item $(H,\Delta,\varepsilon)$ is a coalgebra:  $(\Delta \otimes \id)\comp \Delta = (\id \otimes \Delta )\comp \Delta$,  and 
		$ (\varepsilon \otimes \id) \comp \Delta = \id = (\id \otimes \varepsilon)\comp \Delta$. 
		
		\item $(H,\mu,\eta,\Delta,\varepsilon)$ is a bialgebra: $\Delta$ and $\varepsilon$ are algebra homomorphisms (equivalently $\mu$ and $\eta$ are coalgebra homomorphisms).
		
		\item The antipode $S$ satisfies: $ \mu \comp (S\otimes \id) \comp \Delta =\eta \comp \varepsilon = \mu \comp (\id \otimes S)\comp \Delta$.
	\end{enumerate} 
	\end{definition}
	
	We will use the notation $\mu (x\otimes y) =xy$ and $\eta(1)=1_H$ for multiplication and unit. We will also adopt the Sweedler's notation $\Delta(u)=\sum_{(u)}u^{(1)}\otimes u^{(2)}:=\sum_{i}u_i^{(1)}\otimes u_i^{(2)}$.
	The comultiplication law ensures that $(\Delta\otimes \id)\comp \Delta(u) = (\id \otimes \Delta)\comp \Delta (u)=\sum_{(u)} u^{(1)}\otimes u^{(2)}\otimes u^{(3)}$. In general, we define $\Delta_1=\Delta$ and $\Delta_n=(\id \otimes \cdots \otimes \id \otimes \Delta)\comp \Delta_{n-1}$, then $\Delta_n(u)=\sum_{(u)} u^{(1)} \otimes \cdots \otimes u^{(n+1)}$. From the definition above, we have the following useful identity:
	\begin{align}
	    \sum_{(x)}\varepsilon(x^{(1)})x^{(2)}=x&=\sum_{(x)}x^{(1)}\varepsilon(x^{(2)}).
	\end{align}
	The swap operation is defined as $\tau(x\otimes y)=y\otimes x$, by which one denotes  $\mu^{\rm op}=\mu \comp \tau$ and $\Delta^{\rm op} =\tau \comp \Delta$. The opposite Hopf algebra $H^{\rm op}$ is defined as $(H,\mu^{\rm op},\eta,\Delta,\varepsilon, S^{-1})$, and the coopposite weak Hopf algebra $H^{\rm cop}$ is defined as $(H, \mu,\eta, \Delta^{\rm op},\varepsilon,S^{-1})$.

	A Hopf algebra $H$ is called {\it simple} (or {\it irreducible}) if its algebra $(H,\mu,\eta) $ does not have nontrivial ideal, and it is called {\it semisimple} if its algebra can be written as a direct sum of simple algebras.
	From the Larson-Radford theorem, $H$ is semisimple if and only if $S^2=\id$ (also equivalent to the simplicity of the dual Hopf algebra $\hat{H}$).
	A $*$-Hopf algebra $(H,*)$ is a Hopf algebra $H$ equipped with a star structure  $*: H\to H$ such that $\Delta$ is a $*$-homomorphism. That is
	\begin{equation}
		(x^*)^*=x, (x+y)^*=x^*+y^*, (c x)^*=\bar{c} x^*, (xy)^*=y^*x^*,
		\Delta(x)^*=\Delta (x^*),
	\end{equation}
	for all $x,y\in H$ and $c\in \mathbb{C}$. $(H,*)$ is called a $C^*$  Hopf algebra if there exists a faithful $*$-representation $\rho: H\to \mathbf{B}(\mathcal{H})$ for some operator space over a Hilbert space $\mathcal{H}$. Note that a finite-dimensional $C^*$ Hopf algebra is semisimple. 
	It can be proved that 
	\begin{equation}
	    S(S(x^*)^*)=x.
	\end{equation}
	In particular, for semisimple Hopf algebra, this implies that $S(x^*)=S(x)^*$.

	For a finite-dimensional $C^*$ Hopf algebra, the antipode satisfies the following properties:
	\begin{align}
		S(xy)=S(y)S(x), \quad S(1_H)&=1_H, \quad 	S^2=\id,\quad   \varepsilon \comp S =\varepsilon,\\
		\sum_{(S(x))} S(x)^{(1)} \otimes S(x)^{(2)} &= \sum_{(x)} S(x^{(2)}) \otimes S(x^{(1)}), \\
		\sum_{(x)}x^{(1)}S(x^{(2)})=\varepsilon(&x)1_H=\sum_{(x)}S(x^{(1)})x^{(2)}.
	\end{align}
	These properties are useful for our calculation. Note it follows that $S:H^{\rm op}\to H^{\rm cop}$ is an isomorphism of Hopf algebras.

	Another notion we will use is the {\it Haar integral} $h\in H$, which is defined as a normalized two-sided integral.
	A left (resp. right) integral of $H$ is an element $\ell$ (resp. $r$) satisfying $x\ell=\varepsilon(x) \ell$ (resp. $rx=r\varepsilon(x)$) for all $x\in H$. 
	If $h$ is simultaneously left and right integral, it is called a two-sided integral.
	A left (resp. right) integral $\ell$ (resp. $r$) is called normalized if $\varepsilon(\ell)=1$ (resp. $\varepsilon(r)=1$). We see that  normalized integral is idempotent $\ell^2=\varepsilon(\ell) \ell=\ell$ (resp.  $r^2=r\varepsilon(r) =r$).
	A Haar integral of $H$ is a normalized two-sided integral. 
	Notice that a Haar integral, if exists,  must be unique. To see this, suppose that $h,h'$ are two Haar integrals, then $h'=\varepsilon (h)h' = hh'=h\varepsilon(h')=h$. It is easy to see that $S(h)$ is a Haar integral if $h$ is, thus from uniqueness we see that the Haar integral is $S$-invariant, i.e. $S(h)=h$.
	A $C^*$ Hopf algebra always has a unique Haar integral $h$ which satisfies $h^*=h$,  $h^2=h$ and $S(h)=h$. 
	An element $x\in H$ is called cocommutative if $\Delta^{\rm op}(x)=\Delta(x)$; the set of all cocommutative elements in $H$ is denoted as $\operatorname{Cocom}(H)$.  It can be proved that the Haar integral $h$ is always cocommutative $\Delta (h) =\Delta^{\rm op}(h)$.

	For a given  finite-dimensional $C^*$ Hopf algebra $H$, its dual space $\hat{H}:=\Hom (H,\mathbb{C})=H^{\vee}$ has a canonical  finite-dimensional $C^*$ Hopf algebra structure induced by the canonical pairing $\langle\bullet, \bullet \rangle: \hat{H}\times H\to \mathbb{C}$, $\langle \varphi, h\rangle:=\varphi(h)$. More precisely, 
	\begin{align}
		&	\langle \hat{\mu}(\varphi\otimes \psi),x\rangle=(\varphi\otimes \psi) (\Delta(x)),\\
		&	\langle \hat{\eta} (1),x\rangle= \varepsilon(x),    \hat{1}=\varepsilon,\\
		&   \hat{\Delta}(\varphi) (x\otimes y)=\langle \varphi, \mu(x\otimes y) \rangle,\\
		&    \hat{\varepsilon}(\varphi)=\langle \varphi, \eta(1)\rangle=\varphi(1_H),\\
		&    \langle \hat{S}(\varphi),x\rangle =\langle \varphi, S(x)\rangle.
	\end{align}
	The star operation on $\hat{H}$ is defined as 
	\begin{equation}\label{eq:Sstar}
		\langle \varphi^*, x\rangle=\overline{ \langle \varphi, S(x)^*\rangle }. 
	\end{equation}
	It is easily checked that $(H^{\rm op})^{\vee}\cong (H^{\vee})^{\rm cop}$ and  $(H^{\rm cop})^{\vee}\cong (H^{\vee})^{\rm op}$ as Hopf algebras.
	If $H$ is a $C^*$ Hopf algebra, then there exists a unique Haar integral $h$.
	The pairing $\langle \varphi,\psi\rangle:=\langle \varphi^*\psi,h\rangle=:\int_h \varphi^*\psi$ for $\varphi,\psi \in \hat{H}$, is an inner product  making $\hat{H}$ a Hilbert space.
	Since $\hat{H}$ is also a $C^*$ Hopf algebra, there also exists a  unique Haar integral  $\varphi$ which induces a Hilbert space structure on $H$, 
	\begin{equation}\label{eq:inner}
		\langle x ,y \rangle_H:=\langle \varphi, x^*y\rangle
	\end{equation}
	The Haar integral of $\hat{H}$ is also called the Haar measure on $H$.
	
	The quantum double of $H$ is the vector space $(H^{\vee})^{\rm cop} \otimes H$ equipped with a Hopf algebra structure\,\footnote{There are several different constructions of the quantum double, see \cite{majid2000foundations}. Here we choose the one built from $(H^{\vee})^{\rm cop} \otimes H$. See also Sec.~\ref{sec:bdII} for more general construction. }; we denote it as $D(H)=(H^{\vee})^{\rm cop} \Join H$. The multiplication is given by
	\begin{equation} \label{eq:double-prod}
		(\varphi \otimes x)(\psi \otimes y):= \sum_{(x)} \varphi  \psi (S^{-1} (x^{(3)}) \bullet x^{(1)}  )   \otimes  x^{(2)} y,
	\end{equation}
	where ``$\bullet$'' denotes the argument of the function. 
	The other data are given by
	\begin{align}
		&	1_{D(H)} =\hat{1} \otimes 1,\\
		&	\Delta_{D(H)} (\varphi \otimes x) =\sum_{(\varphi)} \sum_{(x)} (\varphi^{(2)} \otimes x^{(1)}  ) \otimes (\varphi^{(1)} \otimes x^{(2)}  ) ,\\
		&\varepsilon_{D(H)}  (\varphi \otimes x) =\varepsilon(x) \varphi ( 1_H),\\
		&S_{D(H)} (\varphi \otimes x) = \sum_{(\varphi)} \sum_{(x)} \langle \varphi^{(1)}  \otimes \varphi^{(3)}, h^{(3)} \otimes S^{-1} (h^{(1)}) \rangle  \hat{S}^{-1}(\varphi^{(2)}) \otimes S(h^{(2)}) .
	\end{align}
    From expression of antipode, we see that the antipode of $D(H)$ is involutive if and only if the antipode of $H$ is involutive. Then using the Larson-Radford theorem, $D(H)$ is semisimple if and only if $H$ is semisimple.
    Both $(H^{\vee})^{\rm cop}, H$ can be embedded in $D(H)$ as Hopf subalgebras by $\varphi\mapsto \varphi\otimes 1$ and $x\mapsto \hat{1}\otimes x$, respectively. 
    For more details, the reader can consult Refs.~\cite{majid2000foundations, drinfel1988quantum,kassel2012quantum,abe2004hopf}. In particular, we have the following useful straightening formula
	\begin{equation} \label{eq:straightening}
	    x\varphi = \sum_{(x)}\varphi(S^{-1}(x^{(3)})\bullet x^{(1)})x^{(2)},
	\end{equation}
	which will be useful in dealing with computation involving ribbon operators.

	The above is the input data for the bulk. For boundaries and domain walls, we need to introduce the notion of $H$-comodule algebra.
	
	\begin{definition}
	Let $H$ be a Hopf algebra and $M$ a vector space. If there is a linear map $\beta_{M}:M\to H\otimes M$ such that 
	\begin{equation}
	    (\Delta \otimes \id_M)\comp \beta_M=(\id_H\otimes \beta_M)\comp \beta_M, \quad (\varepsilon\otimes \id_M)\comp\beta_M=\id_M,
	\end{equation}
	then $M$ is called a left $H$-comodule, and $\beta_M$ is called a left coaction.
	The right $H$-comodule can be defined similarly.
	\end{definition}

	\begin{definition}
	Let $H$ be a Hopf algebra and $A$ an algebra. If $A$ is an $H$-module such that $h\triangleright (xy)=\sum_{(h)}(h^{(1)}\triangleright x)(h^{(2)}\triangleright y)$ and $h\triangleright 1_A=\varepsilon(h)1_A$, then $A$ is called an $H$-module algebra.
	\end{definition}

	\begin{definition}
	Let $H$ be a Hopf algebra and $A$ an algebra. If $A$ is a left $H$-comodule with left coaction $\beta_{A}:A\to H\otimes A$ such that $\beta_A(xy)=\beta_A(x)\beta_A(y)$ and $\beta(1_A)=1_H\otimes 1_A$, then $A$ is called a left $H$-comodule algebra. The right $H$-comodule algebra can be defined similarly.
	\end{definition}

	\begin{theorem}[See \cite{andruskiewitsch2007module}]
	\label{thm:appAMod}
	The indecomposable exact module categories over the UFC $\mathsf{Rep}(H)$ of representations of a finite-dimensional Hopf algebra $H$ are classified by $H$-comodule algebras up to Morita equivalence.
	\end{theorem}

	\begin{definition}[Smash product]
	Let $H$ be a Hopf algebra and $A$ a left $H$-module algebra. The smash product algebra $A\# H$ is defined as follows. 
	\begin{enumerate}
	    \item As a vector space $A\#H=A\otimes H$; we denote $a\# h$ for element $a\otimes h$.
	    \item The product is given by
	    \begin{equation}
	        (a\#h)(b\#g)=\sum_{(h)}a(h^{(1)}b) \# h^{(2)}g,
	    \end{equation}
	    and the identity is $1\#1$.
	\end{enumerate}
	\end{definition}
When $\mathfrak{M}$ is a left $H$-module algebra, $\mathfrak{M}^{\rm op}$ is a left $H^{\rm cop}$-module algebra.
The smash product $\mathfrak{A}=\mathfrak{M}^{\rm op}\# H^{\rm cop}$ is a left $H$-comodule algebra with coaction given by
\begin{equation}
    \beta_{\mathfrak{A}}(a\# h)=\sum_{(h)}h^{(1)} \otimes (a\# h^{(2)}),
\end{equation}
where the comultiplication label is taken in $H$.

	\begin{definition}
A Frobenius algebra is a quintuple $(\mathfrak{F}, \mu,\eta,\Delta, \varepsilon)$ with $\mu:\mathfrak{F}\otimes \mathfrak{F}\to \mathfrak{F}$, $\eta:\mathbb{C}\to \mathfrak{F}$, $\Delta :\mathfrak{F}\to \mathfrak{F}\otimes \mathfrak{F}$ and $\varepsilon: \mathfrak{F}\to \mathbb{C}$, such that  $(\mathfrak{F}, \mu,\eta)	$ forms an algebra and $(\mathfrak{F}, \Delta, \varepsilon)$ forms a coalgebra and
	\begin{equation}
		(\id_{\mathfrak{F}}\otimes \mu)\comp (\Delta\otimes \id_{\mathfrak{F}})=\Delta\comp\mu= (\mu\otimes \id_{\mathfrak{F}})\comp (\id_{\mathfrak{F}} \otimes \Delta).
	\end{equation}
	A Frobenius algebra is called special if there exist $\alpha,\beta \in \mathbb{C}^{\times}$ such that
	\begin{equation}
		\mu \comp \Delta = \alpha  \id_{\mathfrak{F}}, \quad \varepsilon \comp \eta =\beta \id_{\one}.
	\end{equation}
	When $\alpha=1$ and $\beta=\dim \mathfrak{F}$, it is said normalized-special (or separable).
	A Frobenius algebra is said symmetric if 
	\begin{equation}
	    (\varepsilon_{\FF} \otimes \id_{\FF^{\vee}})\comp  (\mu_{\FF} \otimes \id_{\FF^{\vee}})\comp (\id_{\FF} \otimes \operatorname{coev}_{\FF})=    ( \id_{\FF^{\vee}}\otimes \varepsilon_{\FF} )\comp( \id_{\FF^{\vee}} \otimes \mu_{\FF})\comp (\operatorname{coev}_{\FF^{\vee}} \otimes\id_{\FF}  ),
	\end{equation}
	where $\operatorname{coev}_{\FF}$ is the coevaluation map.
\end{definition}

    \section{Geometric objects on $2d$ lattice}
    \label{sec:AA}
	In this appendix, following Ref.~\cite{Bombin2008family}, we provide a detailed discussion of geometric objects that appeared in the quantum double model.
	
	\begin{definition}[Lattice on a surface]
	Let $\Sigma$ be a $2d$ surface, a lattice $C(\Sigma)$ on it is a simple graph embedded in $\Sigma$, where, by simple graph we mean a graph for which no edge starts and ends at the same vertex. 
	\end{definition}
	
	It can be proved that, for an arbitrary $2d$ surface, this kind of lattice embedding always exists.  Another way to understand the lattice is from the perspective of cellulation, and the result coincides with the graph embedding. 
	A typical example is a triangulation, which always exists for topological manifolds with dimensions less than or equal to three \cite{manolescu2016lecture}.
	A slightly more complicated case is to consider the lattice as a ribbon graph embedded in the surface \cite{meusburger2017kitaev,meusburger2021hopf}, but there is no big difference. We can always break the ribbon graph into a simple graph by breaking edges and adding some new edges.
	Notice that for a given surface, different lattices on it have the same topological invariant (e.g., the same Euler characteristics), thus they are topologically equivalent.
	For the surface with boundaries, similar results hold.
	We first make a graph embedding of the boundary, which is nothing but an $m$-vertex cycle. The bulk graph shares these vertices and edges on the boundary.
	
	\begin{definition}
	For a lattice $C(\Sigma)$ on the surface $\Sigma$, we denote its vertex set, edge set, and face set as $V(\Sigma)$, $E(\Sigma)$ and $F(\Sigma)$ respectively.
	The dual lattice $\tilde{C}(\Sigma)$ is the Poincar\'{e} dual of $C(\Sigma)$, for which the role of face set and vertex set interchange.
	We also introduce a direction of each edge $e\in E(\Sigma)$, and the direction of the dual edge is obtained by rotating the direct edge counterclockwise with $\pi/2$. The inverse edge $\tilde{e}$ is obtained by reversing the direction of edge $e$.
	Two ends of edge $e$ are denoted as $\partial_i e$ with $i=0,1$ indicating the starting and terminal ends.
	A direct path $p$ is a list $p=(v_0,e_1,\cdots,e_n,v_n)$ for which $\partial_0 e_k=v_{k-1}$ and $\partial e_k=v_k$. The dual path is thus a list $\tilde{p}=(f_0,\tilde{e}_1,\cdots,\tilde{e}_n,f_n)$ such that $\partial_0 e_k=f_{k-1}$ and $\partial e_k=f_k$.
	\end{definition}

    \begin{definition}
    A site is a pair $s=(v,f)$ with adjacent $v\in V(\Sigma)$ and $f\in F(\Sigma)$.
    A direct triangle $\tau=(s_0,s_1,e)$ consists of two sites $s_0, s_1$ and a direct edge $e$, for which $s_0$ and $s_1$ share a common face. A dual triangle $\tilde{\tau} =(s_0,s_1,\tilde{e})$ consists of two sites $s_0,s_1$ and a dual edge $\tilde{e}$, for which $s_0$ and $s_1$ share a common vertex. The two sites of the triangle are denoted as $\partial_i \tau$ with $i=0,1$ indicating the starting and terminal ends.
    Notice that each edge has its direction, which may or may not match the direction of the triangle. 
    A left-handed (right-handed) triangle is one for which the edge is on the left-hand (right-hand) side when we pass through the triangle along its positive direction.
    Two triangles overlap if there share some common area. Two direct (dual) triangles overlap only if they are the same triangle. A direct triangle overlaps with a dual triangle when their direct edge and dual edge coincide.
    \end{definition}

    \begin{definition}
    A strip is an alternating sequence of triangles $\rho=(\tau_1,\cdots, \tau_n)$ such that $\partial_1\tau_k=\partial_0 \tau_{k+1}$. The starting and terminal site are denoted as $\partial_0\rho =\partial_0 \tau_1$ and $\partial_1 \rho =\partial_1 \tau_n$.
    Notice that a strip may have self-overlapping.
    A strip without self-overlapping is called a ribbon.
    \begin{itemize}
        \item A strip (ribbon) is called empty if its length is zero (no triangle).
        \item A strip  (ribbon)  is called direct (dual) if it consists only of direct (dual) triangles. 
        \item A strip  (ribbon)  is called proper if it contains both direct and dual triangles. For a proper strip  (ribbon), there is an associated direct path $p_{\rho}$ and a dual path $\tilde{p}_{\rho}$.
        \item A strip  (ribbon) is called closed if $\partial_0\rho=\partial_1 \rho $, \emph{viz.}, its starting site and terminal site are the same site. The only end site of the closed strip (ribbon) is denoted as $\partial \rho$.
        \item A strip  (ribbon) is called open if
        $\partial_0 \rho$ and $\partial_1 \rho$ have no overlap ($v_{\partial_0 \rho } \neq v_{\partial_1 \rho }$ and $v_{\partial_0 \rho } \neq v_{\partial_1 \rho }$).
     \end{itemize}
    \end{definition}

    \begin{definition}
    A directed strip (ribbon) is called type-A (resp. type-B) if all the direct triangles in it are left-handed (resp. right-handed).
    \end{definition}

    \begin{definition}
    A direct closed ribbon is a closed ribbon consisting of only direct triangles. A dual closed ribbon is a closed ribbon consisting of only dual triangles.
    \end{definition}

	\section{Properties of ribbon operators}
	\label{sec:AB}
	
	In this appendix, we provide some detailed discussion and calculations of ribbon operators.
	\subsection{Triangle operators}
	
	The building blocks for the ribbon operators are eight triangle operators given in Eqs.~\eqref{eq:tri1}-\eqref{eq:tri16}.
	We see that they are determined by the left-handed edge operators $T_{\pm}$, $L_{\pm}$, and right-handed edge operators $\tilde{T}_{\pm}$ and $\tilde{L}_{\pm}$.
	Recall that the inner product on the Hopf algebra $H$ is given by $\langle x,y\rangle=\varphi_{\hat{H}}(x^*y)$, where $\varphi_{\hat{H}}$ is the Haar integral of $\hat{H}$.
	
	The left-handed operator and right-handed operator are independent, and they won't appear in a given ribbon simultaneously. Only when considering the overlap of two ribbon operators with different chiralities, need we be concerned with the relation between them.
	
	It is easy to verify that the left-handed operators satisfy
	\begin{equation} \label{eq:LTS}
		L^h_{-}=S\comp L_+^h \comp S^{-1}, \quad T_-^{\phi} =S^{-1}\comp T_+^{\phi}\comp S.
	\end{equation}
	These edge operators are also compatible with  $*$-structures of $H$ and $\hat{H}$:
	\begin{align} \label{eq:TL}
	S^{\dagger}=S,\quad
	(L_{\pm}^{h})^{\dagger} = L_{\pm}^{h^*},\quad 	(T_{\pm}^{\varphi})^{\dagger} = T_{\pm}^{\varphi^*}.
	\end{align}
	For completeness, here we give proof of Eq.~\eqref{eq:TL}. 
	First consider $\langle S^{\dagger}(x),y\rangle=\varphi_{\hat{H}} (x^* S(y)) = S(\varphi_{\hat{H}}) (S(x^* S(y)) =\varphi_{\hat{H}} (S(x^*) y)=\langle S(x),y\rangle$, thus $S^{\dagger}=S$.
	The proof of $(L^h_{\pm})^{\dagger}=L_{\pm}^{h^*}$ is almost straightforward (see Sec.~III B of \cite{Buerschaper2013a} for a detailed proof).  We will prove the $T_{\pm}^{\phi}$ part, which is relatively technical. Namely, we have 

         \begin{align}
            \langle x, T^{\phi}_{+} y\rangle &= \sum_{(y)} \varphi_{\hat{H}} (x^* y^{(1)} \phi( y^{(2)})) \nonumber \\
             &=\sum_{(y),(\phi),(x^*)} \varphi_{\hat{H}}( {x^*}^{(1)} \varepsilon({x^*}^{(2)}) \phi^{(2)}(1_H)  
             \phi^{(1)}(y^{(2)}) y^{(1)}) \nonumber \\
             &=\sum_{(y),(\phi),(x^*)} \varphi_{\hat{H}}( {x^*}^{(1)} y^{(1)}   1_{\hat{H}} ( {x^*}^{(2)})  \hat{\varepsilon}(\phi^{(2)}) \phi^{(1)}( y^{(2)})  ) \nonumber \\
             &=\sum_{(y),(\phi),(x^*)} \varphi_{\hat{H}}( {x^*}^{(1)} y^{(1)}  [\phi^{(2)} S(\phi^{(3)}) ]({x^*}^{(2)})    \phi^{(1)}( y^{(2)})  ) \nonumber \\
             &=\sum_{(y),(\phi),(x^*)} \varphi_{\hat{H}}( {x^*}^{(1)} y^{(1)}  \phi^{(2)}({x^*}^{(2)}) S(\phi^{(3)}) ({x^*}^{(3)})  \phi^{(1)}( y^{(2)})  )  \\
            &=\sum_{(y),(\phi),(x^*)} \varphi_{\hat{H}}( (y{x^*}^{(1)})^{(1)} \phi^{(1)}((y{x^*}^{(1)} )^{(2)})   S(\phi^{(2)}) ({x^*}^{(2)}) ) \nonumber \\
            &=\sum_{(y),(\phi),(x^*)} \varphi_{\hat{H}}( (y{x^*}^{(1)})^{(1)} \phi^{(1)}((y{x^*}^{(1)} )^{(2)})   S(\phi^{(2)}) ({x^*}^{(2)}) ) \nonumber \\
            &=\sum_{(\phi),(x^*)} \varphi_{\hat{H}}( {x^*}^{(1)} y \hat{\varepsilon}(\phi^{(1)})  S(\phi^{(2)}) ({x^*}^{(2)}) ) \nonumber \\
            &=\langle T^{\phi^*}_+ x,y\rangle, \nonumber 
         \end{align}
	where we have used $\varphi_{\hat{H}} \phi^{(1)} =\hat{\varepsilon} (\phi^{(1)})\varphi_{\hat{H}}$ and Eq.~\eqref{eq:Sstar}. Using Eq.~\eqref{eq:LTS}, we  have $(T^{\phi}_{-})^{\dagger}= S^{\dagger} \comp T^{\phi^*}_{+} \comp S^{\dagger} = T^{\phi^*}_{-}$.

	Using the definition of these operators, it is not difficult to work out the commutation relations among these operators,
	\begin{align}
		&T^{\varphi}_+ L_+^h =  \sum_{(h)} L_+^{h^{(1)} }  T_+^{\varphi (h^{(2)}  \bullet )  },\quad T^{\varphi}_+ L_-^h = \sum_{(S(h))} L_{-}^{S(h)^{(1)}} T_+^{\varphi(\bullet S(h)^{(2)})}, \label{eq:TL-comm1}\\
		& T_{-}^{\varphi} L_+^h = \sum_{(h)}L_+^{h^{(2)}}   T_-^{\varphi ( \bullet S(h^{(1)}   )    )}, \quad 
		T_{-}^{\varphi} L_{-}^h =\sum_{(h)} L_{-}^{h^{(2)} } T_{-}^{\varphi(h^{(2)} \bullet)},\label{eq:TL-comm2}
	\end{align}
	where ``$\bullet$'' denotes the argument of the function.

	The right-handed edge operators and left-handed edge operators are related via antipode by
	\begin{equation}\label{eq:LR-rel}
	    \tilde{L}_{\pm}^h=L^{S(h)}_{\mp}, \quad \tilde{T}^{\phi}_{\pm}=T^{\hat{S}(\phi)}_{\mp}.
	\end{equation}
	Thus from the properties of left-handed edge operators, we similarly have:
	\begin{equation} \label{eq:rLTS}
		\tilde{L}^h_{+}=S\comp \tilde{L}_{-}^h \comp S^{-1}, \quad \tilde{T}_{+}^{\phi} =S\comp \tilde{T}_{-}^{\phi}\comp S^{-1},
	\end{equation}
	and
	\begin{align} \label{eq:rTL}
	(\tilde{L}_{\pm}^{h})^{\dagger} = \tilde{L}_{\pm}^{h^*},\quad 	(\tilde{T}_{\pm}^{\varphi})^{\dagger} = \tilde{T}_{\pm}^{\varphi^*}.
	\end{align}
	The commutation relations for right-handed edge operators (and the commutation relations between left-handed and right-handed edge operators) can also be obtained from Eqs.~\eqref{eq:LR-rel}, \eqref{eq:TL-comm1} and \eqref{eq:TL-comm2}.
	
	It is easy to prove that
        \begin{equation}
           \begin{aligned}
	    L^h_{\pm}L^g_{\pm}=L^{hg}_{\pm}, \quad T_{\pm}^{\phi}T_{\pm}^{\psi}=T_{\pm}^{\phi \psi},\\
	    \tilde{L}^h_{\pm}\tilde{L}^g_{\pm}=\tilde{L}^{gh}_{\pm}, \quad \tilde{T}^{\phi}_{\pm}T^{\psi}_{\pm}=\tilde{T}^{ \psi \phi}_{\pm}.
	\end{aligned} 
        \end{equation}
	From the above equalities and the definition of triangle operators, it is clear that 
	\begin{align}
	    F^{h,\phi}(\tau_L) F^{g,\psi}(\tau_L) = F^{hg,\psi \phi}(\tau_L) , \quad F^{h,\phi}(\tilde{\tau}_R) F^{g,\psi}(\tilde{\tau}_R)= F^{hg,\psi \phi}(\tilde{\tau}_R),\\
	    F^{h,\phi}(\tau_R) F^{g,\psi}(\tau_R)= F^{gh,\phi \psi }(\tau_R), \quad F^{h,\phi}(\tilde{\tau}_L) F^{g,\psi}(\tilde{\tau}_L)= F^{gh,\phi \psi }(\tilde{\tau}_L),
	\end{align}
	This means that, as algebras, the triangle operator algebras satisfy $\mathcal{A}_{\tau_L} \cong \mathcal{A}_{\tilde{\tau}_R} \cong H\otimes \hat{H}^{\rm op}$, and $\mathcal{A}_{\tau_R} \cong \mathcal{A}_{\tilde{\tau}_L} \cong H^{\rm op} \otimes \hat{H}$.

		\subsection{Vertex and face operators}

	The type-B construction of the quantum double model is based on the left-module structure of $H$.
	The vertex operators ${A}^h(s)$ are built from $L_{\pm}$ in counterclockwise order, and the face operators ${B}^{\varphi}(s)$ are built from $T_{\pm}$ in a counterclockwise order.
	It is easy to verify that 
	\begin{equation}
	    {A}^h(s) {B}^{\varphi}(s)= \sum_{(h)}{B}^{\varphi(S^{-1}(h^{(3)}) \bullet h^{(1)} )}(s) {A}^{h^{(2)}}(s).
	\end{equation}
	See Sec.~III A in Ref.~\cite{Buerschaper2013a}. They form a representation of $D_B(H)=\hat{H}^{\rm cop} \Join  H $.
	
	The type-A construction of the quantum double model is based on the right-module structure of $H$.
	The vertex operators $\tilde{A}^h(s)$ are built from $\tilde{L}_{\pm}$ in a clockwise order, and the face operators $\tilde{B}^{\varphi}(s)$ are built from $\tilde{T}_{\pm}$ in a clockwise order.
	It is easy to verify that 
	\begin{equation}
	    \tilde{A}^h(s) \tilde{B}^{\varphi}(s)=\sum_{(h)} \tilde{B}^{\varphi(h^{(3)} \bullet S^{-1}(h^{(1)}) )}(s) \tilde{A}^{h^{(2)}}(s).
	\end{equation}
	They form a representation of $D_A(H)=\hat{H} \Join H^{\rm cop}\simeq D_B(H)^{\rm cop}$.

\subsection{Properties of ribbon operators}
    
    There are several crucial properties for open ribbon operators that will be used for constructing topological excitations of the model.
    
    As we have mentioned before, the topological excitations are given at two ends of the ribbon operators, thus the commutation relations between the vertex and face operators and ribbon operators are crucial:  
    \begin{enumerate}
        \item At the starting points of ribbons $\rho_A$ and $\rho_B$, we have 
        \begin{align}
            A^g(s_0)F^{h,\varphi}(\rho_A)&=\sum_{(g)}F^{g^{(1)}hS(g^{(3)}),\varphi(S(g^{(2)})\bullet)}(\rho_A)A^{g^{(4)}}(s_0),\label{eq:comm1}\\
            A^g(s_0)F^{h,\varphi}(\rho_B)&=\sum_{(g)}F^{g^{(2)}hS(g^{(4)}),\varphi(S(g^{(3)})\bullet)}(\rho_B)A^{g^{(1)}}(s_0), \label{eq:comm2} \\
            B^{\psi}(s_0)F^{h,\varphi}(\rho_A)&=\sum_{(h)}F^{h^{(2)},\varphi}(\rho_A)B^{\psi(\bullet S(h^{(1)}))}(s_0), \label{eq:comm3} \\
            B^{\psi}(s_0)F^{h,\varphi}(\rho_B)&=\sum_{(h)}F^{h^{(2)},\varphi}(\rho_B)B^{\psi(S(h^{(1)})\bullet)}(s_0).  \label{eq:comm4} 
        \end{align}
        \item Similarly, at the ending points, we have 
        \begin{align}
            A^g(s_1)F^{h,\varphi}(\rho_A)&=\sum_{(g)}F^{h,\varphi(\bullet g^{(2)})}(\rho_A)A^{g^{(1)}}(s_1), \label{eq:comm5} \\
            A^g(s_1)F^{h,\varphi}(\rho_B)&=\sum_{(g)}F^{h,\varphi(\bullet g^{(1)})}(\rho_B)A^{g^{(2)}}(s_1), \label{eq:comm6} \\
            B^{\psi}(s_1)F^{h,\varphi}(\rho_A)&=\sum_{(h)}\sum_{k,(k)}\varphi(k^{(2)})F^{h^{(1)},\hat{k}}(\rho_A)B^{\psi(S(k^{(3)})h^{(2)}k^{(1)}\bullet)}(s_1), \label{eq:comm7} \\
            B^{\psi}(s_1)F^{h,\varphi}(\rho_B)&=\sum_{(h)}\sum_{k,(k)}\varphi(k^{(2)})F^{h^{(1)},\hat{k}}(\rho_B)B^{\psi(\bullet S(k^{(3)})h^{(2)}k^{(1)})}(s_1).  \label{eq:comm8}  
        \end{align}
    \end{enumerate}
     See \cite{chen2021ribbon} for detailed proofs (although the convention we use here is different from the one in  \cite{chen2021ribbon}, but the same result can be obtained).

\begin{proposition}
 Let $\rho$ be a closed ribbon with end site $s=\partial\rho$.
    \begin{itemize}
        \item[1.] If $\rho=\rho_A$ is of type-A, we have
            \begin{align}
               A^h(s)F^{g,\psi}(\rho_A)&=\sum_{(h)}F^{h^{(1)}gS(h^{(3)}),\psi(S(h^{(2)})\bullet h^{(5)})}(\rho_A)A^{h^{(4)}}(s), \\  
               B^{\varphi}(s)F^{g,\psi}(\rho_A)&= \sum_{(g)}F^{g^{(2)},\psi}(\rho_A)B^{\varphi(g^{(3)}\bullet S(g^{(1)}))}(s).          \label{eq:b_closed}
            \end{align}
        \item[2.] If $\rho=\rho_B$ is of type-B, we have
            \begin{align}
               A^h(s)F^{g,\psi}(\rho_B)&=\sum_{(h)}F^{h^{(3)}gS(h^{(5)}),\psi(S(h^{(4)})\bullet h^{(1)})}(\rho_B)A^{h^{(2)}}(s), \label{eq:a_closed}\\
               B^{\varphi}(s)F^{g,\psi}(\rho_B)&=\sum_{(g)}F^{g^{(2)},\psi}(\rho_B)B^{\varphi(S(g^{(1)})\bullet g^{(3)})}(s).           \label{eq:b_closed1}
            \end{align}
    \end{itemize}
\end{proposition}
	
	\begin{proof}
	   Decompose $\rho$ into $\rho=\rho_1\cup\rho_2$ with $\partial_1\rho_1=\partial_0\rho_2$ and $\partial_0\rho_1=\partial_1\rho_2=s$. Assume $\rho$ is of type-B first. Using Eqs.~\eqref{eq:comm2} and \eqref{eq:comm6}, we have 
	   \begin{align}
	      &\ \ \ \ A^h(s)F^{g,\psi}(\rho_B) \nonumber \\  &=\sum_{k,(k),(g)}A^h(s)F^{g^{(1)},\hat{k}}(\rho_1)F^{S(k^{(3)})g^{(2)}k^{(1)},\psi(k^{(2)}\bullet)}(\rho_2)\nonumber\\
	      &=\sum_{k,(k),(g)}\sum_{(h)}F^{h^{(2)}g^{(1)}S(h^{(4)}),\hat{k}(S(h^{(3)})\bullet)}(\rho_1)A^{h^{(1)}}(s)F^{S(k^{(3)})g^{(2)}k^{(1)},\psi(k^{(2)}\bullet)}(\rho_2)\nonumber\\
	      &=\sum_{k(k),(g)}\sum_{(h)}F^{h^{(3)}g^{(1)}S(h^{(5)}),\hat{k}(S(h^{(4)})\bullet)}(\rho_1)F^{S(k^{(3)})g^{(2)}k^{(1)},\psi(k^{(2)}\bullet h^{(1)})}(\rho_2)A^{h^{(2)}}(s)\nonumber\\
	      &=\sum_{k,(k),(g)}\sum_{(h)}\sum_{j}F^{h^{(3)}g^{(1)}S(h^{(5)}),\hat{k}(S(h^{(4)})j)\hat{j}}(\rho_1)F^{S(k^{(3)})g^{(2)}k^{(1)},\psi(k^{(2)}\bullet h^{(1)})}(\rho_2)A^{h^{(2)}}(s)\nonumber\\
	      &= \sum_{(h)}\sum_{j,(j),(g)}F^{h^{(3)}g^{(1)}S(h^{(7)}),\hat{j}}(\rho_1)F^{S(j^{(3)})h^{(4)}g^{(2)}S(h^{(6)})j^{(1)},\psi(S(h^{(5)})j^{(2)}\bullet h^{(1)})}(\rho_2)A^{h^{(2)}}(s)\nonumber\\
	      &=\sum_{(h)}F^{h^{(3)}gS(h^{(5)}),\psi(S(h^{(4)})\bullet h^{(1)})}(\rho_B)A^{h^{(2)}}(s). 
	   \end{align}
	Using straightening formula, 
Eq.~\eqref{eq:comm8} can be rewritten as 
\begin{equation}
B^{\psi}(s_1)F^{h,\varphi}(\rho_B)=\sum_{(h)}F^{h^{(1)},\varphi}(\rho_B)B^{\psi(\bullet h^{(2)})}(s_1). 
\end{equation}
Thus, by Eqs.~\eqref{eq:comm4} and the above, we have 
    	\begin{align}
	       &\ \ \ \ B^\varphi(s)F^{g,\psi}(\rho_B) \nonumber\\ 
	       &= \sum_{k,(k),(g)}B^\varphi(s)F^{g^{(1)},\hat{k}}(\rho_1)F^{S(k^{(3)})g^{(2)}k^{(1)},\psi(k^{(2)}\bullet)}(\rho_2) \nonumber \\
	       &= \sum_{k,(k),(g)} F^{g^{(2)},\hat{k}}(\rho_1)B^{\varphi(S(g^{(1)})\bullet)}(s) F^{S(k^{(3)})g^{(3)}k^{(1)},\psi(k^{(2)}\bullet)}(\rho_2) \nonumber\\
	       &= \sum_{k,(k),(g)} F^{g^{(2)},\hat{k}}(\rho_1)F^{S(k^{(5)})g^{(3)}k^{(1)},\psi(k^{(3)}\bullet)}(\rho_2)B^{\varphi(S(g^{(1)})\bullet S(k^{(4)})g^{(4)}k^{(2)})}(s) \nonumber \\
	       &=\sum_{(g)}\sum_{k,(k),(g^{(2)})}F^{g^{(2)},\hat{k}}(\rho_1)F^{S(k^{(3)})g^{(3)}k^{(1)},\psi(k^{(2)}\bullet)}(\rho_2)B^{\varphi(S(g^{(1)})\bullet g^{(4)})}(s) \nonumber \\
	       &=\sum_{(g)}F^{g^{(2)},\psi}(\rho)B^{\varphi(S(g^{(1)})\bullet g^{(3)})}(s). 
    	\end{align}
	Here, in the fourth equality, we apply the straightening formula. 
	This completes the proof for the assertion of closed ribbons of type-B. The proof for the assertion of closed ribbons of type-A is similar. 
	\end{proof}

	\begin{proposition}\label{prop:ribb_local_op}
	The ribbon operator $F^{g,\psi}(\rho)$ for an open ribbon $\rho$ commutes with all operators $A_{v}^{H}=A^{h_H} (s)$ and $B^{\hat{H}}_f=B^{\varphi_{\hat{H}}}(s)$ for $s\neq\partial_0 \rho,\partial_1 \rho$, where $h_H$ and $\varphi_{\hat{H}}$ are Haar integrals.
	    
	\begin{proof}
	   We prove the assertion for type-B ribbons. Denote $h=h_H$ and $\varphi=\varphi_{\hat{H}}$ for simplicity. Decompose $\rho$ into $\rho=\tau_1\cup\tau_2$, with $\partial_1\tau_1=\partial_0\tau_2=s$: 
	   \begin{equation*}
	        \begin{aligned}
				\begin{tikzpicture}
					\draw[-latex,black] (-2,0) -- (0,0);						
					\draw[-latex,black] (0,0) -- (2,0); 
					\draw[line width=0.5pt, red] (0,0) -- (1,1);
					\draw[line width=0.5pt, black] (1,1) -- (2,0);
					\draw[line width=0.5pt, black] (0,0) -- (-1,1);
					\draw[line width=0.5pt, black] (-2,0) -- (-1,1);
					\draw[line width=0.5pt, black] (2,0) -- (3,1);
					\draw[-latex, dashed, black] (1,1) -- (-1,1);
					\draw[-latex, dashed, black] (3,1) -- (1,1);
					\draw [fill = black] (0,0) circle (1.2pt);
					\draw [fill = black] (1,1) circle (1.2pt);
					\node[ line width=0.2pt, dashed, draw opacity=0.5] (a) at (0.3,0.6){$s$};	
					\draw[-stealth,gray, line width=3pt] (1.2,0.5) -- (2,0.5);
					\draw[-stealth,gray, line width=3pt] (-1,0.5) -- (-0.2,0.5);
					\node[ line width=0.2pt, dashed, draw opacity=0.5] (a) at (-0.5,0.2){$\tau_1$};
					\node[ line width=0.2pt, dashed, draw opacity=0.5] (a) at (1.5,0.2){$\tau_2$};
				    \node[ line width=0.2pt, dashed, draw opacity=0.5] (a) at (-2.2,0.5){$\cdots$};
					\node[ line width=0.2pt, dashed, draw opacity=0.5] (a) at (3.2,0.5){$\cdots$};
				\end{tikzpicture}
			\end{aligned}
		\end{equation*}
	Using Eqs.~\eqref{eq:comm2} and \eqref{eq:comm6}, one gets 
	    \begin{align}
	         &\quad A^h(s)F^{g,\psi}(\rho) \nonumber \\
          & = \sum_{k,(k),(g)}A^h(s)F^{g^{(1)},\hat{k}}(\tau_1)F^{S(k^{(3)})g^{(2)}k^{(1)},\psi(k^{(2)}\bullet)}(\tau_2) \nonumber \\
          & = \sum_{k,(k),(g)}\sum_{(h)}F^{g^{(1)},\hat{k}(\bullet h^{(1)})}(\tau_1)A^{h^{(2)}}(s)F^{S(k^{(3)})g^{(2)}k^{(1)},\psi(k^{(2)}\bullet)}(\tau_2) \nonumber \\
          & = \sum_{k,(k),(g)}\sum_{(h)}F^{g^{(1)},\hat{k}(\bullet h^{(1)})}(\tau_1)F^{h^{(3)}S(k^{(3)})g^{(2)}k^{(1)}S(h^{(5)}),\psi(k^{(2)}S(h^{(4)})\bullet)}(\tau_2)A^{h^{(2)}}(s) \nonumber \\
          & = \sum_{k,(k),(g)}\sum_{(h)}\sum_jF^{g^{(1)},\hat{k}(jh^{(1)})\hat{j}}(\tau_1)F^{h^{(3)}S(k^{(3)})g^{(2)}k^{(1)}S(h^{(5)}),\psi(k^{(2)}S(h^{(4)})\bullet)}(\tau_2)A^{h^{(2)}}(s) \nonumber \\
          & = \sum_{j,(j),(g)}\sum_{(h)}F^{g^{(1)},\hat{j}}(\tau_1)F^{h^{(5)}S(h^{(3)})S(j^{(3)})g^{(2)}j^{(1)}h^{(1)}S(h^{(7)}),\psi(j^{(2)}h^{(2)}S(h^{(6)})\bullet)}(\tau_2)A^{h^{(4)}}(s) \nonumber \\
          & = \sum_{j,(j),(g)}\sum_{(h)}F^{g^{(1)},\hat{j}}(\tau_1)F^{h^{(4)}S(h^{(2)})S(j^{(3)})g^{(2)}j^{(1)},\psi(j^{(2)}h^{(1)}S(h^{(5)})\bullet)}(\tau_2)A^{h^{(3)}}(s) \nonumber \\
          & = \sum_{j,(j),(g)}\sum_{(h)}F^{g^{(1)},\hat{j}}(\tau_1)F^{h^{(3)}S(h^{(1)})S(j^{(3)})g^{(2)}j^{(1)},\psi(j^{(2)}\bullet)}(\tau_2)A^{h^{(2)}}(s) \nonumber \\
          & = \sum_{j,(j),(g)}F^{g^{(1)},\hat{j}}(\tau_1)F^{S(j^{(3)})g^{(2)}j^{(1)},\psi(j^{(2)}\bullet)}(\tau_2)A^{h}(s) \nonumber \\
          & = F^{g,\psi}(\rho)A^h(s). 
	    \end{align}
     Here, in the fourth equality, we use $\bullet=\sum_j\hat{j}(\bullet)j$ for an orthogonal basis $\{j\}$; in the fifth equality, we use $jh^{(1)}=\sum_k\hat{k}(jh^{(1)})k$; in the sixth equality, we use the fact that $h=h_H$ is cocommutative which enable us to rotate $h^{(7)}$ to $h^{(1)}$ and $h^{(i)}$ to $h^{(i+1)}$ for $1\leq i \leq 6$; the seventh and eighth equalities use the same trick as in the sixth.

	 For the second statement, using Eqs.~\eqref{eq:comm4} and \eqref{eq:comm8}, one has
\begin{align}
    &\quad  B^\varphi(s)F^{g,\psi}(\rho) \nonumber \\
    & = \sum_{k,(k),(g)} B^\varphi(s) F^{g^{(1)},\hat{k}}(\rho_1)F^{S(k^{(3)})g^{(2)}k^{(1)},\psi(k^{(2)}\bullet)}(\rho_2) \nonumber  \\
    & = \sum_{k,(k),(g)}\sum_{j,(j)}\hat{k}(j^{(2)})  F^{g^{(1)},\hat{j}}(\rho_1) B^{\varphi(\bullet S(j^{(3)})g^{(2)}j^{(1)})}(s) F^{S(k^{(3)})g^{(3)}k^{(1)},\psi(k^{(2)}\bullet)}(\rho_2) \nonumber  \\
    & = \sum_{k,(k),(g)}\sum_{j,(j)}\hat{k}(j^{(2)})  F^{g^{(1)},\hat{j}}(\rho_1)  F^{S(k^{(4)})g^{(4)}k^{(2)},\psi(k^{(3)}\bullet)}(\rho_2)B^{\varphi(S(k^{(1)})S(g^{(3)})k^{(5)}\bullet S(j^{(3)})g^{(2)}j^{(1)})}(s)  \nonumber \\
    & = \sum_{(g)}\sum_{j,(j)}  F^{g^{(1)},\hat{j}}(\rho_1)  F^{S(j^{(5)})g^{(4)}j^{(3)},\psi(j^{(4)}\bullet)}(\rho_2)B^{\varphi(S(j^{(2)})S(g^{(3)})j^{(6)}\bullet S(j^{(7)})g^{(2)}j^{(1)})}(s) \nonumber  \\
    & = \sum_{(g)}\sum_{j,(j)}  F^{g^{(1)},\hat{j}}(\rho_1)  F^{S(j^{(5)})g^{(4)}j^{(3)},\psi(j^{(4)}\bullet)}(\rho_2)B^{\varphi(S(j^{(7)})g^{(2)}j^{(1)}S(j^{(2)})S(g^{(3)})j^{(6)}\bullet )}(s)  \nonumber \\
    & = \sum_{(g)}\sum_{j,(j)}  F^{g^{(1)},\hat{j}}(\rho_1)  F^{S(j^{(3)})g^{(2)}j^{(1)},\psi(j^{(2)}\bullet)}(\rho_2)B^{\varphi}(s)  \nonumber \\
    & = F^{g,\psi}(\rho)B^\varphi(s). 
\end{align}
	 Here, in the fifth equality, we use the fact that $\varphi$ is cocommutative.
	 This completes the proof of the assertion for type-B ribbons. The assertion for type-A ribbons can be proved similarly \cite{chen2021ribbon}.
	  \end{proof}
	\end{proposition}

	\begin{proposition}
	   The ribbon operator $F^{h,\phi}(\rho)$ of a  closed ribbon $\rho$ commutes with all stabilizers $A_v^H=A^{h_H}(s)$ and $B_{f}^{\hat{H}}=B^{\varphi_H}(s)$ with $s\neq\partial\rho$.
	\end{proposition}
	
	\begin{proof}
	    Denote $s_0=\partial\rho$ and $s=(v,f)$. We only need to consider the case that $s$ lies on $\rho$. Note that we can always find an $s_1$ in $\rho$ such that $\rho=\rho_1\cup\rho_2$ with $\partial_0\rho_1=\partial_1\rho_2=s_0, \partial_1\rho_1=\partial_0\rho_2=s_1$ and $s\neq s_1$. Then $\rho_1$ and $\rho_2$ are open ribbons and $s$ is not an end site of them. It follows from Proposition \ref{prop:ribb_local_op} that $A^{H}_v$ and $B_f^{\hat{H}}$ commute with $F^{(h\otimes\varphi)^{(1)}}(\rho_1)$ and $F^{(h\otimes\varphi)^{(2)}}(\rho_2)$, and hence with $F^{h,\varphi}(\rho)$ by the recursive formula of ribbon operator.
	\end{proof}

\subsection*{Proof of the commutation relations (\ref{eq:bdd-comm3}-\ref{eq:bdd-comm8})}

We will prove the commutation relations for type-A ribbons. The proof for type-B ribbons is similar.

$\bullet$ \eqref{eq:bdd-comm3} for short ribbon: 

\begin{equation*}
    \begin{tikzpicture}
        \draw[line width = 0.5pt, red] (0.75,0.75) -- (0,0);
        \draw[line width = 0.5pt, red] (1.5,0) -- (0.75,0.75);
        \draw[-latex,black] (1.5,0) -- (0,0);
        \draw[-latex,line width=1.6pt,black] (0,0) -- (0,1.5);
        \draw[-latex,line width=1.6pt,black] (0,-1.2) -- (0,0);
        \draw[-latex,black] (0,1.5) -- (1.5,1.5);
        \draw[-latex,black] (1.5,1.5) -- (1.5,0);
        \draw[-latex,black] (1.5,-1.2) -- (1.5,0);
        \draw[-stealth,gray, line width=2pt] (0.95,0.3) -- (0.45,0.3); 
        \node[ line width=0.2pt, dashed, draw opacity=0.5] (a) at (-0.3,0.75){$a_1$};
        \node[ line width=0.2pt, dashed, draw opacity=0.5] (a) at (0.8,-0.25){$a_3$};
        \node[ line width=0.2pt, dashed, draw opacity=0.5] (a) at (0.2,0.45){$s_b$};
        \node[ line width=0.2pt, dashed, draw opacity=0.5] (a) at (-0.3,-0.75){$a_2$};
    \end{tikzpicture}
\end{equation*}
\begin{align*}
    &\quad A^k(s_b)F^{[h],\varphi}(\rho_A) |a_1,a_2,a_3\rangle \\ 
    & = \sum A^k(s_1)\varepsilon(h)\varphi(S(a_3^{(2)}))|a_1,a_2,a_3^{(1)}\rangle \\ 
    & = \sum \varepsilon(h)\varphi(S(a_3^{(2)}))|a_1S(k^{(1)}),k^{(2)}a_3,k^{(3)}a_3^{(1)}\rangle \\ 
    & = \sum \varepsilon(h)\varphi(S(S(k^{(5)})k^{(4)}a_3^{(2)}))|a_1S(k^{(1)}),k^{(2)}a_2,k^{(3)}a_3^{(1)}\rangle \\ 
    & = \sum \varepsilon(h)\varphi(S(k^{(4)}a^{(2)}_3)k^{(5)})|a_1S(k^{(1)}),k^{(2)}a_2,k^{(3)}a_3^{(1)}\rangle \\ 
    & = \sum F^{[h],\varphi(\bullet k^{(4)})}(\rho) |a_1S(k^{(1)}),k^{(2)}a_2,k^{(3)}a_3\rangle \\  
    & = \sum F^{[h],\varphi(\bullet k^{(2)})}(\rho_A)A^{k^{(1)}}(s_b) |a_1,a_3,a_3\rangle. 
\end{align*}

$\bullet$ \eqref{eq:bdd-comm4} for short ribbon: 

\begin{equation*}
    \begin{tikzpicture}
        \draw[line width = 0.5pt, red] (0.75,0.75) -- (0,0);
        \draw[line width = 0.5pt, red] (1.5,0) -- (0.75,0.75);
        \draw[line width = 0.5pt, red] (1.5,0) -- (2.25,0.75);
        \draw[black,dashed] (0.75,0.75) -- (2.25,0.75);
        \draw[-latex,line width=1.6pt,black] (1.5,0) -- (0,0);
        \draw[-latex,line width=1.6pt,black] (1.5,-1.2) -- (1.5,0);
        \draw[-latex,black] (1.5,1.5) -- (1.5,0);
        \draw[-latex,black] (0,1.5) -- (1.5,1.5);
        \draw[-latex,black] (3,0) -- (1.5,0);
        \draw[-latex,black] (0,0) -- (0,1.5);
        \draw[-stealth,gray, line width=2pt] (1.4,0.35) -- (0.9,0.35); 
        \node[ line width=0.2pt, dashed, draw opacity=0.5] (a) at (-0.2,0.75){$a_4$};
        \node[ line width=0.2pt, dashed, draw opacity=0.5] (a) at (0.8,-0.25){$a_1$};
        \node[ line width=0.2pt, dashed, draw opacity=0.5] (a) at (0.2,0.45){$s_b$};
        \node[ line width=0.2pt, dashed, draw opacity=0.5] (a) at (0.75,1.68){$a_3$};
        \node[ line width=0.2pt, dashed, draw opacity=0.5] (a) at (1.8,0.96){$a_2$};
    \end{tikzpicture}
\end{equation*}
\begin{align*}
    &\quad B^\psi(s_1)F^{[h],\varphi}(\rho) |a_1,a_2,a_3,a_4 \rangle \\  
    & = \sum B^\psi(s_1) F^{[h^{(1)}],\hat{k}}(\rho_1)F^{[S(k^{(3)})h^{(2)}k^{(1)}],\varphi(k^{(2)}\bullet)}(\rho_2) |a_1,a_2,a_3,a_4 \rangle \\  
    & = \sum B^\psi(s_1) \varepsilon(\hat{k})\varepsilon(S(k^{(3)})h^{(2)}k^{(1)})\varphi(k^{(2)}S(a_1^{(2)})) |a_1^{(1)},h^{(1)}h_Ka_2,a_3,a_4 \rangle \\  
    & = \sum B^\psi(s_1)\varphi(S(a_1^{(2)})) |a_1^{(1)},hh_Ka_2,a_3,a_4 \rangle \\  
    & = \sum \varphi(S(a_1^{(3)}))\psi(a_1^{(2)}h^{(2)}h_K^{(2)}a_2^{(2)}a_3^{(2)}a_4^{(2)}) |a_1^{(1)},h^{(1)}h_K^{(1)}a_2^{(1)},a_3^{(1)},a_4^{(1)} \rangle \\
    & = \sum \varphi(S(a_1^{(3)}))\psi(a_1^{(2)}h^{(2)}h_K^{(2)}h_Ka_2^{(2)}a_3^{(2)}a_4^{(2)}) |a_1^{(1)},h^{(1)}h_K^{(1)}a_2^{(1)},a_3^{(1)},a_4^{(1)} \rangle \\
    & = \sum \varphi(S(a_1^{(3)})) \psi(a_1^{(2)}h^{(2)}a_2^{(2)}a_3^{(2)}a_4^{(2)})|a_1^{(1)},h^{(1)}h_Ka_2^{(1)},a_3^{(1)},a_4^{(1)} \rangle \\  
    & = \sum \varphi(S(a_1^{(3)})) \psi(a_1^{(2)}h^{(2)}S(a_1^{(4)})a_1^{(5)}a_2^{(2)}a_3^{(2)}a_4^{(2)}) \\
    & \quad\quad |a_1^{(1)},h^{(1)}h_Ka_2^{(1)},a_3^{(1)},a_4^{(1)} \rangle \\  
    & = \sum \varphi(k^{(2)})\hat{k}(S(a_1^{(2)}))\psi(S(k^{(3)})h^{(2)}k^{(1)}a_1^{(3)}a_2^{(2)}a_3^{(2)}a_4^{(2)})  \\
    & \quad\quad |a_1^{(1)},h^{(1)}h_Ka_2^{(1)},a_3^{(1)},a_4^{(1)} \rangle \\  
    & = \sum \varphi(k^{(2)})\hat{k}(j^{(2)}S(a_1^{(2)}))\psi(S(k^{(3)})h^{(3)}k^{(1)}a_1^{(3)}a_2^{(2)}a_3^{(2)}a_4^{(2)}) \\
    & \quad\quad \varepsilon(S(j^{(3)}))\varepsilon(h^{(2)})\varepsilon(j^{(1)})\varepsilon(\hat{j})|a_1^{(1)},h^{(1)}h_Ka_2^{(1)},a_3^{(1)},a_4^{(1)} \rangle \\  
    & = \sum \varphi(k^{(2)}) F^{[h^{(1)}],\hat{j}}(\rho_1)F^{[S(j^{(3)})h^{(2)}j^{(1)}],\hat{k}(j^{(2)}\bullet)}(\rho_2)\\
    &\quad \quad \psi(S(k^{(3)})h^{(3)}k^{(1)}a_1^{(3)}a_2^{(2)}a_3^{(2)}a_4^{(2)}) |a_1^{(1)},a_2^{(1)},a_3^{(1)},a_4^{(1)} \rangle \\ 
    & = \sum \varphi(k^{(2)}) F^{[h^{(1)}],\hat{k}}(\rho)B^{\psi(S(k^{(3)})h^{(2)}k^{(1)}\bullet)}(s_1)|a_1,a_2,a_3,a_4 \rangle. 
\end{align*}
Note that for any $h\in H$, since $h = hh_K+h(1-h_K)$ and $(1-h_K)\in K^+$, $[h]=[hh_K]$ in $H/HK^+$, and since $\psi\in(H/HK^+)^\vee$ only depends on the equivalence classes, the fifth and sixth equalities follow. 

Let us prove the commutation relations for long ribbons using the recursive formula of ribbon operators. Let $\rho_A= \rho_1\cup\rho_2$ be a decomposition, where $\rho_2$ is a short ribbon such that the relations \eqref{eq:bdd-comm3} and \eqref{eq:bdd-comm4} hold. Note that there is no overlap for $\rho_1$ and $s_b=\partial_1\rho_A$. Then one computes 
\begin{align*}
    &\quad A^k(s_b)F^{[h],\varphi}(\rho_A) \\ 
    & = \sum A^k(s_b)F^{[h^{(1)}],\hat{\ell}}(\rho_1)F^{[S(\ell^{(3)})h^{(2)}\ell^{(1)}],\varphi(\ell^{(2)}\bullet)}(\rho_2) \\
    & = \sum F^{[h^{(1)}],\hat{\ell}}(\rho_1)A^k(s_b)F^{[S(\ell^{(3)})h^{(2)}\ell^{(1)}],\varphi(\ell^{(2)}\bullet)}(\rho_2) \\ 
    & = \sum F^{[h^{(1)}],\hat{\ell}}(\rho_1)F^{[S(\ell^{(3)})h^{(2)}\ell^{(1)}],\varphi(\ell^{(2)}\bullet k^{(2)})}(\rho_2) A^{k^{(1)}}(s_b)\\
    & = \sum F^{[h],\varphi(\bullet k^{(2)})}(\rho_A) A^{k^{(1)}}(s_b).  
\end{align*}
Similarly, one has  
\begin{align*}
    & \quad B^\psi(s_b)F^{[h],\varphi}(\rho_A) \\ 
    & = \sum  B^\psi(s_b)F^{[h^{(1)}],\hat{k}}(\rho_1)F^{[S(k^{(3)})h^{(2)}k^{(1)}],\varphi(k^{(2)}\bullet)}(\rho_2) \\ 
    & = \sum F^{[h^{(1)}],\hat{k}}(\rho_1)B^\psi(s_b)F^{[S(k^{(3)})h^{(2)}k^{(1)}],\varphi(k^{(2)}\bullet)}(\rho_2) \\ 
    & = \sum \varphi(k^{(3)}j^{(2)})F^{[h^{(1)}],\hat{k}}(\rho_1)F^{[S(k^{(5)})h^{(2)}k^{(1)}],\hat{j}}(\rho_2)B^{\psi(S(j^{(3)})S(k^{(4)})h^{(3)}k^{(2)}j^{(1)}\bullet)}(s_b) \\ 
    & = \sum \varphi(\ell^{(2)}) F^{[h^{(1)}],\hat{k}}(\rho_1)F^{[S(k^{(5)})h^{(2)}k^{(1)}],\hat{\ell}(k^{(2)}j)\hat{j}}(\rho_2)B^{\psi(S(\ell^{(3)})h^{(3)}\ell^{(1)}\bullet)}(s_b) \\
    & = \sum \varphi(\ell^{(2)}) F^{[h^{(1)}],\hat{k}}(\rho_1)F^{[S(k^{(5)})h^{(2)}k^{(1)}],\hat{\ell}(k^{(2)}\bullet)}(\rho_2)B^{\psi(S(\ell^{(3)})h^{(3)}\ell^{(1)}\bullet)}(s_b) \\
    & = \sum \varphi(\ell^{(2)}) F^{[h^{(1)}],\hat{\ell}}(\rho)B^{\psi(S(\ell^{(3)})h^{(3)}\ell^{(1)}\bullet)}(s_b). 
\end{align*}
Here $\{k\}$, $\{j\}$ and $\{\ell\}$ are bases of $K$. We used $k^{(2)}j = \sum_\ell \hat{\ell}(k^{(2)}j)\ell$ in the fourth equality, and used $\bullet = \hat{j}(\bullet)j$ in the fifth equality.

\section{Gapped boundary lattice based on generalized quantum double}
\label{sec:bdII}	
	
Let us now introduce a general construction of a gapped boundary model based on a generalized quantum double.
The cellulation of the surface for this construction is slightly different from the one given above.
As shown in Fig.~\ref{fig:bd}, the boundary is drawn as a bold solid line. The boundary vertices are vertices on the line, which we denote as $v_b$. The boundary face is the right end of the dual boundary edges. There is only one boundary face for each boundary and we denote it as $f_b$. There are two kinds of boundary sites: the one $s_b=(v_b,f_b)$ consists of boundary vertex and boundary face, which we call the outer boundary site; the one $s_b=(v_b,f)$ consists of boundary vertex and inner face, which we call inner boundary site.

\begin{definition}
     A pairing $\lambda=\langle \bullet, \bullet \rangle: J\otimes K\to \mathbb{C}$ between two Hopf algebras $J,K$ is a bilinear map  satisfying
     \begin{align}
       &  \langle hg,a\rangle = \sum_{(a)} \langle h,a^{(1)}\rangle \langle g,a^{(2)}\rangle,\\
    &     \langle h,ab\rangle =\sum_{(h)} \langle h^{(1)} ,a\rangle \langle h^{(2)},b\rangle,\\
     &    \langle 1_J,a\rangle =\varepsilon_K(a),\quad
         \langle h,1_K\rangle =\varepsilon_J(h).
     \end{align}
\end{definition}

There is a convolution algebra structure over $C:=\Hom(J\otimes K,\mathbb{C})$, with the convolution defined as 
	\begin{equation}
	\begin{split}
	    (\varphi * \psi)(h\otimes a) :&= (\mu_{\mathbb{C}}\comp (\varphi\otimes \psi)\comp \Delta_{J\otimes K})(h\otimes a)\\
	    &=\sum_{(h),(a)} \varphi(h^{(1)} \otimes a^{(1)}) \psi(h^{(2)} \otimes a^{(2)}).
	\end{split}
	\end{equation}
The multiplication and unit of $C$ are defined as 
	\begin{align}
	    \mu_C(\varphi\otimes \psi)= \varphi *\psi,\quad \eta_{C}=\eta_{\mathbb{C}} \comp \varepsilon_{J\otimes K}=\varepsilon_J \varepsilon_K.
	\end{align} 
	It is easy to verify that (for semisimple Hopf algebras), the convolutional inverse of $\lambda$ is
	\begin{equation}
	    \lambda^{-1}(h\otimes a)
	    =\langle S(h),a\rangle
	    =\langle h,S(a)\rangle
	    =\langle S^{-1}(h),a \rangle
	    =\langle h,S^{-1}(a)\rangle.
	\end{equation}
	Hereinafter, we assume that the pairing is also consistent with the $C^*$ structure in the sense that $\lambda(\varphi^*\otimes h )=\overline{\lambda(\varphi\otimes S(h)^*)}$.

	\begin{definition}[Generalized quantum double \cite{majid2000foundations}]
	For a given pairing $\lambda: J\otimes K\to \mathbb{C}$, the generalized quantum double $D_{\lambda} (J^{\rm cop},K):=J^{\rm cop}\Join_{\lambda} K$ is a Hopf algebra built on $J^{\rm cop}\otimes K$ with Hopf algebra multiplication given by 
	\begin{equation}
	    (h\otimes a)(g\otimes b)= \sum_{(a),(g)} hg^{(2)} \otimes a^{(2)} b \lambda(g^{(1)}\otimes a^{(1)}) \lambda^{-1} (g^{(3)}\otimes a^{(3)}),
	\end{equation}
	where the comultiplication of $g$ is taken in $J^{\rm cop}$.
	\end{definition}
	
When taking $J=\hat{H}$, $K=H$ and $\lambda(\varphi\otimes x)=\varphi(x)$, one recovers the Drinfeld quantum double $D(H)=\hat{H}^{\rm cop}\Join H$ as shown in Appendix \ref{sec:app_hopf}. 

To construct the topological boundary of the Hopf quantum double model with the bulk determined by a Hopf algebra $H$, we need to impose some conditions on $J,K$ such that they are related to $H$.
Here we will assume that there exists a left $H$-comodule algebra $\FA$ such that the category of representations of $J^{\rm cop}\Join_{\lambda} K$ is equivalent to the category of $H$-covariant $\FA|\FA$-bimodules, that is, $\mathsf{Rep}(J^{\rm cop}\Join_{\lambda} K)\simeq {_{ \FA}^H}\mathsf{Mod}_{\FA}$.
The physical meaning of this restriction has been illustrated in Sec.~\ref{sec:bdd-top-exc}.
If $\FA$ is a left $H$-comodule algebra with coaction $\beta_{\FA}$, by $\bar{\FA}$ we mean the right $H$-comodule structure with opposite underlying algebra $\FA^{\rm op}$ and coaction $\bar{\beta}_{\bar{\FA}}(x)=\sum_{(x)}x^{[0]}\otimes S^{-1}(x^{[-1]})$ (where $\beta_{\FA}(x)=\sum_{(x)} x^{[-1]}\otimes x^{[0]}$).
If $\FA$ is a right $H$-comodule algebra, similarly, $\bar{\FA}$ is a left $H$-comodule algebra.
For a right $H$-comodule $(M,\beta_M)$ and a left $H$-comodule $(N,\beta_N)$, the cotensor product $M\Box_H N$ is the equalizer of the morphisms $\beta_M\otimes \id_N$ and $\id_{M}\otimes \beta_{N}$, in the sense as follows: 
\begin{equation}\label{eq:boxtensor}
    M\Box_H N=\left\{\sum_{i}m_i\otimes n_i~\bigg|~\sum_i \beta_M(m_i)\otimes n_i=\sum_i m_i\otimes \beta_N(n_i)\right\}.
\end{equation}
When $\FA$ is a Galois-Hopf extension, the representation category is equivalent to the left module category over $\bar{\FA}\Box_H \FA$. This will be helpful for constructing boundary ribbon operators.
As we will see later, the above restriction for $J,K$ is equivalent to that the representation category of the generalized quantum double is equivalent to the $H$-covariant $\FA|\FA$-bimodule category.
We would like to stress that this restriction is not artificial, since for any such $J^{\rm cop}\Join_{\lambda} K$, starting from the representation category  $\mathsf{Rep}(J^{\rm cop}\Join_{\lambda} K)$, we can construct such an $H$-comodule algebra $\FA$ which satisfies the above assumption. Details have been given Sec.~\ref{sec:bdd-top-exc}.

	\begin{figure}[t]
		\centering
		\includegraphics[width=7cm]{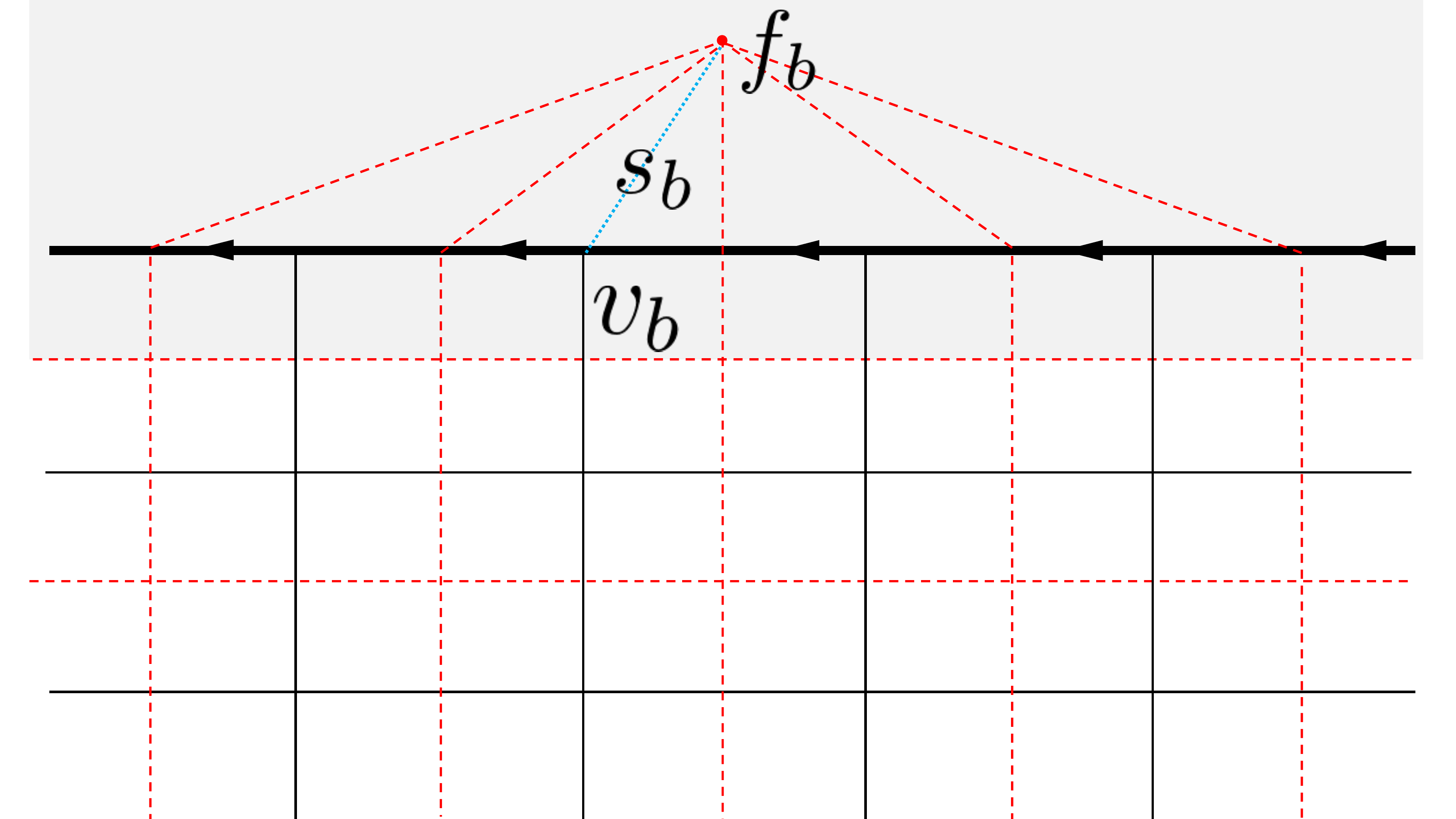}
		\caption{The depiction of the boundary, boundary face, boundary vertex, and boundary site. \label{fig:bd}}
	\end{figure}
	
	Now we are ready to present the construction of the boundary Hamiltonian. The boundary is characterized by the following collection of data: the bulk Hopf algebra $H$ and the boundary Hopf algebras $J,K,W$. They satisfy the following conditions:
	\begin{itemize}
	    \item Two finite-dimensional $C^*$ Hopf algebras $K$ and $J$ for which there exists a pairing $\lambda: J\otimes K \to \mathbb{C}$. And there is a pairing between $\hat{H}$ and $K$, $\gamma: \Hhat \otimes K\to \mathbb{C}$.
	    
	    \item A finite-dimensional $C^*$ Hopf algebra $W$, which is a $K$-bimodule coalgebra. 
	    The left and right $K$-actions are denoted as $h\triangleright w$ and $w \triangleleft h$ respectively for $w\in W,~h\in K$. The right action induces a left $K$-module structure by $w \triangleleft S(h)$. We denote the operators corresponding to these two actions as
	    \begin{equation}
	         L_+^{K,h}w=   h \triangleright w=hw, \quad L_{-}^{K,h}w= w \triangleleft  S(h)=wS(h).
	    \end{equation}
	    We also assume that the antipode law holds: $S(h\triangleright w)=S(w) \triangleleft S(h)$ and $S(w \triangleleft h) =S(h)\triangleright S(w)$;
	    for $C^*$ structure, we assume $(h\triangleright w)^*= w^{*} \triangleleft h^*$ and $(w \triangleleft h)^*=h^{*}\triangleright w^*$; and for Haar integral $h_W\in W$ and $k\in K$, we assume $kh_W=\varepsilon(k)h_W=h_W k$.

	    \item $H$ is a $K$-bimodule coalgebra that satisfies the antipode law.

	   \item There is a pairing $\zeta: J\otimes W\to \mathbb{C}$ between $J$ and $W$.
	   The pairing induces the $J$-bimodule structures of $W$: left module $\varphi \rightharpoonup w=\sum_{(w)}w^{(1)} \zeta (\varphi\otimes w^{(2)})$, and right module $w\leftharpoonup \varphi =\sum_{(w)}w^{(2)}\zeta(\varphi \otimes  w^{(1)})$. The right module structure induces a left module structure by $w\leftharpoonup S(\varphi)$.
	   The corresponding two left module structure operators are denoted as
	\begin{align} \label{eq:pairT}
	\begin{aligned}
	 T^{J,\varphi}_+ &w=\varphi \rightharpoonup w=\sum_{(w)}w^{(1)} \zeta(\varphi\otimes w^{(2)}), \\  T^{J,\varphi}_{-}w&=w\leftharpoonup S(\varphi) =\sum_{(w)}w^{(2)} \zeta^{-1}(\varphi\otimes w^{(1)}).
	 \end{aligned}
	\end{align}
	For the Haar integral $\varphi_{\hat{W}}$, we assume that $\sum_{(x)}\langle \varphi_{\hat{W}} ,x^{(1)}\rangle \zeta(
	\psi,x^{(2)}) =\varepsilon_{J}(\psi) \langle \varphi_{\hat{W}},x\rangle$ for all $\psi\in J$ and $x\in W$.

	\item There is a pairing $\alpha: \hat{H} \otimes W\to \mathbb{C}$ which also induces an $\hat{H}$-bimodule structure over $W$ in a similar way as in Eq.~\eqref{eq:pairT}.
	It is easy to verify that $W$ is also an $\hat{H}$-module algebra, and this module algebra structure induces an $H$-comodule algebra over $W$.
	For the Haar integral $\varphi_{\hat{W}}$, we assume that $\sum_{(x)}\langle \varphi_{\hat{W}} ,x^{(1)}\rangle \alpha(
	\psi,x^{(2)}) =\varepsilon_{\hat{H}}(\psi) \langle \varphi_{\hat{W}},x\rangle$ for all $\psi \in \hat{H}$ and $x\in W$.

	\item Different pairings are consistent when acting on two Hopf algebras for which there is a module structure. For example,  $\lambda: J\otimes K \to \mathbb{C}$ and $\zeta: J\otimes W\to \mathbb{C}$  are two pairings and $W$ has left $K$-module structures $h\triangleright w$ and $w \triangleleft S(h)$, then the consistency conditions read
	\begin{align}
	    \sum_{(\varphi)}\lambda(\varphi^{(1)}\otimes h )\zeta(\varphi^{(2)}\otimes w)&=\zeta(\varphi\otimes (h\triangleright w)), \\
	    \sum_{(\varphi)}\lambda(\varphi^{(1)}\otimes h )\zeta^{-1}(\varphi^{(2)}\otimes w)&=\zeta(\varphi\otimes (w\triangleleft S(h))).
	\end{align}
	\end{itemize}
	
    \begin{remark} \label{rmk:right-mod-str}
        As in the bulk, we can also construct the boundary using the right-module structure; the generalization is straightforward.	
    \end{remark} 
	
	For each boundary edge $e_b$, we assign a Hopf algebra $\mathcal{H}_{e_b}=W$. The total boundary space becomes $ \mathcal{H}(\partial \Sigma)=\otimes _{e_b\in E(\partial \Sigma)} \mathcal{H}_{e_b}$, and therefore, the total Hilbert space becomes $\mathcal{H}_{tot}=\mathcal{H}(\Sigma\setminus \partial \Sigma) \otimes \mathcal{H}(\partial \Sigma)$.
    As shown in Fig.~\ref{fig:bd}, the boundary corresponds to the shaded region, and there is only one boundary face outside a given boundary.
	For boundary site $s_b=(v_b,f_b)$ and $h\in K$, we can define (in counterclockwise order),
	\begin{equation}
		A^{h}(s_b) = \sum_{(h)}   L^{K,h^{(1)}} (j_1,v_b) \otimes \cdots \otimes L^{K,h^{(n)}} (j_n,v_b),
	\end{equation}
	where, for bulk edges, $L^{K,h^{(\ell)}}(j_\ell,v_b)$ act on $H$, and for boundary edges, act on $W$. The convention to choose $L_{+},L_{-}$ is the same as in the bulk.

	For the outer boundary site $s_b$ and $\varphi \in J$, we introduce an operator which acts non-trivially only on all boundary edges (in counterclockwise order),
	\begin{equation}
	    B^{\varphi}(s_b)=\sum_{(\varphi)}T^{J,\varphi^{(1)}}(j_1,f_b)\otimes T^{J,\varphi^{(2)}}(j_2,f_b)\otimes \cdots \otimes T^{J,\varphi^{(n)}}(j_n,f_b).
	\end{equation}
	The convention to choose $T_{+},T_{-}$ is also the same as in the bulk. For example,
		\begin{equation}
		\begin{aligned}
			\begin{tikzpicture}
				\draw[-latex,black,line width=1.5pt] (1,0) -- (0,0); 
				\draw[-latex,black,line width=1.5pt] (0,0) -- (-0.5,0.5);
				\draw[-latex,black,line width=1.5pt] (1.5,0.5) -- (1,0);
				\draw[-latex,black,line width=1.5pt] (1,1) -- (1.5,0.5);
				\draw[-latex,black,line width=1.5pt] (0,1) -- (1,1);
				\draw[dotted,line width=1.0pt, red] (0,0) -- (0.5,0.5);
				\draw [fill = black] (0,0) circle (1.2pt);
				\draw [fill = black] (0.5,0.5) circle (1.2pt);
				\draw[-latex,black,line width=0.5pt] (0,-1) -- (0,0); 
				\draw[-latex,black,line width=0.5pt] (1,-1) -- (1,0); 
				\draw[-latex,black,line width=0.5pt] (1.5,-0.5) -- (1.5,0.5); 
				\node[ line width=0.2pt, dashed, draw opacity=0.5] (a) at (0.2,0.3){$s_b$};
				\node[ line width=0.2pt, dashed, draw opacity=0.5] (a) at (-0.2,-0.2){$v_b$};
				\node[ line width=0.2pt, dashed, draw opacity=0.5] (a) at (0.7,0.5){$f_b$};
				\node[ line width=0.2pt, dashed, draw opacity=0.5] (a) at (0.5,-0.3){$w_1$};
				\node[ line width=0.2pt, dashed, draw opacity=0.5] (a) at (1.3,-0.2){$w_2$};
				\node[ line width=0.2pt, dashed, draw opacity=0.5] (a) at (1.4,0.9){$w_3$};
				\node[ line width=0.2pt, dashed, draw opacity=0.5] (a) at (0.5,1.2){$w_4$};
				\node[ line width=0.2pt, dashed, draw opacity=0.5] (a) at (-0.2,-1.1){$x_1$};
				\node[ line width=0.2pt, dashed, draw opacity=0.5] (a) at (0.8,-1.1){$x_2$};
				\node[ line width=0.2pt, dashed, draw opacity=0.5] (a) at (1.5,-0.8){$x_3$};
				\node[ line width=0.2pt, dashed, draw opacity=0.5] (a) at (-0.6,0.1){$w_n$};
				\node[ line width=0.2pt, dashed, draw opacity=0.5] (a) at (-0.3,0.75){\begin{turn}{90}$\ddots$
				\end{turn}
				};
			\end{tikzpicture}
		\end{aligned}
		\quad 
		\begin{aligned}
			&B^{\varphi}(s_b)|w_1,\cdots, w_n\rangle\\
			=&\sum_{(w_1),\cdots,(w_n),(\varphi)}\zeta(\varphi^{(1)}\otimes w_1^{(2)} ) \cdots \zeta(\varphi^{(n)}\otimes w_n^{(2)} ) |w_1^{(1)},\cdots,w_n^{(1)}\rangle \\
			=& \sum_{(w_1),\cdots,(w_n)}\zeta (\varphi \otimes w_1^{(2)}\cdots w_n^{(2)})|w_1^{(1)},\cdots,w_n^{(1)}\rangle,
		\end{aligned}\label{eq:boundarycycle}
	\end{equation}
	where comultiplication components $\varphi^{(i)}$ are taken in $J$.

	\begin{proposition}\label{prop:double}
	    For the outer boundary site $s_b$, the operators $A^{h}(s_b)$ and $B^{\varphi}(s_b)$ satisfy
	   \begin{equation}\label{eq:bicrossproduct}
	       A^h(s_b)B^{\varphi}(s_b)= \sum_{(h),(\varphi)}\lambda(\varphi^{(3)} \otimes h^{(1)}) \lambda^{-1} (\varphi^{(1)}\otimes h^{(3)}) B^{\varphi^{(2)}}(s_b) A^{h^{(2)}}(s_b).
	   \end{equation}
	   Therefore, by mapping $\varphi \otimes h \in J^{\rm cop}\Join_{\lambda} K$ to $B^{\varphi}(s_b)A^{h}(s_b)$, we obtain a representation of $J^{\rm cop}\Join_{\lambda} K$ over $\mathcal{H}
	   (s_b)=\otimes_{e\in \partial s_b} \mathcal{H}_e$. 
	\end{proposition}
	\begin{proof}
	    Consider the setting as in Eq.~\eqref{eq:boundarycycle}, we see that 
	    \begin{equation}
	    \begin{aligned}
	         & A^h(s_b)B^{\varphi}(s_b) |w_1,\cdots,w_n\rangle |x_1\rangle\\
	          =&  \sum_{(w_i),(h)} \zeta (\varphi \otimes w_1^{(2)}\cdots w_n^{(2)})|h^{(3)}w_1^{(1)},\cdots,w_n^{(1)}S(h^{(1)})\rangle |h^{(2)}x_1\rangle.
	    \end{aligned}
	    \end{equation}
	  On the other hand, we have
	    \begin{equation}
	        \begin{aligned}
	          & \sum_{(h),(\varphi)} \lambda(\varphi^{(3)} \otimes h^{(1)}) \lambda^{-1} (\varphi^{(1)}\otimes h^{(3)}) B^{\varphi^{(2)}}(s_b) A^{h^{(2)}}(s_b)|w_1,\cdots,w_n\rangle |x_1\rangle\\
	            =&\sum_{(w_i),(h),(\varphi)}  \lambda(\varphi^{(3)} \otimes h^{(1)}) \lambda^{-1} (\varphi^{(1)}\otimes h^{(5)}) B^{\varphi^{(2)}}(s_b) |h^{(4)} w_1,\cdots,w_nS(h^{(2)})\rangle |h^{(3)} x_1\rangle\\
	            =&\sum_{(w_i),(h),(\varphi)} \lambda(\varphi^{(3)} \otimes h^{(1)}) \lambda^{-1} (\varphi^{(1)}\otimes h^{(5)}) \zeta(\varphi^{(2)}\otimes  (h^{(4)} w_1)^{(2)} w_2^{(2)}\cdots (w_nS(h^{(2)}))^{(2)}  )\\
	            &|(h^{(4)} w_1)^{(1)}, w_2^{(1)},\cdots, (w_nS(h^{(2)}))^{(1)}  )\rangle |h^{(3)} x_1\rangle \\
	           =&\sum_{(w_i),(h)}  \zeta(\varphi \otimes S^{-1}(h^{(7)}) h^{(6)} w_1^{(2)} w_2^{(2)} \cdots w_n^{(2)} S(h^{(2)}) h^{(1)}) \\
	           &| h^{(5)} w_1^{(1)}, w_2^{(1)},\cdots, w_n^{(1)} S(h^{(3)}) \rangle |h^{(4)} x_1\rangle\\
	           =&\sum_{(w_i),(h)} \zeta (\varphi \otimes w_1^{(2)}\cdots w_n^{(2)})|h^{(3)}w_1^{(1)},\cdots,w_n^{(1)}S(h^{(1)})\rangle |h^{(2)}x_1\rangle.
	        \end{aligned}
	    \end{equation}
	    This completes the proof of Eq.~\eqref{eq:bicrossproduct}.
	    Using Eq.~\eqref{eq:bicrossproduct}, it is easy to check that $B^{\varphi}(s_b)A^h(s_b)$ is a representation of $J^{\rm cop}\Join K$.
	\end{proof}

	Since $K$ and $J$ are both $C^*$ Hopf algebras, they both have unique Haar integrals $h_K$ and $\varphi_{J}$. Then we can define the boundary stabilizer operators as
	\begin{equation}
	    A^K_{v_b}=A^{h_K}(s_b), \quad B^J_{f_b}=B^{\varphi_J}(s_b).
	\end{equation}
	Since both of $h_K$ and $\varphi_J$ are cocommutative, the vertex and face operators only depend on their respective vertices and faces.
	
	\begin{proposition}
	All boundary stabilizer operators are projectors, and satisfy 
	\begin{align}
	   ( A^{h_K}_{v_b})^{\dagger} = A^{h_K^*}_{v_b}= A^{h_K}_{v_b},\quad  (A^{h_K}_{v_b})^2= A^{h_K^2}_{v_b}= A^{h_K}_{v_b},\\
	   (B^{\varphi_J}_{f_b})^{\dagger} =B^{\varphi_J^*}_{f_b}=B^{\varphi_J}_{f_b}, \quad  (B^{\varphi_J}_{f_b})^2=B^{\varphi_J^2}_{f_b}=B^{\varphi_J}_{f_b}.
	\end{align}
	 All boundary stabilizer operators commute with each other and with all bulk stabilizer operators. 
	\end{proposition}
	\begin{proof}
	 The inner product of $W$ is determined by the Haar integral $\varphi_{\hat{W}}$, i.e., $\langle x,y\rangle:=\varphi_{\hat{W}}(x^*y)$.
	 We see that $\langle x,L_+^{h} y\rangle=\varphi_{\hat{W}}(x^*(h\triangleright y))=\varphi_{\hat{W}}((h^*\triangleright x)^*y) =\langle L_+^{h^*} x,y\rangle$, thus $(L_+^{h})^{\dagger}=L_+^{h^*}$. Similarly, $(L^{h}_-)^{\dagger}=L^{h^*}_-$.
	 This implies that $(A^{h_K}_{v_b})^{\dagger}=\sum_{(h_K)}\otimes_j (L^{h^{(j)}_K}(j,v_b))^{\dagger}=\sum_{(h_K)}\otimes_j L^{{h^{(j)}_K}^*}(j,v_b)=A^{h_K^*}_{v_b}=A^{h_K}_{v_b}$. 
	 
     Similarly, it follows that
     \begin{equation}
         \begin{aligned}
            \langle x, T^{\varphi}_{+} y\rangle &=\sum_{(y)} \varphi_{\hat{W}} (x^* y^{(1)} \zeta (\varphi \otimes y^{(2)}))\\
             &=\sum_{(y),(x^*),(\varphi)}\varphi_{\hat{W}}( {x^*}^{(1)} \zeta(1_J\otimes {x^*}^{(2)}) \varepsilon_J(\varphi^{(2)}) \zeta(\varphi^{(1)}\otimes y^{(2)}) y^{(1)})\\
             &=\sum_{(y),(x^*),(\varphi)}\varphi_{\hat{W}}( {x^*}^{(1)} \zeta(\varphi^{(2)} S(\varphi^{(3)})  \otimes {x^*}^{(2)})  \zeta(\varphi^{(1)}\otimes y^{(2)}) y^{(1)} )\\
             &=\sum_{(y),(x^*),(\varphi)}\varphi_{\hat{W}}(
             {x^*}^{(1)}y^{(1)} \zeta (\varphi^{(2)}\otimes {x^*}^{(2)}) \zeta (S(\varphi ^{(3)})\otimes {x^*}^{(3)}) \zeta(\varphi^{(1)}\otimes y^{(2)}) )\\
             &=\sum_{(y),(x^*),(\varphi)}\varphi_{\hat{W}}(
            y^{(1)} {x^*}^{(1)}\zeta (\varphi^{(1)}\otimes y^{(2)} {x^*}^{(2)})\zeta (S(\varphi ^{(2)})\otimes {x^*}^{(3)}))\\
             &= \sum_{(x^*),(y),(\varphi)}\varphi_{\hat{W}}(
             (y{x^*}^{(1)})^{(1)} \zeta (\varphi^{(1)} \otimes (y{x^*}^{(1)})^{(2)}  ) \zeta (S(\varphi ^{(2)})\otimes {x^*}^{(3)}) )\\
             &= \sum_{(x^*)}\varphi_{\hat{W}}( {x^*}^{(1)} y \zeta(S(\varphi)\otimes {x^*}^{(2)} ) )=\langle T_+^{\varphi^*} x ,y\rangle. 
         \end{aligned}
     \end{equation}
     This means that $(T^{\varphi}_+)^{\dagger}=T^{\varphi^*}_+$. We can similarly show that $(T^{\varphi}_-)^{\dagger}=T^{\varphi^*}_-$.
     Thus we see that $(B^{\varphi_J}_{f_b})^{\dagger} =B^{\varphi_J^*}_{f_b}=B^{\varphi_J}_{f_b}$.
     
	 Using Proposition \ref{prop:double} and the fact that $h_K^2=h_K$ and $\varphi_J^2=\varphi_J$, we see that $A^{h_K}_{v_b}$ and $B_{f_b}^{\varphi_J}$ are idempotent.
	 Further using the fact that $h_K$ and $\varphi_J$ are cocommutative, we can show that all stabilizers are commutative.
	\end{proof}

	The boundary Hamiltonian is given by the local commuting projector Hamiltonian
	\begin{equation} \label{eq:BdII}
		H_{K,{J}}(\partial \Sigma)= \sum_{v_b\in V(\partial \Sigma)}(I- A_{v_b}^K) +\sum_{f_b\in F(\partial \Sigma)} (I-B^{J}_{f_b}).
	\end{equation}
	The total Hamiltonian is thus
	\begin{equation}
		H_{H}^{K,J}=H_H(\Sigma\setminus \partial \Sigma) + H_{K,J} (\partial \Sigma).
	\end{equation}
	This model will be called the extended Hopf algebraic quantum double model.

	The ground state degeneracy now depends on both the topology of the $2d$ surface $\Sigma$ and its boundaries. It is still independent of the choice of the cellulation $C(\Sigma)$. 
	
	\begin{equation}
		\operatorname{GSD}=\operatorname{Tr} (\prod_{v\in V(\Sigma \setminus \partial \Sigma) } A^H_v \prod_{f\in F(\Sigma \setminus \partial \Sigma)} B_f^{\hat{H}}\prod_{v_b\in V(\partial \Sigma ) } A^K_{v_b} \prod_{f_b\in F(\partial\Sigma )} B_{f_b}^{{J}}  )
	\end{equation}
	When $\Sigma$ is a sphere, by cutting boundaries, the ground state space will become degenerate.
	
	To shed further light on this construction of boundaries, let us consider two examples.
\begin{example}
    Consider the case that $W=H$, $J \leq \hat{H}$ and $K\leq H$ are Hopf subalgebras. We see that they satisfy all the requirements for constructing the boundary. In this kind of model, it is more convenient to discuss boundary condensation and confinement using the ribbon operators.
	\end{example}
    
    \begin{example}
    For finite group algebra $H=\mathbb{C}[G]$,
    take $W=\mathbb{C}[G]$, $K=\mathbb{C}[M]$,  $J=\mathbb{C}^{G/N}$ with $M$ a subgroup and $N$ a normal subgroup of $G$. We can also obtain a corresponding generalized model.    
    This model is similar to the one given in \cite{Bombin2008family}.
    \end{example}

\emph{Ribbon operator.} ---
Let us briefly discuss the construction of ribbon operators. Assume that $J\leq\hat{H}$ and $K\leq H$ are subalgebras. 
The comultiplication of the dual $D_\lambda(J^{\rm cop},K)^\vee\cong \hat{J}\otimes \hat{K}$ is given by 
\begin{equation*}
    \Delta(h\otimes\varphi) = \sum_{k,j}\sum_{(k),(j),(h)} (h^{(1)}\otimes \hat{k})\otimes (\hat{j}\otimes \varphi(k^{(2)}\bullet)) \lambda(j^{(3)}\otimes k^{(1)})\lambda^{-1}(j^{(1)}\otimes k^{(3)})h^{(2)}(j^{(2)}),
\end{equation*}
for $h\in\hat{J}$ and $\varphi\in\hat{K}$, where $\{k\}\subset K$ and $\{j\}\subset J$ are bases. Define the operators $L^{\varphi}_\pm,\,T^h_\pm:W\to W$ by 
\begin{equation*}
    L^{\hat{k}}_\pm:=L^{K,k}_\pm,\quad T^{\hat{j}}_\pm: = T^{J,j}_\pm,
\end{equation*}
on bases and then by linearity. Since $J\subset \hat{H}$ and $K\subset H$, one can define the operators $L^{\varphi}_\pm,\,T^h_\pm:H\to H$ by 
\begin{equation*}
     L^{\hat{k}}_\pm:=L^{k}_\pm,\quad T^{\hat{j}}_\pm: = T^{j}_\pm,
\end{equation*}
on bases and then by linearity. One can also defines $\tilde{L}^\varphi_\pm,\tilde{T}^h_\pm$ using the right module structures in Remark~\ref{rmk:right-mod-str}. 

Having defined these operators, we can define the triangle operators similar to Eqs.~\eqref{eq:tri1}-\eqref{eq:tri16}. For a bulk-to-boundary ribbon $\rho=\rho_\downarrow=\rho_1\cup\rho_2$, the ribbon operator $F^{h,\varphi}(\rho)$ for $h\in \hat{J}$ and $\varphi\in\hat{K}$ is defined by 
\begin{equation*}
    F^{h,\varphi}(\rho) = \sum_{k,j}\sum_{(k),(j),(h)} \lambda(j^{(3)}\otimes k^{(1)})\lambda^{-1}(j^{(1)}\otimes k^{(3)})h^{(2)}(j^{(2)})F^{h^{(1)}, \hat{k}}(\rho_1)F^{\hat{j},\varphi(k^{(2)}\bullet)}(\rho_2).  
\end{equation*}
This is independent of the choice of the decomposition $\rho=\rho_1\cup\rho_2$. 
One can compute the commutation relations of the local operators and the ribbon operators with the same strategy as we have discussed before.
	
\section{Domain wall lattice based on generalized quantum double}	\label{sec:app-domain-wall}

Suppose that two $2d$ quantum double phases, which are respectively characterized by Hopf algebras $H_1$ and $H_2$, are separated by a $1d$ domain wall, see Fig.~\ref{fig:wall} (a).
	To construct the domain wall Hamiltonian, we need the following input data (all Hopf algebras here are finite-dimensional $C^*$ Hopf algebras):
	\begin{itemize}
	    \item Two domain wall face Hopf algebras $J_1$, $J_2$ and one domain wall vertex Hopf algebra $K$, and there exist pairings, $\lambda_1:J_1\otimes K \to \mathbb{C}$ and  $\lambda_2:J_2\otimes K \to \mathbb{C}$.
	    \item A domain wall Hopf algebra $W$ which is a $K$-bimodule. The left actions of $K$ on $W$ induced by $h\triangleright w$ and $w \triangleleft S(h)$ are denoted as $L^{K,h}_+$ and $L_{-}^{K,h}$ respectively. 
	    \item Two bulk Hopf algebras $H_1$, $H_2$ which are $K$-bimodules. The bimodule structures also induce two inequivalent left module structures on $H_i$. 
	    \item There are pairings $\zeta_1: J_1\otimes W\to \mathbb{C}$ and $\zeta_2: J_2\otimes W\to \mathbb{C}$. 
	    These two pairings induce the $J_1$- and $J_2$-bimodule structures on $W$: $\phi\rightharpoonup w=\sum_{(w)} w^{(1)} \zeta_i(\phi\otimes w^{(2)})$ and $w \leftharpoonup \phi=\sum_{(w)} w^{(2)} \zeta_i(\phi\otimes w^{(1)})$.
	    This further induces two left $J_i$-module structure operators
	    \begin{equation}
	         T^{J_i,\phi}_+ w=\phi \rightharpoonup w, \quad  T^{J_i,\phi}_- w=w\leftharpoonup S(\phi).  
	    \end{equation}
	    \item There are pairings $\xi_{1} : J_1\otimes H_1\to \mathbb{C}$ and $\xi_{2} : J_2\otimes H_2\to \mathbb{C}$, which also induce two inequivalent left $J_i$-module structures on $H_i$.
	\end{itemize}
    Similar to the gapped boundary construction, we will also assume some consistency conditions for these pairings and actions.
	
	For each edge in the first and second bulks $C(\Sigma_1)$ and $C(\Sigma_2)$, we assign the Hopf algebras $H_1$ and $H_2$ respectively. Two bulk total Hilbert spaces are thus $\mathcal{H}(\Sigma_i)=\otimes_{e\in C(\Sigma_i)} H_i$ with $i=1,2$.
	For each domain wall edge $e_{d} \in C
	(\Sigma_d)$ (where $\Sigma_d=\partial \Sigma_1=\partial \Sigma_2$), we attach the Hopf algebra $W$ to it.
	The total domain wall Hilbert space is $\mathcal{H}(\Sigma_d) =\otimes _{e_d\in C(\Sigma_d)} W$.
	We define the domain wall vertex as the vertex on the wall and denote them as $v_d$. The domain wall face on the first bulk is the face near the wall in the first bulk, and it is denoted as $f_{d,1}$. Similarly, we denote $f_{d,2}$ for the domain wall face on the second bulk.

	For a domain wall site  $s_d=(v_{d},f_d)$ and $h\in K$, we defined (in counterclockwise order)
	\begin{equation}
		A^{h}(s_d) = \sum_{(h)}   L^{K,h^{(1)}} (j_1,v_d) \otimes \cdots \otimes L^{K,h^{(n)}} (j_n,v_d),
	\end{equation}
	where, for first and second bulk edges, $L^{K,h^{(\ell)}} (j_\ell,v_d)$ act on $H_i$, and for domain wall edges, they act on $W$. The convention for choosing $L_+,L_-$ is the same as in the bulk.
	
	For the domain wall site $s_{d,i}$ on the $i$-th bulk and $\varphi \in J_i^{\rm cop}$, we defined the face operator (in counterclockwise order)
	\begin{equation}
	    B^{\varphi}(s_{d,i})=\sum_{(\varphi)}T^{J_i,\varphi^{(1)}}(j_1,f_{d,i})\otimes T^{J_i,\varphi^{(2)}}(j_2,f_{d,i})\otimes \cdots \otimes T^{J_i,\varphi^{(n)}}(j_n,f_{d,i}),
	\end{equation}
   	where the comultiplication is taken in $J_i$. The convention for choosing $T_+,T_-$ is the same as in the bulk.

	\begin{proposition}
	   For the domain wall site $s_{d,i}$, the operators $A^h(s_{d,i})$ and $B^{\varphi}(s_{d,i})$ satisfy
	     \begin{equation}\label{eq:wallbicrossproduct}
	       A^h(s_{d,i})B^{\varphi}(s_{d,i})= \sum_{(h),(\varphi)}\lambda_i(\varphi^{(3)} \otimes h^{(1)}) \lambda_i^{-1} (\varphi^{(1)}\otimes h^{(3)}) B^{\varphi^{(2)}}(s_{d,i}) A^{h^{(2)}}(s_{d,i}).
	   \end{equation}
	   Therefore, by mapping $\varphi \otimes h \in J_i^{\rm cop}\Join_{\lambda_i} K$ to $B^{\varphi}(s_{d,i})A^{h}(s_{d,i})$, we obtain a representation of $J^{\rm cop}_i\Join_{\lambda_i} K$ over $\mathcal{H}_i
	   (s_{d,i})=\otimes_{e\in \partial s_{d,i}} \mathcal{H}_e$. 
	\end{proposition}
\begin{proof}
    The proof is similar to the one in Proposition \ref{prop:double}, so we omit it here.
\end{proof}

	Now by choosing the Haar integrals $h_K\in K$ and $\varphi_{J_i}\in J_i$, we define domain wall stabilizers as
		\begin{equation}
	    A^K_{v_d}=A^{h_K}(s_d), \quad B^{J_i}_{f_{d,i}}=B^{\varphi_{J_i}}(s_{d,i}).
	\end{equation}
		All domain wall stabilizer operators are projectors and satisfy 
	\begin{align}
	   ( A^{h_K}_{v_b})^{\dagger} = A^{h_K^*}_{v_b}&= A^{h_K}_{v_b},\quad  (A^{h_K}_{v_b})^2= A^{h_K^2}= A^{h_K}_{v_b},\\
	   (B^{\varphi_{J_i}}_{f_{d,i}})^{\dagger} =B^{\varphi_{J_i}^*}_{f_{d,i}}&=B^{\varphi_{J_i}}_{f_{d,i}}, \quad  (B^{\varphi_{J_i}}_{f_{d,i}})^2=B^{\varphi_{J_i}^2}_{f_{d,i}}=B^{\varphi_{J_i}}_{f_{d,i}}.
	\end{align}
	 All boundary stabilizer operators commute with each other and with all bulk stabilizer operators. 
	The domain wall Hamiltonian can be constructed as 
	\begin{equation}
	  H^{K,J_1,J_2}(\Sigma_d)=\sum_{v_d} (I-A^K_{v_d}) +\sum_{f_{d,1}}(I- B^{J_1}_{f_{d,1}})+\sum_{f_{d,2}} (I-B^{J_2}_{f_{d,2}}).
	\end{equation}
	The total Hamiltonian is thus
	\begin{equation}
	    H_{H_1,H_2}^{K,J_1,J_2}=H_{H_1}(\Sigma_1)+H_{H_2}(\Sigma_2)+ H^{K,J_1,J_2}(\Sigma_d).
	\end{equation}
	
	\begin{example}
    To take a closer look at the construction, let us consider a simple example.
    Choose $J_1\subseteq \Hhat_1$,$J_2\subseteq\Hhat_2$ and $W=H_1\otimes H_2$.
    The pairings are set as
    \begin{align}
        \langle \varphi_1, h_1\otimes h_2\rangle &:=\langle \varphi_1,h_1\rangle \langle \varepsilon,h_2\rangle =\varphi_1(h_1) \varepsilon(h_2),\\
        \langle \varphi_2, h_1\otimes h_2\rangle &:=\langle \varepsilon,h_1\rangle \langle \varphi_2,h_2\rangle =\varepsilon(h_1) \varphi_2(h_2),
    \end{align}
    for $\varphi_i\in J_i$. It is easy to verify that these give well-defined pairings.
   For the edge operators induced by these pairings, we have
    \begin{align}
        T^{J_1,\varphi}_{\pm} h_1\otimes h_2= (T^{J_1,\varphi}_{\pm} h_1) \otimes h_2, \\
        T^{J_2,\varphi}_{\pm} h_1\otimes h_2= h_1\otimes (T^{J_2,\varphi}_{\pm} h_2).
    \end{align}
    The action decouples from the tensor product of two phases.
    We could introduce a domain wall Hopf algebra $K$ to sew them together. The actions of $K$ on $H_i$ naturally induce an action of $K$ on $W=H_1\otimes H_2$ as $k\triangleright (h_1\otimes h_2)=\sum_{(k)} k^{(1)} h_1\otimes k^{(2)} h_2$. To simplify the notation, we will denote the action of $K$ on $H_i$ as $kh$ and $h_iS(k)$.
    For example
    \begin{equation}
        \begin{aligned}
			\begin{tikzpicture}
				\draw[-latex,black,line width = 1.6pt] (-1,0) -- (0,0); 
				\draw[-latex,black,line width = 1.6pt] (0,0) -- (1,0); 
				\draw[-latex,black] (0,0) -- (0,1); 
				\draw[-latex,black] (0,-1) -- (0,0); 
				\draw[line width=0.5pt, red] (0,0) -- (0.5,0.5);
				\draw [fill = black] (0,0) circle (1.2pt);
				\draw [fill = black] (0.5,0.5) circle (1.2pt);
				\node[ line width=0.2pt, dashed, draw opacity=0.5] (a) at (0.2,0.6){$s$};
				\node[ line width=0.2pt, dashed, draw opacity=0.5] (a) at (-0.3,-0.3){$v_d$};
				\node[ line width=0.2pt, dashed, draw opacity=0.5] (a) at (0.9,0.7){$f_{d,1}$};
				\node[ line width=0.2pt, dashed, draw opacity=0.5] (a) at (-1.6,0){$h_3\otimes g_2$};
				\node[ line width=0.2pt, dashed, draw opacity=0.5] (a) at (0,-1.2){$g_3$};
				\node[ line width=0.2pt, dashed, draw opacity=0.5] (a) at (0.8,-0.3){$h_1\otimes g_1$};
				\node[ line width=0.2pt, dashed, draw opacity=0.5] (a) at (-0.2,0.5){$h_2$};
			\end{tikzpicture}
		\end{aligned} \quad 
		\begin{aligned}
		    	&	A^{k} |h_1\otimes g_1,h_2,h_3\otimes g_2,g_3 \rangle\\ 
		    		=&\sum_{(k)} | h_2 S(k^{(1)}), 
		    		k^{(2)} (h_3\otimes g_2), k^{(3)} g_3,
		    		(h_1\otimes g_1) S(k^{(4)})\rangle.
		\end{aligned}
    \end{equation}
    We see that the gapped domain wall is determined by $K,J_1,J_2$.
\end{example}

\emph{$N$-bulk domain wall and folding trick.} ---
The above discussion about the domain wall between two bulks can be generalized to the $N$-bulk scenario.
This $N$-bulk domain wall is crucial for investigating the higher-dimensional Hopf quantum double model.

In this $N$-bulk case, it turns out that we need extra input data, that is the direction for each bulk (denoted as $\sigma(i)$).
If we fix the direction of the domain wall, then using the right-hand screw rule, there is an induced direction $\tilde{\sigma}(i)$ for each bulk.
This induced direction may or may not coincide with the direction of that bulk.
We assign a Hopf algebra $H_i$ (equivalently, a UFC $\EC_{i}=\mathsf{Rep}(H_i)$) for each bulk.
An $N$-comodule algebra $\FB$ can be defined as follows: for arbitrary pair of bulks $H_i,H_j$, if $\sigma(i)=\tilde{\sigma}(i)$ and $\sigma(j)\neq \tilde{\sigma}(j)$, there is an $H_1\otimes H_2^{\rm cop}$-comodule structure over $\FB$, and all other situations can be treated similarly. 
The $N$-bulk domain wall is thus characterized by this $N$-comodule algebra $\FB$.

The lattice construction is similar to the two-bulk case: we assign a Hopf algebra $W$ for each wall edge; a Hopf algebra $K$ for each wall vertex; and a Hopf algebra $J_i$ for boundary face in the $H_i$-bulk.
They satisfy similar compatibility conditions to the two-bulk case.
The wall Hamiltonian becomes 
\begin{equation}
    H_{H_1,\cdots,H_N}^{K,J_1,\cdots ,J_N}(\Sigma_d)=\sum_{v_d}(I-A_{v_d}^K)+\sum_{i}\sum_{f_{d,i}}(I-B^{J_i}_{f_{d,i}}).
\end{equation}
The wall excitations are given by the $N$-module functor category $\mathsf{Fun}_{\EC_1|\cdots|\EC_N}({_{\FB}}\mathsf{Mod}, {_{\FB}}\mathsf{Mod})$.
The folding trick also holds in this case, namely, by stacking all bulk together, the domain wall will become a gapped boundary.

\bibliographystyle{apsrev4-1-title}
\bibliography{mybib.bib}

\begin{thebibliography}{93}%
\makeatletter
\providecommand \@ifxundefined [1]{%
 \@ifx{#1\undefined}
}%
\providecommand \@ifnum [1]{%
 \ifnum #1\expandafter \@firstoftwo
 \else \expandafter \@secondoftwo
 \fi
}%
\providecommand \@ifx [1]{%
 \ifx #1\expandafter \@firstoftwo
 \else \expandafter \@secondoftwo
 \fi
}%
\providecommand \natexlab [1]{#1}%
\providecommand \enquote  [1]{``#1''}%
\providecommand \bibnamefont  [1]{#1}%
\providecommand \bibfnamefont [1]{#1}%
\providecommand \citenamefont [1]{#1}%
\providecommand \href@noop [0]{\@secondoftwo}%
\providecommand \href [0]{\begingroup \@sanitize@url \@href}%
\providecommand \@href[1]{\@@startlink{#1}\@@href}%
\providecommand \@@href[1]{\endgroup#1\@@endlink}%
\providecommand \@sanitize@url [0]{\catcode `\\12\catcode `\$12\catcode
  `\&12\catcode `\#12\catcode `\^12\catcode `\_12\catcode `\%12\relax}%
\providecommand \@@startlink[1]{}%
\providecommand \@@endlink[0]{}%
\providecommand \url  [0]{\begingroup\@sanitize@url \@url }%
\providecommand \@url [1]{\endgroup\@href {#1}{\urlprefix }}%
\providecommand \urlprefix  [0]{URL }%
\providecommand \Eprint [0]{\href }%
\providecommand \doibase [0]{http://dx.doi.org/}%
\providecommand \selectlanguage [0]{\@gobble}%
\providecommand \bibinfo  [0]{\@secondoftwo}%
\providecommand \bibfield  [0]{\@secondoftwo}%
\providecommand \translation [1]{[#1]}%
\providecommand \BibitemOpen [0]{}%
\providecommand \bibitemStop [0]{}%
\providecommand \bibitemNoStop [0]{.\EOS\space}%
\providecommand \EOS [0]{\spacefactor3000\relax}%
\providecommand \BibitemShut  [1]{\csname bibitem#1\endcsname}%
\let\auto@bib@innerbib\@empty
\bibitem [{\citenamefont {Wen}\ and\ \citenamefont {Niu}(1990)}]{Wen1990}%
  \BibitemOpen
  \bibfield  {author} {\bibinfo {author} {\bibfnamefont {X.~G.}\ \bibnamefont
  {Wen}}\ and\ \bibinfo {author} {\bibfnamefont {Q.}~\bibnamefont {Niu}},\
  }\bibfield  {title} {\enquote {\bibinfo {title} {Ground-state degeneracy of
  the fractional quantum {H}all states in the presence of a random potential
  and on high-genus riemann surfaces},}\ }\href {\doibase
  10.1103/PhysRevB.41.9377} {\bibfield  {journal} {\bibinfo  {journal} {Phys.
  Rev. B}\ }\textbf {\bibinfo {volume} {41}},\ \bibinfo {pages} {9377}
  (\bibinfo {year} {1990})}\BibitemShut {NoStop}%
\bibitem [{\citenamefont {Wen}(2004)}]{Wen2004}%
  \BibitemOpen
  \bibfield  {author} {\bibinfo {author} {\bibfnamefont {X.-G.}\ \bibnamefont
  {Wen}},\ }\href
  {http://www.oxfordscholarship.com/view/10.1093/acprof:oso/9780199227259.001.0001/acprof-9780199227259}
  {\emph {\bibinfo {title} {Quantum field theory of many-body systems: from the
  origin of sound to an origin of light and electrons}}}\ (\bibinfo
  {publisher} {Oxford University Press on Demand},\ \bibinfo {year}
  {2004})\BibitemShut {NoStop}%
\bibitem [{\citenamefont {Dennis}\ \emph {et~al.}(2002)\citenamefont {Dennis},
  \citenamefont {Kitaev}, \citenamefont {Landahl},\ and\ \citenamefont
  {Preskill}}]{Dennis2002topological}%
  \BibitemOpen
  \bibfield  {author} {\bibinfo {author} {\bibfnamefont {E.}~\bibnamefont
  {Dennis}}, \bibinfo {author} {\bibfnamefont {A.}~\bibnamefont {Kitaev}},
  \bibinfo {author} {\bibfnamefont {A.}~\bibnamefont {Landahl}}, \ and\
  \bibinfo {author} {\bibfnamefont {J.}~\bibnamefont {Preskill}},\ }\bibfield
  {title} {\enquote {\bibinfo {title} {Topological quantum memory},}\ }\href
  {\doibase 10.1063/1.1499754} {\bibfield  {journal} {\bibinfo  {journal}
  {Journal of Mathematical Physics}\ }\textbf {\bibinfo {volume} {43}},\
  \bibinfo {pages} {4452–4505} (\bibinfo {year} {2002})},\ \Eprint
  {http://arxiv.org/abs/quant-ph/0110143} {arXiv:quant-ph/0110143 [quant-ph]}
  \BibitemShut {NoStop}%
\bibitem [{\citenamefont {Terhal}(2015)}]{Terhal2015quantum}%
  \BibitemOpen
  \bibfield  {author} {\bibinfo {author} {\bibfnamefont {B.~M.}\ \bibnamefont
  {Terhal}},\ }\bibfield  {title} {\enquote {\bibinfo {title} {Quantum error
  correction for quantum memories},}\ }\href {\doibase
  10.1103/RevModPhys.87.307} {\bibfield  {journal} {\bibinfo  {journal} {Rev.
  Mod. Phys.}\ }\textbf {\bibinfo {volume} {87}},\ \bibinfo {pages} {307}
  (\bibinfo {year} {2015})},\ \Eprint {http://arxiv.org/abs/1302.3428}
  {arXiv:1302.3428 [quant-ph]} \BibitemShut {NoStop}%
\bibitem [{\citenamefont {Kitaev}(2003)}]{Kitaev2003}%
  \BibitemOpen
  \bibfield  {author} {\bibinfo {author} {\bibfnamefont {A.}~\bibnamefont
  {Kitaev}},\ }\bibfield  {title} {\enquote {\bibinfo {title} {Fault-tolerant
  quantum computation by anyons},}\ }\href {\doibase
  https://doi.org/10.1016/S0003-4916(02)00018-0} {\bibfield  {journal}
  {\bibinfo  {journal} {Annals of Physics}\ }\textbf {\bibinfo {volume}
  {303}},\ \bibinfo {pages} {2 } (\bibinfo {year} {2003})},\ \Eprint
  {http://arxiv.org/abs/quant-ph/9707021} {arXiv:quant-ph/9707021 [quant-ph]}
  \BibitemShut {NoStop}%
\bibitem [{\citenamefont {Freedman}\ \emph {et~al.}(2002)\citenamefont
  {Freedman}, \citenamefont {Larsen},\ and\ \citenamefont
  {Wang}}]{freedman2002modular}%
  \BibitemOpen
  \bibfield  {author} {\bibinfo {author} {\bibfnamefont {M.~H.}\ \bibnamefont
  {Freedman}}, \bibinfo {author} {\bibfnamefont {M.}~\bibnamefont {Larsen}}, \
  and\ \bibinfo {author} {\bibfnamefont {Z.}~\bibnamefont {Wang}},\ }\bibfield
  {title} {\enquote {\bibinfo {title} {A modular functor which is universal for
  quantum computation},}\ }\href
  {https://link.springer.com/article/10.1007/s002200200645} {\bibfield
  {journal} {\bibinfo  {journal} {Communications in Mathematical Physics}\
  }\textbf {\bibinfo {volume} {227}},\ \bibinfo {pages} {605} (\bibinfo {year}
  {2002})},\ \Eprint {http://arxiv.org/abs/quant-ph/0001108}
  {arXiv:quant-ph/0001108 [quant-ph]} \BibitemShut {NoStop}%
\bibitem [{\citenamefont {Nayak}\ \emph {et~al.}(2008)\citenamefont {Nayak},
  \citenamefont {Simon}, \citenamefont {Stern}, \citenamefont {Freedman},\ and\
  \citenamefont {Das~Sarma}}]{Nayak2008}%
  \BibitemOpen
  \bibfield  {author} {\bibinfo {author} {\bibfnamefont {C.}~\bibnamefont
  {Nayak}}, \bibinfo {author} {\bibfnamefont {S.~H.}\ \bibnamefont {Simon}},
  \bibinfo {author} {\bibfnamefont {A.}~\bibnamefont {Stern}}, \bibinfo
  {author} {\bibfnamefont {M.}~\bibnamefont {Freedman}}, \ and\ \bibinfo
  {author} {\bibfnamefont {S.}~\bibnamefont {Das~Sarma}},\ }\bibfield  {title}
  {\enquote {\bibinfo {title} {Non-{A}belian anyons and topological quantum
  computation},}\ }\href {\doibase 10.1103/RevModPhys.80.1083} {\bibfield
  {journal} {\bibinfo  {journal} {Rev. Mod. Phys.}\ }\textbf {\bibinfo {volume}
  {80}},\ \bibinfo {pages} {1083} (\bibinfo {year} {2008})},\ \Eprint
  {http://arxiv.org/abs/0707.1889} {arXiv:0707.1889 [cond-mat.str-el]}
  \BibitemShut {NoStop}%
\bibitem [{\citenamefont {Mesaros}\ and\ \citenamefont
  {Ran}(2013)}]{Mesaros2013classification}%
  \BibitemOpen
  \bibfield  {author} {\bibinfo {author} {\bibfnamefont {A.}~\bibnamefont
  {Mesaros}}\ and\ \bibinfo {author} {\bibfnamefont {Y.}~\bibnamefont {Ran}},\
  }\bibfield  {title} {\enquote {\bibinfo {title} {Classification of symmetry
  enriched topological phases with exactly solvable models},}\ }\href {\doibase
  10.1103/PhysRevB.87.155115} {\bibfield  {journal} {\bibinfo  {journal} {Phys.
  Rev. B}\ }\textbf {\bibinfo {volume} {87}},\ \bibinfo {pages} {155115}
  (\bibinfo {year} {2013})},\ \Eprint {http://arxiv.org/abs/1212.0835}
  {arXiv:1212.0835 [cond-mat.str-el]} \BibitemShut {NoStop}%
\bibitem [{\citenamefont {Bullivant}\ \emph {et~al.}(2017)\citenamefont
  {Bullivant}, \citenamefont {Hu},\ and\ \citenamefont
  {Wan}}]{Bullivant2017twisted}%
  \BibitemOpen
  \bibfield  {author} {\bibinfo {author} {\bibfnamefont {A.}~\bibnamefont
  {Bullivant}}, \bibinfo {author} {\bibfnamefont {Y.}~\bibnamefont {Hu}}, \
  and\ \bibinfo {author} {\bibfnamefont {Y.}~\bibnamefont {Wan}},\ }\bibfield
  {title} {\enquote {\bibinfo {title} {Twisted quantum double model of
  topological order with boundaries},}\ }\href {\doibase
  10.1103/PhysRevB.96.165138} {\bibfield  {journal} {\bibinfo  {journal} {Phys.
  Rev. B}\ }\textbf {\bibinfo {volume} {96}},\ \bibinfo {pages} {165138}
  (\bibinfo {year} {2017})},\ \Eprint {http://arxiv.org/abs/1706.03611}
  {arXiv:1706.03611 [cond-mat.str-el]} \BibitemShut {NoStop}%
\bibitem [{\citenamefont {Hu}\ \emph {et~al.}(2013)\citenamefont {Hu},
  \citenamefont {Wan},\ and\ \citenamefont {Wu}}]{Hu2012twisted}%
  \BibitemOpen
  \bibfield  {author} {\bibinfo {author} {\bibfnamefont {Y.}~\bibnamefont
  {Hu}}, \bibinfo {author} {\bibfnamefont {Y.}~\bibnamefont {Wan}}, \ and\
  \bibinfo {author} {\bibfnamefont {Y.-S.}\ \bibnamefont {Wu}},\ }\bibfield
  {title} {\enquote {\bibinfo {title} {Twisted quantum double model of
  topological phases in two dimensions},}\ }\href {\doibase
  10.1103/PhysRevB.87.125114} {\bibfield  {journal} {\bibinfo  {journal} {Phys.
  Rev. B}\ }\textbf {\bibinfo {volume} {87}},\ \bibinfo {pages} {125114}
  (\bibinfo {year} {2013})},\ \Eprint {http://arxiv.org/abs/1211.3695}
  {arXiv:1211.3695 [cond-mat.str-el]} \BibitemShut {NoStop}%
\bibitem [{\citenamefont {Dijkgraaf}\ and\ \citenamefont
  {Witten}(1990)}]{dijkgraaf1990topological}%
  \BibitemOpen
  \bibfield  {author} {\bibinfo {author} {\bibfnamefont {R.}~\bibnamefont
  {Dijkgraaf}}\ and\ \bibinfo {author} {\bibfnamefont {E.}~\bibnamefont
  {Witten}},\ }\bibfield  {title} {\enquote {\bibinfo {title} {Topological
  gauge theories and group cohomology},}\ }\href
  {https://link.springer.com/article/10.1007/BF02096988} {\bibfield  {journal}
  {\bibinfo  {journal} {Communications in Mathematical Physics}\ }\textbf
  {\bibinfo {volume} {129}},\ \bibinfo {pages} {393} (\bibinfo {year}
  {1990})}\BibitemShut {NoStop}%
\bibitem [{\citenamefont {Kuperberg}(1996)}]{kup1996noninvolutary}%
  \BibitemOpen
  \bibfield  {author} {\bibinfo {author} {\bibfnamefont {G.}~\bibnamefont
  {Kuperberg}},\ }\bibfield  {title} {\enquote {\bibinfo {title} {Noninvolutory
  {H}opf algebras and {$3$}-manifold invariants},}\ }\href {\doibase
  10.1215/S0012-7094-96-08403-3} {\bibfield  {journal} {\bibinfo  {journal}
  {Duke Math. J.}\ }\textbf {\bibinfo {volume} {84}},\ \bibinfo {pages} {83}
  (\bibinfo {year} {1996})},\ \Eprint {http://arxiv.org/abs/q-alg/9712047}
  {arXiv:q-alg/9712047 [math.QA]} \BibitemShut {NoStop}%
\bibitem [{\citenamefont {Levin}\ and\ \citenamefont {Wen}(2005)}]{Levin2005}%
  \BibitemOpen
  \bibfield  {author} {\bibinfo {author} {\bibfnamefont {M.~A.}\ \bibnamefont
  {Levin}}\ and\ \bibinfo {author} {\bibfnamefont {X.-G.}\ \bibnamefont
  {Wen}},\ }\bibfield  {title} {\enquote {\bibinfo {title} {String-net
  condensation: A physical mechanism for topological phases},}\ }\href
  {\doibase 10.1103/PhysRevB.71.045110} {\bibfield  {journal} {\bibinfo
  {journal} {Phys. Rev. B}\ }\textbf {\bibinfo {volume} {71}},\ \bibinfo
  {pages} {045110} (\bibinfo {year} {2005})},\ \Eprint
  {http://arxiv.org/abs/cond-mat/0404617} {arXiv:cond-mat/0404617
  [cond-mat.str-el]} \BibitemShut {NoStop}%
\bibitem [{\citenamefont {Turaev}\ and\ \citenamefont
  {Viro}(1992)}]{turaev1992state}%
  \BibitemOpen
  \bibfield  {author} {\bibinfo {author} {\bibfnamefont {V.~G.}\ \bibnamefont
  {Turaev}}\ and\ \bibinfo {author} {\bibfnamefont {O.~Y.}\ \bibnamefont
  {Viro}},\ }\bibfield  {title} {\enquote {\bibinfo {title} {State sum
  invariants of 3-manifolds and quantum 6$j$-symbols},}\ }\href
  {https://www.sciencedirect.com/science/article/pii/004093839290015A}
  {\bibfield  {journal} {\bibinfo  {journal} {Topology}\ }\textbf {\bibinfo
  {volume} {31}},\ \bibinfo {pages} {865} (\bibinfo {year} {1992})}\BibitemShut
  {NoStop}%
\bibitem [{\citenamefont {Barrett}\ and\ \citenamefont
  {Westbury}(1996)}]{barrett1996invariants}%
  \BibitemOpen
  \bibfield  {author} {\bibinfo {author} {\bibfnamefont {J.}~\bibnamefont
  {Barrett}}\ and\ \bibinfo {author} {\bibfnamefont {B.}~\bibnamefont
  {Westbury}},\ }\bibfield  {title} {\enquote {\bibinfo {title} {Invariants of
  piecewise-linear 3-manifolds},}\ }\href
  {https://www.ams.org/journals/tran/1996-348-10/S0002-9947-96-01660-1/}
  {\bibfield  {journal} {\bibinfo  {journal} {Transactions of the American
  Mathematical Society}\ }\textbf {\bibinfo {volume} {348}},\ \bibinfo {pages}
  {3997} (\bibinfo {year} {1996})},\ \Eprint
  {http://arxiv.org/abs/hep-th/9311155} {arXiv:hep-th/9311155 [hep-th]}
  \BibitemShut {NoStop}%
\bibitem [{\citenamefont {Buerschaper}\ \emph
  {et~al.}(2013{\natexlab{a}})\citenamefont {Buerschaper}, \citenamefont
  {Mombelli}, \citenamefont {Christandl},\ and\ \citenamefont
  {Aguado}}]{Buerschaper2013a}%
  \BibitemOpen
  \bibfield  {author} {\bibinfo {author} {\bibfnamefont {O.}~\bibnamefont
  {Buerschaper}}, \bibinfo {author} {\bibfnamefont {J.~M.}\ \bibnamefont
  {Mombelli}}, \bibinfo {author} {\bibfnamefont {M.}~\bibnamefont
  {Christandl}}, \ and\ \bibinfo {author} {\bibfnamefont {M.}~\bibnamefont
  {Aguado}},\ }\bibfield  {title} {\enquote {\bibinfo {title} {A hierarchy of
  topological tensor network states},}\ }\href {\doibase 10.1063/1.4773316}
  {\bibfield  {journal} {\bibinfo  {journal} {Journal of Mathematical Physics}\
  }\textbf {\bibinfo {volume} {54}},\ \bibinfo {pages} {012201} (\bibinfo
  {year} {2013}{\natexlab{a}})},\ \Eprint {http://arxiv.org/abs/1007.5283}
  {arXiv:1007.5283 [cond-mat.str-el]} \BibitemShut {NoStop}%
\bibitem [{\citenamefont {Chang}(2014)}]{chang2014kitaev}%
  \BibitemOpen
  \bibfield  {author} {\bibinfo {author} {\bibfnamefont {L.}~\bibnamefont
  {Chang}},\ }\bibfield  {title} {\enquote {\bibinfo {title} {Kitaev models
  based on unitary quantum groupoids},}\ }\href
  {https://aip.scitation.org/doi/abs/10.1063/1.4869326} {\bibfield  {journal}
  {\bibinfo  {journal} {Journal of Mathematical Physics}\ }\textbf {\bibinfo
  {volume} {55}},\ \bibinfo {pages} {041703} (\bibinfo {year} {2014})},\
  \Eprint {http://arxiv.org/abs/1309.4181} {arXiv:1309.4181 [math.QA]}
  \BibitemShut {NoStop}%
\bibitem [{\citenamefont {Buerschaper}\ \emph
  {et~al.}(2013{\natexlab{b}})\citenamefont {Buerschaper}, \citenamefont
  {Christandl}, \citenamefont {Kong},\ and\ \citenamefont
  {Aguado}}]{buerschaper2013electric}%
  \BibitemOpen
  \bibfield  {author} {\bibinfo {author} {\bibfnamefont {O.}~\bibnamefont
  {Buerschaper}}, \bibinfo {author} {\bibfnamefont {M.}~\bibnamefont
  {Christandl}}, \bibinfo {author} {\bibfnamefont {L.}~\bibnamefont {Kong}}, \
  and\ \bibinfo {author} {\bibfnamefont {M.}~\bibnamefont {Aguado}},\
  }\bibfield  {title} {\enquote {\bibinfo {title} {Electric--magnetic duality
  of lattice systems with topological order},}\ }\href
  {https://www.sciencedirect.com/science/article/abs/pii/S0550321313004367?via%3Dihub}
  {\bibfield  {journal} {\bibinfo  {journal} {Nuclear Physics B}\ }\textbf
  {\bibinfo {volume} {876}},\ \bibinfo {pages} {619} (\bibinfo {year}
  {2013}{\natexlab{b}})},\ \Eprint {http://arxiv.org/abs/1006.5823}
  {arXiv:1006.5823 [cond-mat.str-el]} \BibitemShut {NoStop}%
\bibitem [{\citenamefont {Yan}\ \emph {et~al.}(2022)\citenamefont {Yan},
  \citenamefont {Chen},\ and\ \citenamefont {Cui}}]{chen2021ribbon}%
  \BibitemOpen
  \bibfield  {author} {\bibinfo {author} {\bibfnamefont {B.}~\bibnamefont
  {Yan}}, \bibinfo {author} {\bibfnamefont {P.}~\bibnamefont {Chen}}, \ and\
  \bibinfo {author} {\bibfnamefont {S.}~\bibnamefont {Cui}},\ }\bibfield
  {title} {\enquote {\bibinfo {title} {Ribbon operators in the generalized
  {K}itaev quantum double model based on {H}opf algebras},}\ }\href
  {https://iopscience.iop.org/article/10.1088/1751-8121/ac552c/meta} {\bibfield
   {journal} {\bibinfo  {journal} {Journal of Physics A: Mathematical and
  Theoretical}\ } (\bibinfo {year} {2022})},\ \Eprint
  {http://arxiv.org/abs/2105.08202} {arXiv:2105.08202 [cond-mat.str-el]}
  \BibitemShut {NoStop}%
\bibitem [{\citenamefont {Bais}\ \emph {et~al.}(2003)\citenamefont {Bais},
  \citenamefont {Schroers},\ and\ \citenamefont {Slingerland}}]{bais2003hopf}%
  \BibitemOpen
  \bibfield  {author} {\bibinfo {author} {\bibfnamefont {A.~F.}\ \bibnamefont
  {Bais}}, \bibinfo {author} {\bibfnamefont {B.~J.}\ \bibnamefont {Schroers}},
  \ and\ \bibinfo {author} {\bibfnamefont {J.~K.}\ \bibnamefont
  {Slingerland}},\ }\bibfield  {title} {\enquote {\bibinfo {title} {Hopf
  symmetry breaking and confinement in (2+1)-dimensional gauge theory},}\
  }\href {https://doi.org/10.1088/1126-6708/2003/05/068} {\bibfield  {journal}
  {\bibinfo  {journal} {Journal of High Energy Physics}\ }\textbf {\bibinfo
  {volume} {2003}},\ \bibinfo {pages} {068} (\bibinfo {year} {2003})},\ \Eprint
  {http://arxiv.org/abs/hep-th/0205114} {arXiv:hep-th/0205114 [hep-th]}
  \BibitemShut {NoStop}%
\bibitem [{\citenamefont {Meusburger}(2017)}]{meusburger2017kitaev}%
  \BibitemOpen
  \bibfield  {author} {\bibinfo {author} {\bibfnamefont {C.}~\bibnamefont
  {Meusburger}},\ }\bibfield  {title} {\enquote {\bibinfo {title} {Kitaev
  lattice models as a {H}opf algebra gauge theory},}\ }\href
  {https://link.springer.com/article/10.1007%2Fs00220-017-2860-7} {\bibfield
  {journal} {\bibinfo  {journal} {Communications in Mathematical Physics}\
  }\textbf {\bibinfo {volume} {353}},\ \bibinfo {pages} {413} (\bibinfo {year}
  {2017})},\ \Eprint {http://arxiv.org/abs/1607.01144} {arXiv:1607.01144
  [math.QA]} \BibitemShut {NoStop}%
\bibitem [{\citenamefont {Meusburger}\ and\ \citenamefont
  {Wise}(2021)}]{meusburger2021hopf}%
  \BibitemOpen
  \bibfield  {author} {\bibinfo {author} {\bibfnamefont {C.}~\bibnamefont
  {Meusburger}}\ and\ \bibinfo {author} {\bibfnamefont {D.~K.}\ \bibnamefont
  {Wise}},\ }\bibfield  {title} {\enquote {\bibinfo {title} {Hopf algebra gauge
  theory on a ribbon graph},}\ }\href
  {https://www.worldscientific.com/doi/abs/10.1142/S0129055X21500161}
  {\bibfield  {journal} {\bibinfo  {journal} {Reviews in Mathematical Physics}\
  ,\ \bibinfo {pages} {2150016}} (\bibinfo {year} {2021})},\ \Eprint
  {http://arxiv.org/abs/1512.03966} {arXiv:1512.03966 [math.QA]} \BibitemShut
  {NoStop}%
\bibitem [{\citenamefont {Buerschaper}\ and\ \citenamefont
  {Aguado}(2009)}]{Buerschaper2009mapping}%
  \BibitemOpen
  \bibfield  {author} {\bibinfo {author} {\bibfnamefont {O.}~\bibnamefont
  {Buerschaper}}\ and\ \bibinfo {author} {\bibfnamefont {M.}~\bibnamefont
  {Aguado}},\ }\bibfield  {title} {\enquote {\bibinfo {title} {Mapping
  {K}itaev's quantum double lattice models to {L}evin and {W}en's string-net
  models},}\ }\href {\doibase 10.1103/PhysRevB.80.155136} {\bibfield  {journal}
  {\bibinfo  {journal} {Phys. Rev. B}\ }\textbf {\bibinfo {volume} {80}},\
  \bibinfo {pages} {155136} (\bibinfo {year} {2009})},\ \Eprint
  {http://arxiv.org/abs/0907.2670} {arXiv:0907.2670 [cond-mat.str-el]}
  \BibitemShut {NoStop}%
\bibitem [{\citenamefont {Hu}\ \emph {et~al.}(2018{\natexlab{a}})\citenamefont
  {Hu}, \citenamefont {Geer},\ and\ \citenamefont {Wu}}]{Hu2018full}%
  \BibitemOpen
  \bibfield  {author} {\bibinfo {author} {\bibfnamefont {Y.}~\bibnamefont
  {Hu}}, \bibinfo {author} {\bibfnamefont {N.}~\bibnamefont {Geer}}, \ and\
  \bibinfo {author} {\bibfnamefont {Y.-S.}\ \bibnamefont {Wu}},\ }\bibfield
  {title} {\enquote {\bibinfo {title} {Full dyon excitation spectrum in
  extended {L}evin-{W}en models},}\ }\href {\doibase
  10.1103/PhysRevB.97.195154} {\bibfield  {journal} {\bibinfo  {journal} {Phys.
  Rev. B}\ }\textbf {\bibinfo {volume} {97}},\ \bibinfo {pages} {195154}
  (\bibinfo {year} {2018}{\natexlab{a}})},\ \Eprint
  {http://arxiv.org/abs/1502.03433} {arXiv:1502.03433 [cond-mat.str-el]}
  \BibitemShut {NoStop}%
\bibitem [{\citenamefont {Wang}\ \emph {et~al.}(2020)\citenamefont {Wang},
  \citenamefont {Li}, \citenamefont {Hu},\ and\ \citenamefont
  {Wan}}]{wang2020electric}%
  \BibitemOpen
  \bibfield  {author} {\bibinfo {author} {\bibfnamefont {H.}~\bibnamefont
  {Wang}}, \bibinfo {author} {\bibfnamefont {Y.}~\bibnamefont {Li}}, \bibinfo
  {author} {\bibfnamefont {Y.}~\bibnamefont {Hu}}, \ and\ \bibinfo {author}
  {\bibfnamefont {Y.}~\bibnamefont {Wan}},\ }\bibfield  {title} {\enquote
  {\bibinfo {title} {Electric-magnetic duality in the quantum double models of
  topological orders with gapped boundaries},}\ }\href
  {https://link.springer.com/article/10.1007%2FJHEP02%282020%29030} {\bibfield
  {journal} {\bibinfo  {journal} {Journal of High Energy Physics}\ }\textbf
  {\bibinfo {volume} {2020}},\ \bibinfo {pages} {1} (\bibinfo {year} {2020})},\
  \Eprint {http://arxiv.org/abs/1910.13441} {arXiv:1910.13441
  [cond-mat.str-el]} \BibitemShut {NoStop}%
\bibitem [{\citenamefont {Kitaev}\ and\ \citenamefont
  {Kong}(2012)}]{Kitaev2012a}%
  \BibitemOpen
  \bibfield  {author} {\bibinfo {author} {\bibfnamefont {A.}~\bibnamefont
  {Kitaev}}\ and\ \bibinfo {author} {\bibfnamefont {L.}~\bibnamefont {Kong}},\
  }\bibfield  {title} {\enquote {\bibinfo {title} {Models for gapped boundaries
  and domain walls},}\ }\href {\doibase 10.1007/s00220-012-1500-5} {\bibfield
  {journal} {\bibinfo  {journal} {Communications in Mathematical Physics}\
  }\textbf {\bibinfo {volume} {313}},\ \bibinfo {pages} {351} (\bibinfo {year}
  {2012})},\ \Eprint {http://arxiv.org/abs/1104.5047} {arXiv:1104.5047
  [cond-mat.str-el]} \BibitemShut {NoStop}%
\bibitem [{\citenamefont {Bravyi}\ and\ \citenamefont
  {Kitaev}(1998)}]{bravyi1998quantum}%
  \BibitemOpen
  \bibfield  {author} {\bibinfo {author} {\bibfnamefont {S.~B.}\ \bibnamefont
  {Bravyi}}\ and\ \bibinfo {author} {\bibfnamefont {A.~Y.}\ \bibnamefont
  {Kitaev}},\ }\bibfield  {title} {\enquote {\bibinfo {title} {Quantum codes on
  a lattice with boundary},}\ }\href@noop {} {\  (\bibinfo {year} {1998})},\
  \Eprint {http://arxiv.org/abs/quant-ph/9811052} {arXiv:quant-ph/9811052
  [quant-ph]} \BibitemShut {NoStop}%
\bibitem [{\citenamefont {Bombin}\ and\ \citenamefont
  {Martin-Delgado}(2008)}]{Bombin2008family}%
  \BibitemOpen
  \bibfield  {author} {\bibinfo {author} {\bibfnamefont {H.}~\bibnamefont
  {Bombin}}\ and\ \bibinfo {author} {\bibfnamefont {M.~A.}\ \bibnamefont
  {Martin-Delgado}},\ }\bibfield  {title} {\enquote {\bibinfo {title} {Family
  of non-{A}belian {K}itaev models on a lattice: Topological condensation and
  confinement},}\ }\href {\doibase 10.1103/PhysRevB.78.115421} {\bibfield
  {journal} {\bibinfo  {journal} {Phys. Rev. B}\ }\textbf {\bibinfo {volume}
  {78}},\ \bibinfo {pages} {115421} (\bibinfo {year} {2008})},\ \Eprint
  {http://arxiv.org/abs/0712.0190} {arXiv:0712.0190 [cond-mat.str-el]}
  \BibitemShut {NoStop}%
\bibitem [{\citenamefont {Freedman}\ and\ \citenamefont
  {Meyer}(2001)}]{freedman2001projective}%
  \BibitemOpen
  \bibfield  {author} {\bibinfo {author} {\bibfnamefont {M.~H.}\ \bibnamefont
  {Freedman}}\ and\ \bibinfo {author} {\bibfnamefont {D.~A.}\ \bibnamefont
  {Meyer}},\ }\bibfield  {title} {\enquote {\bibinfo {title} {Projective plane
  and planar quantum codes},}\ }\href
  {https://link.springer.com/article/10.1007/s102080010013} {\bibfield
  {journal} {\bibinfo  {journal} {Foundations of Computational Mathematics}\
  }\textbf {\bibinfo {volume} {1}},\ \bibinfo {pages} {325} (\bibinfo {year}
  {2001})},\ \Eprint {http://arxiv.org/abs/quant-ph/9810055}
  {arXiv:quant-ph/9810055 [quant-ph]} \BibitemShut {NoStop}%
\bibitem [{\citenamefont {Beigi}\ \emph {et~al.}(2011)\citenamefont {Beigi},
  \citenamefont {Shor},\ and\ \citenamefont {Whalen}}]{Beigi2011the}%
  \BibitemOpen
  \bibfield  {author} {\bibinfo {author} {\bibfnamefont {S.}~\bibnamefont
  {Beigi}}, \bibinfo {author} {\bibfnamefont {P.~W.}\ \bibnamefont {Shor}}, \
  and\ \bibinfo {author} {\bibfnamefont {D.}~\bibnamefont {Whalen}},\
  }\bibfield  {title} {\enquote {\bibinfo {title} {The quantum double model
  with boundary: Condensations and symmetries},}\ }\href {\doibase
  10.1007/s00220-011-1294-x} {\bibfield  {journal} {\bibinfo  {journal}
  {Communications in Mathematical Physics}\ }\textbf {\bibinfo {volume}
  {306}},\ \bibinfo {pages} {663} (\bibinfo {year} {2011})},\ \Eprint
  {http://arxiv.org/abs/1006.5479} {arXiv:1006.5479 [quant-ph]} \BibitemShut
  {NoStop}%
\bibitem [{\citenamefont {Levin}(2013)}]{Levin2013protected}%
  \BibitemOpen
  \bibfield  {author} {\bibinfo {author} {\bibfnamefont {M.}~\bibnamefont
  {Levin}},\ }\bibfield  {title} {\enquote {\bibinfo {title} {Protected edge
  modes without symmetry},}\ }\href {\doibase 10.1103/PhysRevX.3.021009}
  {\bibfield  {journal} {\bibinfo  {journal} {Phys. Rev. X}\ }\textbf {\bibinfo
  {volume} {3}},\ \bibinfo {pages} {021009} (\bibinfo {year} {2013})},\ \Eprint
  {http://arxiv.org/abs/1301.7355} {arXiv:1301.7355 [cond-mat.str-el]}
  \BibitemShut {NoStop}%
\bibitem [{\citenamefont {Cong}\ \emph {et~al.}(2017)\citenamefont {Cong},
  \citenamefont {Cheng},\ and\ \citenamefont {Wang}}]{Cong2017}%
  \BibitemOpen
  \bibfield  {author} {\bibinfo {author} {\bibfnamefont {I.}~\bibnamefont
  {Cong}}, \bibinfo {author} {\bibfnamefont {M.}~\bibnamefont {Cheng}}, \ and\
  \bibinfo {author} {\bibfnamefont {Z.}~\bibnamefont {Wang}},\ }\bibfield
  {title} {\enquote {\bibinfo {title} {Hamiltonian and algebraic theories of
  gapped boundaries in topological phases of matter},}\ }\href {\doibase
  10.1007/s00220-017-2960-4} {\bibfield  {journal} {\bibinfo  {journal}
  {Communications in Mathematical Physics}\ }\textbf {\bibinfo {volume}
  {355}},\ \bibinfo {pages} {645} (\bibinfo {year} {2017})},\ \Eprint
  {http://arxiv.org/abs/1707.04564} {arXiv:1707.04564 [cond-mat.str-el]}
  \BibitemShut {NoStop}%
\bibitem [{\citenamefont {Etingof}\ and\ \citenamefont
  {Ostrik}(2004)}]{etingof2003finite}%
  \BibitemOpen
  \bibfield  {author} {\bibinfo {author} {\bibfnamefont {P.}~\bibnamefont
  {Etingof}}\ and\ \bibinfo {author} {\bibfnamefont {V.}~\bibnamefont
  {Ostrik}},\ }\bibfield  {title} {\enquote {\bibinfo {title} {Finite tensor
  categories},}\ }\href {http://mi.mathnet.ru/mmj167} {\bibfield  {journal}
  {\bibinfo  {journal} {Moscow Mathematical Journal}\ }\textbf {\bibinfo
  {volume} {4}},\ \bibinfo {pages} {627–654} (\bibinfo {year} {2004})},\
  \Eprint {http://arxiv.org/abs/math/0301027} {arXiv:math/0301027 [math.QA]}
  \BibitemShut {NoStop}%
\bibitem [{\citenamefont {Ostrik}(2003)}]{ostrik2003module}%
  \BibitemOpen
  \bibfield  {author} {\bibinfo {author} {\bibfnamefont {V.}~\bibnamefont
  {Ostrik}},\ }\bibfield  {title} {\enquote {\bibinfo {title} {Module
  categories, weak {H}opf algebras and modular invariants},}\ }\href
  {https://link.springer.com/content/pdf/10.1007%2Fs00031-003-0515-6.pdf}
  {\bibfield  {journal} {\bibinfo  {journal} {Transformation Groups}\ }\textbf
  {\bibinfo {volume} {8}},\ \bibinfo {pages} {177} (\bibinfo {year} {2003})},\
  \Eprint {http://arxiv.org/abs/math/0111139} {arXiv:math/0111139 [math.QA]}
  \BibitemShut {NoStop}%
\bibitem [{\citenamefont {Andruskiewitsch}\ and\ \citenamefont
  {Mombelli}(2007)}]{andruskiewitsch2007module}%
  \BibitemOpen
  \bibfield  {author} {\bibinfo {author} {\bibfnamefont {N.}~\bibnamefont
  {Andruskiewitsch}}\ and\ \bibinfo {author} {\bibfnamefont {J.~M.}\
  \bibnamefont {Mombelli}},\ }\bibfield  {title} {\enquote {\bibinfo {title}
  {On module categories over finite-dimensional {H}opf algebras},}\ }\href
  {https://www.sciencedirect.com/science/article/pii/S0021869307002426}
  {\bibfield  {journal} {\bibinfo  {journal} {Journal of Algebra}\ }\textbf
  {\bibinfo {volume} {314}},\ \bibinfo {pages} {383} (\bibinfo {year}
  {2007})},\ \Eprint {http://arxiv.org/abs/math/0608781} {arXiv:math/0608781
  [math.QA]} \BibitemShut {NoStop}%
\bibitem [{\citenamefont {Natale}(2017)}]{natale2017equivalence}%
  \BibitemOpen
  \bibfield  {author} {\bibinfo {author} {\bibfnamefont {S.}~\bibnamefont
  {Natale}},\ }\bibfield  {title} {\enquote {\bibinfo {title} {On the
  equivalence of module categories over a group-theoretical fusion category},}\
  }\href {https://www.emis.de/journals/SIGMA/2017/042/} {\bibfield  {journal}
  {\bibinfo  {journal} {SIGMA. Symmetry, Integrability and Geometry: Methods
  and Applications}\ }\textbf {\bibinfo {volume} {13}},\ \bibinfo {pages} {042}
  (\bibinfo {year} {2017})},\ \Eprint {http://arxiv.org/abs/1608.04435}
  {arXiv:1608.04435 [math.QA]} \BibitemShut {NoStop}%
\bibitem [{\citenamefont {Kong}\ and\ \citenamefont
  {Wen}(2014)}]{kong2014braided}%
  \BibitemOpen
  \bibfield  {author} {\bibinfo {author} {\bibfnamefont {L.}~\bibnamefont
  {Kong}}\ and\ \bibinfo {author} {\bibfnamefont {X.-G.}\ \bibnamefont {Wen}},\
  }\bibfield  {title} {\enquote {\bibinfo {title} {Braided fusion categories,
  gravitational anomalies, and the mathematical framework for topological
  orders in any dimensions},}\ }\href {https://arxiv.org/abs/1405.5858}
  {\bibfield  {journal} {\bibinfo  {journal} {arXiv preprint arXiv:1405.5858}\
  } (\bibinfo {year} {2014})}\BibitemShut {NoStop}%
\bibitem [{\citenamefont {Kong}\ \emph {et~al.}(2017)\citenamefont {Kong},
  \citenamefont {Wen},\ and\ \citenamefont {Zheng}}]{kong2017boundary}%
  \BibitemOpen
  \bibfield  {author} {\bibinfo {author} {\bibfnamefont {L.}~\bibnamefont
  {Kong}}, \bibinfo {author} {\bibfnamefont {X.-G.}\ \bibnamefont {Wen}}, \
  and\ \bibinfo {author} {\bibfnamefont {H.}~\bibnamefont {Zheng}},\ }\bibfield
   {title} {\enquote {\bibinfo {title} {Boundary-bulk relation in topological
  orders},}\ }\href
  {https://www.sciencedirect.com/science/article/pii/S0550321317302183}
  {\bibfield  {journal} {\bibinfo  {journal} {Nuclear Physics B}\ }\textbf
  {\bibinfo {volume} {922}},\ \bibinfo {pages} {62} (\bibinfo {year} {2017})},\
  \Eprint {http://arxiv.org/abs/1702.00673} {arXiv:1702.00673
  [cond-mat.str-el]} \BibitemShut {NoStop}%
\bibitem [{\citenamefont {Haldane}(1995)}]{Haldane1995stability}%
  \BibitemOpen
  \bibfield  {author} {\bibinfo {author} {\bibfnamefont {F.~D.~M.}\
  \bibnamefont {Haldane}},\ }\bibfield  {title} {\enquote {\bibinfo {title}
  {Stability of chiral luttinger liquids and abelian quantum hall states},}\
  }\href {\doibase 10.1103/PhysRevLett.74.2090} {\bibfield  {journal} {\bibinfo
   {journal} {Phys. Rev. Lett.}\ }\textbf {\bibinfo {volume} {74}},\ \bibinfo
  {pages} {2090} (\bibinfo {year} {1995})}\BibitemShut {NoStop}%
\bibitem [{\citenamefont {Kane}\ and\ \citenamefont
  {Fisher}(1997)}]{Kane1997quantized}%
  \BibitemOpen
  \bibfield  {author} {\bibinfo {author} {\bibfnamefont {C.~L.}\ \bibnamefont
  {Kane}}\ and\ \bibinfo {author} {\bibfnamefont {M.~P.~A.}\ \bibnamefont
  {Fisher}},\ }\bibfield  {title} {\enquote {\bibinfo {title} {Quantized
  thermal transport in the fractional quantum {H}all effect},}\ }\href
  {\doibase 10.1103/PhysRevB.55.15832} {\bibfield  {journal} {\bibinfo
  {journal} {Phys. Rev. B}\ }\textbf {\bibinfo {volume} {55}},\ \bibinfo
  {pages} {15832} (\bibinfo {year} {1997})},\ \Eprint
  {http://arxiv.org/abs/cond-mat/9603118} {arXiv:cond-mat/9603118 [cond-mat]}
  \BibitemShut {NoStop}%
\bibitem [{\citenamefont {Kapustin}\ and\ \citenamefont
  {Saulina}(2011)}]{kapustin2011topological}%
  \BibitemOpen
  \bibfield  {author} {\bibinfo {author} {\bibfnamefont {A.}~\bibnamefont
  {Kapustin}}\ and\ \bibinfo {author} {\bibfnamefont {N.}~\bibnamefont
  {Saulina}},\ }\bibfield  {title} {\enquote {\bibinfo {title} {Topological
  boundary conditions in abelian chern--simons theory},}\ }\href
  {https://linkinghub.elsevier.com/retrieve/pii/S0550321310006723} {\bibfield
  {journal} {\bibinfo  {journal} {Nuclear Physics B}\ }\textbf {\bibinfo
  {volume} {845}},\ \bibinfo {pages} {393} (\bibinfo {year} {2011})},\ \Eprint
  {http://arxiv.org/abs/1008.0654} {arXiv:1008.0654 [hep-th]} \BibitemShut
  {NoStop}%
\bibitem [{\citenamefont {Barkeshli}\ \emph {et~al.}(2013)\citenamefont
  {Barkeshli}, \citenamefont {Jian},\ and\ \citenamefont
  {Qi}}]{Barkeshli2013theory}%
  \BibitemOpen
  \bibfield  {author} {\bibinfo {author} {\bibfnamefont {M.}~\bibnamefont
  {Barkeshli}}, \bibinfo {author} {\bibfnamefont {C.-M.}\ \bibnamefont {Jian}},
  \ and\ \bibinfo {author} {\bibfnamefont {X.-L.}\ \bibnamefont {Qi}},\
  }\bibfield  {title} {\enquote {\bibinfo {title} {Theory of defects in abelian
  topological states},}\ }\href {\doibase 10.1103/PhysRevB.88.235103}
  {\bibfield  {journal} {\bibinfo  {journal} {Phys. Rev. B}\ }\textbf {\bibinfo
  {volume} {88}},\ \bibinfo {pages} {235103} (\bibinfo {year} {2013})},\
  \Eprint {http://arxiv.org/abs/1305.7203} {arXiv:1305.7203 [cond-mat.str-e]}
  \BibitemShut {NoStop}%
\bibitem [{\citenamefont {Kong}(2014)}]{Kong2014}%
  \BibitemOpen
  \bibfield  {author} {\bibinfo {author} {\bibfnamefont {L.}~\bibnamefont
  {Kong}},\ }\bibfield  {title} {\enquote {\bibinfo {title} {Anyon condensation
  and tensor categories},}\ }\href {\doibase
  https://doi.org/10.1016/j.nuclphysb.2014.07.003} {\bibfield  {journal}
  {\bibinfo  {journal} {Nuclear Physics B}\ }\textbf {\bibinfo {volume}
  {886}},\ \bibinfo {pages} {436 } (\bibinfo {year} {2014})},\ \Eprint
  {http://arxiv.org/abs/1307.8244} {arXiv:1307.8244 [cond-mat.str-el]}
  \BibitemShut {NoStop}%
\bibitem [{\citenamefont {Wang}\ and\ \citenamefont
  {Wen}(2015)}]{Wang2015boundary}%
  \BibitemOpen
  \bibfield  {author} {\bibinfo {author} {\bibfnamefont {J.~C.}\ \bibnamefont
  {Wang}}\ and\ \bibinfo {author} {\bibfnamefont {X.-G.}\ \bibnamefont {Wen}},\
  }\bibfield  {title} {\enquote {\bibinfo {title} {Boundary degeneracy of
  topological order},}\ }\href {\doibase 10.1103/PhysRevB.91.125124} {\bibfield
   {journal} {\bibinfo  {journal} {Phys. Rev. B}\ }\textbf {\bibinfo {volume}
  {91}},\ \bibinfo {pages} {125124} (\bibinfo {year} {2015})},\ \Eprint
  {http://arxiv.org/abs/1212.4863} {arXiv:1212.4863 [cond-mat.str-e]}
  \BibitemShut {NoStop}%
\bibitem [{\citenamefont {Lan}\ \emph {et~al.}(2015)\citenamefont {Lan},
  \citenamefont {Wang},\ and\ \citenamefont {Wen}}]{Lan2015gapped}%
  \BibitemOpen
  \bibfield  {author} {\bibinfo {author} {\bibfnamefont {T.}~\bibnamefont
  {Lan}}, \bibinfo {author} {\bibfnamefont {J.~C.}\ \bibnamefont {Wang}}, \
  and\ \bibinfo {author} {\bibfnamefont {X.-G.}\ \bibnamefont {Wen}},\
  }\bibfield  {title} {\enquote {\bibinfo {title} {Gapped domain walls, gapped
  boundaries, and topological degeneracy},}\ }\href {\doibase
  10.1103/PhysRevLett.114.076402} {\bibfield  {journal} {\bibinfo  {journal}
  {Phys. Rev. Lett.}\ }\textbf {\bibinfo {volume} {114}},\ \bibinfo {pages}
  {076402} (\bibinfo {year} {2015})},\ \Eprint {http://arxiv.org/abs/1408.6514}
  {arXiv:1408.6514 [cond-mat.str-e]} \BibitemShut {NoStop}%
\bibitem [{\citenamefont {Seiberg}\ and\ \citenamefont
  {Witten}(2016)}]{seiberg2016gapped}%
  \BibitemOpen
  \bibfield  {author} {\bibinfo {author} {\bibfnamefont {N.}~\bibnamefont
  {Seiberg}}\ and\ \bibinfo {author} {\bibfnamefont {E.}~\bibnamefont
  {Witten}},\ }\bibfield  {title} {\enquote {\bibinfo {title} {Gapped boundary
  phases of topological insulators via weak coupling},}\ }\href
  {https://academic.oup.com/ptep/article/2016/12/12C101/2624087} {\bibfield
  {journal} {\bibinfo  {journal} {Progress of Theoretical and Experimental
  Physics}\ }\textbf {\bibinfo {volume} {2016}},\ \bibinfo {pages} {12C101}
  (\bibinfo {year} {2016})},\ \Eprint {http://arxiv.org/abs/1602.04251}
  {arXiv:1602.04251 [cond-mat.str-e]} \BibitemShut {NoStop}%
\bibitem [{\citenamefont {Hu}\ \emph {et~al.}(2018{\natexlab{b}})\citenamefont
  {Hu}, \citenamefont {Luo}, \citenamefont {Pankovich}, \citenamefont {Wan},\
  and\ \citenamefont {Wu}}]{hu2018boundary}%
  \BibitemOpen
  \bibfield  {author} {\bibinfo {author} {\bibfnamefont {Y.}~\bibnamefont
  {Hu}}, \bibinfo {author} {\bibfnamefont {Z.-X.}\ \bibnamefont {Luo}},
  \bibinfo {author} {\bibfnamefont {R.}~\bibnamefont {Pankovich}}, \bibinfo
  {author} {\bibfnamefont {Y.}~\bibnamefont {Wan}}, \ and\ \bibinfo {author}
  {\bibfnamefont {Y.-S.}\ \bibnamefont {Wu}},\ }\bibfield  {title} {\enquote
  {\bibinfo {title} {Boundary hamiltonian theory for gapped topological phases
  on an open surface},}\ }\href
  {https://link.springer.com/article/10.1007/JHEP01(2018)134} {\bibfield
  {journal} {\bibinfo  {journal} {Journal of High Energy Physics}\ }\textbf
  {\bibinfo {volume} {2018}},\ \bibinfo {pages} {1} (\bibinfo {year}
  {2018}{\natexlab{b}})},\ \Eprint {http://arxiv.org/abs/1706.03329}
  {arXiv:1706.03329 [cond-mat.str-el]} \BibitemShut {NoStop}%
\bibitem [{\citenamefont {Lan}\ \emph {et~al.}(2020)\citenamefont {Lan},
  \citenamefont {Wen}, \citenamefont {Kong},\ and\ \citenamefont
  {Wen}}]{Lan2020Gapped}%
  \BibitemOpen
  \bibfield  {author} {\bibinfo {author} {\bibfnamefont {T.}~\bibnamefont
  {Lan}}, \bibinfo {author} {\bibfnamefont {X.}~\bibnamefont {Wen}}, \bibinfo
  {author} {\bibfnamefont {L.}~\bibnamefont {Kong}}, \ and\ \bibinfo {author}
  {\bibfnamefont {X.-G.}\ \bibnamefont {Wen}},\ }\bibfield  {title} {\enquote
  {\bibinfo {title} {Gapped domain walls between 2+1d topologically ordered
  states},}\ }\href {\doibase 10.1103/PhysRevResearch.2.023331} {\bibfield
  {journal} {\bibinfo  {journal} {Phys. Rev. Res.}\ }\textbf {\bibinfo {volume}
  {2}},\ \bibinfo {pages} {023331} (\bibinfo {year} {2020})},\ \Eprint
  {http://arxiv.org/abs/1911.08470} {arXiv:1911.08470 [cond-mat.str-e]}
  \BibitemShut {NoStop}%
\bibitem [{\citenamefont {Freed}\ and\ \citenamefont
  {Teleman}(2021)}]{freed2021gapped}%
  \BibitemOpen
  \bibfield  {author} {\bibinfo {author} {\bibfnamefont {D.~S.}\ \bibnamefont
  {Freed}}\ and\ \bibinfo {author} {\bibfnamefont {C.}~\bibnamefont
  {Teleman}},\ }\bibfield  {title} {\enquote {\bibinfo {title} {Gapped boundary
  theories in three dimensions},}\ }\href
  {https://link.springer.com/article/10.1007/s00220-021-04192-x} {\bibfield
  {journal} {\bibinfo  {journal} {Communications in Mathematical Physics}\
  }\textbf {\bibinfo {volume} {388}},\ \bibinfo {pages} {845} (\bibinfo {year}
  {2021})},\ \Eprint {http://arxiv.org/abs/2006.10200} {arXiv:2006.10200
  [math.QA]} \BibitemShut {NoStop}%
\bibitem [{\citenamefont {Ganeshan}\ and\ \citenamefont
  {Levin}(2021)}]{ganeshan2021ungappable}%
  \BibitemOpen
  \bibfield  {author} {\bibinfo {author} {\bibfnamefont {S.}~\bibnamefont
  {Ganeshan}}\ and\ \bibinfo {author} {\bibfnamefont {M.}~\bibnamefont
  {Levin}},\ }\bibfield  {title} {\enquote {\bibinfo {title} {Ungappable edge
  theories with finite dimensional {H}ilbert spaces},}\ }\href@noop {} {\
  (\bibinfo {year} {2021})},\ \Eprint {http://arxiv.org/abs/2109.11539}
  {arXiv:2109.11539 [cond-mat.str-el]} \BibitemShut {NoStop}%
\bibitem [{\citenamefont {Cowtan}\ and\ \citenamefont
  {Majid}(2022)}]{cowtan2021quantum}%
  \BibitemOpen
  \bibfield  {author} {\bibinfo {author} {\bibfnamefont {A.}~\bibnamefont
  {Cowtan}}\ and\ \bibinfo {author} {\bibfnamefont {S.}~\bibnamefont {Majid}},\
  }\bibfield  {title} {\enquote {\bibinfo {title} {Quantum double aspects of
  surface code models},}\ }\href {\doibase 10.1063/5.0063768} {\bibfield
  {journal} {\bibinfo  {journal} {Journal of Mathematical Physics}\ }\textbf
  {\bibinfo {volume} {63}},\ \bibinfo {pages} {Paper No. 042202, 49} (\bibinfo
  {year} {2022})},\ \Eprint {http://arxiv.org/abs/2107.04411} {arXiv:2107.04411
  [quant-ph]} \BibitemShut {NoStop}%
\bibitem [{\citenamefont {Balsam}\ and\ \citenamefont
  {Kirillov}(2012)}]{balsam2012kitaevs}%
  \BibitemOpen
  \bibfield  {author} {\bibinfo {author} {\bibfnamefont {B.}~\bibnamefont
  {Balsam}}\ and\ \bibinfo {author} {\bibfnamefont {A.}~\bibnamefont
  {Kirillov}},\ }\href@noop {} {\enquote {\bibinfo {title} {Kitaev's lattice
  model and {T}uraev-{V}iro {TQFT}s},}\ } (\bibinfo {year} {2012}),\ \Eprint
  {http://arxiv.org/abs/1206.2308} {arXiv:1206.2308 [math.QA]} \BibitemShut
  {NoStop}%
\bibitem [{\citenamefont {Gelaki}\ and\ \citenamefont
  {Nikshych}(2008)}]{gelaki2008nilpotent}%
  \BibitemOpen
  \bibfield  {author} {\bibinfo {author} {\bibfnamefont {S.}~\bibnamefont
  {Gelaki}}\ and\ \bibinfo {author} {\bibfnamefont {D.}~\bibnamefont
  {Nikshych}},\ }\bibfield  {title} {\enquote {\bibinfo {title} {Nilpotent
  fusion categories},}\ }\href
  {https://www.sciencedirect.com/science/article/pii/S0001870807002290}
  {\bibfield  {journal} {\bibinfo  {journal} {Advances in Mathematics}\
  }\textbf {\bibinfo {volume} {217}},\ \bibinfo {pages} {1053} (\bibinfo {year}
  {2008})},\ \Eprint {http://arxiv.org/abs/math/0610726} {arXiv:math/0610726
  [math.QA]} \BibitemShut {NoStop}%
\bibitem [{\citenamefont {Burciu}(2012)}]{burciu2012irreducible}%
  \BibitemOpen
  \bibfield  {author} {\bibinfo {author} {\bibfnamefont {S.}~\bibnamefont
  {Burciu}},\ }\bibfield  {title} {\enquote {\bibinfo {title} {On the
  irreducible representations of generalized quantum doubles},}\ }\href@noop {}
  {\  (\bibinfo {year} {2012})},\ \Eprint {http://arxiv.org/abs/1202.4315}
  {arXiv:1202.4315 [math.QA]} \BibitemShut {NoStop}%
\bibitem [{\citenamefont {Burciu}(2017)}]{burciu2017grothendieck}%
  \BibitemOpen
  \bibfield  {author} {\bibinfo {author} {\bibfnamefont {S.}~\bibnamefont
  {Burciu}},\ }\bibfield  {title} {\enquote {\bibinfo {title} {On the
  grothendieck rings of generalized drinfeld doubles},}\ }\href
  {https://www.sciencedirect.com/science/article/pii/S0021869317303046}
  {\bibfield  {journal} {\bibinfo  {journal} {Journal of Algebra}\ }\textbf
  {\bibinfo {volume} {486}},\ \bibinfo {pages} {14} (\bibinfo {year}
  {2017})}\BibitemShut {NoStop}%
\bibitem [{\citenamefont {Montgomery}(1993)}]{montgomery1993hopf}%
  \BibitemOpen
  \bibfield  {author} {\bibinfo {author} {\bibfnamefont {S.}~\bibnamefont
  {Montgomery}},\ }\href {https://bookstore.ams.org/cbms-82} {\emph {\bibinfo
  {title} {Hopf algebras and their actions on rings}}},\ \bibinfo {number}
  {82}\ (\bibinfo  {publisher} {American Mathematical Soc.},\ \bibinfo {year}
  {1993})\BibitemShut {NoStop}%
\bibitem [{\citenamefont {Majid}(2000)}]{majid2000foundations}%
  \BibitemOpen
  \bibfield  {author} {\bibinfo {author} {\bibfnamefont {S.}~\bibnamefont
  {Majid}},\ }\href
  {https://www.cambridge.org/core/books/foundations-of-quantum-group-theory/BDBBAB645399E72AA1A01BDECAFC7E8C#}
  {\emph {\bibinfo {title} {Foundations of quantum group theory}}}\ (\bibinfo
  {publisher} {Cambridge university press, Cambridge},\ \bibinfo {year}
  {2000})\ pp.\ \bibinfo {pages} {x+607}\BibitemShut {NoStop}%
\bibitem [{\citenamefont {Fuchs}\ and\ \citenamefont
  {Stigner}(2009)}]{fuchs2009frobenius}%
  \BibitemOpen
  \bibfield  {author} {\bibinfo {author} {\bibfnamefont {J.}~\bibnamefont
  {Fuchs}}\ and\ \bibinfo {author} {\bibfnamefont {C.}~\bibnamefont
  {Stigner}},\ }\bibfield  {title} {\enquote {\bibinfo {title} {On {F}robenius
  algebras in rigid monoidal categories},}\ }\href@noop {} {\bibfield
  {journal} {\bibinfo  {journal} {Arabian J. Sci. Eng.}\ }\textbf {\bibinfo
  {volume} {33-2C}},\ \bibinfo {pages} {175} (\bibinfo {year} {2009})},\
  \Eprint {http://arxiv.org/abs/0901.4886} {arXiv:0901.4886 [math.CT]}
  \BibitemShut {NoStop}%
\bibitem [{\citenamefont {Jia}\ \emph {et~al.}(2023)\citenamefont {Jia},
  \citenamefont {Tan}, \citenamefont {Kaszlikowski},\ and\ \citenamefont
  {Chang}}]{jia2023weak}%
  \BibitemOpen
  \bibfield  {author} {\bibinfo {author} {\bibfnamefont {Z.}~\bibnamefont
  {Jia}}, \bibinfo {author} {\bibfnamefont {S.}~\bibnamefont {Tan}}, \bibinfo
  {author} {\bibfnamefont {D.}~\bibnamefont {Kaszlikowski}}, \ and\ \bibinfo
  {author} {\bibfnamefont {L.}~\bibnamefont {Chang}},\ }\bibfield  {title}
  {\enquote {\bibinfo {title} {On weak {H}opf symmetry and weak {H}opf quantum
  double model},}\ }\href {https://arxiv.org/abs/2302.08131} {\bibfield
  {journal} {\bibinfo  {journal} {arXiv preprint arXiv:2302.08131}\ } (\bibinfo
  {year} {2023})}\BibitemShut {NoStop}%
\bibitem [{\citenamefont {Aguiar}(2000)}]{aguiar2000note}%
  \BibitemOpen
  \bibfield  {author} {\bibinfo {author} {\bibfnamefont {M.}~\bibnamefont
  {Aguiar}},\ }\bibfield  {title} {\enquote {\bibinfo {title} {A note on
  strongly separable algebras},}\ }\href@noop {} {\bibfield  {journal}
  {\bibinfo  {journal} {Bol. Acad. Nac. Cienc.(C{\'o}rdoba)}\ }\textbf
  {\bibinfo {volume} {65}},\ \bibinfo {pages} {51} (\bibinfo {year}
  {2000})}\BibitemShut {NoStop}%
\bibitem [{\citenamefont {Koppen}(2020)}]{koppen2020defects}%
  \BibitemOpen
  \bibfield  {author} {\bibinfo {author} {\bibfnamefont {V.}~\bibnamefont
  {Koppen}},\ }\bibfield  {title} {\enquote {\bibinfo {title} {Defects in
  {K}itaev models and bicomodule algebras},}\ }\href
  {https://arxiv.org/abs/2001.10578} {\bibfield  {journal} {\bibinfo  {journal}
  {arXiv preprint arXiv:2001.10578}\ } (\bibinfo {year} {2020})}\BibitemShut
  {NoStop}%
\bibitem [{\citenamefont {Jia}\ and\ \citenamefont
  {Kaszlikowski}(2022)}]{Jia2022electric}%
  \BibitemOpen
  \bibfield  {author} {\bibinfo {author} {\bibfnamefont {Z.}~\bibnamefont
  {Jia}}\ and\ \bibinfo {author} {\bibfnamefont {D.}~\bibnamefont
  {Kaszlikowski}},\ }\href {\doibase 10.48550/ARXIV.2201.12361} {\enquote
  {\bibinfo {title} {Electric-magnetic duality and $\mathbb{Z}_2$ symmetry
  enriched {A}belian lattice gauge theory},}\ } (\bibinfo {year} {2022}),\
  \Eprint {http://arxiv.org/abs/2201.12361} {arXiv:2201.12361 [quant-ph]}
  \BibitemShut {NoStop}%
\bibitem [{\citenamefont {Etingof}\ \emph {et~al.}(2016)\citenamefont
  {Etingof}, \citenamefont {Gelaki}, \citenamefont {Nikshych},\ and\
  \citenamefont {Ostrik}}]{etingof2016tensor}%
  \BibitemOpen
  \bibfield  {author} {\bibinfo {author} {\bibfnamefont {P.}~\bibnamefont
  {Etingof}}, \bibinfo {author} {\bibfnamefont {S.}~\bibnamefont {Gelaki}},
  \bibinfo {author} {\bibfnamefont {D.}~\bibnamefont {Nikshych}}, \ and\
  \bibinfo {author} {\bibfnamefont {V.}~\bibnamefont {Ostrik}},\ }\href
  {https://bookstore.ams.org/surv-205} {\emph {\bibinfo {title} {Tensor
  categories}}},\ Vol.\ \bibinfo {volume} {205}\ (\bibinfo  {publisher}
  {American Mathematical Soc.},\ \bibinfo {year} {2016})\BibitemShut {NoStop}%
\bibitem [{\citenamefont {Eilenberg}(1960)}]{eilenberg1960abstract}%
  \BibitemOpen
  \bibfield  {author} {\bibinfo {author} {\bibfnamefont {S.}~\bibnamefont
  {Eilenberg}},\ }\bibfield  {title} {\enquote {\bibinfo {title} {Abstract
  description of some basic functors},}\ }\href@noop {} {\bibfield  {journal}
  {\bibinfo  {journal} {J. Indian Math. Soc}\ }\textbf {\bibinfo {volume}
  {24}},\ \bibinfo {pages} {231} (\bibinfo {year} {1960})}\BibitemShut
  {NoStop}%
\bibitem [{\citenamefont {Watts}(1960)}]{watts1960intrinsic}%
  \BibitemOpen
  \bibfield  {author} {\bibinfo {author} {\bibfnamefont {C.~E.}\ \bibnamefont
  {Watts}},\ }\bibfield  {title} {\enquote {\bibinfo {title} {Intrinsic
  characterizations of some additive functors},}\ }\href@noop {} {\bibfield
  {journal} {\bibinfo  {journal} {Proceedings of the American Mathematical
  Society}\ }\textbf {\bibinfo {volume} {11}},\ \bibinfo {pages} {5} (\bibinfo
  {year} {1960})}\BibitemShut {NoStop}%
\bibitem [{\citenamefont {Yan}\ and\ \citenamefont
  {Zhu}(1998)}]{yan1998stabilizer}%
  \BibitemOpen
  \bibfield  {author} {\bibinfo {author} {\bibfnamefont {M.}~\bibnamefont
  {Yan}}\ and\ \bibinfo {author} {\bibfnamefont {Y.}~\bibnamefont {Zhu}},\
  }\bibfield  {title} {\enquote {\bibinfo {title} {Stabilizer for {H}opf
  algebra actions},}\ }\href
  {https://www.tandfonline.com/doi/abs/10.1080/00927879808826382?journalCode=lagb20}
  {\bibfield  {journal} {\bibinfo  {journal} {Communications in Algebra}\
  }\textbf {\bibinfo {volume} {26}},\ \bibinfo {pages} {3885} (\bibinfo {year}
  {1998})}\BibitemShut {NoStop}%
\bibitem [{\citenamefont {Jia}\ and\ \citenamefont {Tan}()}]{jia2023extended}%
  \BibitemOpen
  \bibfield  {author} {\bibinfo {author} {\bibfnamefont {Z.}~\bibnamefont
  {Jia}}\ and\ \bibinfo {author} {\bibfnamefont {S.}~\bibnamefont {Tan}},\
  }\bibfield  {title} {\enquote {\bibinfo {title} {Gapped boundary and domain
  wall theories for extended string-net model: dyonic excitations and
  electromagnetic duality},}\ }\href@noop {} {\bibinfo  {journal} {In
  preparation}\ }\BibitemShut {NoStop}%
\bibitem [{\citenamefont {Caenepeel}\ \emph {et~al.}(2007)\citenamefont
  {Caenepeel}, \citenamefont {Crivei}, \citenamefont {Marcus},\ and\
  \citenamefont {Takeuchi}}]{caenepeel2007morita}%
  \BibitemOpen
\bibfield  {journal} {  }\bibfield  {author} {\bibinfo {author} {\bibfnamefont
  {S.}~\bibnamefont {Caenepeel}}, \bibinfo {author} {\bibfnamefont
  {S.}~\bibnamefont {Crivei}}, \bibinfo {author} {\bibfnamefont
  {A.}~\bibnamefont {Marcus}}, \ and\ \bibinfo {author} {\bibfnamefont
  {M.}~\bibnamefont {Takeuchi}},\ }\bibfield  {title} {\enquote {\bibinfo
  {title} {Morita equivalences induced by bimodules over {H}opf--{G}alois
  extensions},}\ }\href
  {https://www.sciencedirect.com/science/article/pii/S002186930700141X}
  {\bibfield  {journal} {\bibinfo  {journal} {Journal of Algebra}\ }\textbf
  {\bibinfo {volume} {314}},\ \bibinfo {pages} {267} (\bibinfo {year}
  {2007})},\ \Eprint {http://arxiv.org/abs/math/0608572} {arXiv:math/0608572
  [math.RA]} \BibitemShut {NoStop}%
\bibitem [{\citenamefont
  {Schauenburg}(2002{\natexlab{a}})}]{schauenburg02extension}%
  \BibitemOpen
  \bibfield  {author} {\bibinfo {author} {\bibfnamefont {P.}~\bibnamefont
  {Schauenburg}},\ }\bibfield  {title} {\enquote {\bibinfo {title} {Hopf
  algebra extensions and monoidal categories},}\ }\href@noop {} {\bibfield
  {journal} {\bibinfo  {journal} {New directions in Hopf algebras}\ }\textbf
  {\bibinfo {volume} {43}},\ \bibinfo {pages} {321} (\bibinfo {year}
  {2002}{\natexlab{a}})}\BibitemShut {NoStop}%
\bibitem [{\citenamefont
  {Schauenburg}(2002{\natexlab{b}})}]{schauenburg02coquasibialg}%
  \BibitemOpen
  \bibfield  {author} {\bibinfo {author} {\bibfnamefont {P.}~\bibnamefont
  {Schauenburg}},\ }\bibfield  {title} {\enquote {\bibinfo {title} {Hopf
  bimodules, coquasibialgebras, and an exact sequence of {K}ac},}\ }\href
  {https://www.sciencedirect.com/science/article/pii/S000187080192016X}
  {\bibfield  {journal} {\bibinfo  {journal} {Advances in mathematics}\
  }\textbf {\bibinfo {volume} {165}},\ \bibinfo {pages} {194} (\bibinfo {year}
  {2002}{\natexlab{b}})}\BibitemShut {NoStop}%
\bibitem [{\citenamefont {Etingof}\ \emph {et~al.}(2021)\citenamefont
  {Etingof}, \citenamefont {Kinser},\ and\ \citenamefont
  {Walton}}]{etingof2021tensor}%
  \BibitemOpen
  \bibfield  {author} {\bibinfo {author} {\bibfnamefont {P.}~\bibnamefont
  {Etingof}}, \bibinfo {author} {\bibfnamefont {R.}~\bibnamefont {Kinser}}, \
  and\ \bibinfo {author} {\bibfnamefont {C.}~\bibnamefont {Walton}},\
  }\bibfield  {title} {\enquote {\bibinfo {title} {Tensor algebras in finite
  tensor categories},}\ }\href
  {https://academic.oup.com/imrn/article/2021/24/18529/5674136} {\bibfield
  {journal} {\bibinfo  {journal} {International Mathematics Research Notices}\
  }\textbf {\bibinfo {volume} {2021}},\ \bibinfo {pages} {18529} (\bibinfo
  {year} {2021})},\ \Eprint {http://arxiv.org/abs/1906.02828} {arXiv:1906.02828
  [math.QA]} \BibitemShut {NoStop}%
\bibitem [{\citenamefont {Kong}(2017)}]{Kong2019private}%
  \BibitemOpen
  \bibfield  {author} {\bibinfo {author} {\bibfnamefont {L.}~\bibnamefont
  {Kong}},\ }\bibfield  {title} {\enquote {\bibinfo {title} {A short course on
  tensor categories and topological orders},}\ }\href
  {https://kongliang.wordpress.com/2017/08/11/a-short-course-on-tensor-categories-and-topological-orders}
  {\bibfield  {journal} {\bibinfo  {journal}
  {https://kongliang.wordpress.com/2017/08/11/a-short-course-on-tensor-categories-and-topological-orders}\
  } (\bibinfo {year} {2017})}\BibitemShut {NoStop}%
\bibitem [{\citenamefont {Turaev}\ and\ \citenamefont
  {Virelizier}(2017)}]{turaev2017monoidal}%
  \BibitemOpen
  \bibfield  {author} {\bibinfo {author} {\bibfnamefont {V.}~\bibnamefont
  {Turaev}}\ and\ \bibinfo {author} {\bibfnamefont {A.}~\bibnamefont
  {Virelizier}},\ }\href
  {https://link.springer.com/book/10.1007/978-3-319-49834-8} {\emph {\bibinfo
  {title} {Monoidal categories and topological field theory}}},\ \bibinfo
  {series} {Progress in Mathematics}, Vol.\ \bibinfo {volume} {322}\ (\bibinfo
  {year} {2017})\ pp.\ \bibinfo {pages} {xii+523}\BibitemShut {NoStop}%
\bibitem [{\citenamefont {Jia}\ \emph {et~al.}(2020)\citenamefont {Jia},
  \citenamefont {Wei}, \citenamefont {Wu}, \citenamefont {Guo},\ and\
  \citenamefont {Guo}}]{jia2020entanglement}%
  \BibitemOpen
  \bibfield  {author} {\bibinfo {author} {\bibfnamefont {Z.-A.}\ \bibnamefont
  {Jia}}, \bibinfo {author} {\bibfnamefont {L.}~\bibnamefont {Wei}}, \bibinfo
  {author} {\bibfnamefont {Y.-C.}\ \bibnamefont {Wu}}, \bibinfo {author}
  {\bibfnamefont {G.-C.}\ \bibnamefont {Guo}}, \ and\ \bibinfo {author}
  {\bibfnamefont {G.-P.}\ \bibnamefont {Guo}},\ }\bibfield  {title} {\enquote
  {\bibinfo {title} {Entanglement area law for shallow and deep quantum neural
  network states},}\ }\href
  {https://iopscience.iop.org/article/10.1088/1367-2630/ab8262/meta} {\bibfield
   {journal} {\bibinfo  {journal} {New Journal of Physics}\ }\textbf {\bibinfo
  {volume} {22}},\ \bibinfo {pages} {053022} (\bibinfo {year}
  {2020})}\BibitemShut {NoStop}%
\bibitem [{\citenamefont {Chen}\ \emph {et~al.}(2018)\citenamefont {Chen},
  \citenamefont {Hung}, \citenamefont {Li},\ and\ \citenamefont
  {Wan}}]{chen2018entanglement}%
  \BibitemOpen
  \bibfield  {author} {\bibinfo {author} {\bibfnamefont {C.}~\bibnamefont
  {Chen}}, \bibinfo {author} {\bibfnamefont {L.-Y.}\ \bibnamefont {Hung}},
  \bibinfo {author} {\bibfnamefont {Y.}~\bibnamefont {Li}}, \ and\ \bibinfo
  {author} {\bibfnamefont {Y.}~\bibnamefont {Wan}},\ }\bibfield  {title}
  {\enquote {\bibinfo {title} {Entanglement entropy of topological orders with
  boundaries},}\ }\href
  {https://link.springer.com/article/10.1007/JHEP06(2018)113} {\bibfield
  {journal} {\bibinfo  {journal} {Journal of High Energy Physics}\ }\textbf
  {\bibinfo {volume} {2018}},\ \bibinfo {pages} {1} (\bibinfo {year} {2018})},\
  \Eprint {http://arxiv.org/abs/1804.05725} {arXiv:1804.05725 [hep-th]}
  \BibitemShut {NoStop}%
\bibitem [{\citenamefont {Lou}\ \emph {et~al.}(2019)\citenamefont {Lou},
  \citenamefont {Shen},\ and\ \citenamefont {Hung}}]{lou2019ishibashiI}%
  \BibitemOpen
  \bibfield  {author} {\bibinfo {author} {\bibfnamefont {J.}~\bibnamefont
  {Lou}}, \bibinfo {author} {\bibfnamefont {C.}~\bibnamefont {Shen}}, \ and\
  \bibinfo {author} {\bibfnamefont {L.-Y.}\ \bibnamefont {Hung}},\ }\bibfield
  {title} {\enquote {\bibinfo {title} {{I}shibashi states, topological orders
  with boundaries and topological entanglement entropy. {P}art {I}},}\
  }\href@noop {} {\bibfield  {journal} {\bibinfo  {journal} {Journal of High
  Energy Physics}\ }\textbf {\bibinfo {volume} {2019}},\ \bibinfo {pages} {1}
  (\bibinfo {year} {2019})},\ \Eprint {http://arxiv.org/abs/1901.08238}
  {arXiv:1901.08238 [hep-th]} \BibitemShut {NoStop}%
\bibitem [{\citenamefont {Shen}\ \emph {et~al.}(2019)\citenamefont {Shen},
  \citenamefont {Lou},\ and\ \citenamefont {Hung}}]{shen2019ishibashiII}%
  \BibitemOpen
  \bibfield  {author} {\bibinfo {author} {\bibfnamefont {C.}~\bibnamefont
  {Shen}}, \bibinfo {author} {\bibfnamefont {J.}~\bibnamefont {Lou}}, \ and\
  \bibinfo {author} {\bibfnamefont {L.-Y.}\ \bibnamefont {Hung}},\ }\bibfield
  {title} {\enquote {\bibinfo {title} {Ishibashi states, topological orders
  with boundaries and topological entanglement entropy. {P}art {II}. {C}utting
  through the boundary},}\ }\href
  {https://link.springer.com/article/10.1007/JHEP11(2019)168} {\bibfield
  {journal} {\bibinfo  {journal} {Journal of High Energy Physics}\ }\textbf
  {\bibinfo {volume} {2019}},\ \bibinfo {pages} {1} (\bibinfo {year} {2019})},\
  \Eprint {http://arxiv.org/abs/1908.07700} {arXiv:1908.07700 [hep-th]}
  \BibitemShut {NoStop}%
\bibitem [{\citenamefont {Hu}\ and\ \citenamefont
  {Wan}(2019)}]{hu2019entanglement}%
  \BibitemOpen
  \bibfield  {author} {\bibinfo {author} {\bibfnamefont {Y.}~\bibnamefont
  {Hu}}\ and\ \bibinfo {author} {\bibfnamefont {Y.}~\bibnamefont {Wan}},\
  }\bibfield  {title} {\enquote {\bibinfo {title} {Entanglement entropy,
  quantum fluctuations, and thermal entropy in topological phases},}\ }\href
  {https://link.springer.com/article/10.1007/JHEP05(2019)110} {\bibfield
  {journal} {\bibinfo  {journal} {Journal of High Energy Physics}\ }\textbf
  {\bibinfo {volume} {2019}},\ \bibinfo {pages} {1} (\bibinfo {year} {2019})},\
  \Eprint {http://arxiv.org/abs/1901.09033} {arXiv:1901.09033
  [cond-mat.str-el]} \BibitemShut {NoStop}%
\bibitem [{\citenamefont {Brown}\ \emph {et~al.}(2013)\citenamefont {Brown},
  \citenamefont {Bartlett}, \citenamefont {Doherty},\ and\ \citenamefont
  {Barrett}}]{Brown2013topological}%
  \BibitemOpen
  \bibfield  {author} {\bibinfo {author} {\bibfnamefont {B.~J.}\ \bibnamefont
  {Brown}}, \bibinfo {author} {\bibfnamefont {S.~D.}\ \bibnamefont {Bartlett}},
  \bibinfo {author} {\bibfnamefont {A.~C.}\ \bibnamefont {Doherty}}, \ and\
  \bibinfo {author} {\bibfnamefont {S.~D.}\ \bibnamefont {Barrett}},\
  }\bibfield  {title} {\enquote {\bibinfo {title} {Topological entanglement
  entropy with a twist},}\ }\href {\doibase 10.1103/PhysRevLett.111.220402}
  {\bibfield  {journal} {\bibinfo  {journal} {Phys. Rev. Lett.}\ }\textbf
  {\bibinfo {volume} {111}},\ \bibinfo {pages} {220402} (\bibinfo {year}
  {2013})},\ \Eprint {http://arxiv.org/abs/1303.4455} {arXiv:1303.4455
  [quant-ph]} \BibitemShut {NoStop}%
\bibitem [{\citenamefont {Shi}\ and\ \citenamefont
  {Kim}(2021)}]{Shi2021domain}%
  \BibitemOpen
  \bibfield  {author} {\bibinfo {author} {\bibfnamefont {B.}~\bibnamefont
  {Shi}}\ and\ \bibinfo {author} {\bibfnamefont {I.~H.}\ \bibnamefont {Kim}},\
  }\bibfield  {title} {\enquote {\bibinfo {title} {Domain wall topological
  entanglement entropy},}\ }\href {\doibase 10.1103/PhysRevLett.126.141602}
  {\bibfield  {journal} {\bibinfo  {journal} {Phys. Rev. Lett.}\ }\textbf
  {\bibinfo {volume} {126}},\ \bibinfo {pages} {141602} (\bibinfo {year}
  {2021})},\ \Eprint {http://arxiv.org/abs/2008.11794} {arXiv:2008.11794
  [cond-mat.str-el]} \BibitemShut {NoStop}%
\bibitem [{\citenamefont {Barkeshli}\ \emph {et~al.}(2019)\citenamefont
  {Barkeshli}, \citenamefont {Bonderson}, \citenamefont {Cheng},\ and\
  \citenamefont {Wang}}]{Barkeshli2019symmetry}%
  \BibitemOpen
  \bibfield  {author} {\bibinfo {author} {\bibfnamefont {M.}~\bibnamefont
  {Barkeshli}}, \bibinfo {author} {\bibfnamefont {P.}~\bibnamefont
  {Bonderson}}, \bibinfo {author} {\bibfnamefont {M.}~\bibnamefont {Cheng}}, \
  and\ \bibinfo {author} {\bibfnamefont {Z.}~\bibnamefont {Wang}},\ }\bibfield
  {title} {\enquote {\bibinfo {title} {Symmetry fractionalization, defects, and
  gauging of topological phases},}\ }\href {\doibase
  10.1103/PhysRevB.100.115147} {\bibfield  {journal} {\bibinfo  {journal}
  {Phys. Rev. B}\ }\textbf {\bibinfo {volume} {100}},\ \bibinfo {pages}
  {115147} (\bibinfo {year} {2019})},\ \Eprint {http://arxiv.org/abs/1410.4540}
  {arXiv:1410.4540 [cond-mat.str-el]} \BibitemShut {NoStop}%
\bibitem [{\citenamefont {Barkeshli}\ and\ \citenamefont
  {Cheng}(2020)}]{barkeshli2020relative}%
  \BibitemOpen
  \bibfield  {author} {\bibinfo {author} {\bibfnamefont {M.}~\bibnamefont
  {Barkeshli}}\ and\ \bibinfo {author} {\bibfnamefont {M.}~\bibnamefont
  {Cheng}},\ }\bibfield  {title} {\enquote {\bibinfo {title} {Relative
  anomalies in (2+1)d symmetry enriched topological states},}\ }\href
  {https://scipost.org/SciPostPhys.8.2.028} {\bibfield  {journal} {\bibinfo
  {journal} {SciPost Physics}\ }\textbf {\bibinfo {volume} {8}},\ \bibinfo
  {pages} {028} (\bibinfo {year} {2020})},\ \Eprint
  {http://arxiv.org/abs/1906.10691} {arXiv:1906.10691 [cond-mat.str-el]}
  \BibitemShut {NoStop}%
\bibitem [{\citenamefont {Barkeshli}\ \emph {et~al.}(2020)\citenamefont
  {Barkeshli}, \citenamefont {Bonderson}, \citenamefont {Cheng}, \citenamefont
  {Jian},\ and\ \citenamefont {Walker}}]{barkeshli2020reflection}%
  \BibitemOpen
  \bibfield  {author} {\bibinfo {author} {\bibfnamefont {M.}~\bibnamefont
  {Barkeshli}}, \bibinfo {author} {\bibfnamefont {P.}~\bibnamefont
  {Bonderson}}, \bibinfo {author} {\bibfnamefont {M.}~\bibnamefont {Cheng}},
  \bibinfo {author} {\bibfnamefont {C.-M.}\ \bibnamefont {Jian}}, \ and\
  \bibinfo {author} {\bibfnamefont {K.}~\bibnamefont {Walker}},\ }\bibfield
  {title} {\enquote {\bibinfo {title} {Reflection and time reversal symmetry
  enriched topological phases of matter: path integrals, non-orientable
  manifolds, and anomalies},}\ }\href
  {https://link.springer.com/article/10.1007/s00220-019-03475-8} {\bibfield
  {journal} {\bibinfo  {journal} {Communications in Mathematical Physics}\
  }\textbf {\bibinfo {volume} {374}},\ \bibinfo {pages} {1021} (\bibinfo {year}
  {2020})},\ \Eprint {http://arxiv.org/abs/1612.07792} {arXiv:1612.07792
  [cond-mat.str-el]} \BibitemShut {NoStop}%
\bibitem [{\citenamefont {Williamson}\ \emph {et~al.}(2017)\citenamefont
  {Williamson}, \citenamefont {Bultinck},\ and\ \citenamefont
  {Verstraete}}]{williamson2017symmetryenriched}%
  \BibitemOpen
  \bibfield  {author} {\bibinfo {author} {\bibfnamefont {D.~J.}\ \bibnamefont
  {Williamson}}, \bibinfo {author} {\bibfnamefont {N.}~\bibnamefont
  {Bultinck}}, \ and\ \bibinfo {author} {\bibfnamefont {F.}~\bibnamefont
  {Verstraete}},\ }\href@noop {} {\enquote {\bibinfo {title} {Symmetry-enriched
  topological order in tensor networks: Defects, gauging and anyon
  condensation},}\ } (\bibinfo {year} {2017}),\ \Eprint
  {http://arxiv.org/abs/1711.07982} {arXiv:1711.07982 [quant-ph]} \BibitemShut
  {NoStop}%
\bibitem [{\citenamefont {Wang}\ and\ \citenamefont
  {Cheng}(2021)}]{wang2021exactly}%
  \BibitemOpen
  \bibfield  {author} {\bibinfo {author} {\bibfnamefont {Q.-R.}\ \bibnamefont
  {Wang}}\ and\ \bibinfo {author} {\bibfnamefont {M.}~\bibnamefont {Cheng}},\
  }\href@noop {} {\enquote {\bibinfo {title} {Exactly solvable models for
  {U}(1) symmetry-enriched topological phases},}\ } (\bibinfo {year} {2021}),\
  \Eprint {http://arxiv.org/abs/2103.13399} {arXiv:2103.13399
  [cond-mat.str-el]} \BibitemShut {NoStop}%
\bibitem [{\citenamefont {Heinrich}\ \emph {et~al.}(2016)\citenamefont
  {Heinrich}, \citenamefont {Burnell}, \citenamefont {Fidkowski},\ and\
  \citenamefont {Levin}}]{Heinrich2016symmetry}%
  \BibitemOpen
  \bibfield  {author} {\bibinfo {author} {\bibfnamefont {C.}~\bibnamefont
  {Heinrich}}, \bibinfo {author} {\bibfnamefont {F.}~\bibnamefont {Burnell}},
  \bibinfo {author} {\bibfnamefont {L.}~\bibnamefont {Fidkowski}}, \ and\
  \bibinfo {author} {\bibfnamefont {M.}~\bibnamefont {Levin}},\ }\bibfield
  {title} {\enquote {\bibinfo {title} {Symmetry-enriched string nets: Exactly
  solvable models for set phases},}\ }\href {\doibase
  10.1103/PhysRevB.94.235136} {\bibfield  {journal} {\bibinfo  {journal} {Phys.
  Rev. B}\ }\textbf {\bibinfo {volume} {94}},\ \bibinfo {pages} {235136}
  (\bibinfo {year} {2016})},\ \Eprint {http://arxiv.org/abs/1606.07816}
  {arXiv:1606.07816 [cond-mat.str-el]} \BibitemShut {NoStop}%
\bibitem [{\citenamefont {Hamma}\ \emph {et~al.}(2005)\citenamefont {Hamma},
  \citenamefont {Zanardi},\ and\ \citenamefont {Wen}}]{Hamma2005string}%
  \BibitemOpen
  \bibfield  {author} {\bibinfo {author} {\bibfnamefont {A.}~\bibnamefont
  {Hamma}}, \bibinfo {author} {\bibfnamefont {P.}~\bibnamefont {Zanardi}}, \
  and\ \bibinfo {author} {\bibfnamefont {X.-G.}\ \bibnamefont {Wen}},\
  }\bibfield  {title} {\enquote {\bibinfo {title} {String and membrane
  condensation on three-dimensional lattices},}\ }\href {\doibase
  10.1103/PhysRevB.72.035307} {\bibfield  {journal} {\bibinfo  {journal} {Phys.
  Rev. B}\ }\textbf {\bibinfo {volume} {72}},\ \bibinfo {pages} {035307}
  (\bibinfo {year} {2005})},\ \Eprint {http://arxiv.org/abs/cond-mat/0411752}
  {arXiv:cond-mat/0411752 [cond-mat.str-el]} \BibitemShut {NoStop}%
\bibitem [{\citenamefont {Kong}\ \emph {et~al.}(2020)\citenamefont {Kong},
  \citenamefont {Tian},\ and\ \citenamefont {Zhang}}]{kong2020defects}%
  \BibitemOpen
  \bibfield  {author} {\bibinfo {author} {\bibfnamefont {L.}~\bibnamefont
  {Kong}}, \bibinfo {author} {\bibfnamefont {Y.}~\bibnamefont {Tian}}, \ and\
  \bibinfo {author} {\bibfnamefont {Z.-H.}\ \bibnamefont {Zhang}},\ }\bibfield
  {title} {\enquote {\bibinfo {title} {Defects in the 3-dimensional toric code
  model form a braided fusion 2-category},}\ }\href
  {https://link.springer.com/article/10.1007/JHEP12(2020)078} {\bibfield
  {journal} {\bibinfo  {journal} {Journal of High Energy Physics}\ }\textbf
  {\bibinfo {volume} {2020}},\ \bibinfo {pages} {1} (\bibinfo {year} {2020})},\
  \Eprint {http://arxiv.org/abs/2009.06564} {arXiv:2009.06564
  [cond-mat.str-el]} \BibitemShut {NoStop}%
\bibitem [{\citenamefont {Delcamp}\ and\ \citenamefont
  {Schuch}(2021)}]{delcamp2021tensornet}%
  \BibitemOpen
  \bibfield  {author} {\bibinfo {author} {\bibfnamefont {C.}~\bibnamefont
  {Delcamp}}\ and\ \bibinfo {author} {\bibfnamefont {N.}~\bibnamefont
  {Schuch}},\ }\bibfield  {title} {\enquote {\bibinfo {title} {On tensor
  network representations of the $(3+1)d$ toric code},}\ }\href
  {https://quantum-journal.org/papers/q-2021-12-16-604/} {\bibfield  {journal}
  {\bibinfo  {journal} {Quantum}\ }\textbf {\bibinfo {volume} {5}},\ \bibinfo
  {pages} {604} (\bibinfo {year} {2021})},\ \Eprint
  {http://arxiv.org/abs/2012.15631} {arXiv:2012.15631 [cond-mat.str-el]}
  \BibitemShut {NoStop}%
\bibitem [{\citenamefont {Drinfel'd}(1988)}]{drinfel1988quantum}%
  \BibitemOpen
  \bibfield  {author} {\bibinfo {author} {\bibfnamefont {V.~G.}\ \bibnamefont
  {Drinfel'd}},\ }\bibfield  {title} {\enquote {\bibinfo {title} {Quantum
  groups},}\ }\href {https://link.springer.com/article/10.1007/BF01247086}
  {\bibfield  {journal} {\bibinfo  {journal} {Journal of Soviet mathematics}\
  }\textbf {\bibinfo {volume} {41}},\ \bibinfo {pages} {898} (\bibinfo {year}
  {1988})}\BibitemShut {NoStop}%
\bibitem [{\citenamefont {Kassel}(1995)}]{kassel2012quantum}%
  \BibitemOpen
  \bibfield  {author} {\bibinfo {author} {\bibfnamefont {C.}~\bibnamefont
  {Kassel}},\ }\href {https://link.springer.com/book/10.1007/978-1-4612-0783-2}
  {\emph {\bibinfo {title} {Quantum groups}}},\ \bibinfo {series} {Graduate
  Texts in Mathematics}, Vol.\ \bibinfo {volume} {155}\ (\bibinfo  {publisher}
  {Springer-Verlag, New York},\ \bibinfo {year} {1995})\ pp.\ \bibinfo {pages}
  {xii+531}\BibitemShut {NoStop}%
\bibitem [{\citenamefont {Abe}(2004)}]{abe2004hopf}%
  \BibitemOpen
  \bibfield  {author} {\bibinfo {author} {\bibfnamefont {E.}~\bibnamefont
  {Abe}},\ }\href@noop {} {\emph {\bibinfo {title} {Hopf algebras}}},\ \bibinfo
  {series} {Cambridge Tracts in Mathematics}, Vol.~\bibinfo {volume} {74}\
  (\bibinfo  {publisher} {Cambridge University Press},\ \bibinfo {year}
  {2004})\ pp.\ \bibinfo {pages} {xii+284}\BibitemShut {NoStop}%
\bibitem [{\citenamefont {Manolescu}(2016)}]{manolescu2016lecture}%
  \BibitemOpen
  \bibfield  {author} {\bibinfo {author} {\bibfnamefont {C.}~\bibnamefont
  {Manolescu}},\ }\bibfield  {title} {\enquote {\bibinfo {title} {Lectures on
  the triangulation conjecture},}\ }\href@noop {} {\  (\bibinfo {year}
  {2016})},\ \Eprint {http://arxiv.org/abs/1607.08163} {arXiv:1607.08163
  [math.GT]} \BibitemShut {NoStop}%
\end{thebibliography}%

\end{document}